\documentclass[11pt,letterpaper,reqno]{amsart}
\usepackage[foot]{amsaddr}
\usepackage{ifxetex}
\ifxetex
  \usepackage[no-math]{fontspec}
\else
\fi
\usepackage{tikz}
\usetikzlibrary{arrows.meta, positioning, shapes, calc}
\usepackage{amsmath}
\usepackage{amsfonts}
\usepackage{amssymb}
\usepackage{amsthm}

\usepackage{fullpage}
\usepackage{microtype}

\usepackage{ifthen}

\def\fontsettingup{1} % 1: classic?; 2: Mathpazo

\ifthenelse{\equal{\fontsettingup}{1}}{%
  \ifxetex 
    \usepackage[libertine]{newtxmath}
  \else
    \usepackage{newtxmath}
  \fi
  \usepackage[tt=false]{libertine} 
}{%
  \usepackage[T1]{fontenc}
  \usepackage{ifthen}
  \usepackage{mathpazo}
  \usepackage{tgpagella}
}    

\usepackage{caption}
\usepackage{tcolorbox}
\usepackage{xspace}
\usepackage{bbm}
\usepackage[linesnumbered,boxed,ruled,vlined]{algorithm2e}
\usepackage{bm}
\usepackage{bbm}
\usepackage[numbers]{natbib}
\usepackage{xcolor}
\usepackage{enumerate} 
\usepackage{enumitem}
\usepackage{ragged2e}
\usepackage{tabularx}
\usepackage{array}
\usepackage{longtable}
\usepackage{hhline}
\usepackage{multirow}
\usepackage{xcolor}
\usepackage[hypertexnames=false]{hyperref}
\hypersetup{colorlinks=true,citecolor=blue, linkcolor=blue, urlcolor=blue}
\usepackage{cleveref}
\usepackage{thmtools}

\newcolumntype{L}[1]{>{\raggedright\arraybackslash}p{#1}}
\newcolumntype{C}[1]{>{\centering\arraybackslash}m{#1}}
\newcolumntype{R}[1]{>{\raggedleft\arraybackslash}p{#1}}

\usepackage{makecell}

\usepackage{footnote}
\usepackage{float}
\makesavenoteenv{tabular}

\renewcommand{\epsilon}{\varepsilon}

\declaretheorem[name=Theorem,numberwithin=section]{theorem}
\declaretheorem[sibling=theorem,name=Lemma]{lemma}

\declaretheorem[sibling=theorem,name=Corollary]{corollary}
\declaretheorem[sibling=theorem,name=Observation]{observation}

\declaretheorem[sibling=theorem,name=Condition]{condition}

\declaretheorem[sibling=theorem,name=Fact]{fact}
\declaretheorem[sibling=theorem,name=Property]{property}
\declaretheorem[name=Claim,numbered=no]{claim*}
\declaretheorem[name=Remark,style=remark,numbered=no]{remark*}
\declaretheorem[sibling=theorem,name=Definition,style=definition]{definition}
\declaretheorem[sibling=theorem,name=Remark,style=remark]{remark}

\renewcommand{\Pr}[2][]{ \ifthenelse{\isempty{#1}}
  {{\textnormal{Pr}}\left[#2\right]} {{\textnormal{Pr}}_{#1}\left[#2\right]} }
\newcommand{\Var}[2][]{ \ifthenelse{\isempty{#1}}
  {{\textnormal{Var}}\left[#2\right]} {{\textnormal{Var}}_{#1}\left[#2\right]} }
\newcommand{\Cov}[2][]{ \ifthenelse{\isempty{#1}}
  {{\textnormal{Cov}}\left[#2\right]} {{\textnormal{Cov}}_{#1}\left[#2\right]} }
\newcommand{\E}[2][]{ \ifthenelse{\isempty{#1}}
  {{\mathbb{E}}\left[#2\right]} {{\mathbb{E}}_{#1}\left[#2\right]} }

\newcommand{\norm}[1]{\left\Vert#1\right\Vert}
\newcommand{\abs}[1]{\left\vert#1\right\vert}
\newcommand{\set}[1]{\left\{#1\right\}}
\newcommand{\tuple}[1]{\left(#1\right)} 
\newcommand{\inbr}[1]{\left[#1\right]}

\newcommand{\tp}{\tuple}

\newcommand{\defeq}{\triangleq}
\newcommand{\supp}{\mathrm{supp}}

\renewcommand{\d}{\,\-d}

\newcommand{\dmu}[1]{\mu^{\-{GD}}_{T, #1, \*b}}
\def\*#1{\bm{#1}} 
\newcommand{\dpsi}[1]{\psi^{\-{GD}}_{T, #1, \*b}}
\def\*#1{\bm{#1}} 
\def\+#1{\mathcal{#1}} 
\def\-#1{\mathrm{#1}} 
\def\=#1{\mathbb{#1}} 
\def\`#1{\mathfrak{#1}}

\usepackage[textsize=tiny]{todonotes}
\newcommand{\yxtodo}[1]{\todo[color=gray!30]{yyx: #1}}
\usepackage{xifthen}

\makeatletter
\def\prob#1#2#3{\goodbreak\begin{list}{}{\labelwidth\z@ \itemindent-\leftmargin
                        \itemsep\z@  \topsep6\p@\@plus6\p@
                        \let\makelabel\descriptionlabel}
                \item[\it Name]#1
               \item[\it Instance]                #2
                \item[\it Output]#3
                \end{list}}
\makeatother

\newcommand{\true}{\textnormal{\textsf{True}}}
\newcommand{\false}{\textnormal{\textsf{False}}}
\newcommand{\upd}{\textnormal{\textsf{pred}}}
\newcommand{\ts}{\textnormal{\textsf{TS}}}
\newcommand{\vbl}{\mathsf{vbl}}
\newcommand{\vbad}{V^{\-{bad}}}
\newcommand{\cbad}{\+C^{\-{bad}}}
\newcommand{\tbad}{\+T^{\-{bad}}}
\newcommand{\ebad}{E^{\-{bad}}}
\newcommand{\sbad}{\+S^{\-{bad}}}
\usepackage{enumitem}

\newcommand{\dist}{\mathrm{dist}}

\newcommand{\Erdos}{Erd\H{o}s\xspace}

\newcommand{\Lovasz}{Lov\'asz\xspace}

\newcommand{\BadTS}[1]{\textnormal{\textsf{BadTS}}\tuple{#1}}

\newcommand{\decompositionScheme}{$\*b$}
\newcommand{\psiAdaptive}[2]{\psi_{#1}^{#2, \bot}}
\newcommand{\conditionalDecayConstant}{\widehat{N}}
\newcommand{\conditionalTriangleConstant}{\widehat{M}}
\newcommand{\kcnfPFO}{0^\blacktriangle}
\newcommand{\kcnfPFI}{1^\blacktriangle}
\newcommand{\kcnfPF}{\blacktriangle}
\newcommand{\coF}{\blacktriangle}
\newcommand{\coS}{\star}
\newcommand{\coRHO}{608 q \Delta^2 k^5}
\newcommand{\coGamma}{\frac{1}{16 \Delta^2 k^5}}
\newcommand{\complexColorSet}[1]{f^{-1}_{#1}(1^{\coF})}

\newcommand{\cnfGamma}{\frac{1}{2000\Delta^2 k^5}}
\newcommand{\cnfS}{2000\Delta^2 k^5}

\newcommand{\fsRHO}{40 \Delta^2 (k+1)^5}
\newcommand{\fsGamma}{\frac{1}{32\Delta^2 (k+1)^5}}
\newcommand{\dependencyGraph}[1]{G^{\textnormal{dep}}_{#1}}
\newcommand{\witnessGraph}{G_{\Phi}^{c^*}}
\newcommand{\witnessTreeRoot}{(\ts(\vbl(c^*),0), c^*)}
\newcommand{\markedK}{k_{\textnormal{mk}}}
\newcommand{\unmarkedK}{k_{\textnormal{umk}}}

\newcommand{\coloringOneRandomVariable}{\+X_{\textnormal{co}}}

\newcommand{\coloringCondition}{q\ge 700\Delta^{\frac{5}{k-10}}}
\newcommand{\kcnfCondition}{k\ge 12\log(\Delta) + 24\log(k) + 57}
\newcommand{\connectedComponentSize}{(\Delta k + 1)\log (4\Delta |V|^2)}
\newcommand{\connectedComponentVarSize}{k(\Delta k + 1)\log (4\Delta |V|^2)}
\newcommand{\lcltConcentrationProb}{C_{\text{\ref{lemma:lclt-J-is-good}}}}
\newcommand{\lcltExpectation}{C_{\text{\ref{lemma:lclt-concentration-expectation-lower-bound}}}}
\newcommand{\lcltHighPhaseConstant}{C_{\text{\ref{lemma:lclt-high-fourier-phases}}}}
\newcommand{\lcltConditionVioProb}{p_{\-{cond}}}
\newcommand{\lcltProjectedVioProb}{p_{\-{proj}}}
\newcommand{\lcltCltTail}{\xi(q, q^*, \lambda, \Delta, k, N^*)}
\usepackage{graphicx}

\title{Zero-free regions and concentration inequalities for hypergraph colorings in the local lemma regime}
\newbool{doubleblind}
\setbool{doubleblind}{true}
\author{Jingcheng Liu, Yixiao Yu}
\address{State Key Laboratory for Novel Software Technology, New Cornerstone Science Laboratory, Nanjing University, 163 Xianlin Avenue, Nanjing, Jiangsu Province, 210023, China. \textnormal{E-mail: \texttt{liu@nju.edu.cn}, \texttt{yixiaoyu@smail.nju.edu.cn}}. JL is supported by the National Natural Science Foundation of China under Grant No. 62472212.}
\begin{document}

\begin{abstract}
We show that for $q$-colorings in $k$-uniform hypergraphs with maximum degree $\Delta$, if $k\ge 50$ and $\coloringCondition$, there is a ``Lee-Yang'' zero-free strip around the interval $[0,1]$ of the partition function, which includes the special case of uniform enumeration of hypergraph colorings.
As an immediate consequence, we obtain Berry-Esseen type inequalities for hypergraph $q$-colorings under such conditions, demonstrating the asymptotic normality for the size of any color class in a uniformly random coloring. Our framework also extends to the study of ``Fisher zeros'', leading to deterministic algorithms for approximating the partition function in the zero-free region.

Our approach is based on extending the recent work of [Liu, Wang, Yin, Yu, STOC 2025] to general constraint satisfaction problems (CSP).
We focus on partition functions defined for CSPs by introducing external fields to the variables. A key component in our approach is a projection-lifting scheme, which enables us to essentially lift information percolation type analysis for Markov chains from the real line to the complex plane.
 Last but not least, we also show a Chebyshev-type inequality under the sampling LLL condition for atomic CSPs.

\end{abstract}
\maketitle
\tableofcontents
\setcounter{page}{0}
\clearpage

\section{Introduction}
%{
%\color{blue}
%I plan to state the zero-freeness, then the information percolation.
%Focusing on ``extending more general information percolation arguments into complex settings.''
%}
The location of complex zeros of partition functions is intimately connected to the study of phase transitions in statistical physics, dating back to the famous Lee-Yang program~\cite{Lee1952StatisticalTO}. Since then, it has also found applications in combinatorics~\cite{Heilmann1972TheoryOM,wagner2009weighted}, Chernoff bounds~\cite{kyng2018matrix}, asymptotic normality~\cite{godsil1981matching}, central limit theorems~\cite{lebowitz2016central,michelen2024central,jain2022approximate}, and in designing efficient approximate counting algorithms~\cite{barvinok2016combinatorics,patel2017deterministic,Liu2017TheIP}. 
Establishing a zero-free region has therefore attracted a lot of attention, and tight zero-free results around the real axis have been obtained for numerous graph polynomials, notably the independent set polynomial~\cite{Scott2003TheRL,peter2019conjecture} and monomer-dimer model~\cite{Heilmann1972TheoryOM}. The case of hypergraphs is much less understood, especially for spin systems on $k$-uniform hypergraphs.

Hypergraph $q$-colorings, introduced by \Erdos{}~\cite{erdos1963combinatorial}, have long been one of the most fruitful testbeds in probabilistic combinatorics. To name a few, they featured in the \Lovasz local lemma (LLL)~\cite{erdHos1975problems}, the nibble technique~\cite{rodl1985packing}, the container method~\cite{saxton2015hypergraph,balogh2018method}, and frequented the textbook by Alon and Spencer~\cite{alon2016probabilistic}. 
A standard application of LLL shows that for $k$-uniform  hypergraphs with maximum degree $\Delta$, there is a proper $q$-coloring as soon as $q> C \Delta^{\frac{1}{k-1}}$ for some constant $C$, and such coloring can be found by efficient algorithms~\cite{beck1991algorithmic,moser2009constructive,moser2010constructive}. For simple hypergraphs (two hyperedges can share at most one vertex), Frieze et al.~\cite{frieze2011randomly,frieze2017randomly} showed that standard Glauber dynamics mixes rapidly provided that $q\ge \max\set{C_k \ln n, 500k^3 \Delta^{\frac{1}{k-1}}}$. For non-simple hypergraphs, however, the space of $q$-colorings can become disconnected~\cite{frieze2011randomly}, causing Glauber dynamics to fail. Such a connectivity barrier is not unique to hypergraph coloring but appears broadly in CSPs under local lemma like conditions.

More recently, there is a beautiful line of work studying a counting/sampling variant of LLL for CSPs.
Formally, one considers a uniformly random assignment to the CSP: let $p$ be the \emph{maximum violation probability} over all constraints, and $D$ be the maximum degree of the dependency graph (the number of constraints that share variables with a given constraint).
While the classical LLL~\cite{erdHos1975problems,She85} says that the satisfiability threshold is asymptotically $\mathrm{e} p D \le 1$,  numerous exciting recent developments suggest that $p D^c \lesssim 1$ for some constant $c\ge 2$ would be sufficient so that uniform sampling of solutions under LLL can be done efficiently.
Examples such as hypergraph independent sets~\cite{hermon2019rapid,qiu2022perfect},  $k$-CNF~\cite{moitra2019approximate,feng2021fast}, hypergraph $q$-colorings~\cite{guo2018counting}, random CNF formulas~\cite{galanis2022fast,he2023improved}, as well as general CSP instances~\cite{feng2021sampling,jain2021sampling,Vishesh21towards,he2021perfect,feng2023towards}, have been considered.
Two notable cases are $k$-CNF and hypergraph $q$-colorings, representative for boolean CSP and for CSPs with sufficiently large domains, respectively.
A recent breakthrough result by Wang and Yin~\cite{wang2024sampling} established the sampling LLL under the condition $p D^{2+o_q(1)} \lesssim 1$ for domain size $q$. This is asymptotically optimal for large enough $q$, as NP-hardness has been shown by~\cite{bezakova2019approximation,galanis2023inapproximability} for $pD^2 \gtrsim 1 $.

In this paper, we study ``Lee-Yang zeros'' of the hypergraph $q$-coloring partition functions in the sampling LLL regime.
An interesting prospect from establishing zero-free regions is the concentration phenomenon and central limit theorems. Specifically, in the case of hypergraph $q$-colorings, we are interested in the size of any color class under a uniformly random coloring.
To some extent, one could view the LLL conditions as a relaxation of the independence assumption, but does the bounded dependence structure of LLL give rise to concentration? 
Near the existence threshold of $\mathrm{e}pD \le 1$, concentration would be rare to find, because many random CSPs exhibit a rigid phenomenon with high probability near the satisfiability threshold, where nearly all the variables are frozen (see, e.g.,~\cite{achlioptas2008algorithmic}). We show that under the stronger condition of $p D^{2+o_q(1)} \lesssim 1/q$, a Chebyshev type inequality can still be derived for atomic CSPs: in a uniformly random satisfying assignment, the size of any color class lies within $1\pm O(\frac1n)$ times its expectation. %by combining local uniformity with the recursive coupling of~\cite{wang2024sampling}. 
This leads to a natural question, can one also prove central limit theorems (CLT) and normal approximations  in similar sampling LLL regimes? CLTs under various local dependence structures have also been studied before. Indeed, if one is only concerned about the number of violated constraints in a random assignment, then Stein's method would already suffice (e.g.,~\cite{chen2004normal,chatterjee2008new}). To obtain a CLT for the size of color classes, we follow a different route through complex zero-freeness~\cite{lebowitz2016central,michelen2024central,jain2022approximate,liu2024phase}.
Specifically, we introduce ``external fields'' to the variables of a CSP to form a partition function and show that in a complex-plane neighborhood of the point $1$, the partition function does not vanish. 
%{\color{red} Specifically, we introduce ``external fields'' to the variables of a CSP to form a partition function and show that in a complex-plane neighborhood of the interval $[0,1]$, the partition function does not vanish.}
Then, non-asymptotic central limit theorems (CLT) and local central limit theorems can be derived analytically  in a neighborhood of $1$.
CLT and local CLT are desirable in many applications as they combine both concentration aspects (the random variable does not deviate too much from the mean) and anti-concentration aspects (the random variable is non-degenerate and can still fluctuate like a Gaussian).
%We are able to show a zero-free region around the point $1$ for hypergraph $q$-colorings, l
%Unfortunately, a zero-free region around the point $1$ appears to be too much to ask for in general CSP under LLL-like conditions.
%Instead, we show that there is always a large sub-instance of the CSP with a zero-free region around the point $1$.
We also mention in passing that while concentration bounds for the number of resampling steps in the Moser-Tardos algorithm have also been established~\cite{kolipaka2011moser,achlioptas2016random,haeupler2017parallel}, these do not seem to translate to the concentration of the LLL distribution itself in general, unless for extremal LLL instances.

Classical techniques for establishing zero-free regions include the Lee-Yang approach~\cite{Lee1952StatisticalTO} and its extensions such as \cite{Asano1970LeeYangTA,ruelle1971extension}, as well as the contraction method~\cite{peter2019conjecture,liu2019correlation,Shao2019ContractionAU} based on self-avoiding walk tree construction~\cite{godsil1981matchings,Scott2003TheRL} and the decay of correlation~\cite{weitz06counting,bandyopadhyay08counting}.
More recently, there have also been several works extending these approaches beyond pairwise interactions. 
Notably, \cite{barvinok2024zeros} studied \emph{Fisher} zeros under many-body \emph{soft} interactions, where one relaxes hard constraints and introduces a ``penalty factor'' for every violated constraint instead. Then, one studies complex zeros in terms of the ``penalty factor'' around the point $1$. Partition functions in terms of Fisher zeros are better suited for studying optimization landscapes in Max-CSPs as opposed to solutions in CSPs. This is mainly due to relaxing  with a ``penalty factor'' around the point $1$ (corresponding to ``almost no interactions''), and it is unclear how to extract useful information about the solution space of the original CSP (corresponding to ``maximum interactions'').
To encode the solutions of the original CSP, the standard approach is to consider the \emph{Lee-Yang} zeros, where  an external field is introduced to the variables. Then, a neighborhood around the point $1$ still corresponds to ``almost-uniform random distribution over CSP solutions''. Two notable works along the same line concern the hypergraph independent set models. \cite{galvin2024zeroes} studied zero-free disk mainly around the origin, while~\cite{liu2024phase} showed a zero-free strip around the real interval $[0,1]$ under the asymptotically optimal sampling LLL condition $pD^2 \lesssim 1$. Notably, \cite{liu2024phase} introduced the framework of complex extensions of the Markov chain for locating complex zeros of the hypergraph independent set partition function, which matches, up to lower order factors, where Glauber dynamics mixes rapidly~\cite{hermon2019rapid}. Lee-Yang zeros will also be the focus of our work, while our approach can be extended to the study of Fisher zeros around the entire interval of $[0,1]$.

Indeed, the framework of complex extensions of Markov chain~\cite{liu2024phase} is a promising starting point for studying Lee-Yang zeros of CSPs under LLL-like conditions. An obvious obstacle, however, is that CSPs under LLL-like conditions may have disconnected solution space, ruling out standard Glauber dynamics Markov chains (see, e.g., \cite[Section 5]{frieze2011randomly} for the case of hypergraph $q$-coloring). 
Existing Markov chain approaches crucially rely on Glauber dynamics on a projection of the original instance, in which every step of the Glauber dynamics is (approximately) simulated with a rejection sampling~\cite{feng2021sampling,jain2021sampling}. 
By applying the complex extension of Glauber dynamics on the projected CSP instance, one only gains control on the projected measure, while proving zero-freeness requires lifting it back to the original measure. %Specifically, one needs to bound the complex measure that a new constraint is violated inductively in the original CSP. 
For hypergraph $q$-colorings, we are able to do both projection and lifting without loss, and also carry out the analysis of the complex extension of Glauber dynamics on the projection, leading to a zero-free region around the point $1$ of any given color under LLL-like conditions.
For more general CSPs like $k$-CNF, we are currently able to show zero-freeness around the point $1$ only for a large sub-instance. We leave as an interesting open problem how to get zero-freeness on the original instance without requiring projection. 
%To do so, we first apply the complex extension of Glauber dynamics on the projected CSP instance, then we introduce a \emph{projection-lifting} scheme to lift back the original measure.

\subsection{Our results}
We formally state our results next. We focus on hypergraph $q$-colorings, while we also remark that our framework in~\Cref{section:CSP} is applicable to more general CSPs that are \emph{atomic} (that is, each constraint can be violated by at most one assignment to its variables). For general CSPs, we turn them into atomic CSPs by adding atomic constraints. 
% We formally state our results next. We focus on hypergraph $q$-colorings and $k$-CNF, while we also note that our framework in~\Cref{section:CSP} is applicable to more general CSPs that are \emph{atomic} (that is, each constraint can be violated by at most one assignment to its variables).

% {\color{red}
% Let $\Phi = (V, [q]^V, C)$ be an atomic CSP formula, and let $\Omega_\Phi$ be the set of all satisfying assignments of $\Phi$.
% By symmetry, we introduce a complex external field on the symbol $1$. Formally, the partition function is given by:
% \[
% Z_{\Phi}(\lambda)\defeq\sum_{\sigma\in[q]^V: \sigma \in \Omega_{\Phi}} \lambda^{ \abs{ \sigma^{-1}\left(
% 1\right)}} , \qquad \hbox{where }\sigma^{-1}\left(
% 1\right)\defeq\set{v\in V\mid \sigma(v) = 1}. 
% \]
% \begin{restatable}[Zero-freeness for atomic $(k, \Delta)$-CSP formulas]{theorem}{CSPOneSpecialZeroIntro}
% \label{theorem:CSP-one-special-zero-free-intro}
% Fix any integers $k\ge 50$ and $q\ge ?\Delta^{\frac{?}{k-?}}$.
% Let $\Phi = (V, [q]^V, C)$ be an atomic $(k, \Delta)$-CSP formula.
% % Let $H=(V, \+E)$ be a $k$-uniform hypergraph with maximum degree $\Delta$.
% Then $Z_G(\lambda)\neq 0$ for any $\lambda \in \+D(1, ?)$. 
% \end{restatable}
% }

Let $H=(V, \+E)$ be a $k$-uniform hypergraph with maximum degree $\Delta$.  A coloring on the vertices of $H$ is said to be \emph{proper} if every hyperedge is not monochromatic.
 To study Lee-Yang zeros, %in hypergraph $q$-colorings on $H$,
we introduce a complex external field on the color $1$. 
The choice of color $1$ is arbitrary due to symmetry.
Formally, the partition function is given by:
\[
Z^{\-{co}}_H(\lambda)\defeq\sum_{\sigma\in[q]^V: \sigma \text{ is a proper coloring in }H} \lambda^{ \abs{ \sigma^{-1}\left(
1\right)}} , \qquad \hbox{where }\sigma^{-1}\left(
1\right)\defeq\set{v\in V\mid \sigma(v) = 1}. 
\]
\begin{restatable}{theorem}{colorOneSpecialZeroIntro}
\label{theorem:coloring-one-special-zero-free-intro}
Fix any integers $k\ge 50$ and $\coloringCondition$.
Let $H=(V, \+E)$ be a $k$-uniform hypergraph with maximum degree $\Delta$.
Then $Z^{\-{co}}_H(\lambda)\neq 0$ for any $\lambda\in \=C$ satisfying $\exists \lambda_c\in[0,1]$, such that $|\lambda - \lambda_c|\le \coGamma$. %in the disk centered at $1$ with radius $\coGamma$, i.e., $\lambda \in \+D(1, \coGamma)$. 
\end{restatable}
% \colorOneSpecialZeroIntro*
% \begin{theorem}[Zero-freeness for hypergraph $q$-coloring with one special color]

% \end{theorem}

% {
%     \color{red}
Let $\mu$ be the uniform distribution over all hypergraph $q$-colorings, and let $\coloringOneRandomVariable$ be the random variable counting the number of vertices whose color is $1$ under $\mu$.
\begin{restatable}[name=CLT for a color class in hypergraph $q$-coloring]{theorem}{CLTColorOneIntro}
\label{theorem:coloring-one-special-clt-intro}
Fix any integers $k\ge 50$ and $\coloringCondition$. Let $\bar{\mu}$ be the expectation of $\coloringOneRandomVariable$, and $\bar{\sigma}$ be the standard deviation of $\coloringOneRandomVariable$.
% We define $\coloringOneRandomVariable^*\defeq (\coloringOneRandomVariable - \bar{\mu})\bar{\sigma}^{-1}$. 
Then, it holds that 
\[
\sup_{t\in \=R}|\Pr{(\coloringOneRandomVariable - \bar{\mu})\bar{\sigma}^{-1} \le t} - \Pr{\+Z \le t}| \le O_{q, \Delta, k}\tp{\frac{\ln |V|}{\sqrt{|V|}}},
\]
where $\+Z\sim N(0, 1)$ is a standard Gaussian random variable.
\end{restatable}
% }
%{
%    \color{red}
\begin{restatable}[name=Local CLT for a color class in hypergraph $q$-coloring]{theorem}{LCLTColorOneIntro}
\label{theorem:coloring-one-special-lclt-intro}
Fix any integers $k\ge 50$ and $\coloringCondition$. Let $\bar{\mu}$ be the expectation of $\coloringOneRandomVariable$, and $\bar{\sigma}$ be the standard deviation of $\coloringOneRandomVariable$. It holds that 
\[
\sup_{t\in \=Z} \abs{\Pr{\coloringOneRandomVariable = t} - \bar{\sigma}^{-1}\+N((t - \bar{\mu})/\bar{\sigma})} \le O_{q, \Delta, k}\tp{\frac{\ln^{7/2}|V|}{|V|}},
\]
where $\+N(x) = \mathrm{e}^{-x^2/2}/\sqrt{2\pi}$ denote the density of the standard normal distribution.
\end{restatable}
We remark that both CLT and local CLT extend to a constant external field $\lambda$, provided that $q$ is  sufficiently large and satisfies $q\gtrsim \max\tp{\tp{\frac{1}{\lambda}}^{\frac{1+o_q(1)}{k-2}},\lambda} \Delta^{\frac{5}{k-10}}$. 
We also show a Chebyshev-type inequality that works more generally.
Specifically, we consider any \emph{atomic} $(k, \Delta)$-CSP formula $\Phi=(V,[q]^V, C)$, where each constraint involves $k$ variables, each variable appears in at most $\Delta$ constraints, and each constraint forbids at most one of its assignments (atomic). 
Let $\mu$ be the uniform distribution over all satisfying assignments of $\Phi$.
We will be interested in how many variables get assigned value $1$. The choice of value $1$ is arbitrary and works with any other value.
Let $\sigma\sim \mu$.
For each variable $v\in V$, we write $\sigma_v$ for its assignment under $\sigma$, and we define the random variable $X_v \defeq \=I[\sigma_v = 1]$, $\+X \defeq \sum_{v\in V}X_v$.
%Our next concentration inequality is for $\+X$.
\begin{theorem}[Chebyshev type inequality for atomic CSP formulas]
\label{theorem:concentration-atomic-csp}
Let $\Phi=(V,[q]^V, C)$ be an atomic $(k, \Delta)$-CSP formula. Suppose $(8\mathrm{e})^3\cdot \frac{1}{q^{k-1}}\cdot (\Delta k+ 1)^{2 + \zeta} \le 1$, where $\zeta = \frac{2\ln{(2 - 1/q)}}{\ln{(q)} - \ln{(2 - 1/q)}} = o_q(1)$. Then,
\[
\forall \delta > 0,\quad \Pr[\mu]{|\+X - \=E_{\mu}[\+X]|\ge \delta\cdot \=E_{\mu}[\+X]}\le \frac{1}{\delta^2}\cdot O_{q,\Delta, k}\tp{\frac{1}{|V|}}.
\]
% \[
% \forall \zeta > 0,\quad \-{Pr}[|\+X - \=E[\+X]|\ge \zeta\cdot \=E[\+X]]\le \frac{\Delta k(\Delta k + 1)}{\zeta^2 \cdot |V|\tp{\frac{1}{q}-\frac{1}{q\Delta k}}^2}.
% \]
\end{theorem}
It is worth noting that the above theorem essentially says that $p D^{2+o_q(1)} \lesssim 1/q$ is sufficient for a Chebyshev type inequality in atomic CSP.

%For hypergraph $q$-colorings, however, we remark that we can improve the Chebyshev inequality slightly in the small deviation regime, by taking advantage of the CLT (\Cref{theorem:coloring-one-special-clt-intro}). \yxtodo{should I include a theorem?}

Last but not least, our framework in \Cref{section:CSP} can also be used to study Fisher zeros (see \Cref{section:fisher-zeros} for a reduction to Lee-Yang zeros). 
While deterministic FPTASes for counting hypergraph $q$-colorings under local lemma conditions are known~\cite{guo2018counting,Vishesh21towards,wang2024sampling} through extending Moitra's LP approach~\cite{moitra2019approximate} and~\cite{feng2023towards} through the ``coupling towards the past'' (CTTP) framework, our zero-freeness result gives an alternative FPTAS through Barvinok's interpolation method~\cite{barvinok2016combinatorics,patel2017deterministic,Liu2017TheIP}, which works on the complex plane as well.
Currently, our FPTAS based on zero-freeness works only in the regime $q\gtrsim \Delta^{\frac{5}{k-2.5}}$. We leave it as an open problem to close the gap to $q\gtrsim \Delta^{\frac{2 + o_q(1)}{k-2}}$ achieved in~\cite{wang2024sampling}.

\subsection{Technical overview}
Our starting point for proving zero-freeness is the framework of~\cite{liu2024phase}. We explain our main techniques in the context of hypergraph $q$-coloring in the following. 

%Concretely, let $f:[q]\to[B]$
%be a projection that maps the $q$ colors into $B$ color buckets (see \Cref{definition:state-compression}).  We write a bucket index as a superscript $\blacktriangle$ (for example, $1^{\blacktriangle}$).  By color symmetry, we may assume that the first $s$ colors belong to bucket $1^{\blacktriangle}$, i.e.\ $f^{-1}(1^{\blacktriangle})=[s]$.

%Introduce a complex parameter $\lambda$ (the \emph{complex external field}) that acts only on colors in bucket $1^{\blacktriangle}$.  The partition function is
%\[
%Z_H(\lambda)\;=\;\sum_{\sigma\ \text{proper}}\lambda^{|\sigma^{-1}([s])|}.
%\]

\subsubsection{Zero-freeness via constraint-wise self-reduction}
A standard self-reduction can reduce proving zero-freeness of the whole instance into establishing upper bounds on certain complex marginal measures~\cite{Lee1952StatisticalTO,Asano1970LeeYangTA,ruelle1971extension,peter2019conjecture,liu2019correlation,Shao2019ContractionAU,liu2024phase}. Let $\+E=\{e_1,\dots,e_m\}$ and, for $0\le i\le m$, write $\+E_i=\{e_1,\dots,e_i\}$ and $H_i=(V,\+E_i)$.  To show $Z^{\-{co}}_H(\lambda)\neq 0$, it suffices to check $Z^{\-{co}}_{H_0}(\lambda)\neq 0$ (trivial since $H_0$ has no edges) and to show, for every $0\le i<m$, $\frac{Z^{\-{co}}_{H_{i+1}}(\lambda)}{Z^{\-{co}}_{H_i}(\lambda)}\neq 0$.
Assuming inductively that $Z^{\-{co}}_{H_i}(\lambda)\neq 0$, the ratio admits the interpretation that $
\frac{Z^{\-{co}}_{H_{i+1}}(\lambda)}{Z^{\-{co}}_{H_i}(\lambda)}=\mu_{H_i}\bigl(e_{i+1}\text{ is not monochromatic}\bigr)$,
where $\mu_{H_i}$ is the corresponding complex Gibbs measure and is well-defined provided that $Z^{\-{co}}_{H_i}(\lambda)\neq 0$.
Hence it suffices to prove that for each $i$,
\begin{equation}\label{eq:overview-induction}
\bigl|\mu_{H_i}\bigl(e_{i+1}\text{ is monochromatic}\bigr)\bigr|< 1,
\end{equation}
which completes the induction.

\subsubsection{Projection and lifting scheme}
Liu et~al.~\cite{liu2024phase} introduced complex extensions of Glauber dynamics as a proxy for similar marginals for hypergraph independent sets.  As alluded to in the introduction, standard Glauber dynamics is not ergodic for hypergraph colorings~\cite{frieze2011randomly,wigderson2019mathematics}, 
and one has to work with a \emph{state-compression scheme} instead~\cite{feng2021sampling,jain2021sampling,he2021perfect}.  
Specifically, a projection is designed to compress the domain alphabet into smaller ones so that: a) the Glauber dynamics on the projected alphabet is rapid mixing; b) the projected Glauber dynamics can be lifted back to the original instance through rejection sampling. Crucially, conditioned on the projected symbols, the CSP formula shatters into small connected components with high probability. This enables efficient rejection sampling.

To work with complex measures, we start by analyzing a \emph{projected} complex measure.  We consider the following projection $f$, which maps a coloring $\sigma\in[q]^V$ to a color bucket assignment $\tau=f(\sigma)\in\set{1^\coF, 2^\coF, \dots, B^\coF}^V$, where $\tau_v=f(\sigma_v)$ for any $v\in V$. 
The first color is projected to the first color bucket $1^\coF$ (we use $\coF$ to denote the color bucket). The remaining $q-1$ colors are evenly partitioned among the other $B-1$ buckets, so that each bucket (except $1^\coF$) contains approximately the same number of original colors.
Then, the projected measure on $\set{1^\coF, 2^\coF, \dots, B^\coF}^V$ is given by $\psi_{H_i}(\tau)=\sum_{\sigma:\,f(\sigma)=\tau}\mu_{H_i}(\sigma)$, and will be analyzed through complex Markov chains. But the trickier part is to extract bounds on  $\mu_{H_i}$ from $\psi_{H_i}$.

To do so, we introduce a \emph{projection--lifting} scheme (formalized in \Cref{subsection:projection-lifting-scheme}).  Concretely, we partition the projected space into disjoint events $
S_1,\dots,S_\ell\subseteq\set{1^\coF, 2^\coF, \dots, B^\coF}^V$, with $\bigcup_{j=1}^\ell S_j=\set{1^\coF, 2^\coF, \dots, B^\coF}^V$. Notice that by definition, $\psi_{H_i}(S_j)=\mu_{H_i}(S_j)$ for every $j$.   By the law of total measure,
\begin{equation}\label{eq:overview-projection-law-of-total-measure}
\mu_{H_i}\bigl(e_{i+1}\text{ is monochromatic}\bigr)
=\sum_{j=1}^\ell \psi_{H_i}(S_j)\cdot
\mu_{H_i}\bigl(e_{i+1}\text{ is monochromatic}\mid S_j\bigr).
\end{equation}
One of the main difficulties in generalizing probabilistic proofs to the complex plane is that, many probabilistic arguments rely on monotonicity\footnote{in the sense of $\abs{\mu(A)} \le \abs{\mu(B)}$ for every $A \subseteq B$.} in a crucial way, while complex measures do not enjoy monotonicity.   To regain control, we design the partitions $(S_j)$ so that, for every $j$ with $\psi_{H_i}(S_j)\neq 0$,
\[
\mu_{H_i}\bigl(e_{i+1}\text{ is monochromatic}\mid S_j\bigr) \in [0,1].
\]
Combining this with \eqref{eq:overview-projection-law-of-total-measure} and the triangle inequality gives
\begin{equation}\label{eq:overview-projection-lifting-scheme}
\bigl|\mu_{H_i}\bigl(e_{i+1}\text{ is monochromatic}\bigr)\bigr|
\le \sum_{j=1}^\ell \bigl|\psi_{H_i}(S_j)\bigr|.
\end{equation}
Thus the problem reduces to upper-bounding the total norm of a collection of projected events.

Next, we give some intuitions about how we define these events.
We define the events $S_j$ by a local combinatorial object called the \emph{bad cluster} (see \Cref{definition:csp-bad-things}). 
The bad cluster $\sbad$ is a subset of connected hyperedges containing $e_{i+1}$, and all vertices in $\sbad$ are in the same color bucket, while each hyperedge crossing $\sbad$ contains different color buckets (see \Cref{fig:bad-cluster}). 
The events $S_j$'s are defined by the bad cluster $\sbad$, together with the color bucket assignments within.
For a subset of hyperedges $S$, we use $\vbl(S)$ to denote the set of vertices in $S$.
Fix a subset of hyperedges $S$ and a color bucket assignment $\tau_{\vbl(S)}\in \set{1^\coF, 2^\coF, \dots, B^\coF}^{\vbl(S)}$, we will show that if $\psi_{H_i}(\sbad=S \land \tau_{\vbl(S)}) \neq 0$, then we have 
\[
\mu_{H_i}\bigl(e_{i+1}\text{ is monochromatic}\mid \sbad=S\land\tau_{\vbl(S)}\bigr) \in [0,1].
\]
% Let $G_{H_i}=(V_{H_i},E_{H_i})$ be the dependency graph whose vertices are $\+E_i\cup\{e_{i+1}\}$ and where two hyperedges are adjacent if they share at least one vertex.  For a projected configuration $\tau\in[B]^V$, the bad cluster $\sbad(\tau)$ is the maximal connected component of $G_{H_i}$ containing $e_{i+1}$ such that every edge in the component is monochromatic at the bucket level (i.e., all its vertices lie in the same bucket), while edges crossing the component are not monochromatic at the bucket level.  For a set $S\subseteq V_{H_i}$, we let $\vbl(S)=\bigcup_{e\in S} e$ denote the set of original vertices involved in $S$.
\begin{figure}[ht]
    \centering
    \includegraphics[trim={0 12.5cm 0 10.5cm},clip,width=.7\linewidth]{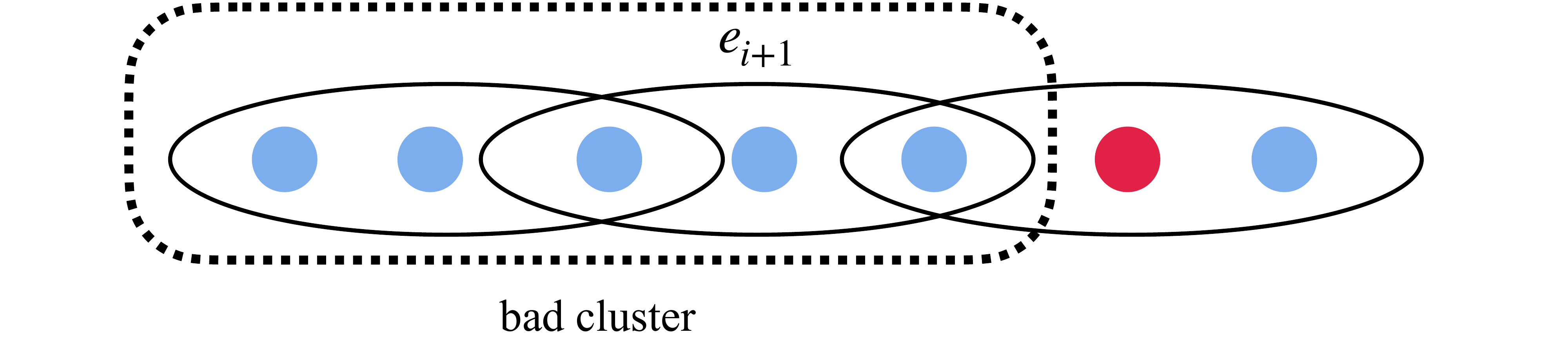}
    \caption{Each displayed color represents one color bucket in the state-compression scheme. The two leftmost hyperedges form a bad cluster since their vertices belong to the same color bucket.}
    \label{fig:bad-cluster}
\end{figure}
% The key property of our construction is that, for any connected $S\subseteq V_{H_i}$ and any partial projected configuration $\tau\in[B]^{\vbl(S)}$, whenever $\psi_{H_i}(\sbad=S\land\tau)\neq 0$, we have
% \[
% \bigl|\mu_{H_i}\bigl(e_{i+1}\text{ is monochromatic}\mid \sbad=S\land\tau\bigr)\bigr|\le 1.
% \]

Intuitively, conditioning on $\sbad=S$ and the color bucket assignments on $\vbl(S)$ fixes all contributions from the external field (because external fields are only on color buckets), so the conditional marginal of the event ``$e_{i+1}$ is monochromatic'' is real and lies in $[0,1]$ (see \Cref{lemma:csp-conditional-measure-probabilty}).  Consequently, we obtain 
\begin{equation}\label{eq:overview-projection-lifting-scheme-bad-cluster}
\bigl|\mu_{H_i}\bigl(e_{i+1}\text{ is monochromatic}\bigr)\bigr|
\le \sum_{S:|S|\ge 1}\ \sum_{\tau_{\vbl(S)}\in\set{1^\coF, \dots, B^\coF}^{\vbl(S)}}
\bigl|\psi_{H_i}\bigl(\sbad=S\land \tau_{\vbl(S)}\bigr)\bigr|.
\end{equation}

\subsubsection{Analyzing the projected measure via complex Glauber dynamics}
To bound the right-hand side of \eqref{eq:overview-projection-lifting-scheme-bad-cluster}, we extend the complex Markov chain framework developed in \cite{liu2024phase} to general CSPs. Intuitively,   to analyze a marginal measure $\psi$ of an event $A$, one analyzes the event $A$ with respect to the resulting state of a complex Glauber dynamics instead. 

Fix any integer $T > 0$, we consider a $T$-step complex Glauber dynamics with stationary measure $\psi_{H_i}$.
 Let $\mu: \Omega\to \=C$ be any \emph{complex measure} (formally defined in the Preliminary),
and $P\in \=C^{\Omega\times\Omega}$ be a \emph{complex-valued transition matrix}, the complex measure of the $T$-step complex Markov chain is given by $\mu P^T$, and will be denoted by $\psi^{\-{GD}}_T$. 
We think of time $0$ as the ``current time'', and the Markov chain is started at time $-T$.
Then, for any ``outcome'' $\sigma_0$ at time $0$, $\psi^{\-{GD}}_T (\sigma_0)$ is well-defined, and we will simply write $\psi^{\-{GD}}_T (\sigma_0 \in A)$ for the summation over the measures of outcomes leading to the event $A$.
%We show that this 
If one can show convergence, in the sense that 
\begin{equation}
\label{eq:overview-convergence}
\psi_{H_i}(A) = \lim_{T\to\infty} \psi^{\-{GD}}_{T}(\sigma_0\in A),
\end{equation}
then, bounding the measure $\psi^{\-{GD}}_T$ also bounds the right-hand side of \eqref{eq:overview-projection-lifting-scheme-bad-cluster}.

A key insight from \cite{liu2024phase} is to use a ``percolation'' style argument for both establishing the convergence and bounding $\psi^{\-{GD}}_T$.
In their setting of hypergraph independent sets, they first decompose each transition of the Markov chain into an \emph{oblivious} update part and an \emph{adaptive} (non-oblivious) update part.  Roughly, this corresponds to decomposing the complex transition matrix into the sum of two parts.
% The measure $\psi^{\-{GD}}_T$ is defined on this decomposition.
Then, convergence follows from showing that the oblivious updates ``dominate'' the contributions to the measure $\psi^{\-{GD}}_T$. Furthermore, contributions from the oblivious updates are easier to control, so that essentially, one can upper bound $\abs{\psi^{\-{GD}}_T}$ by summing over trajectories that end with an oblivious update.

For a ``percolation'' style argument to work in the complex plane, \cite{liu2024phase} formulated the notion of \emph{witness sequence}, conditioned on which the complex measure $\psi^{\-{GD}}_T$ exhibits a \emph{zero-one law}: $\psi^{\-{GD}}_T( \sigma_0\in A \mid \hbox{witness}) \in \set{0,1}$. Intuitively, a witness sequence has two objectives: a) it encodes the ``trajectory'' of a complex Markov chain so that the final event $\sigma_0\in A$ is independent of the initial state; b) the measure of the encoding cannot be much larger than the measure of oblivious updates.

To analyze our projected measure $\psi_{H_i}$, we follow~\cite{liu2024phase} closely in spirit, but we also have to innovate in the following aspects.

\textbf{Decomposition of transition measures:} The first step  is to carefully decompose transition measures into oblivious and non-oblivious transitions, so that oblivious ones are close to the true transition measures.

In the context of hypergraph $q$-colorings for real probability measures, similar decomposition is known through \emph{local uniformity}~\cite{guo2018counting,he2021perfect,feng2021sampling,jain2021sampling}, which roughly means that under local lemma conditions, local marginals on variables are close to the uniform product distribution. A complex analogue of local uniformity seems highly non-trivial, as existing arguments rely on LLL~\cite{haeupler2011new} on the real line.
Furthermore, in order to introduce external fields to a given color (as opposed to external fields on color buckets), our state-compression scheme requires designating one special color, under which standard local uniformity no longer holds even in the real line.

We bypass these difficulties with two observations. The first is to work with ``local uniformity'' of the \emph{transition measure} (as opposed to that of the \emph{stationary measure}). 
In the transition measure, the difference between real-valued or complex-valued external field $\lambda$ can be bounded in a neighborhood of the interval $[0,1]$.
%, so that bounds on the real-valued transition measure carry over. %when updating a vertex $v$, the bucket assignments of other vertices are fixed; consequently, the contributions of the complex external field are fixed except at $v$. 
Concretely, given a bucket assignment $\tau\in \set{1^\coF, 2^\coF, \dots, B^\coF}^{V\setminus {v}}$, for each $i\in [B]$, let $X_i$ be the number of  colorings consistent with $\tau$ that assign $v$ with bucket $i^\coF$.  Then the transition measure $\psi^{\tau}_v$ is given by: 
\begin{equation}
\label{eq:overview-transition-expression}
\psi^{\tau}_v(1^\coF) = \frac{\lambda X_1}{\lambda X_1 + \sum_{i=2}^B X_i},\qquad \hbox{ and }\forall i\in [B]\setminus \set{1},\quad \psi^{\tau}_v(i^\coF) = \frac{X_i}{\lambda X_1 + \sum_{i=2}^B X_i}.
\end{equation}
Therefore,  $\lambda \approx 1 \implies \psi^{\tau}_v(i^\coF)\approx \frac{X_i}{\sum_{i=1}^B X_i}$. It remains to bound $\frac{X_i}{\sum_{i=1}^B X_i}$. As noted earlier,  the first bucket $1^\coF$ contains only one color in our state-compression scheme, and standard local uniformity no longer holds. 

Our second observation is that a relaxed local uniformity still holds: instead of asking the marginals to be close to uniform, we ask uniformity for every color bucket $i^\coF, j^\coF$ other than the special color $1^\coF$, that is, $\psi^{\tau}_v(i^\coF) \approx \psi^{\tau}_v(j^\coF)$; and for $1^\coF$, we ask that either $\psi^{\tau}_v(1^\coF) =0$ or $\psi^{\tau}_v(1^\coF) \approx 1/q$. To see why this holds, under a bucket assignment $\tau$, if there is a hyperedge containing $v$ and $k-1$ vertices in bucket $1^\coF$, one can remove this hyperedge and restrict that $v$ is not in $1^\coF$.
On this reduced hypergraph without $1^\coF$,  each vertex still has large buckets and standard local uniformity holds on this reduced hypergraph.

We design a decomposition of transition measure specifically tailored to the relaxed local uniformity.
It is worth noting that $\psi^{\-{GD}}_T$, which is generated by the vector-matrix product $\mu P^T$, can also be seen as summing over walks of length $T$ over a ``space-time slab'', on which every node is weighted by the corresponding entry in the transition matrix.
Intuitively, we would like to show that, when revealing these ``space-time walks'' (corresponding to ``trajectories'' of the Glauber dynamics) backwards in time, the dominating contributions to $\psi^{\-{GD}}_T$ come from those that have encountered certain oblivious updates before reaching the starting state. This is formalized through {witness sequences} in~\cite{liu2024phase}.

\textbf{Design of witness sequences}: 
We start by reviewing the intuition behind \emph{witness sequences}. \cite{liu2024phase} use them to encode ``space-time walks'' that encounter oblivious updates before reaching the starting state. %They are designed so that any witness sequence $W$ satisfies.
\vspace{-15pt}
\begin{property}
\label{prop:desirable-witness}\emph{Witness sequences} designed by~\cite{liu2024phase} are such that if $W$ is a witness sequence, then
\begin{enumerate}
    \item $\psi^{\-{GD}}_T(\sigma_0 \in A \land W)$ is independent of the starting state;
    \item ``Zero-one'' law: by viewing $W$ as an ``event'', it fully determines the event $(\sigma_0 \in A)$. In particular,
    $\psi^{\-{GD}}_T(\sigma_0 \in A \mid W)\in \set{0,1}$ when well-defined;
    \item $\abs{\psi^{\-{GD}}_T(W)}$ can be explicitly calculated (or bounded), and  is
    independent of the starting state; %$\psi^{\-{GD}}_T(W)$
    \item $\lim_{T\to \infty}\abs{\psi^{\-{GD}}_T(\hbox{non-witness sequences})} = 0$.
\end{enumerate}
\end{property}
Given these desirable properties, one can proceed as follows:
\begin{align*}
% \label{eq:overview-bad-component-total-measure}
\psi^{\-{GD}}_{T}(\sigma_0 \in A) =& \sum_{W: \hbox{witnesses}} \psi^{\-{GD}}_{T}(\sigma_0 \in A \land W) +  \psi^{\-{GD}}_T(\hbox{non-witness sequences})\\
=& \sum_{W: \hbox{witnesses}} \psi^{\-{GD}}_{T}(W)\psi^{\-{GD}}_{T}(\sigma_0 \in A | W) +  \psi^{\-{GD}}_T(\hbox{non-witness sequences}).
\end{align*}
%We note that part (1) and (2) also implies $\psi^{\-{GD}}_T(W)$ is independent of the starting state.

Then, by a triangle inequality one can deduce
\begin{align}
\label{eq:overview-01}
\abs{\psi^{\-{GD}}_{T}(\sigma_0 \in A)} \le& \sum_{W: \hbox{witnesses}} \abs{\psi^{\-{GD}}_{T}(W)} \cdot \abs{\psi^{\-{GD}}_{T}(\sigma_0 \in A | W)}+  \abs{\psi^{\-{GD}}_T(\hbox{non-witness sequences})}\nonumber \\ %\qquad \hbox{by a triangle inequality}\\
\hbox{(by ``zero-one'' law)} \qquad\le& \sum_{W: \hbox{witnesses}} \abs{\psi^{\-{GD}}_{T}(W)}+  \abs{\psi^{\-{GD}}_T(\hbox{non-witness sequences})}.  \end{align}
Taking limits, we show in~\Cref{lemma:convergence} that $\psi^{\-{GD}}_{T}$ converges to $\psi_{H_i}$, with a proof in~\Cref{section:missing-analytic-percolation}. Therefore,
\begin{align*}
\abs{\psi_{H_i}(A)} = \lim_{T\to \infty}  \abs{\psi^{\-{GD}}_{T}(\sigma_0 \in A)} \le \lim_{T\to\infty} \sum_{W: \hbox{witnesses}} \abs{\psi^{\-{GD}}_{T}(W)}.
\end{align*}
It is worth noting that, the above derivation, specifically in \cref{eq:overview-01}, uses a ``zero-one'' law (part (2) in~\Cref{prop:desirable-witness}) instead of the monotonicity of measures. In contrast, for standard probability measures, one can directly bound $\psi^{\-{GD}}_{T}(\sigma_0 \in A \land W) \le \psi^{\-{GD}}_{T}(W)$ thanks to the monotonicity. %, one only needs the first condition for the encodings $W$. Thanks to the monotonicity of probability measures, one can directly bound $\psi^{\-{GD}}_{T}(\sigma_0 \in A \land W) \le \psi^{\-{GD}}_{T}(W)$, and then proceed using explicit bounds on $\psi^{\-{GD}}_{T}(W)$. The above argument allows one to control complex measures that do not enjoy monotonicity.

\cite{liu2024phase} characterized witness sequences in the setting of hypergraph independent sets formally through the notion of \emph{bad component}. This naturally generalizes to hypergraph colorings with one special color. However,  the witness sequences characterized by analogous bad components only satisfy part (1) in~\Cref{prop:desirable-witness}. To restore the ``zero-one'' law, we also need to encode extra information so that combined with bad components, will play the role of part (2) and (3), which can then be used to establish part (4). Essentially, the extra encoding only introduces an additional conditioning during the derivation of~\cref{eq:overview-01}.

Next, we describe some intuitions for how to satisfy part (1) in~\Cref{prop:desirable-witness}, and restoring the ``zero-one'' law in hypergraph colorings.
We recall the relaxed local uniformity, which says that all but the special color $1^\coF$ have approximately equal marginals. So the oblivious transitions would be when trying to update to a non-special color bucket uniformly at random.
We use $r_t \in \set{2^\coF, 3^\coF, \dots, B^\coF, \bot}$ to denote the result of the oblivious transition at time $t$, where $\bot$ is a placeholder indicating that the oblivious transition alone at time $t$ cannot fully determine the update value, and a non-oblivious update may be needed. %For technical reasons, we restrict that $1^\coF$ cannot be the result of an oblivious transition. 

%Based on $r_t$'s, we give the intuition about the witness sequence.

To see how to satisfy part (1) of~\Cref{prop:desirable-witness}, we consider a transition step at time $t$ with a vertex $v$ to update, and we reveal the oblivious color bucket assignments of $r_{s}$ for all neighboring vertices of $v$ at time $s$ just prior to $t$. If all hyperedges containing $v$ are not monochromatic by these $r_s$, then we know the update at time $t$ is independent of the initial state. In fact, conditioned on these $r_s$, since all the neighboring constraints are already satisfied, the complex measure \emph{factorizes} (see \Cref{observation:csp-conditional-independence-cbad-o_t} for a formal statement). 
Similarly, for an update at time $t$, if for all hyperedges containing $v$, their last updates are independent of the initial state, then we know the update at time $t$ is also independent of the initial state.
By revealing the oblivious color bucket assignments backwards in time, this forms a percolation on the space-time diagram, and its structure is formally captured by a
 \emph{bad component} $\cbad$  (see \Cref{definition:csp-conditional-measure-bad-things}). Intuitively, a bad component is ``small'' exactly when the corresponding $r_t$'s is  a witness sequence (that is, satisfying part (1)).

\textbf{Zero-one law}: Compared to~\cite{liu2024phase}, there are two complications in trying to restore a ``zero-one'' law. Firstly, \emph{bad component} itself does not provide enough information for general CSP, and $\psi^{\-{GD}}_{T}(\sigma_0 \in A \land \cbad)$ is still a highly non-trivial complex measure; secondly, the fact that we have to work with a projected measure also leads to much more sophisticated events $A$ (corresponding to a bad cluster during projection-lifting), for which we need to restore the ``zero-one'' law.
Our idea is to also encode the actual updated values within a bad component. Let $\BadTS{\cbad}$ denote the update times involved in the bad component $\cbad$, and $\*o_{\BadTS{\cbad}}$ denote the updated values in $\BadTS{\cbad}$. We show that given $\cbad$ and $\*o_{\BadTS{\cbad}}$, we can still recover a zero-one law (\Cref{lemma:csp-determination-bad-cluster-by-bad-component}).

Then, to satisfy part (3) in~\Cref{prop:desirable-witness}, we also need a much more refined analysis compared to~\cite{liu2024phase} in order to have room for both the more complicated events $A$ (which can depend on an unbounded number of variables) and the extra encoding.

We restore part (4) in~\Cref{prop:desirable-witness} by combining part (1), (2) and (3) to show that contributions from ``large'' bad components diminish to $0$ as $T\to \infty$. This part of the property is formalized as~\Cref{condition:convergence}, which we formally verify in~\Cref{subsubsection:convergence}.
\subsubsection{CLTs, LCLTs and concentration inequalities under LLL-like conditions}
Given our zero-freeness, central limit theorems follow from~\cite{michelen2024central} together with a lower bound on the variance. By the definition of variance, it suffices to show a lower bound of variance on each vertex and a lower bound on the covariance (see \cref{eq:clt-variance-decomposition}). 
The first lower bound follows from local uniformity (\Cref{lemma:clt-marginal-bound,lemma:clt-local-variance}). For the second one, we make use of a recent breakthrough in sampling LLL~\cite{wang2024sampling} to bound the \emph{total influence} (\Cref{lemma:clt-bounded-total-influence}) which leads to a lower bound of the covariance (see \cref{eq:clt-covariace-lower-bound}).
In \Cref{section:lclt}, following the spirit of~\cite{jain2022approximate}, we boost CLTs into LCLTs through a Fourier inversion (see \Cref{lemma:establish-lclt-by-characteristic-function}). We extend their analysis to general CSPs for constant $\lambda$. 

Our simple Chebyshev-type inequality consists of two components: a lower bound on the expectation and an upper bound on the variance.
The lower bound on expectation follows from local uniformity, while the upper bound again makes use of a recent breakthrough in sampling LLL~\cite{wang2024sampling}. Among others, they showed a constant total influence bound with respect to a constraint. Translating this to a bound with respect to a variable only loses factors of $k$, which can then be used to bound the covariances between variables. We also lose a factor of $q$ due to pinning on vertices introduced in the covariances calculation. This leads to a Chebyshev inequality under $p D^{2+o_q(1)} \lesssim 1/q$.

\section{Preliminaries}
\subsection{Complex normalized measures}
We recall the notions of complex normalized measure~\cite{liu2024phase}.
Let $\mu:\Omega\to\=C$ be a complex measure over a measurable space $(\Omega, \+F)$, where $\Omega$ is a finite set and elements in $\+F$ are called events. 
Let $\supp(\mu)$ be the support of $\mu$, i.e. $\supp(\mu)\defeq \{x\in \Omega\mid \mu(x)\neq 0\}$. We say $\mu$ is a \emph{complex normalized measure} if $\sum_{\omega\in \Omega} \mu(\omega) = 1$. 
For any event $A\in \+F$, the measure on $A$ is defined as $\mu(A)\defeq \sum_{\omega\in A}\mu(\omega)$. 

Similar to probability, for any event $A\in \+F$ with $\mu(A)\neq 0$, we define \emph{conditional measure} of $\mu$ on $A$ as a restricted measure $\mu(\cdot\mid A)$ over the measurable space $(\Omega, \+F_A)$ where $\+F_A = \{B\cap A\mid B\in \+F\}$ such that for any event $B\in \+F$,
%\[
$\mu(B\mid A) = \frac{\mu(B\cap A)}{\mu(A)}$.
%\]
Note that the conditional measure $\mu(\cdot\mid A)$ is always normalized when well-defined.

We say that two events $A_1, A_2 \in \+F$ are \emph{independent} if  $\mu(A_1\cap A_2) = \mu(A_1)\cdot \mu(A_2)$.
More generally, for a finite sequence of events $A_1, A_2,\dots, A_m\in\+F$, we say that they are \emph{mutually independent} if, for any finite subset $I\subseteq \{1, 2,\dots, m\}$, it holds that
%\[
$\mu\tp{\bigcap_{i\in I} A_i} = \prod_{i\in I} \mu(A_i)$.
%\]

For a finite sequence of events $A_1,A_2,\dots, A_m \in \+F$, we say that they are \emph{mutually disjoint} if for any $i\neq j$, $A_i\cap A_j = \emptyset$. The \emph{law of total measure} also holds for these complex measures. Let $A_1,A_2,\dots, A_m \in \+F$ be a finite sequence of mutually disjoint events with $\bigcup_{i=1}^m A_i = \Omega$. Then for any event $B \in \+F$, we have that 
\[
\mu(B) = \sum_{i = 1}^m \mu(B\cap A_i).
\]

\subsection{Graphical models and their conditional measure}
Let $V$ be a finite set, $\*Q = (Q_v)_{v\in V}$ be a series of finite sets.
Let $H = (V, \+E)$ be a hypergraph, where each vertex $v\in V$ represents a random variable that takes the value from the set $Q_v$ and each hyperedge $e\in \+E$ represents a local constraint on the set of variables $e\subseteq V$. 
For each $v\in V$, there is a function $\phi_v: Q_v\to \=C$ that expresses the external field, and for each $e\in \+E$, there is a function $\phi_e: \bigotimes_{v\in e}Q_v \to \=C$ that expresses the interaction. 
A \emph{graphical model} is specified by the tuple $\+G = (H, (\phi_v)_{v\in V}, (\phi_e)_{e\in \+E})$. For each configuration $\sigma\in \*Q$, we define its weight as 
% \[
% w_{\+G}(\sigma)\defeq \exp\tp{\sum_{e\in \+E} \phi_e(\sigma_e) + \sum_{v\in V}\phi_v(\sigma_v)}.
% \]
\[
w_{\+G}(\sigma)\defeq \prod_{e\in \+E} \phi_e(\sigma_e) \cdot \prod_{v\in V}\phi_v(\sigma_v).
\]
And the \emph{partition function} $Z_{\+G}$ of the graphical model $\+G$ is given by 
%\[
$Z = Z_{\+G} \defeq \sum_{\sigma\in \*Q} w_{\+G}(\sigma)$.
%\]

When the partition function $Z_{\+G}$ is non-zero, we define the \emph{Gibbs measure} $\mu = \mu_{\+G}$ on the measurable space $(\*Q, 2^{\*Q})$, where
%\[
$\forall \sigma\in \*Q, \mu(\sigma)\defeq \frac{w_{\+G}(\sigma)}{Z_{\+G}}$.
%\]
For any subset of variables $\Lambda \subseteq V$, we define $\*Q_{\Lambda} \defeq \bigotimes_{v\in \Lambda}Q_v$ and for a partial restriction $\sigma\in \*Q_{\Lambda}$ over $\Lambda$, we say $\sigma$ is a \emph{pinning} on $\Lambda$. And we say $\sigma$ is \emph{admissible} if its measure is non-zero. 
For a pinning $\tau\in \*Q_\Lambda$ on $\Lambda$, we will also denote $\mu(X_\Lambda = \tau)\defeq \sum_{\sigma \in \*Q:\ \sigma \succeq \tau} \mu(\sigma)$.
For any disjoint $S, \Lambda \subseteq V$ and any admissible pinning $\sigma\in \*Q_\Lambda$ on $\Lambda$, we use $\mu_S^\sigma$ to denote the marginal measure induced by $\mu$ on $S$ conditioned on $\sigma$, i.e.,
%\yxtodo{Is this $X$ defined? Is this OK?}
%\[
$\forall \tau \in \*Q_S, %\quad 
\mu^\sigma_S(\tau) = \frac{\mu(X_S = \tau\land X_\Lambda = \sigma)}{\mu(X_\Lambda = \sigma)}$.
%\]

\subsection{Glauber dynamics: random scan and systematic scan}
We recall the standard Glauber dynamics. Let $\mu$ be a distribution on $\*Q$ with $V = \{v_1, v_2, \dots, v_n\}$. The \emph{Glauber dynamics} is a canonical construction of Markov chains with stationary distribution $\mu$. Starting from an initial state $X_0\in \*Q$ with $\mu(X_0) \neq 0$, the chain proceeds as follows at each step $t$:
\begin{itemize}
    \item pick a variable $v \in V$ uniformly at random and set $X_t(u) = X_{t-1}(u)$ for all $u\neq v$;
    \item update $X_t(v)$ by sampling from the distribution $\mu_v^{X_{t-1}(V\setminus \{v\})}$.
\end{itemize}

The \emph{systematic scan Glauber dynamics} is a variant of the standard Glauber dynamics. In each step, instead of updating a variable at random, it updates them in a canonical order. Specifically, at each step $t$, we choose the variable $v = v_{i(t)}$, where
\begin{equation}
\label{eq:scan-i(t)}
i(t)\defeq (t\mod n) + 1.
\end{equation}
Then $X_{t-1}$ is updated to $X_t$ using the same rule as in Glauber dynamics, based on the chosen $v$.

The Glauber dynamics is well known to be both aperiodic and reversible with respect to $\mu$. The systematic scan Glauber dynamics is not time-homogeneous, as variables are accessed in a cyclic order. However, by bundling $n$ consecutive updates, we obtain a time-homogeneous Markov chain that is aperiodic and reversible.

\subsection{CSP formulas}
A constraint-satisfaction problem (CSP) formula $\Phi$ is a tuple $(V, \*Q, C)$.
$V$ is the \emph{variable set}. 
$\*Q = (Q_v)_{v\in V}$ is the \emph{domain set} and for each variable $v\in V$, $Q_v$ is finite and $v$ takes a value in $Q_v$.
And $C$ is the \emph{constraint set}.
For any subset $S\subseteq V$ of variables, we use $\bm{Q}_S = \bigotimes_{v\in S} Q_v$ to denote the domain on $S$.
Each constraint $c\in C$ is defined on a subset of variables $\vbl(c)\subseteq V$ and maps every assignment $\sigma_{\vbl(c)}\in \*Q_{\vbl(c)}$ to ``True'' or ``False'', which indicates whether $c$ is \emph{satisfied} or \emph{violated}. 
Let $c^{-1}(\-{False})$ be the set of partial assignments on $\vbl(c)$ that violate the constraint $c$.
Let $\Omega_{\Phi}$ be the set of all satisfying assignments, i.e.,
\[
\Omega_{\Phi} = \left\{\sigma \in \bigotimes_{v\in V} Q_v~\right \vert \left. \bigwedge_{c\in C} c(\sigma_{\vbl(c)}) \right\}.
\]

Note that every CSP formula can be expressed as a graphical model.
% Each variable $v\in V$ takes value in $Q_v$. 
We focus on $k$-uniform bounded degree CSP formulas.
Let $k, \Delta$ be two integers. A \emph{$(k, \Delta)$-CSP formula} is a CSP formula $\Phi = (V, \*Q, C)$ such that for each constraint $c\in C$, it holds that $|\vbl(c)|=k$ and each variable belongs to at most $\Delta$ constraints.

The \emph{dependency graph} $\dependencyGraph{\Phi}$ of a CSP formula $\Phi = (V, \*Q, C)$ is defined on the vertex set $C$, such that any two constraints $c, c'\in C$ are adjacent if $\vbl(c)\cap\vbl(c')\neq \emptyset$ and $c \neq c'$. We use $\Gamma(c)\defeq \{c'\in C\setminus \{c\}\mid \vbl(c)\cap\vbl(c')\neq \emptyset\}$ to denote the neighborhood of $c\in C$.

\subsection{Lov\'asz local lemma}
%We now include the Lov\'asz local lemma.
Let $\+R = \{R_1, R_2, \dots, R_n\}$ be a collection of mutually independent random variables. For any event $E$, let $\vbl(E)\subseteq \+R$ be the set of variables determining the event $E$. Let $\+B = \{B_1, B_2, \dots, B_m\}$ be a collection of ``bad'' events. For each event $B\in\+B$, we define $\Gamma(B)\defeq \{B'\in \+B\mid B'\neq B \text{ and } \vbl(B')\cap \vbl(B)\neq \emptyset\}$. Let $\+P[\cdot]$ denote the product distribution of variables in $\+R$. The Lov\'asz local lemma gives a sufficient condition for avoiding all bad events with positive probability.
\begin{theorem}[{\cite{erdHos1975problems}}]
\label{theorem:LLL}
If there is a function $x:\+B\to(0,1)$ such that
% for any $B\in \+B$,
\begin{equation}
\label{eq:LLL}
\forall B\in \+B, \quad \+P[B]\le x(B)\prod_{B'\in \Gamma(B)}(1 - x(B')),
\end{equation}
then it holds that
\[
\+P\left[  \bigwedge_{B\in \+B} \bar{B} \right] \ge \prod_{B\in\+B}(1 - x(B)) > 0. 
\]
\end{theorem}
We also include the following lopsided version of the Lov\'asz local lemma. For a constraint set $C$, we use $\+P[\cdot \mid C]$ to denote the product distribution $\+P[\cdot]$ conditioned on the event that all constraints in $C$ are satisfied.
\begin{theorem}[{\cite[Theorem 2.1]{haeupler2011new}}]
    \label{theorem:HSS}
    Given a CSP formula $\Phi = (V, \*Q, C)$, let $C = \set{c_1, c_2, \dots, c_{|C|}}$ and $\+B = \set{B_1, B_2, \dots, B_{|C|}}$ where $B_i$ is the event that the constraint $c_i$ is violated, if there exists a function $x: \+B \to (0, 1)$ such that \cref{eq:LLL} holds, then for any event $\+A$ that is determined by the assignment on a subset of variables $\vbl(\+A) \subseteq V$,
    \[
    \+P[\+A\mid C]\le \+P[\+A]\cdot \prod_{\substack{c\in C\\ \vbl(c)\cap \vbl(\+A) \neq \emptyset}} (1 - x(c))^{-1}.
    \]
\end{theorem}
The celebrated Moser-Tardos algorithm~\cite{moser2010constructive}  constructs an assignment of all random variables in $\+P$ that avoids all the bad events in $\+B$ under the same assumption.
\begin{theorem}[{\cite{moser2010constructive}}]
\label{theorem:moser-tardos}
Suppose the asymmetric local lemma condition \cref{eq:LLL} in \Cref{theorem:LLL} holds with the function $x:\+B\to(0,1)$. Upon termination, the Moser-Tardos algorithm returns an assignment that avoids all the bad events. The expected total resampling steps for the Moser-Tardos algorithm are at most $\sum_{B\in\+B} \frac{x(B)}{1 - x(B)}$.
\end{theorem}

% \color{blue}
% We also introduce the lopsided version of the Lov\'asz local lemma.
% \begin{theorem}[{\cite{haeupler2011new}}]
% \label{theorem:HSS}
% Given a CSP formula $\Phi = (V, \*Q, \+C)$, if \cref{eq:LLL} holds, then for any event $A$ that is determined by the assignment on a subset of variables $\vbl(A)\subseteq V$,
% \[
% \+P\inbr{}
% \]
% \end{theorem}
% \color{black}
% \subsection{$k$-CNF}
\subsection{$2$-tree}

We also need the notion of $2$-trees~\cite{Alon91}. 
Given a graph $G=(V,E)$, its square graph $G^2=(V,E_2)$ has the same vertex set, while an edge $(u,v)\in E_2$ if and only if $1\leq \dist_G(u,v)\leq 2$. 

\begin{definition}[$2$-tree]\label{definition:2-tree}
Let $G=(V, E)$ be a graph. A set of vertices $T\subseteq V$ is called a \emph{$2$-tree} of $G$, if
\begin{itemize}
    \item  for any $u,v\in T$, $\text{dist}_G(u,v)\geq 2$, and
    \item  $T$ is connected on $G^2$.
\end{itemize}
\end{definition}

Intuitively, a $2$-tree is an independent set that does not spread far away. We can construct a large $2$-tree in any connected graph as follows.

\begin{definition}[Construction of a maximal $2$-tree {\cite[Lemma 4.5]{Vishesh21towards}}]\label{definition:2-tree-construction}
Let $G = (V, E)$ be a connected graph of maximum degree $D$ and $v\in V$. We can deterministically construct a $2$-tree $T$ of $V$ containing $v$
such that $\abs{T} \geq \lfloor |V|/(D + 1) \rfloor$ as follows:
\begin{itemize}
\item order the vertices in $V$ in the lexicographical order. Start with $T=\{v\}$ and $U=V\setminus N^+(v)$, where $N^+(v)\defeq N(v)\cup \{v\}$ and $N(v)\defeq \{u\in V\mid (u,v)\in E\}$ ;
\item repeat until $U=\emptyset$: let $u$ be the vertex in $U$ with the smallest distance to $T$, with ties broken by the order on $V$. Set $T\gets T\cup \{u\}$ and $U\gets U\setminus N^+(u)$.
\end{itemize}
\end{definition}
The following lemma bounds the number of subtrees and $2$-trees of a certain size containing a given vertex, respectively.

\begin{lemma}[\text{\cite[Lemma 2.1]{borgs2013left}},{\cite[Corollary 5.7]{feng2021rapid}}]\label{lemma:2-tree-number-bound}
Let $G=(V, E)$ be a graph with maximum degree $D$, and $v\in V$ be a vertex. 
The number of subtrees in $G$ of size $k \ge 2$ containing $v$ is at most $\frac{(\mathrm{e}D)^{k-1}}{2}$, and the number of 2-trees in $G$ of size $k \geq 2$ containing $v$ is at most $\frac{(\mathrm{e}D^2)^{k-1}}{2}$.
\end{lemma}

\section{Controlling marginal measures via analytic percolation}
\label{section:analytic-percolation}
% \begin{comment}
%     In this section, I plan to outline the analytic percolation process.
% \end{comment}
% \textcolor{red}{(I still need to put the complex Markov chain here.)}
% In \cite{liu2024phase}, they analyze the complex measure through complex Markov chains. 
% They analyze complex Markov chains by providing a complex extension of the percolation argument which is widely used for analyzing Markov chains in the real case \cite{hermon2019rapid,he2021perfect,qiu2022perfect,feng2023towards}.
% They decompose each transition so that transferring the complex Markov chain into a subcritical percolation.

A crucial step in establishing zero-freeness is to bound the norm of certain complex marginal measures.
In this section, we review the theory of complex Markov chains developed in~\cite{liu2024phase}, which provides a novel tool for analyzing complex normalized measures via complex Markov chains.
% We then extend their framework by introducing a new decomposition scheme (\Cref{definition:decomposition-scheme}), which generalizes their $\*b$-decomposition and enables more refined information percolation analysis.
We demonstrate how to analyze the convergence of complex Markov chains (\Cref{lemma:convergence}) and how to bound the norm of marginal measures (\Cref{lemma:bounding-marginal-measure}).

For cases where the theory of complex Markov chains breaks down—such as when the support set is disconnected—we propose a \emph{projection-lifting scheme}.
In this approach, we first analyze a ``projected complex normalized measure'' using the complex Markov chain theory, and then lift the result back to the original complex normalized measure via the law of total measure.

% A crucial step to establish zero-freeness is to bound the norm of some complex marginal measures. 
% In this section, we first outline the theory of complex Markov chains in~\cite{liu2024phase} which is a new tool to analyze the complex normalized measure through complex Markov chains.
% And we extend their framework by provide a new decomposition scheme (\Cref{definition:decomposition-scheme}) which allows more sophisticated information percolation analysis extending their $\*b$-decomposition scheme.
% We show how to analyze the convergence of the complex Markov chain (\Cref{lemma:convergence}) and how to bound the norm for marginal measures (\Cref{lemma:bounding-marginal-measure}).

% For the cases that the theory of complex Markov chains fail, such that the case that the support set is not connected, we provide a \emph{projection-lifting scheme} that we use the theory of complex Markov chains to analyze a ``projected complex normalized measure'' then we lift it back to the original complex normalized measure by the law of total measure.
% We also extend their framework 
% \subsection{Complex Markov chain}

\subsection{Projection-lifting scheme}
\label{subsection:projection-lifting-scheme}
To introduce our \emph{projection-lifting scheme}, we first define the \emph{projected measure}.
Let $\*\Sigma = (\Sigma_v)_{v\in V}$ where for any $v\in V$, $\Sigma_v$ is a finite \emph{alphabet}.
Let $\*f = (f_v)_{v\in V}$ be a series of projections where $f_v: Q_v\to \Sigma_v$. 
We naturally interpret $\*f$ as a function, for any $\sigma\in\*Q$, we have $\*f(\sigma) = (f_v(\sigma_v))_{v\in V}$.
We refer to $\*\Sigma$ as \emph{the projected alphabet} or \emph{projected domain}.
The \emph{projected measure} $\psi$ on the projected alphabet $\*\Sigma$ is defined as:
\[
\forall \sigma\in\*\Sigma,\quad \psi(\sigma) = \sum_{\tau\in \*Q: \sigma = \*f(\tau)} \mu(\tau).
\] 

Assume that we can analyze the complex marginal measure of $\psi$, for example, through the theory of complex Markov chains.
And recall that a crucial step to prove zero-freeness is to bound the norm of complex marginal measures. Let $A\subseteq \*Q$ be an event. Now we demonstrate how to upper bound $|\mu(A)|$ through the marginals of the projected measure $\psi$.

We first define a series of events $\+B_1, \+B_2, \dots, \+B_m$, where for each $i\le m$, $\+B_i \subseteq \*\Sigma$. We remark that by the definition of the projected measure, for any $1\le i\le m$, the marginal $\mu(\+B_i)$ is well-defined and $\mu(\+B_i) = \psi(\+B_i)$. We give the following condition about our projection-lifting scheme. Next, we first show how to use this condition to bound $|\mu(A)|$, then we outline its intuition.
\begin{condition}
\label{condition:projection-lifting-scheme}
Let $A\subseteq \bm{Q}$ be an event.
It holds that
\begin{enumerate}
\item $\+B_1, \+B_2, \dots, \+B_m$ are mutually disjoint, that is,
 for any $i\neq j$,  $\+B_i \cap \+B_j = \emptyset; $\label{item:projection-lifting-scheme-partition}
\item the event $A\cap (\*\Sigma \setminus (\bigcup_{j=1}^m \+B_j))$ has zero measure, that is, any outcome in $A\cap \*\Sigma$ must occur in one of the events $\+B_1, \+B_2,\dots, \+B_m$;\label{item:projection-lifting-scheme-cover}
\item for any $1\le i\le m$, we have 
\[
\psi(\+B_i) = 0 \Rightarrow \mu(A\land \+B_i) = 0,\quad \psi(\+B_i) \neq 0 \Rightarrow \mu(A\mid \+B_i) \in [0,1].
\]\label{item:projection-lifting-scheme-condition}
\end{enumerate}
\end{condition}
% We first show how to use this condition to upper bound $|\mu(A)|$, then we outline its intuition.
By \Cref{condition:projection-lifting-scheme} and the triangle inequality, we can upper bound the $|\mu(A)|$ as,
\[
|\mu(A)| = \abs{\sum_{i=1}^m \mu(A \land \+B_i)} \le \sum_{i=1}^m |\mu(A \land \+B_i)| \le \sum_{i=1}^m |\mu(\+B_i)| = \sum_{i=1}^m |\psi(\+B_i)|,
\]
where the first equality is due to \Cref{condition:projection-lifting-scheme}-(\ref{item:projection-lifting-scheme-partition}) and \Cref{condition:projection-lifting-scheme}-(\ref{item:projection-lifting-scheme-cover}), the first inequality is due to the triangle inequality, and the second inequality is due to \Cref{condition:projection-lifting-scheme}-(\ref{item:projection-lifting-scheme-condition}) and the last equality is due to the definition of $\psi$.

The intuition for \Cref{condition:projection-lifting-scheme} is that we find disjoint events (\ref{item:projection-lifting-scheme-partition}) on the projected measure that ``cover the event $A$'' (\ref{item:projection-lifting-scheme-cover}). So that we can use the law of total measure and the triangle inequality to upper bound the norm of $A$. Next, if we can fix the contributions from complex external fields by the events in $\psi$, then condition on that event, the conditional measure becomes a probability which is in $[0,1]$ (\ref{item:projection-lifting-scheme-condition}). 

Later, we use a percolation-style analysis to construct these events $\+B_1, \+B_2, \dots, \+B_m$. Then, in \Cref{subsection:csp-zero-freeness}, we use the theory of complex Markov chains for the analysis of $|\psi(\+B_i)|$.
% , in particular, we use the \emph{complex local uniformity} (\Cref{condition:complex-local-uniformity}) to bound the norm of marginals under arbitrary pinnings so that we can bound $|\psi(\+B_i)|$. 

% They hold that they are disjoint and their union is $\*\Sigma$, i.e.,
% \
% \forall i\neq j, \quad \+B_i \cap \+B_j = \emptyset, \quad \bigcup_{i=1}^m \+B_i = \*\Sigma'.
% \]

\subsection{Complex Extensions of Markov Chains}
We review the notion of complex extension of Markov chains from~\cite{liu2024phase}, where we make the modifications in this paper: the Glauber dynamics is applied only to a projected measure.
% ; secondly, we generalize the notion of decomposition to handle decomposition based on local uniformity and beyond.
\subsubsection{Complex Markov Chain}
%We recall the definition of complex-valued transition matrices on a finite state space.
Let $\Omega$ be a finite state space. 
A complex-valued transition matrix $P \in \=C^{\Omega \times \Omega}$ satisfies that
%\[
$\forall \sigma \in \Omega,  \sum_{\tau \in \Omega} P(\sigma, \tau) = 1$.
%\] 
% which is a direct extension of the classical row-stochastic matrix.
% Now we define time-homogeneous complex Markov chains, whose definition generalizes naturally to time-inhomogeneous chains.
Fix $T \ge 1$. For a measurable space $(\Omega, \+F)$ with finite $\Omega$, 
we write $\Omega^T$ for the Cartesian product, and $\+F^{T}$ for the product $\sigma$-algebra. 
Let $\`P$ be a  complex normalized measure on $(\Omega^{T}, \+F^{T})$  and 
let $X_1, X_2, \ldots, X_T$  be a sequence  of measurable functions taking values over $\Omega$ following the measure~$\`P$. 
The sequence $(X_t)_{t=1}^T$ is said to be a $T$-step discrete-time complex Markov chain if there exists a complex-valued transition matrix $P \in \=C^{\Omega \times \Omega}$ such that
for any $1< j \le T$ and any $x_1, x_2, \dots, x_{j} \in \Omega$,  
\begin{align*}
 \`P\tp{X_{j} = x_{j} \mid \bigwedge\limits_{i=1}^{j-1} X_{i}=x_i} &= \`P\tp{X_{j} = x_{j} \mid X_{j-1}=x_{j-1}}
 = P(x_{j-1}, x_{j}).
\end{align*}
% \[
% \`P(X_{j} = x_{j} \mid X_{1}=x_1, X_{2}=x_2, \dots, X_{j-1} = x_{j-1}) = \`P(X_{j} = x_{j} \mid X_{j-1} = x_{j-1}) = P(x_{j-1}, x_{j}),
% \]

We use $P$ to refer to the corresponding Markov chain. 
For a complex normalized measure  $\nu \in \=C^{\Omega}$ on~$\Omega$,
the measure $\nu P$ obtained via  a one-step transition of the Markov chain from $\nu$ is given by
\[\forall x \in \Omega,\quad (\nu P)(x) = \sum\limits_{y \in \Omega} \nu(y) P(y, x).\]
A complex measure $\pi$  over $\Omega$  is a \emph{stationary measure} of $P$ if $\pi = \pi P$. 
% It is important to note that for a generic complex row-stochastic matrices $P$,  it may not have a stationary measure\footnote{While there is a left-eigenvector with eigenvalue $1$, it can sum up to $0$, and cannot be normalized to a complex measure. This is also the main reason why convergence alone does not imply zero-freeness, as we need to rule out the possibility of converging to an eigenvector that cannot be normalized.}, and even if it does, it may not be unique.
We define convergence next.

\begin{definition}[Convergence of the complex Markov chains]
\label{definition:convergence-MC}
    A Markov chain with a complex-valued transition matrix $P$ and state space $\Omega$ is said to be \emph{convergent} if,
    for any two complex normalized measures $\mu$ and $\mu^*$ over $\Omega$, it holds that
        %\[
    $\lim_{T \to \infty} \norm{\mu P^T - \mu^* P^T}_1 = 0$.
    %\]
\end{definition}

% \textcolor{red}{The definitions all above generalize naturally to time-inhomogeneous chains by allowing the transition matrix $P$ to depend on time (denoted $P_j$ at step $j$).}
% The measurable structure $ (\Omega^T, \+{F}^T) $ and functions $X_1,\dots,X_T$ remain unchanged, as each $P_j$ acts on $(\Omega, \mathcal{F})$. 
% Systematic-scan dynamics with time-dependent transitions are thus implicitly captured by this framework.

\subsubsection{Complex Glauber Dynamics}
 We include the definition of the complex extension of systematic scan Glauber dynamics for complex normalized measures. 
%  We do so from two equivalent viewpoints: extending the transition matrices with complex transition weights, and a complex dynamics-based formulation. The latter is more convenient for our analysis, and the two are equivalent in the sense that they eventually generate the same complex normalized measures.

\begin{definition}[Complex extension of systematic scan Glauber dynamics] 
\label{definition:complex-MCMC-transition-form}
Let $\mu\in \=C^{\*Q_V}$ be a complex normalized measure.
The complex systematic scan Glauber dynamics for the target measure $\mu$ is defined by a sequence of complex-valued transition matrices $P_t \in \=C^{\Omega \times \Omega}$ for $t\ge 1$, 
where with $v = v_{i(t)}$ (and $i(t)$ is as defined in \cref{eq:scan-i(t)}), the transition matrix  $P_t$ is defined as
\[
P_t(\sigma, \tau) \defeq \begin{cases}
    \mu_v^{\sigma(V \backslash \{v\})}(\tau_v) & \text{ if } \forall u \neq v, \sigma_u = \tau_u,\\
    0 & \text{ otherwise.}
\end{cases}
\]

\end{definition}
Starting from an initial state $\tau\in \supp(\mu)$, 
the complex Markov chain generates an induced complex measure $\mu_t\in \=C^{\*Q_V}$ which we define next.   At time $t=0$, we define $\mu_0(\tau) = 1$ and $\mu_0(\sigma) = 0$ for all $\sigma \in \*Q_V\setminus\{\tau\}$, and for $t\ge 1$, we define 
$\mu_{t} \defeq \mu_{t-1} P_t$.

\begin{remark}[Well-definedness of the complex systematic scan Glauber dynamics]\label{remark:MCMC-well-defined}
The complex systematic scan Glauber dynamics, as defined in \Cref{definition:complex-MCMC-transition-form}, 
is well-defined as long as the conditional measures $\mu_{v}^{\sigma(V\setminus \{v\})}$ are well-defined for each $\sigma\in \supp(\mu)$ and every $v\in V$.
Then, the induced complex measures $\mu_t$ remain normalized at any time.
\end{remark}

We include the dynamics-based formulation of the complex systematic scan Glauber dynamics in \Cref{Alg:complex-GD}. Recall that the dynamics have a stationary measure $\mu$, an initial starting state $\tau$, and we denote the associated induced complex measure by $\mu^{\-{GD}}_{T,\tau}$. For technical convenience, we shift the timeline of the dynamics so that we are starting with a state $\sigma_{-T}$ and the final state is $\sigma_0$.

  We remark that \Cref{Alg:complex-GD} (as well as \Cref{Alg:complex-GD-decomposed}, which are introduced later) serves as an analytic tool instead of an efficient algorithm.

\begin{algorithm}
\caption{Complex systematic scan Glauber dynamics} \label{Alg:complex-GD}
\SetKwInOut{Input}{Input}
\Input{An arbitrary initial configuration $\tau\in\supp(\mu)\subseteq \*Q$ and an integer $T\ge 1$;} 
Set $\sigma_{-T}\gets\tau$\;
\For{$t=-T+1,-T+2,\ldots,0$}{
let $\sigma_{t}\gets\sigma_{t-1}$ and $v\gets v_{i(t)}$, where $i(t)=(t \mod \abs{V}) + 1$\;
let $o_t$ follow the conditional measure $\mu_{v}^{\sigma_{t-1}(V\setminus \{v\})}$\;
update $\sigma_t(v)\gets o_t$\;
}
\end{algorithm}

% \begin{remark} 
%   We remark that \Cref{Alg:complex-GD} (as well as \Cref{Alg:complex-GD-decomposed}, which are introduced later) serves as an analytic tool instead of an efficient algorithm.
% % When we state ``let $c$ follow a complex normalized measure $\mu$'' or ``$c$ is drawn from a complex normalized measure $\mu$'', we mean that the measure of $c$ is the same as $\mu$. 
% % This statement is conceptual rather than operational;
% % we do not attempt to explicitly generate samples during runtime. 
% % Any subsequent operation on $c$ should be understood as a transformation applied to the complex normalized measure of $c$. In particular, $\sigma_t(v) \gets c$  means that the marginal measure of $\sigma_t(v)$ is identical to that of $c$.
% %
% %For example, suppose the value of $\sigma$ is updated based on certain rules involving $c$.
% %In that case, the updated measure of $\sigma$ is derived by applying a suitable transformation to the measure of $\sigma$ before update and the measure of~$c$.
% %
% % It is also worth noting that any finite segment of the complex measures is computable on a deterministic Turing machine in exponential time,
% % provided all the involved measures can be described using Gaussian rational numbers. 
% % This can be achieved by explicitly enumerating all outcomes of the process.
%     \end{remark}

\begin{remark} \label{remark:remark-equivalence-matrix-dynamics}
It is straightforward to verify that the processes described in \Cref{definition:complex-MCMC-transition-form} and \Cref{Alg:complex-GD} are essentially equivalent in the following sense: for any $\tau^* \in \supp(\mu)$,  we have $\mu^\-{GD}_{T, \tau}(\sigma_t = \tau^*) = \mu_{t+T}(\tau^*)$ for all $-T \le t \le 0$.
This can be routinely verified through induction on $t$.
    \end{remark}

% Generally, we lack convergence theorems for Markov chains with complex-valued transition matrices in the literature,  
% so we cannot assert whether the complex measure after $T$ steps converges to the target measure $\mu$ as $T\to\infty$.
% However, we can consider a complex process initialized with the stationary measure $\mu$. 

We also include the definition of the stationary systematic scan Glauber dynamics. 

\begin{definition}[Stationary systematic scan Glauber dynamics]\label{definition:stationary-start}
Consider the process defined in \Cref{Alg:complex-GD}, 
but now with the initial state $\sigma_{-T}$ following the measure $\mu$. 
We call this modified process the $T$-step stationary systematic scan Glauber dynamics,
and denote its induced measure as $\mu^{\-{GD}}_{T}$.
\end{definition}

It is straightforward to verify that for all $t\in [-T,0]$, the measure induced on $\sigma_{t}$  under $\mu^{\-{GD}}_{T}$ precisely follows  the measure $\mu$.
% \Cref{definition:stationary-start} will play an essential intermediate role in our proof of the convergence of Glauber dynamics. Note that the measure $\mu^{\-{GD}}_{T}$ is a linear combination of the measures $\mu^{\-{GD}}_{T,\sigma}$ over all starting states $\sigma\in \supp(\mu)$.  
As it lacks general convergence theorems for Markov chains with complex-valued transition matrices, we include a sufficient condition (\Cref{condition:convergence}) and prove (in \Cref{lemma:convergence}) by comparing with the stationary process in \Cref{definition:stationary-start}. And under this condition, the Glauber dynamics starting from any initial state converges to a unique limiting measure, which is precisely the stationary measure $\mu$.
\subsubsection{Decomposition scheme}
% In this , we provide our framework for analyzing the complex Markov chains. 
% We extend their decomposition method by allowing more sophisticated decompositions. 
% Then we provide necessary conditions for the convergence.

%We now introduce the decomposition scheme. 
We also include the decomposition scheme in~\cite[Definition 3.7]{liu2024phase}.
We decompose each transition of a complex systematic scan Glauber dynamics into two steps: an oblivious step, where this step does not depend on the current configuration; and an adaptive step, which tries to make up to the right transition measure.

For $t\in \=Z_{\le 0}$, let $i(t)$ be the index of the vertex updated at time $t$, and let $v = v_{i(t)}$. 
% We decompose the transition at time $t$ into two steps.
% We decompose each transition into two steps.
% The first step is oblivious to the configurations of its neighbors. 
To describe the first oblivious step, we associate the vertex $v$ with a complex normalized measure $b_v$ whose support set is $Q_v\cup \set{\bot}$. 
And the choice of $b_v$ and $S_v$ depends on specific models.
In the first oblivious step, let $r_t \in Q_v\cup\set{\bot}$ be a random variable that follows the complex normalized measure $b_v$. 

The second step is an adaptive step depending on $r_t$ sampled in the first step and the current configuration. 
For $t\in\=Z_{\le t}$, let $\sigma_t$ be the configuration at time $t$ of the complex Markov chain.
% Recall we use $\sigma_{t - 1}$ to denote the configuration at time $t-1$ of a complex systematic scan Glauber dynamics.
% Let $v = i(t)$ and $\+N(v)$ be the neighbors of $v$.
Let $v = v_{i(t)}$.
% For any $r_t\in S_v$ and any feasible $\sigma_{t-1}\in Q^{\+N(v)}$(meaning that $\sigma_{t-1}^{\+N(v)}$ can be extended to a configuration $\sigma \in \supp(\mu)$), let $\Phi^v_{r_t, \sigma_{t-1}^{\+N(v)}}$ be a complex normalized measure whose support set is $Q$. And its choice also depends on the concrete models.
% We use $\Phi_{r_t, \sigma_{t-1}^{\+N(v)}}^v$ to describe the second step. 
% We sample a random element in $Q$ following the complex normalized measure $\Phi^v_{r_t, \sigma_{t-1}^{\+N(v)}}$ and denote it as $o_t$.

For any $r_t\in Q_v\cup\set{\bot}$ and any admissible $\sigma_{t-1}$, let $\tau$ be the configuration of $\sigma_{t-1}$ projected on $V\setminus \{v\}$.
If $r_t \neq \bot$, we set $\forall u\in V\setminus\{v\}$, $\sigma_t(u) = \sigma_{t-1}(u)$ and $\sigma_t(v) = r_t$ directly.
Otherwise, let $\mu_v^{\tau, \bot}: Q_v\to \=C$ be a function, and if $b_v(r_t)\neq 0$, it further holds that $\mu_v^{\tau, \bot}$ is a normalized complex measure on $Q_v$. 
% And its choice also depends on the concrete models.
If $b_v(r_t)\neq 0$, we use $\mu_v^{\tau, \bot}$ to describe the second adaptive step. 
Let $o_t\in Q_v$ be a random variable that follows the complex normalized measure $\mu_v^{\tau, \bot}$.
Then we set $\forall u\in V\setminus\{v\}$, $\sigma_t(u) = \sigma_{t-1}(u)$ and $\sigma_t(v) = o_t$.
To make $\sigma_t(v)$ follow the correct marginal measure, we further restrict that 
\[
\forall c\in Q_v, \quad \mu^\tau_v(c) = b_v(c) + b_v(\bot)\cdot \mu_v^{\tau, \bot}(c).
\]

% \yxtodo{I think it is better to choose another work $\+F$.}
For any subset $C\subseteq V$, let $\+F_C$ be the set of all \emph{extendable} configurations (can be extended to a feasible configuration) on the set $C$.
By selecting $(b_v)_{t\in V}$ and $(\mu_v^{\tau, \bot})_{v\in V, \tau \in \+F_{V\setminus\{v\}}}$ properly, it holds that $o_t$ generated by the above two steps follows the transition measure of our complex Markov chain.

The above intuition leads to the following definition.

\begin{definition}[{Decomposition scheme ~\cite[Definition 3.7]{liu2024phase}}]\label{definition:decomposition-scheme}
Let $\mu\in \mathbb{C}^{\*Q}$ be a complex normalized measure. 
For each $v\in V$, we associate a complex normalized measure $b_v: Q_v\cup\set{\bot} \to \=C$, and let $\bm{b}=(b_v)_{v\in V}$.
%  and a series of complex functions $\adaptiveTransition{v}{}$'s, i.e., for any $A\in S_v$ and for any $\tau \in \+F_{V\setminus \{v\}}$, $\adaptiveTransition{v}{A, \tau}: Q_v\to \=C$. 
% Let $\*S = \{S_v\}_{v\in V}$, $\*b = \{b_v\}_{v\in V}$,
% and let $\*\Psi = \{\adaptiveTransition{v}{A,\tau}\}_{v\in A, A\in S_v, \tau \in \+F_{V\setminus\{v\}}}$.
We define the $\*b$-decomposition scheme on $\mu$ as follows. For each $v\in V$ and each extendable $\tau\in \+F_{V\setminus \set{v}}$, we define the measure $\mu_v^{\tau, \bot}$ as
\begin{equation}
\label{eq:decomposition-adatptive-measure}
\forall c\in Q_v,\quad \mu_v^{\tau, \bot}(c) \defeq \frac{\mu_v^{\tau}(c)-b_v(c)}{b_v(\bot)}.
\end{equation}
Then, the marginal measure $\mu_v^\tau$ can be decomposed as:
\begin{equation}\label{eq:decomposition}
\forall c\in Q_v,\quad \mu_v^{\tau}(c) = b_v(c) + b_v(\bot)\cdot \mu_v^{\tau, \bot}(c),
\end{equation}
% It further holds that for any $A\in S_v$ and any $\tau \in \+F_{V\setminus \{v\}}$, if $b_v(A)\neq 0$, then
% $\adaptiveTransition{v}{A, \tau}: Q_v\to\=C$ is a complex normalized measure. 
where we assume the convention $0\cdot \infty =0$ to ensure that \cref{eq:decomposition} still holds when $b_v(\bot)=0$.
% For each $v\in V$, we associate a complex normalized measure $b_v:Q\cup \{\bot\}\to \mathbb{C}$, and let $\bm{b}=(b_v)_{v\in V}$. 
% The \decompositionScheme-decomposition scheme on $\mu$ holds that: for each $v\in V$ and each feasible $\tau\in \+F_{V\setminus\{v\}}$, 
% the marginal measure $\mu^{\tau}_v$ can be decomposed as:
% \begin{equation}\label{eq:decomposition}
% \forall c \in Q_v, \quad \mu^{\tau}_v(c)=\sum_{A\in S_v}b_v(A)\cdot \adaptiveTransition{v}{A,\tau}(c).
% \end{equation}
% where we assume the convention $0\cdot \infty =0$ to ensure that \cref{eq:decomposition} still holds when $b_v(\bot)=0$.
\end{definition}

% \begin{remark}
% \label{remark:b-decomposition-scheme}
%   Comparing with the $\*b$-decomposition scheme defined in~\cite{liu2024phase}, our decomposition scheme allows more refined decompositions. 
%   Given a $\*b$-decomposition scheme, we recover their definition by setting  $S_v\defeq Q_v\cup\{\bot\}$. And for any feasible configuration $\tau \in \+F_{V\setminus \{v\}}$ on $V\setminus \{v\}$, let
%   \[
%   \forall c\in Q_v,\quad \adaptiveTransition{v}{c, \tau}(c) \defeq 1,\quad \forall a\neq c,\quad \adaptiveTransition{v}{c, \tau}(a)\defeq 0,
%   \]
%   and
%   \[
%   \forall a\in Q_v,\quad \adaptiveTransition{v}{\bot,\tau}(a)\defeq \frac{\mu^{\tau}_v(a) - b_v(a)}{b_v(\bot)},
%   \]
%   where we assume the convention $0\cdot \infty =0$ to ensure that this still holds when $b_v(\bot)=0$.
% \end{remark}

Given a $\bm{b}$-decomposition scheme, the complex systematic scan Glauber dynamics can be reinterpreted as the process in \Cref{Alg:complex-GD-decomposed}. For $T\ge 1$ and $\tau \in \supp(\mu)$, let $\dmu{\tau}$ be its induced measure. 
By \Cref{definition:decomposition-scheme}, it holds that the Glauber dynamics is equivalent to the decomposed Glauber dynamics.
% \textcolor{red}{(here I need to argue the equivalence of the Glauber and the decomposed Glauber)}.

\begin{algorithm}
\caption{$\*b$-decomposed complex systematic scan Glauber dynamics} \label{Alg:complex-GD-decomposed}
\SetKwInOut{Input}{Input}
\Input{ An arbitrary initial configuration $\tau\in\supp(\mu)\subseteq \*Q$ and an integer $T\ge 1$;} 
Set $\sigma_{-T}\gets\tau$\;
\For{$t=-T+1,-T+2,\ldots,0$}{
let $\sigma_{t}\gets\sigma_{t-1}$ and $v\gets v_{i(t)}$, where $i(t)=(t \mod \abs{V}) + 1$\;
let $r_t \in Q_v\cup\set{\bot}$ follow the complex normalized measure $b_{v}$\;
if $r_t\neq \bot$, let $o_t\gets r_t$, otherwise let $o_t$ follow the complex normalized measure $\mu_v^{\tau, \bot}$ , where $\tau$ is the configuration of $\sigma_{t-1}$ projected on $V\setminus \{v\}$\;
% \eIf{$r_t\neq \bot$}{let $o_t\gets r_t$\;}{
%       let $o_t$ follow the measure $ \mu_{v}^{\sigma_t,\bot}$\;}
update $\sigma_t(v)\gets o_t$\;
}
\end{algorithm}

As shown in~\cite{liu2024phase}, by choosing $\*b$ properly, we can infer the outcome of a complex Glauber dynamics $\sigma_0$ without knowing the initial configuration $\sigma_{-T}$, as $T\to\infty$, so that we can show the convergence of the complex Glauber dynamics and bound the norm of the complex marginal measure, which we will use to prove the zero-freeness.
% \begin{itemize}
%   \item Infer the outcome $\sigma_0$ without knowing the initial configuration $\sigma_{-T}$ as $T\to\infty$, showing the convergence of the complex Glauber dynamics.
%   \item Infer the marginal measure of $\sigma_0$, 
% \end{itemize}

\subsubsection{Convergence of systematic scan Glauber dynamics}
In this subsubsection, we outline the idea to show the convergence of the systematic scan Glauber dynamics in~\cite{liu2024phase}.

Recall in the complex systematic scan Glauber dynamics, for $t\in \{-T+1,-T+2, \dots, 0\}$, we use $i(t)$ to denote the index of the vertex updated at time $t$.
Given a $\*b$-decomposed systematic scan Glauber dynamics, let $\*\rho$ be the realization of $\*r = \{r_t\}_{t = -T + 1}^0$ where $r_t\in Q_{v_{i(t)}}\cup\set{\bot}$, as used in \Cref{Alg:complex-GD-decomposed}. We call $\*\rho$ a \emph{witness sequence} if it holds that given $\*\rho$, the outcome $\sigma_0$ does not depend on the initial state.

\begin{definition}[Witness sequence]
\label{definition:witness-sequence}
Fix $T\ge 1$. Consider a $\*b$-decomposed systematic scan Glauber dynamics as in \Cref{Alg:complex-GD-decomposed}. For any event $A\subseteq \*Q$, we say a sequence $\*\rho = \{\rho_t\}_{t = -T + 1}^0$ where $\rho_t \in Q_{v_{i(t)}}\cup\set{\bot}$ is a \emph{witness sequence} with respect to the event $A$ if it satisfies exactly one of the following conditions:
\begin{enumerate}
\item the event $\*r = \*\rho$ has zero measure (note this event is independent of the initial state), i.e.,
\[
  \dmu{\sigma}(\*r = \*\rho) = 0,
\]
\item the event $\*r = \*\rho$ has non-zero measure and
 \[
\forall \sigma, \tau \in \supp(\mu),\quad \dmu{\sigma}(\sigma_0 \in A\mid \*r = \*\rho) = \dmu{\tau}(\sigma_0\in A\mid \*r = \*\rho).
\]
\end{enumerate}
When $\*\rho$ is a \emph{witness sequence} with respect to the event $A$, we denote it as $\*\rho \Rightarrow A$. Otherwise, we denote $\*\rho \nRightarrow A$.
\end{definition}
% \color{red}
Instead of showing the convergence of a complex systematic scan Glauber dynamics, we only interested in the convergence on the marginals. For any event $A\subseteq \bm{Q}$, it suffices to show that the contributions from witness sequences dominate while contributions from non-witness sequences diminish to $0$ as $T\to \infty$. The intuition is that only contributions from non-witness sequences depend on the initial state. The following definition formalizes the above intuition.
% \color{black}
% To show the convergence of a complex systematic scan Glauber dynamics, it suffices to show that for any event $A\subseteq \*Q$, contributions from witness sequences dominate and contributions from non-witness sequences diminish to $0$ as $T\to\infty$. 
% The intuition is that only contributions from non-witness sequences depend on the initial state. The following condition formalizes the above intuition.

\begin{condition}[Sufficient condition for convergence]
  \label{condition:convergence}
Let $A\subseteq \bm{Q}$ be an event.
Assuming the $\*b$-decomposed systematic scan Glauber dynamics in \Cref{Alg:complex-GD-decomposed} is well-defined, there exists a sequence of sets $\{B(T)\}_{T\ge 1}$ such that for each $T\ge 1$, $B(T) \subseteq \bigotimes_{t=-T+1}^0 \tp{Q_{v_{i(t)}}\cup\set{\bot}}$ satisfying the following conditions:
\begin{itemize}
  \item For all $\*\rho \nRightarrow A$, it holds that $\*\rho \in B(T)$; thus, $B(T)$ contains all non-witness sequences for $A$.
  \item For any initial configuration $\sigma\in \supp(\mu)$, the following limits exits and satisfies:
  \[
% \lim_{T\to \infty}\abs{\sum_{\*\rho \in B(T)} \dmu{\sigma}(\*r = \*\rho)\cdot \dmu{\sigma}(\sigma_0 = \tau \mid \*r = \*\rho) } = 0.
\lim_{T\to \infty}\abs{\sum_{\*\rho \in B(T)}  \dmu{\sigma}(\sigma_0 \in A \land \*r = \*\rho) } = 0.
  \]
\end{itemize} 
\end{condition}

The above condition is a sufficient condition for the convergence of the event $A$; we also include the following lemma which is implicit in~\cite{liu2024phase}. And we include its proof in \Cref{section:missing-analytic-percolation} for completeness.

\begin{lemma}[Convergence of complex systematic scan Glauber dynamics]
  \label{lemma:convergence}
  Let $A\subseteq \bm{Q}$ be an event.
  Assume that there exists a $\*b$-decomposition scheme such that \Cref{condition:convergence} holds. Then, it holds that for any $\tau \in \supp(\mu)$, $\mu(A) = \lim_{T\to \infty} \dmu{\tau} \tp{\sigma_0 \in A}$.
  % the complex systematic scan Glauber dynamics converges to $\mu$ as $T\to\infty$, starting from any initial configuration $\sigma\in\supp(\mu)$.
\end{lemma}
% \todo{yyx: modify the proof of this lemma.}

% Given the definition of the witness sequence, all realizations of $\*r$ can be divided into 

% We define the \emph{witness sequence} 

% In this section, we show a necessary condition for convergence of the systematic scan Glauber dynamics. 
% Note same ideas appeared in 
% Now, we provide a 

\begin{comment}
{
\color{olive}
Here we consider the matrix form of the decomposition scheme. 
Let $\mu$ be the Gibbs distribution of a graphical model on $G=(V, E)$ and $\Omega$ be its support set. Let $v\in V$ be a vertex in $G$ and $P_v\in\=C^{\Omega\times \Omega}$ be the complex transition matrix of our complex Glauber dynamics when updating vertex $v$.

We define two complex matrix  $D_v$ and $P^{\bot}_v$ to decompose $P_v$. For $\sigma, \tau \in \Omega$, if $\sigma$ and $\tau$ are the same or differ only at the vertex $v$, we let 
\[
D_{v}(\sigma,\tau)\defeq b_v(\tau_v),\quad P^{\bot}_v(\sigma, \tau) \defeq \mu_v^\sigma(\tau_v) - b_v(\tau_v),
\]
otherwise, we set $D_{v}(\sigma,\tau) = P^{\bot}_v(\sigma, \tau) \defeq 0$.

We have that 
\[
P_v = D_v + P^{\bot}_v.
\]
}
\end{comment}

\subsubsection{Bounding the marginal measures}
Next, we outline how to bound the marginal measure of an event $A\subseteq \*Q$ assuming \Cref{condition:convergence} in \cite{liu2024phase}. 
Later, we use the edge-wise self-reducibility to reduce the zero-freeness to bounding the norm of a marginal measure.
% We now explain how to bound the marginal measure of an event $A\subseteq\*Q_V$. 
For any $T\ge 1$, we consider the stationary complex systematic scan Glauber dynamics (\Cref{definition:stationary-start}). 
For an event $A$, we consider enumerating all possible sequences $\*\rho$'s as realizations of $\*r$ in \Cref{Alg:complex-GD-decomposed} and consider the measure of $\sigma_0 \in A$ conditioned on $\*r = \*\rho$.

By \Cref{condition:convergence}, we consider the following two cases: $\*\rho \in B(T)$; and $\*\rho \notin B(T)$. For $\*\rho \in B(T)$, by \Cref{condition:convergence}, it holds that its contributions diminish to $0$ when $T\to\infty$. So in order to bound $\mu(A)$, it suffices to bound the contributions from $\*\rho \notin B(T)$.
To formalize this argument, we include the following lemma and we also include its proof in \Cref{section:missing-analytic-percolation}.

% For any $T \ge 1$. Let $\+S$ be the set of sequences $\*\rho = (\rho_i)_{i=-T+1}^0$ where each $\rho_i \in Q\cup \set{\bot}$. Note that \Cref{condition:convergence} immediately implies that there is a set $B(T)\subseteq \tp{Q\cup \set{\bot}}^T$ such that for any event $A \subseteq \*Q_V$, and for any $\*\rho \not\in B(T)$, it holds that $\*\rho \Rightarrow A$.  Consider the stationary systematic scan Glauber dynamics (\Cref{definition:stationary-start}), for any event $A\subseteq \*Q_V$, by the triangle inequality,

\begin{lemma}[Bounding the marginal measure]
\label{lemma:bounding-marginal-measure}
Let $A\subseteq \bm{Q}$ be an event.
    Given a \decompositionScheme-decomposition scheme satisfying \Cref{condition:convergence}. For 
    any initial configuration $\tau \in \supp(\mu)$, it holds that

\begin{equation*}
% \abs{\mu(A)} \le  \lim_{T\to\infty} \abs{\sum\limits_{\rho \not\in B(T)} \dmu{\tau}(\bm{r}=\bm{\rho}) \cdot \dmu{\tau}(\sigma_0 \in A \mid  \*r = \*\rho)}.
\abs{\mu(A)} \le  \lim_{T\to\infty} \abs{\sum\limits_{\rho \not\in B(T)} \dmu{\tau}(\sigma_0 \in A \land  \*r = \*\rho)}.
\end{equation*}

\end{lemma}
% Note that the measure of $\*\rho$ follows a product measure regardless of the initial state.
% So it suffices to design witness sequences so that $\dmu{\tau}(\sigma_0\in A\mid \*r = \*\rho)$ is easy to analyze. 
% Later, in the following sections of this paper, we design different witness sequences for concrete models.
% And for bounding the norm of a marginal measure which we leave them for the following sections.

% This, in turn, enables us to establish zero-freeness results using edge-wise self-reducibility.
% It remains to demonstrate how to establish \Cref{condition:convergence} and to upper bound the right-hand side in~\cref{eq:norm-bound}. 
% Note that for any sequence $\bm{\rho}\in (Q\cup \{\bot\})^{T}$ and any initial configuration $\tau\in\-{supp}(\mu)$, the measure $\mu^{\-{GD}}_{T,\tau,\bm{b}}(\bm{r}=\bm{\rho})$ can be computed directly, as $\bm{r}$ follows a product measure. 
% The primary technical challenge then lies in characterizing the bound on the measure $\mu^{\-{GD}}_{T,\tau,\bm{b}}(\sigma_0\in A\mid \bm{r}=\bm{\rho})$ through useful properties of witness sequences~$\bm{\rho}$.
% However, this characterization may depend on the concrete models. 
% Therefore, we do not aim to provide a generic method for establishing \Cref{condition:convergence} and bounding the marginal measure through \cref{eq:norm-bound}. 
% Instead, we will show in the following section how to apply this general framework to the hypergraph independence polynomials, yielding the desired zero-free and convergence results.

\section{Zero-freeness for CSP formulas}
\label{section:CSP}
In this section, we apply our decomposition scheme (\Cref{definition:decomposition-scheme}) to CSP formulas and prove a sufficient condition for zero-freeness (\Cref{theorem:sufficient-condition-zero-freeness}).
We remark that we work with general CSPs, but \Cref{condition:csp-conditional-measure-analysis} implicitly transforms any general CSPs into atomic CSPs.
To establish zero-freeness, we start with a constraint-wise self-reduction, which reduces zero-freeness to analyzing the norm of a complex marginal measure.
Then we apply our \emph{projection-lifting scheme} (\Cref{subsection:projection-lifting-scheme}) to analyze a projected measure constructed by the standard \emph{state-compression} scheme.
The \emph{state-compression} scheme and its special case ``mark/unmark'' method are widely used for sampling satisfying assignments of CSP formulas~\cite{moitra2019approximate, guo2018counting,  feng2021fast, feng2021sampling, jain2021sampling, Vishesh21towards, he2021perfect, galanis2022fast, he2022sampling, he2023deterministic, he2023improved,chen2023algorithms}.
For the real case, the projected distribution has many good properties such as the \emph{local uniformity} which can be used to prove the rapid mixing of Markov chains on the projected distribution and construct local samplers.
For our complex case, we use our decomposition scheme (\Cref{definition:decomposition-scheme}) to analyze the complex systematic scan Glauber dynamics on the projected measure.
Then we lift the analysis for the projected measure back to the original measure and bound the norm of a complex marginal measure.

We start by defining the framework of the \emph{state-compression} scheme.
\begin{definition}[State compression]
\label{definition:state-compression}
For a CSP formula $\Phi = (V, \*Q, C)$, let $\*f = (f_v)_{v\in V}$. And for each variable $v\in V$, $f_v$ is a mapping from its domain $Q_v$ to a finite \emph{alphabet} $\Sigma_v$.
For any $\Lambda\subseteq V$, we define $\Sigma_{\Lambda}=\bigotimes_{v\in \Lambda}\Sigma_v$.
For simplicity, we assume that $\bigcup_{v\in V}Q_v$ and $\bigcup_{v\in V} \Sigma_v$ are disjoint.
\end{definition}

We use $\*Q$ to denote the collection of the original alphabets and $\*\Sigma$ for the projected alphabets.  
Throughout the rest of this paper, these notations will be used consistently without further reminder.  
Later, in \Cref{section:applications}, we will use the symbol $\blacktriangle$ to indicate the projected symbol in the projected alphabet $\*\Sigma$.

Next, we define the \emph{consistency}. 
Let $\Lambda'\subseteq \Lambda\subseteq V$. For $\sigma\in \bigotimes_{v\in \Lambda'}\tp{Q_v\cup \Sigma_v} $ and $ \tau \in \bigotimes_{v\in \Lambda}\tp{Q_v\cup \Sigma_v}$, we say $\sigma$ is \emph{consistent} with $\tau$ and denote it as $\tau \trianglelefteq \sigma$ if $\forall v\in \Lambda'$, we have $\tau_v = \sigma_v$ or $\tau_v = f_v(\sigma_v)$.

For any constraint $c\in \+C$, and any projected partial assignment $\tau\in \*\Sigma_{S}$ on the variables $S$, we say $c$ is \emph{satisfied by the projected partial assignment} $\tau$ iff for any $\sigma\in \*Q_{S}$ with $\tau \trianglelefteq \sigma$, it holds that $c$ is satisfied by $\sigma$.

Let $1^\coF \in \bigcup_{v\in V}\Sigma_v$, and let $\lambda \in \=C$ be a complex number. 
Our framework allows us to add a complex external field $\lambda$ for each variable $v$ with the projected symbol $1^\coF$.

\begin{definition}[Complex extensions of CSP formulas]
\label{definition:complex-extensions-csp}
Given the projection $\*f$ of the state-compression scheme, let $1^\coF \in \bigcup_{v\in V}\Sigma_v$, and let $\lambda \in \=C$ be a complex number (complex external field).  
We define the weight function $w = w(\Phi, \lambda, \*f, 1^\coF)$ as follows:
\[
\forall \sigma\in\Omega_{\Phi}, \quad w(\sigma) \defeq \lambda^{|\{v \in V\mid f_v(\sigma_v) = 1^\coF\}|}.
\] 
We also denote $Z = Z(\Phi, \lambda,  \*f, 1^\coF)$ as the sum of all assignments' weights, i.e.  $Z\defeq \sum_{\sigma\in \Omega_{\Phi}} w(\sigma)$.
When $Z\neq 0$, we naturally define a complex measure $\mu = \mu(\Phi, \lambda, \*f, 1^\coF)$, i.e.,$\forall \sigma\in\Omega_{\Phi}, \mu(\sigma)\defeq \frac{w(\sigma)}{Z}$.
And we denote its support set as $\supp(\mu) \defeq \{\sigma\in \Omega_{\Phi}\mid \mu(\sigma)\neq 0\}$.
\end{definition}
We now define the marginal measure.
For any $\Lambda \subseteq V$ and any $\sigma\in \bigotimes_{v\in \Lambda} \tp{Q_v\cup \Sigma_v}$, let $\mu(\sigma)$ be defined as 
\[
\mu(\sigma)\defeq \sum_{\tau\in\Omega_{\Phi}: \sigma\trianglelefteq \tau} \mu(\tau).
\]

We also consider the \emph{projected measure} under the projection $\*f$, which plays an important role in our framework for deriving zero-freeness. 
Let $\Omega_{\Phi}^{\*f}$ be the set of satisfying assignments under the projection $\*f$, i.e., $\Omega_{\Phi}^{\*f} = \{\*f(\sigma)\mid \sigma\in \Omega_{\Phi}\}$. %\todo{yyx: is this used?}
\begin{definition}[Projected measure]
\label{definition:csp-projected-measure-under-f}
Suppose $Z\neq 0$, let $\psi = \psi(\Phi, \lambda, \*f, 1^\coF)$ be a complex normalized measure, i.e.,
\[
\forall \sigma\in \Omega_{\Phi}^{\*f},\quad \psi(\sigma)\defeq \sum_{\tau \in \Omega_{\Phi}: \*f(\tau) = \sigma} \mu(\tau).
\]
And we denote $\supp(\psi)$ as the support set of $\psi$, i.e., $\supp(\psi) = \{\sigma\in \Omega_{\Phi}^{\*f} \mid \psi(\sigma)\neq 0\}$.
\end{definition}

For any $\Lambda \subseteq V$ and for any $\sigma\in \Sigma_{\Lambda}$, we define $\psi(\sigma)$ as
\[
\psi(\sigma)\defeq \sum_{\tau \in \Omega^{\*f}_{\Phi}: \sigma\trianglelefteq \tau} \psi(\tau).
\]
We also remark that for any $\Lambda \subseteq V$, and for any $\sigma\in \Sigma_{\Lambda}$, the conditional measure $\psi^{\sigma}$ is well-defined if and only if $\psi(\sigma) \neq 0$.
For the consistency of notations, we use $\mu$ to denote the original measure and $\psi$ for the projected measure throughout the rest of this paper.  
% Throughout the rest of this paper, these notations will be used consistently without further reminder.  

Later, we establish some properties of the projected measure $\psi$ which we will use to show the zero-freeness of the original measure $\mu$.
And we provide a sufficient condition that can lift the properties of $\psi$ back to $\mu$.

Here we do not aim to provide a zero-free condition for general CSP formulas. 
We only focus on $(k, \Delta)$-CSP formulas and provide a sufficient condition for our zero-freeness (\Cref{theorem:sufficient-condition-zero-freeness}).
We first introduce some useful notations.
Recall that for any constraint $c\in C$, we use $\vbl(c)$ to denote the variables that $c$ depends on and we use $c^{-1}(\-{False})$ to denote the set of partial assignments on $\vbl(c)$ that violate the constraint $c$.
For any assignment $\sigma\in\bm{\Sigma}$ under the projection, and any constraint $c$, recall that we say $c$ is \emph{satisfied} by $\sigma$ iff for any violating assignment $\tau\in c^{-1}\tp{\false}$, there exists $v\in \vbl(c)$ such that $f_v(\tau_v) \neq \sigma_v$. This means that any original assignment whose projected symbol is $\sigma$ satisfies the constraint $c$.
Let $\Lambda \subseteq V$. 
For an assignment $\sigma \in \*\Sigma_\Lambda$ under the projection, we say $\sigma$ is an \emph{admissible} partial assignment under the projection, if its projected measure is not $0$, i.e., $\psi(\sigma)\neq 0$.
And we say $\sigma$ is an \emph{extendable} partial assignment under the projection, if it can be extended into an admissible projected assignment. 

In this section, we give a sufficient condition for establishing zero-freeness (\Cref{theorem:sufficient-condition-zero-freeness}). 
Our proof is based on the constraint-wise self-reduction.
Let $\Phi = (V, \*Q, C)$ be a $(k, \Delta)$-CSP formula, and let $C = \set{c_1, c_2, \dots, c_m}$.
For any $0\le i \le m$, we define $C_i$ as the subset of $C$ containing the first $i$ constraints, i.e., $C_i = \set{c_1, c_2, \dots, c_{i}}$. Let $\Phi_i = (V, \*Q, C_i)$ be a new CSP formula, and let $Z_i = Z(\Phi_i, \lambda, \*f, 1^\coF)$ be its corresponding partition function.

The base case of the induction is $Z_0\neq 0$, which can be verified immediately, as there are no constraints. 

The induction hypothesis is that for an integer $1\le i\le m$, we assume that $Z_{i-1} \neq 0$. 

For the induction step, we show $Z_i \neq 0$ by showing $\frac{Z_i}{Z_{i-1}}\neq 0$. Next, based on the induction hypothesis $Z_{i-1}\neq 0$, let $\mu_{i-1}$ be the corresponding complex Gibbs measure, and let $\psi_{i-1}$ be the projected measure, we express $\frac{Z_i}{Z_{i-1}}$ as a complex marginal of $\mu_{i-1}$:
\begin{equation}\label{eq:partition-function-divide-to-maginal-measure}
\frac{Z_i}{Z_{i-1}} = \frac{Z_{i-1} - \sum\limits_{\sigma\in\Omega_{\Phi_{i-1}}: \sigma \text{ violates } c_i} w(\sigma)}{Z_{i-1}} = 1 - \mu_{i-1}(c_i \text{ is violated}).
\end{equation}
% (formalized in \eqref{eq:partition-function-divide-to-marginal-measure}). 
Hence, we use the projection-lifting scheme (\Cref{subsection:projection-lifting-scheme}) to analyze the projected measure $\psi_{i-1}$ and show $|\mu_{i-1}(c_i \text{ is violated})| < 1$. So that $|\frac{Z_i}{Z_{i-1}}| > 0$ and $Z_i \neq 0$.
% The induction will be formalized in \Cref{subsection:csp-zero-freeness}.
The next condition contains two key properties for analyzing the projected measure through analytic percolation and completing the induction step.

For any subset $S\subseteq V$, let $\+F_{S}$ be the set of extendable (partial) assignments on $S$ under the projection.
% The \Cref{condition:csp-conditional-measure-analysis} in the next lemma will be introduced later.
\begin{condition}[Induction step]\label{condition:csp-conditional-measure-analysis}
Let $\Phi = (V, \*Q, C)$ be a $(k, \Delta)$-CSP formula with $|C|\ge 1$.
Let $\*f$ be a state-compression scheme (\Cref{definition:state-compression}).
Let $c^* \in C$ be a constraint with the largest index number, and let $\Phi' = (V, \*Q, C\setminus \set{c^*})$.
% For any $1\le i\le m$, 
If $Z(\Phi', \lambda, \*f, 1^\coF)\neq 0$, let $\psi = \psi(\Phi', \lambda, \*f, 1^\coF)$ be its projected measure, it holds that
% Suppose the projected measure $\psi$ induced by $\*f$ is well-defined, that is, $Z(\Phi, \lambda, \*f, 1^\coF)\neq 0$. Then,
\begin{enumerate}
\item (Well-definedness of complex transitions) for any variable $u$ and for any extendable partial assignment $\tau\in \*\Sigma_{V\setminus \{u\}}$ on $V\setminus \{u\}$ under the projection $\*f$, the marginal $\psi^{\tau}_{u}$ on $u$ is well-defined.\label{item:csp-transition-well-define}

\item (Nice decomposition scheme)
There exists a \decompositionScheme-decomposition scheme.
Define
    \begin{align}
        \conditionalDecayConstant \defeq& \max_{c\in C_i} \sum_{\sigma \in c^{-1}(\-{False})} \prod_{u\in \vbl(c)}\tp{|b_u(f_u(\sigma_u))| + |b_{u}(\bot)|\cdot \max_{\tau\in \+F_{V\setminus\{u\}}}\sum_{x\in \Sigma_u} |\psi_{u}^{\tau,\bot}(x)| };\label{item:csp-conditional-meaure-analysis-decay}\\
         \conditionalTriangleConstant \defeq& \max_{u\in V} \tp{\sum_{x\in \Sigma_u} |b_u(x)| + |b_{u}(\bot)|\cdot \max_{\tau\in \+F_{V\setminus\{u\}}}\sum_{x\in \Sigma_u} |\psi_{u}^{\tau,\bot}(x)|}.\label{item:csp-conditional-meaure-analysis-triangle}
    \end{align}
It holds that:
\begin{align}
    4\mathrm{e}\Delta^2 k^4 \cdot \conditionalDecayConstant\cdot \conditionalTriangleConstant^{4\Delta^2 k^5} \le \frac{1}{4}.
    \label{item:csp-conditional-meaure-analysis-numerical}
\end{align}
\label{item:csp-decomposition-decay}
\end{enumerate}
\end{condition}

The next lemma establishes the induction step.

\begin{lemma}[Induction step]
\label{lemma:csp-induction-step}
Let $\Phi = (V, \*Q, C)$ be a $(k, \Delta)$-CSP formula with $|C|\ge 1$. 
Let $c^* \in C$ be a constraint with the largest index number, and let $\Phi' = (V, \*Q, C\setminus \set{c^*})$.

% For any $1\le i\le m$, 
If $Z = Z(\Phi', \lambda, \*f, 1^\coF)\neq 0$ and \Cref{condition:csp-conditional-measure-analysis} hold, let $\mu = \mu(\Phi', \lambda,\*f, 1^\coF)$ be its complex Gibbs measure, then it holds that 
$|\mu(c^*\text{ is violated})| < 1$.
\end{lemma}

Assuming the base case of the induction $Z_0 \neq 0$, the above lemma is the induction step for establishing the following zero-freeness.
Recall that the constraint set is $C = \set{c_1, c_2, \dots, c_m}$. Also recall that for any $0\le i \le m$, we set $C_i = \set{c_1, c_2, \dots, c_i}$ containing the first $i$ constraints, $\Phi_i = (V, \*Q, C_i)$ and $Z_i = Z(\Phi_i, \lambda, \*f, 1^\coF)$.
\begin{theorem}[Zero-freeness]
\label{theorem:sufficient-condition-zero-freeness}
Let $\Phi = (V, \*Q, C)$ be a $(k, \Delta)$-CSP formula. Let $\bm{f}$ be the projection in the state-compression scheme (\Cref{definition:state-compression}), and the complex external field $\lambda$ on the projected symbol $1^\coF$ (\Cref{definition:complex-extensions-csp}).
Suppose
\begin{equation}
\label{eq:csp-Z_0-is-not-0}
    Z_0 = \prod_{v\in V} \sum_{\sigma\in Q_v } \tp{ \=I[f_v(\sigma)\neq 1^\coF] +  \lambda\cdot  \=I[f_v(\sigma) = 1^\coF]} \neq 0,
\end{equation}
and $\Phi_1, \Phi_2, \dots, \Phi_m$ satisfy \Cref{condition:csp-conditional-measure-analysis}, we have $Z = Z(\Phi, \lambda, \*f, 1^\coF)\neq 0$.
\end{theorem}
\begin{proof}
We use induction to show that $Z_m \neq 0$. The base case of the induction is $Z_0\neq 0$ which follows from \cref{eq:csp-Z_0-is-not-0}.
Now, we consider the induction step on $Z_i$ with $1\le i \le m$, by the induction hypothesis $Z_{i-1} \neq 0$, let $\mu_{i-1}$ be its complex Gibbs measure. By \Cref{condition:csp-conditional-measure-analysis} and \Cref{lemma:csp-induction-step}, we set the parameters in \Cref{lemma:csp-induction-step} as $\Phi\gets \Phi_i$ and $c^* \gets c_i$, it holds that $|\mu_{i-1}(c_{i} \text{ is violated})| < 1$. Recall that $\frac{Z_i}{Z_{i-1}} = 1 - \mu_{i-1}(c_i \text{ is violated})$ (\cref{eq:partition-function-divide-to-maginal-measure}).
Combined, we have $|\frac{Z_i}{Z_{i-1}}| > 0$, hence $Z_i \neq 0$ which completes the induction. So we have $Z = Z(\Phi, \lambda, \*f, 1^\coF)\neq 0$.
\end{proof}

\subsection{Proof of the induction step}
\label{subsection:csp-zero-freeness}
In this subsection, we prove \Cref{lemma:csp-induction-step}. 
For simplicity, we consider a $(k, \Delta)$-CSP formula $\Phi = (V, \*Q, C)$ with an extra constraint $c^*$ such that $(V, \*Q, C\cup\set{c^*})$ is also a $(k, \Delta)$-CSP formula.
In order to prove the induction step (\Cref{lemma:csp-induction-step}), we assume the partition function $Z = Z(\Phi, \lambda, \*f, 1^\coF) \neq 0$. Let $\mu = \mu(\Phi, \lambda, \*f, 1^\coF)$ be the Gibbs distribution. It suffices to show $|\mu(c^* \text{ is violated})| < 1$.

Next, we use our \emph{projection-lifting scheme} (\Cref{condition:projection-lifting-scheme}) to bound this norm.
Recall that in our \emph{projection-lifting scheme}, we need to design events with respect to the projected measure $\psi$ such that fixing the contributions from complex external fields (\Cref{condition:projection-lifting-scheme}). Here we use pinnings to fix those contributions, and we use the following combinatorial structures to construct these events
 and later we use the complex Markov chain to bound the norm of each event. We remark that this argument can also be seen as a complex analogue of the lazy perfect sampling scheme~\cite{anand2022perfect,feng2023towards}.
% \yxtodo{I have to make sure this is the TRUE complex analogue.}
% The intuition is that we enumerate assignments $\sigma\in \*\Sigma$ under the two-step projection. By conditioning on some ``proper structures'' of $\sigma$, the conditional measure is a positive real number that is not larger than $1$, and the measure for the ``proper structures'' can be handled by \Cref{condition:complex-local-uniformity}.
% To formalize the ``proper structures'', we give the following definition.

Recall that for any assignment $\sigma\in\bm{\Sigma}$ under the projection, and any constraint $c$, we say $c$ is \emph{satisfied} by $\sigma$ iff for any violating assignment $\tau\in c^{-1}\tp{\false}$ there exists $v\in \vbl(c)$ such that $f_v(\tau_v) \neq \sigma_v$. This means that any original assignment whose projected symbol is $\sigma$ satisfies the constraint $c$. Otherwise, we say $c$ is \emph{not satisfied} by $\sigma$.
\begin{definition}[Bad constraint, bad cluster]
\label{definition:csp-bad-things}
% Let $\Phi = (V, \*Q, C)$ be a CSP formula. Let $\*f$ and $\*g$ be two functions in the two-step projection scheme (\Cref{definition:two-step-projection-scheme}), and let $q\in \bigcup_{v\in V}\Sigma_v$, $\lambda \in \=C$ be two parameters in the complex extensions of CSP formulas (\Cref{definition:complex-extensions-csp}). 
Let $\Phi = (V, \*Q, C)$ be a CSP formula,
% Let $\Phi = (V, \*Q, C)$ be a $(k, \Delta)$-CSP formula, 
with the projection $\*f$ in the state-compression scheme (recall \Cref{definition:state-compression}), and the complex external field $\lambda$ on the symbol $1^\coF$ (\Cref{definition:complex-extensions-csp}). 

Let $c^*$ be an additional constraint. Given an assignment $\sigma\in \*\Sigma_{V}$ under the projection, we define the following structures on the dependency graph  $\dependencyGraph{\Phi'}$ of $\Phi' = (V, \*Q, C\cup\{c^*\})$.
% Let $\dependencyGraph{\Phi'}$ be the dependency graph of $\Phi'$. \todo{yyx: did I define the dependency graph?}
\begin{enumerate}
    \item The set of \emph{bad constraint} $\ebad = \ebad(\sigma)$ is defined as:
    \[
    \ebad \defeq \{c\in C\cup\{c^*\}\mid c \text{ is not satisfied by } \sigma \}.
    \]
    \item Let $\dependencyGraph{\Phi'}[\ebad]$ be the subgraph of $\dependencyGraph{\Phi'}$ induced by $\ebad$.
    \item The \emph{bad cluster} $\sbad = \sbad(\sigma)$ is defined as the maximal connected component in $\dependencyGraph{\Phi'}[\ebad]$ containing the additional constraint $c^*$. If $c^* \notin \ebad$, then we define $\sbad = \emptyset$.
    % \item The \emph{bad tree} $\tbad = \tbad(\sigma)$ is defined as the $2$-tree of the bad cluster $\sbad$ that contains $c^*$, constructed deterministically using \Cref{definition:2-tree-construction}. We further denote this deterministic construction as a mapping $\=T$ from the bad cluster such that $\tbad = \=T(\sbad)$.
\end{enumerate}
\end{definition}

Now, we use bad clusters to define the events in our \emph{projection-lifting scheme} (\Cref{subsection:projection-lifting-scheme}) and note that if ``$c^*$ is violated'', then the bad cluster must not be $\emptyset$.
By the above definition, \Cref{condition:projection-lifting-scheme}-(\ref{item:projection-lifting-scheme-partition}) and  \Cref{condition:projection-lifting-scheme}-(\ref{item:projection-lifting-scheme-cover}) hold.
We use the next lemma to verify \Cref{condition:projection-lifting-scheme}-(\ref{item:projection-lifting-scheme-condition}).
Given a bad cluster $\sbad$, we use $\vbl(\sbad) \defeq \bigcup_{c\in \sbad} \vbl(c)$ to denote the set of variables in $\sbad$.
\begin{lemma}
\label{lemma:csp-conditional-measure-probabilty}
% For any assignment $\sigma\in \*\Sigma_{V}$ under the projection, let $\+S = \sbad(\sigma)$. 
For any partial assignment $\tau \in \*\Sigma_{\vbl(\sbad)}$ under the projection, it holds that 
\begin{enumerate}
\item if $\mu(\sbad = \+S \land \tau)= 0$, then $\mu(c^*\text{ is violated}\land \sbad = \+S \land \tau) = 0$,\label{item:csp-conditional-measure-result-1}
\item if $\mu(\sbad = \+S \land \tau)\neq 0$, then $\mu(c^* \text{ is violated} \mid \sbad = \+S \land \tau) \in [0,1]$. \label{item:csp-conditional-measure-result-2}
\end{enumerate}
% if $\mu(\sbad = S \land \tau)\neq 0$, then it holds that 
% \[
% \mu(c^* \text{ is violated} \mid \sbad = \+S \land \tau) \in [0,1].
% \]
\end{lemma}
\begin{proof}
% We first prove (\ref{item:csp-conditional-measure-result-1}).
To prove this lemma, we consider the weight sum of all satisfying assignments that also satisfy $\sbad =\+S$ and $\tau$. 
We claim that this sum can be factorized into a product between contributions from $V\setminus \vbl(\+S)$ and $\vbl(\+S)$.

To see this, let $\+S^*$ be the set of constraints that are adjacent to some constraints in $\+S$ but not in $\+S$. By \Cref{definition:csp-bad-things}, we have that constraints in $\+S^*$ are all satisfied. 
And by $\tau$, for each constraint $c \in \+S^*$, we know that $c$ is satisfied by $\tau$ on $\vbl(\+S)$ or $c$ is satisfied by $V\setminus \vbl(\+S)$. 
% \todo{yyx: tag here.}

We now construct a new CSP formula $\Phi' = (V', \*Q', C')$ that encodes the contribution from $V\setminus \vbl(\+S)$. Let $C' = C\setminus \+S$ and $V' = V\setminus \vbl(\+S)$ and $\*Q' = \*Q_{V'}$. For each constraint $c$ in $\+S^*$, if $c$ is satisfied by $\tau$, we remove $c$ from $C'$, otherwise, we restrict that this constraint is satisfied by the value on $\vbl(c)\cap V'$ under the projection.
We denote the partition function of $\Phi'$ as $Z'$.

For the contributions from $\vbl(\+S)$, 
by the fact that $\+S^*$ are all satisfied, we only need to consider the constraints in $\+S$. Note that from $\tau$ and our complex extension of CSP formulas (\Cref{definition:complex-extensions-csp}), we can determine the contributions of the complex external field. 
So the contribution of this part can be expressed as the product between $\lambda^{|\{v\in \vbl(\+S)\mid \tau_v = 1^\coF\}|}$ and the number of satisfying assignments with respect to $\+S$ under the partial assignment $\tau$. Let $A_1$ be the number of satisfying assignments on variables $\vbl(\+S)$ with respect to $\+S$ under the partial assignment $\tau$, and  $A_2$ be the number of satisfying assignments on variables $\vbl(\+S)$ with respect to $\+S$ and an additional constraint ``$c^*$ is violated'' under the partial assignment $\tau$.

For (\ref{item:csp-conditional-measure-result-1}), by $\mu(\sbad = \+S\land \tau) = 0$, we have that $A_1\cdot \lambda^{|\{v\in \vbl(\+S)\mid \tau_v = 1^\coF\}|}\cdot Z' = 0$. Let $Z$ be the partition function of $\Phi$, combined with the fact that $A_2 \le A_1$, we have $\mu(c^*\text{ is violated}\land \sbad = \+S \land \tau) = \frac{A_2\cdot \lambda^{|\{v\in \vbl(\+S)\mid \tau_v = 1^\coF\}|}\cdot Z'}{Z} = 0$.

For (\ref{item:csp-conditional-measure-result-2}), we have that $\mu(c^* \text{ is violated} \mid \sbad = \+S \land \tau) = \frac{A_2\cdot \lambda^{|\{v\in \vbl(\+S)\mid \tau_v = 1^\coF\}|}\cdot Z'}{A_1\cdot \lambda^{|\{v\in \vbl(\+S)\mid \tau_v = 1^\coF\}|}\cdot Z'} = \frac{A_2}{A_1}\in [0,1]$. 
\end{proof}
Now we are ready to prove \Cref{theorem:sufficient-condition-zero-freeness} using our \emph{projection-lifting scheme} (\Cref{subsection:projection-lifting-scheme}) and the following lemma that bounds the norm of bad events. 
\begin{lemma}
\label{lemma:csp-bad-events-sum-bound}
Assuming the conditions in \Cref{theorem:sufficient-condition-zero-freeness}, it holds that $\sum_{\+S:\abs{S}\ge 1}\sum_{\tau\in\Sigma_{\vbl(\+S)}}\abs{\psi(\+S^{\-{bad}} = S \land \tau)} < 1$. %\todo{yyx: need to fill in.}
\end{lemma}
We first use this lemma to prove the induction step (\Cref{lemma:csp-induction-step}), and we prove the above lemma in the next subsection through the complex Markov chains.
\begin{proof}[Proof of \Cref{lemma:csp-induction-step}]
It suffices to show $|\mu(c^*\text{ is violated})| < 1$.
By the law of total measure, the triangle inequality and the fact that if the event ``$c^*$ is violated'' happens, then $|\sbad|\ge 1$, we have that
\begin{equation*}
% \label{eq:csp-lifting-eq1}
|\mu(c^*\text{ is violated})| \le \sum_{\+S:|\+S|\ge 1}\sum_{\tau\in \*\Sigma_{\vbl(S)}} |\mu(c^*\text{ is violated}\land \sbad = \+S\land \tau)|.
\end{equation*}
If $\mu(c^*\text{ is violated}\land \sbad = \+S\land \tau)\neq 0$, by the conditional measure, we have 
\[
|\mu(c^*\text{ is violated}\land \sbad = \+S\land \tau)| = |\mu(\sbad = \+S\land \tau)\cdot \mu(c^*\text{ is violated} \mid \sbad = \+S\land \tau)|.
\]
By \Cref{lemma:csp-conditional-measure-probabilty}, the RHS can be upper bounded as $|\mu(\sbad = \+S\land \tau)|$. By the definition of $\psi$ (\Cref{definition:csp-projected-measure-under-f}), it is equal to $|\psi(\sbad = \+S\land \tau)|$. Combined, it holds that
\begin{equation*}
% \label{eq:csp-lifting-eq2}
\begin{aligned}
\sum_{\+S:|\+S|\ge 1}\sum_{\tau\in \*\Sigma_{\vbl(S)}} |\mu(c^*\text{ is violated}\land \sbad = \+S\land \tau)| \le \sum_{\+S:|\+S|\ge 1}\sum_{\tau\in \*\Sigma_{\vbl(S)}} |\psi(\sbad = \+S\land \tau)|.
\end{aligned}
\end{equation*}
Combined with \Cref{lemma:csp-bad-events-sum-bound}, this lemma follows.
\end{proof}

\subsection{Bounding the measure of bad events}
In this subsection, we demonstrate how to bound the measure of bad events $\abs{\psi\tp{\sbad = S\land \tau}}$ and prove \Cref{lemma:csp-bad-events-sum-bound} through the complex systematic scan Glauber dynamics on the projected measure $\psi$ (\Cref{definition:csp-projected-measure-under-f}).
The strategy is that we run the systematic scan Glauber dynamics for a sufficiently long time. Then we can show that the resulting state follows from the projected measure. Let $\sigma_0$ be the resulting state of our systematic scan Glauber dynamics. To get a handle on $\psi$, we consider the measure $\psi^{\-{GD}}_T$ and the corresponding event: $\sbad(\sigma_0) = S$ and for any $v\in \vbl(S)$, $\sigma_0(v) =\tau_v$.

\subsubsection{Complex Glauber dynamics on the projected measure}
In this subsubsection, we first introduce some useful conditions and combinatorial structures for analyzing the systematic scan Glauber dynamics on the projected measure.
% Note, in order to establish the above inequality, we need to bound the norm of a marginal measure.
% Here, we do not aim for a general condition to establish \Cref{condition:complex-local-uniformity} but to demonstrate how to establish it.
% Recall that $\tau$ is a partial assignment defined on $\Lambda\subseteq V$ under the two-step projection $\*h$.
Let $T\ge 1$, and let $\sigma\in \*\Sigma$ be an assignment under the projection $\*f$ with $\psi(\sigma)\neq 0$, we consider a $T$-step complex systematic scan Glauber dynamics on $\psi$ with the initial state $\sigma$. 

To apply the analytic percolation framework to analyze the complex Glauber dynamics, we recall \Cref{condition:csp-conditional-measure-analysis} which contains conditions about the decomposition scheme (\Cref{definition:decomposition-scheme}). 

And we have the following two observations about the \decompositionScheme-decomposition scheme.

\begin{observation}\label{observation:csp-decomposition-scheme-resample-depends-on-neighbors}
For any variable $u\in V$, and any extendable partial assignment $\tau\in \*\Sigma_{V\setminus \{u\}}$ under the projection, it holds that $\psi_u^{\tau, \bot}$ only depends on the $\psi_u^\tau$.
    % \[
    % \forall x\in \Sigma_u,\quad \Psi^u_{x, \tau}(x) = 1,\quad \forall y\in \Sigma_u \text{ and } y\neq x,\quad \Psi^u_{x, \tåau'}(y) = 0,
    % \]
\end{observation}

\begin{observation}\label{observation:csp-decomposition-scheme-r-t-o-t}
% Let \decompositionScheme be the decomposition scheme in \Cref{condition:csp-conditional-measure-analysis}.
Let $T>0$ be an integer. Consider the $T$-step \decompositionScheme-decomposed complex systematic scan Glauber dynamics (\Cref{condition:csp-conditional-measure-analysis}), let $\bm{r}=\set{r_t}_{t=-T+1}^0$ and $\bm{o} = \set{o_t}_{t=-T+1}^0$ be defined in \Cref{Alg:complex-GD-decomposed}. For any $t: -T+1 \le t \le 0$, it holds that if $r_t\neq \bot$, then $o_t = r_t$.
% For any variable $u\in V$, and any feasible partial assignment $\tau\in \*\Sigma_{V\setminus \{u\}}$ under the projection, it holds that $\Psi^u_{\bot, \tau}$ only depends on the $\psi_u^\tau$.
    % \[
    % \forall x\in \Sigma_u,\quad \Psi^u_{x, \tau}(x) = 1,\quad \forall y\in \Sigma_u \text{ and } y\neq x,\quad \Psi^u_{x, \tau'}(y) = 0,
    % \]
\end{observation}

We use the decomposed complex Markov chain to study the bad events we defined before. We follow the basic ideas in~\cite{liu2024phase}. Next, we define some useful structures to analyze the systematic scan Glauber dynamics and then use these structures to characterize the witness sequences and analyze the bad events.

\subsubsection{Witness graph}

A structure useful for analyzing the classical Glauber dynamics~\cite{hermon2019rapid, he2021perfect,qiu2022perfect,feng2023towards} and the complex Glauber dynamics~\cite{liu2024phase} is the so-called \emph{witness graph}, which we introduce next.

For any $u\in V$ and integer $t$, we denote by $\upd_u(t)$ the last time before $t$ at which $u$ is updated, i.e.
\begin{align*}
\upd_u(t)\defeq \max\{s \mid s\leq t \hbox{ such that } v_{i(s)}=u\}.
\end{align*}
For any subset of variables $U\subseteq V$ and $t\in \mathbb{Z}_{\leq 0}$, define
\begin{align*}
		\ts(U,t)\defeq\{\upd_v(t)\mid v\in U\}
\end{align*}
as the collection of ``\emph{timestamps}'' of the latest updates of variables in  $U$ up to time $t$.

Now we give the definition of the \emph{witness graph}.
Each vertex of the witness graph is a tuple containing a collection of ``timestamps'' and the associated constraint. To analyze the complex marginal measure on an additional constraint $c^*$, we also include an additional vertex $\witnessTreeRoot$.
%  Here, $\emptyset$ means that we do not associate any constraint with $\witnessTreeRoot$.

\begin{definition}[Witness graph]\label{definition:csp-witness-graph}
Given a CSP formula $\Phi = (V, \*Q, C)$, and let $c^*$ be an additional constraint with $|\vbl(c^*)| = k$, the \emph{witness graph} $G^{c^*}_{\Phi}=\left(V^{c^*}_{\Phi},E^{c^*}_{\Phi}\right)$ is an infinite graph with the vertex set 
\[
V^{c^*}_{\Phi}=\{(\ts(c,t),c)\mid c\in C,t\in \mathbb{Z}_{\leq 0}\}\cup \{\witnessTreeRoot\},
\] 
and 
for any vertex $x$ in $V_{\Phi}^{c^*}$, we use $x(0)$ to denote the first coordinate of $x$ and $x(1)$ or $c_x$ to denote the second coordinate of $x$. % And we use $t_x$ to denote the largest timestamp in $x(0)$, i.e., $t_x = \max_{t'\in x(0)} t'$. \todo{yyx: is this used?}

$E^{c^*}_{\Phi}$ consists of undirected edges between vertices $x,y \in V^{c^*}_{\Phi}$ such that $x\neq y$ and $x(0)\cap y(0)\neq\emptyset$.
\end{definition}

% Note that in our induction procedure, $c^* = c_{i}$ is one constraint in the original $(k, \Delta)$-CSP formula.
Recall that $(V, \*Q, C\cup \set{c^*})$ is still a $(k, \Delta)$-CSP formula.
The following structural property of the witness graph has been established in~\cite{feng2023towards,liu2024phase} in terms of hypergraphs.

\begin{lemma}[{\cite[Corollary 6.15]{feng2023towards}, \cite[Lemma 4.5]{liu2024phase}}]\label{lemma:csp-witness-graph-degree-bound} 
Assume $\Phi = (V, \*Q, C\cup \set{c^*})$ is a $(k, \Delta)$-CSP formula.
Then, on the witness graph $G^{c^*}_\Phi=\left(V^{c^*}_\Phi,E^{c^*}_\Phi\right)$, for any $x\in V_\Phi^{c^*}$, the degree of $x$ is at most $2\Delta k^2-2$. % Furthermore, the degree of \;$(\ts(S,0),\emptyset)$ is at most $2\Delta k|S|-1$.
\end{lemma}

Following the high-level idea in \Cref{section:analytic-percolation}, we now define some structures in the above \emph{witness graph}, and later we will use them to characterize the witness sequence (\Cref{definition:witness-sequence}).
% Recall the definition of witness sequences from \Cref{definition:witness-sequence}. 
Let $\sigma_0$ be the resulting state of the complex systematic scan Glauber dynamics.
Recall that in order to prove the zero-freeness, for any subset of constraints $S$ and $\tau\in\Sigma_{\vbl(S)}$ defined on $\vbl(S)$ with the projected alphabet, we are interested in the event (1) $\sbad(\sigma_0) = S$ and (2) for any $v\in \vbl(S)$, $\sigma_0(v) = \tau_v$. 

Recall the definitions of $o_t$'s and $r_t$'s in \Cref{Alg:complex-GD-decomposed}. It holds that for any variable $v\in V$, we have that $\sigma_0(v)$ is determined by the corresponding $o_t$'s. For any time $t$, recall that $i(t)$ denotes the index of the variable updated at the time $t$, and recall that $\upd_v(0)$ denotes the last update time of $v$ up to the time $0$. We have that for any variable $v\in V$, we have $\sigma_0(v) = o_{\upd_v(0)}$. 
So we are interested in the $o_t$'s. 

For any time $t$, note that if $r_t \neq \bot$, then we have $o_t = r_t$ (\Cref{observation:csp-decomposition-scheme-r-t-o-t}), otherwise if $r_t = \bot$, then by the definition of the decomposition scheme in \Cref{condition:csp-conditional-measure-analysis}, \Cref{observation:csp-decomposition-scheme-resample-depends-on-neighbors} and \Cref{Alg:complex-GD-decomposed}, let $\sigma_{t-1}$ be the assignment at the time $t-1$, it holds that $o_t$ depends on $\sigma_{t-1}$. Specifically, let $u = v_{i(t)}$ be the updated variable at the time $t$, it suffices to consider constraints $c$ with $u\in\vbl(c)$. 

Next, we further restrict the set of constraints that we need to take into consideration.
Recall that for any time $t'$, if $r_{t'}\neq \bot$, we have $o_{t'} = r_{t'}$. 
Also note that $r_t$'s do not depend on the initial state.
We can use information of $r_t$'s to restrict the set of constraints in order to compute $o_t$.
For a constraint $c$, recall that we use $c^{-1}(\-{False})$ to denote the set of partial assignments violating the constraint $c$. 
By \Cref{observation:csp-decomposition-scheme-resample-depends-on-neighbors}, it holds that for any $c\in C$ containing variable $u$, if $c$ is satisfied by the $r_t$'s on variables other than $u$, namely variables in $\vbl(c) \backslash\{u\}$, i.e.,
\[
\forall c\in C \text{ and } \vbl(c)\ni u, \exists u'\in \vbl(c)\backslash\{u\},\text{ s.t. } r_{\upd_{u'}(t)} \notin f_{u'}(c^{-1}(\-{False})_{u'}) \text{ and } r_{\upd_{u'}(t)}\neq \bot,
\]
then we do not need to consider this constraint when computing $o_t$.

% then we know that $\sigma_t(v)$ does not depend on the initial state because all constraints involving $v$ are satisfied.

So we are interested in finding a ``connected component'' such that all its ``boundaries'' satisfy the above restrictions, then we know all the updates within this ``connected component'' do not depend on $o_t$'s and $r_t$'s outside of this ``connected component''.
Furthermore, if this ``connected component'' is small enough, we can argue that $o_t$'s in this ``connected component'' do not depend on the initial state. So the event (1) $\sbad(\sigma_0) = S$ and (2) for any $v\in \vbl(S)$, $\sigma_0(v) = \tau_v$ does not depend on the initial state.

We formalize the above intuition by the following combinatorial structures on the witness graph (\Cref{definition:csp-witness-graph}). Later, we will show the connection between the following structures and the bad cluster we defined before (\Cref{lemma:csp-determination-bad-cluster-by-bad-component}) and we will use these structures to characterize the witness sequences (\Cref{lemma:csp-conditional-measure-characterization}). 

Recall that for any vertex $x$ in $V_{\Phi}^{c^*}$, we use $x(0)$ to denote the first coordinate of $x$ and $x(1)$ or $c_x$ to denote the second coordinate of $x$. And we use $t_x$ to denote the largest timestamp in $x(0)$, i.e., $t_x = \max_{t'\in x(0)} t'$.
For any $\bm{r}=(r_t)_{t=-T+1}^{0}\in \bigotimes_{t = -T+1}^0 \Sigma_{v_{i(t)}} \cup \set{\bot}$ constructed by \Cref{Alg:complex-GD-decomposed}, we say a vertex $x$ is \emph{satisfied} by $\*r$ if $x(1)$ is satisfied by the partial assignment of $\*r$ on time $x(0)$, i.e., $\forall \sigma\in c_x^{-1}(\-{False})$, there exists a $v\in \vbl(c_x)$ such that $r_{\upd_v(t_x)} \neq f_{v}(\sigma(v))$ and $r_{\upd_v(t_x)} \neq \bot$. Otherwise, we say $x$ is \emph{not satisfied} by $\*r$.
\begin{definition}[Bad vertices, bad components, and bad trees]
    \label{definition:csp-conditional-measure-bad-things}
Let $\Phi = (V, \*Q, C)$ be a CSP formula, and let $c^*$ be an additional constraint with $|\vbl(c^*)|=k$.
Let $G_{\Phi}^{c^*} = (V_{\Phi}^{c^*}, E_{\Phi}^{c^*})$ be the witness graph as in \Cref{definition:csp-witness-graph}.
Let $T\geq 1$, and consider the $T$-step  \decompositionScheme-decomposed complex systematic scan Glauber dynamics  (\Cref{condition:csp-conditional-measure-analysis}). Let $\bm{r}=(r_t)_{t=-T+1}^{0}\in \bigotimes_{t = -T+1}^0 \Sigma_{v_{i(t)}} \cup \set{\bot}$ be constructed by \Cref{Alg:complex-GD-decomposed}. 
\begin{itemize}
\item The set of \emph{bad vertices} $V^{\-{bad}}=V^{\-{bad}}(\bm{r})$ is defined as:
\[
V^{\-{bad}}\defeq \{x\in V^{c^*}_{\Phi}\mid \forall t\in x(0), -T + 1\le t\le 0 \text{ and }  x \text{ is not satisfied by } \*r\},
\]
which contains the vertices in the witness graph whose constraints are not satisfied by the $r_t$'s. 
%  there exists a $j\in[s+1]$, all $r_t$ evaluates to $\bot$ or $j$ for every timestamp $t$ associated with that vertex.
\item Let $G^{c^*}_{\Phi}\left[V^{\-{bad}}\right]$ be the subgraph of the witness graph $G^{c^*}_{\Phi}$ induced by $V^{\-{bad}}$.
\item The \emph{bad component}  $\+C^{\-{bad}}=\+C^{\-{bad}}(\bm{r})\subseteq V^{\-{bad}}$ is defined as the maximal set of vertices in $V^{\-{bad}}$ containing $\witnessTreeRoot$ that is connected in $G^{c^*}_{\Phi}\left[V^{\-{bad}}\right]$.  If $\witnessTreeRoot \notin \vbad$, we set $\cbad = \emptyset$.
\item The \emph{bad tree} $\+T^{\-{bad}}=\+T^{\-{bad}}(\bm{r})\subseteq V^{\-{bad}}$ is defined as the $2$-tree of the induced subgraph $G^{c^*}_{\Phi}\left[\+C^{\-{bad}}\right]$ containing $\witnessTreeRoot$, constructed deterministically using \Cref{definition:2-tree-construction}.

We further denote this deterministic construction as a mapping $\=T$ from the bad component such that $\tbad=\=T(\cbad)$.
% Let $\+C^{\-{bad}}=\emptyset$ if $\ts{}(S,0)\notin V^{\-{bad}}$.
\end{itemize}
\end{definition}

By the definition of the bad component, we know that all constraints crossing the bad component are already satisfied by the $r_t$'s.
Hence, we have the following observation, which means that the update values within the bad component do not depend on update values outside of the bad component.
For any bad component $\cbad$, let $\BadTS{\cbad}$ be the set of timestamps involved in the bad component, i.e., $\BadTS{\cbad} \defeq \bigcup_{x\in \cbad} x(0)$.
\begin{observation}
\label{observation:csp-conditional-independence-cbad-o_t}
For any $\bm{r}=(r_t)_{t=-T+1}^{0}\in \bigotimes_{t = -T+1}^0 \Sigma_{v_{i(t)}} \cup \set{\bot}$, let $\cbad = \cbad(\bm{r})$.
%For any $o_t$ whose timestamp $t$ is in the bad component ($t\in \BadTS{\cbad}$), it holds that $o_t$ does not depend on any $o_{s}$ whose timestamp $s$ is outside of $\cbad$ ($s\notin \BadTS{\cbad}$).
For any $t\in \BadTS{\cbad}$ and $s\notin \BadTS{\cbad}$, it holds that $o_t$ is independent of $o_s$. In other words, 
\[\dpsi{\tau}(o_t\land o_s) = \dpsi{\tau}(o_s)\cdot \dpsi{\tau}(o_t).\]
\end{observation}

\subsubsection{Connection between bad clusters and bad components}
In this subsection, we give the connection between the bad cluster we defined before (\Cref{definition:csp-bad-things}) and the bad component we defined above (\Cref{definition:csp-conditional-measure-bad-things}). Recall the $o_t$'s in our complex systematic scan Glauber dynamics (\Cref{Alg:complex-GD-decomposed}), and recall that $\sigma_0$ is the resulting state of our complex Markov chain. We first show that given the bad component and $o_t$'s involved in the bad component, we can determine $\sbad = \sbad(\sigma_0)$ and for any $v\in \vbl(\sbad)$, the value of $\sigma_0(v)$. Then we provide two technical lemmas which show the convergence and give the upper bounds of norm of marginals measure. Finally, we use these two technical lemmas to prove \Cref{lemma:csp-bad-events-sum-bound}.

We start with a lemma about the determination of $\sbad(\sigma_0)$ and the assignment involved in it.
\begin{lemma}
\label{lemma:csp-determination-bad-cluster-by-bad-component}
Fix the bad component $\cbad$ and $o_t$'s involved in the bad component, i.e., $o_t, \forall t\in \BadTS{\cbad}$. 
There is an algorithm that returns $\sbad = \sbad(\sigma_0)$ and $\tau\in \Sigma_{\vbl(\sbad)}$, satisfying that for any $v\in \vbl(\sbad)$, $\sigma_0(v) = \tau_v$ and if $|\cbad| = 0$ then $|\sbad| = 0$.

So for any subset of constraints $S$ and the partial assignment $\tau \in \Sigma_{\vbl(S)}$, the following event is fully determined:  $\sbad(\sigma_0) = S$ and $\forall v\in \vbl(S)$, $\sigma_0(v) = \tau_v$.
\end{lemma}
\begin{proof}
Let $\bm{o}_{\BadTS{\cbad}}$ be the $o_t$'s involved in the bad component. 
Recall that for any constraint $c$, we use $\ts(c, t)$ to denote the set of timestamps of latest updates of variables in $\vbl(c)$ up to the time $t$.
In this proof, we will focus on $\ts(c, 0)$ which is the set of timestamps of latest updates of variables in $\vbl(c)$ up to the end of our systematic scan Glauber dynamics.
We are interested in the $o_t$'s in the last $|V|$ update times because by the \Cref{Alg:complex-GD-decomposed}, we have that for any $v\in V$, let $t$ be its last update time up to the time $0$, we have $\sigma_0(v) = o_{t}$.
We present an algorithm to construct the $\sbad(\sigma_0)$ and $\sigma_0(\vbl(S))$ by the following procedure, so that we can determine the above event.

\begin{enumerate}
\item Initialize $B = \emptyset$. If $\witnessTreeRoot \in \cbad$ and is also not satisfied by $\bm{o}_{\ts\tp{c^*, 0}}$, then we update $B \gets \set{c^*}$, otherwise we terminate this procedure.
\item If there exists a constraint $c\in C$ satisfying (1) $c\notin B$ and $\vbl(c) \cap \vbl(B) \neq \emptyset$, (2) $\tp{\ts\tp{\vbl(c), 0}, c} \in \cbad$, and (3). $c$ is not satisfied by $\bm{o}_{\ts\tp{c, 0}}$, then we update $B\gets B\cup \set{c}$. And we keep doing this step until there is no such constraint.
\end{enumerate}

Next, we show that $B = \sbad(\sigma_0)$. We first show $B\subseteq \sbad(\sigma_0)$. By the step (2), we have that all constraints in $B$ are not satisfied corresponding $o_t$'s which means that all constraints in $B$ are not satisfied by $\sigma_0$. Recall that $\Phi' = (V, \*Q, C\cup\set{c^*})$, and recall the definition of the dependency graph $\dependencyGraph{\Phi'}$. Also by the step (2), we claim that $B$ is a connected component in the dependency graph $\dependencyGraph{\Phi'}$. To see this, when a new constraint is added into $B$, it holds that this new constraint shares some variables with $B$. By the definition of the bad cluster (\Cref{definition:csp-bad-things}), we have $B\subseteq \sbad(\sigma_0)$.

Then, we show $\sbad(\sigma_0) \subseteq B$ by contradiction. Assume that $\sbad(\sigma_0) \setminus B \neq \emptyset$, by the fact that $B\subseteq \sbad(\sigma_0)$ and the definition of the bad cluster, there exists a constraint $c$ such that $\vbl(B)\cap \vbl(c)\neq \emptyset$ and $c$ is not satisfied by $\sigma_0$ but $c\notin B$. By the condition in step (2), it holds that $\tp{\ts\tp{\vbl(c), 0}, c}\notin \cbad$. So we have that $c$ is satisfied by the corresponding $r_t$'s. By \Cref{observation:csp-decomposition-scheme-r-t-o-t}, it holds that for any $t\le 0$, if $r_t \neq \bot$ we have $o_t = r_t$. So $c$ is satisfied by the corresponding $o_t$'s so that $c$ is satisfied by $\sigma_0$ which contradicts with the assumption $c\in \sbad(\sigma_0)$.

Now, we finish the construction of $\sbad(\sigma_0)$. By the condition (2) and (3) in the step (2), it holds that we only check $o_t$ for $t\in \BadTS{\cbad}$. So the $\sigma_0(\vbl(S))$ follows from the corresponding $o_t$'s.

Note that by step (1), if $|\cbad| = 0$, $\sbad = B = \emptyset$.
\end{proof}

Next, we give two technical lemmas about our complex systematic scan Glauber dynamics. The first one shows the convergence and its proof is deferred in \Cref{subsubsection:convergence}. The second one bounds the sum of norms of marginal measures and its proof is deferred in \Cref{subsubsection:bounding-norm}. Next, we use these two technical lemmas to prove \Cref{lemma:csp-bad-events-sum-bound}.
\begin{lemma}
\label{lemma:csp-projected-measure-marginal-convergence}
Suppose the \Cref{condition:csp-conditional-measure-analysis}.
For any subset of constraints $S$ and any partial assignment $\tau \in \Sigma_{\vbl(S)}$ defined on variables involved in $S$, let $A$ be the event that $\sbad(\sigma_0) = S$ and $\forall v\in \vbl(S)$, $\sigma_0(v) = \tau_v$, we have that \Cref{condition:convergence} holds, and for any $\tau'\in\supp(\psi)$, we have $\psi(A) = \lim_{T\to\infty} \dpsi{\tau'}\tp{\sigma_0 \in A}$.
\end{lemma}

\begin{lemma}
\label{lemma:csp-projected-measure-norm-sum-bound}
Suppose the \Cref{condition:csp-conditional-measure-analysis}. 
It holds that 
\[\lim_{T\to \infty} \sum_{\+C:|C|\ge 1} \sum_{\*y_{\BadTS{\+C}}} \abs{\dpsi{\sigma}\tp{\cbad = \+C \land \*o_{\BadTS{\+C}} = \*y_{\BadTS{\+C}}}} < 1.\]
\end{lemma}

\begin{proof}[Proof of \Cref{lemma:csp-bad-events-sum-bound}]
    % This lemma follows from \cref{eq:csp-bad-cluster-expressed-by-GD} and \Cref{lemma:csp-upper-bound}.
For any $S$ and $\tau \in \Sigma_{\vbl(S)}$, we use $A_{S, \tau}$ to denote the event that $\sbad(\sigma_0) = S$ and $\forall v\in \vbl(S)$, $\sigma_0(v) = \tau_v$. 
Let $T>|V|$. We consider a complex systematic scan Glauber dynamics on the projected measure $\psi$ starting from $\tau'\in \supp(\psi)$. 
By \Cref{lemma:csp-projected-measure-marginal-convergence}, it suffices to upper bound $\lim_{T\to \infty}\sum_{|S|\ge 1}\sum_{\tau\in\Sigma_{\vbl(S)}}\abs{\dpsi{\tau'}\tp{\sigma_0 \in A_{S, \tau}}}$.

Recall \Cref{lemma:csp-determination-bad-cluster-by-bad-component}, given the bad component $\cbad$ and the $\bm{o}_{\BadTS{\cbad}} = \tp{o_t}_{t\in \BadTS{\cbad}}$, we can deterministically construct $\sbad(\sigma_0)$ and $\tau \in \Sigma_{\vbl(\sbad(\sigma_0))}$, and we denote this construct as a function $\mathfrak{P}(\cbad, \bm{o}_{\BadTS{\cbad}}) = (\sbad(\sigma_0), \tau)$. It suffices to upper bound 
\[
\lim_{T\to\infty} \sum_{S:|S|\ge 1}\sum_{\tau\in\Sigma_{\vbl(S)}} \abs{\sum_{\substack{C, \bm{y}_{\BadTS{C}} \\ \mathfrak{P}(C, \bm{y}_{\BadTS{C}}) = (S, \tau)}} \dpsi{\tau'}\tp{\cbad = C\land \bm{o}_{\BadTS{C}} = \bm{y}_{\BadTS{C}}}}.
\]
By the triangle inequality, we can push the absolute norm inwards. Also note that $\`P$ is a deterministic procedure, so that each pair $(C, \bm{y}_{\BadTS{C}})$ contributes to exactly one pair $(S, \tau)$. And by \Cref{lemma:csp-determination-bad-cluster-by-bad-component}, to let $|\sbad|\ge 1$, we need $|\cbad|\ge 1$.  Combined with the triangle inequality, it suffices to upper bound
\[
\lim_{T\to \infty} \sum_{C:|C|\ge 1} \sum_{\bm{y}_{\BadTS{C}}} \abs{\dpsi{\tau'}\tp{\cbad = C \land \bm{o}_{\BadTS{C}} = \bm{y}_{\BadTS{C}}}}.
\]
Then this lemma follows from \Cref{lemma:csp-projected-measure-norm-sum-bound}.
\end{proof}

% And the next lemma shows that we can upper bound the sum of norms for all possible bad components which 
% For any assignment $\tau\in\Sigma_{V}$ under the projection, 

\subsubsection{Convergence}
\label{subsubsection:convergence}
In this subsubsection, we prove \Cref{lemma:csp-projected-measure-marginal-convergence} using \Cref{lemma:convergence}.
It suffices to verify \Cref{condition:convergence}.
We first use the next lemma to characterize the witness sequences.

\begin{lemma}[Characterization of witness sequences]
\label{lemma:csp-conditional-measure-characterization}
Fix $S$ as a subset of constraints, and $\tau \in \Sigma_{\vbl(S)}$ as a partial assignment defined on variables involved in $S$.
Let $\Phi = (V, \*Q, C)$ be a CSP formula, let $c^*$ be the additional constraint, and $A$ be the event that $\sbad(\sigma_0) = S$ and $\forall v\in \vbl(S)$, $\sigma_0(v) = \tau_v$.

%  a $A\subseteq \*\Sigma$ be an event and $S = \vbl(A)$.
Let $G_{\Phi}^{c^*} = (V_{\Phi}^{c^*}, E_{\Phi}^{c^*})$ be the witness graph as in \Cref{definition:csp-witness-graph}.
Let $T\geq |V|$, and consider the $T$-step \decompositionScheme-decomposed complex systematic scan Glauber dynamics  (\Cref{condition:csp-conditional-measure-analysis}). 
% Let $\bm{r}=(r_t)_{t=-T+1}^{0}\in \bigotimes_{t = -T+1}^0 S_t$ be constructed by \Cref{Alg:complex-GD-decomposed}.
% Fix any $T\geq \abs{V}$ and any event $A\subseteq \*\Sigma'$. 
% Consider the witness graph $\witnessGraph = \tp{V_{\Phi}^S, E_{\Phi}^S}$ constructed using the set of variables $S=\vbl(A)$.
Let $\bm{\rho}=(\rho_t)_{t=-T+1}^{0}\in \bigotimes_{t = -T+1}^0 \Sigma_{v_{i(t)}} \cup \set{\bot}$. 
%be a sequence where each $\rho_t\in \{0,1,\bot\}$. 
If the corresponding bad tree $\tbad(\bm{\rho})$ in the witness graph $G^{c^*}_\Phi$ satisfies: 
$\left|\tbad(\bm{\rho})\right|\leq \frac{T}{2\abs{V}}-2$,
then we have that $\bm{\rho}\Rightarrow A$, i.e., $\bm{\rho}$ is a witness sequence with respect to the event $A$.
\end{lemma}
\begin{proof}
When the event $\*r = \*\rho$ has zero measure, this lemma holds trivially from the definition of the witness sequence (\Cref{definition:witness-sequence}).
So we consider the case that the event $\*r = \*\rho$ has non-zero measure. By \Cref{lemma:csp-determination-bad-cluster-by-bad-component}, it suffices to show that $\cbad = \cbad(\bm{\rho})$ and $\bm{o}_{\BadTS{\cbad}}$ are independent of the initial state. For the bad component, this holds trivially because by the definition of the bad component, it depends on $r_t$'s which are independent of the initial state. Next, we show that $\bm{o}_{\BadTS{\cbad}}$ is also independent of the initial state.

By the definition of $\upd_u(t)$, $G^{c^*}_\Phi$ and the systematic scan, we have for each $x\in V^{c^*}_\Phi$,
\begin{align}
\label{eq:csp-conditional-measure-timestamps-max-min}
\max\{t:t\in x(0)\}-\min\{t:t\in x(0)\}\leq |V|.
\end{align}

Note that according to the definition of $\tbad(\bm{\rho})$ in \Cref{definition:csp-conditional-measure-bad-things}, we have $\tbad(\bm{\rho})$ is connected in
the square graph of $G_\Phi^{c^*}$.
So, if we have $|\tbad(\bm{\rho})|\leq T/(2|V|)-2$ then it holds that $t\geq -T+1+2|V|$ for all $t\in x(0), x\in \tbad(\bm{\rho})$.
Furthermore, by \Cref{definition:2-tree}, $t\geq -T+1+2|V|$ for all $t\in x(0),x\in \tbad(\bm{\rho})$ implies that $t\geq -T+1+|V|$ for all $t\in x(0),x\in \cbad(\bm{\rho})$.
To see this, 
first observe that each $x\in \cbad(\*\rho)$ must share timestamps with some $y\in \tbad(\*\rho)$ in $G^{c^*}_\Phi$, according to the construction in \Cref{definition:2-tree-construction}. Then this claim follows by applying \cref{eq:csp-conditional-measure-timestamps-max-min}.

Next, by the \Cref{definition:csp-conditional-measure-bad-things} and \Cref{observation:csp-decomposition-scheme-r-t-o-t}, it holds that for any $t \in x(0)$, $x\in \cbad(\*\rho)$, let $u = v_{i(t)}$, if $c\in C$, $c\ni u$ and $(\ts(c, t),c)\notin \cbad(\*\rho)$, then $c$ is satisfied, i.e.,
for any $\sigma\in c^{-1}(\-{False})$, there exists $t'\in\ts(c, t)$ satisfying that:
$r_{t'} \neq  f_{i(t')}(\sigma_{i(t')})$ and $r_{t'}\neq \bot$. 
% \todo{yyx: handle about the $o_t$ here.}

Now, we show that $o_t$'s in the bad component are independent of the initial state by the induction. We enumerate $t\in \BadTS{\cbad}$ by an increasing order, for the smallest timestamp $t$, let $u=v_{i(t)}$ be the updated variable at the time $t$, it holds that for any constraint $c$ with $u\in \vbl(c)$, $c$ is satisfied by the above argument. So $o_t$ is independent of the initial state. 
Next, assume that for a timestamp $t_0\in \BadTS{\cbad}$, we have that for any timestamp $t < t_0$ and $t\in \BadTS{\cbad}$, $o_t$ is independent of the initial state. We show $o_{t_0}$ is independent of the initial state. Let $u = v_{i(t_0)}$ be the updated variable at the time $t_0$. For any constraint $c$ with $u\in c$, if $\tp{\ts\tp{\vbl(c), t_0}, c} \notin \cbad$, by the argument above, it holds that $c$ is satisfied by the corresponding $r_t$'s. Recall the definition of $o_{t_0}$ in \Cref{Alg:complex-GD-decomposed} and \Cref{observation:csp-decomposition-scheme-resample-depends-on-neighbors}, we can remove this constraint when considering $o_{t_0}$, otherwise if $\tp{\ts\tp{\vbl(c), t_0}, c} \in \cbad$, it holds all timestamps are in $\BadTS{\cbad}$ and not larger than $t_0$. By the induction hypothesis, we know their $o_t$'s are independent of the initial state. Combined, we have $o_{t_0}$ is independent of the initial state.
% assumption 
% Note that a transition depends on the initial state directly if it uses the information at time $-T$.
% By the above two facts, all transitions within the bad component $\cbad(\*\rho)$ do not depend on the initial state (at time $-T$). 
\end{proof}

Let $S$ be a set of constraints and $\tau \in \Sigma_{\vbl(S)}$ be a partial assignment defined on $\vbl(S)$. 
Let $A$ be the event that (1) $\sbad(\sigma_0) = S$ and (2) for any $v\in \vbl(S)$, we have $\sigma_0(v) = \tau_v$.
By the above lemma, in order to verify \Cref{condition:convergence}, it suffices to show that for any initial state $\tau'\in\supp(\psi)$,
\[
\lim_{T\to \infty} \abs{\sum_{\+T: |\+T|> T/(2|V|) - 2}\dpsi{\tau'}\tp{\sigma_0 \in A\land \tbad = \+T} }= 0.
\]
Recall that in \Cref{definition:csp-conditional-measure-bad-things}, we use $\=T$ to denote the function from the bad component to the bad tree. 
 Combined with the triangle inequality, it suffices to show that 
\[
\lim_{T\to\infty} \sum_{\+T: |\+T|> T/(2|V|) - 2} \sum_{C: \=T(C) = \+T} \sum_{\bm{y}_{\BadTS{C}}} \abs{\dpsi{\tau'}\tp{\sigma_0 \in A \land \cbad = C\land \bm{o}_{\BadTS{C}} = \bm{y}_{\BadTS{C}}}} = 0.
\]
Also recall \Cref{lemma:csp-determination-bad-cluster-by-bad-component}, it holds that $A$ is determined by $(\cbad, \bm{o}_{\BadTS{\cbad}})$.
% By the conditional measure, if $\dpsi{\tau'}\tp{\cbad = C\land \bm{o}_{\BadTS{C}}=\bm{y}_{\BadTS{C}}}\neq 0$, then 
That is, if 
\[\dpsi{\tau'}\tp{\cbad = C\land \*o_{\BadTS{C}} = \*y_{\BadTS{C}}} \neq 0,\]
then we have 
\begin{equation}
\label{eq:csp-dpsi-zero-one-law}
\dpsi{\tau'}\tp{\sigma_0 \in A \mid \cbad = C\land \*o_{\BadTS{C}} = \*y_{\BadTS{C}}} \in \set{0, 1}.
\end{equation}
% \todo{Make 0-1 law explicit.}
Combined, we have 
\[
\abs{\dpsi{\tau'}\tp{\sigma_0 \in A\land\cbad = C\land \bm{o}_{\BadTS{C}}=\bm{y}_{\BadTS{C}}}} \le \abs{\dpsi{\tau'}\tp{\cbad = C\land \bm{o}_{\BadTS{C}}=\bm{y}_{\BadTS{C}}}}.
\]
Combined, it suffices to upper bound 
\begin{equation}
\label{eq:csp-convergence-limit}
\lim_{T\to\infty} \sum_{\+T: |\+T|> T/(2|V|) - 2} \sum_{C: \=T(C) = \+T} \sum_{\bm{y}_{\BadTS{C}}} \abs{\dpsi{\tau'}\tp{\cbad = C\land \bm{o}_{\BadTS{C}} = \bm{y}_{\BadTS{C}}}} = 0.
\end{equation}
We use the next lemma, which bounds the contributions of fixed-size bad trees, to show this limit and prove \Cref{lemma:csp-projected-measure-marginal-convergence}.

\begin{lemma}\label{lemma:csp-large-bad-tree-fixed-size-decay}
Suppose the \Cref{condition:csp-conditional-measure-analysis}, and 
% Let $\Lambda \subseteq V$ and $\tau\in\*\Sigma_\Lambda$ be a feasible partial assignment defined on $\Lambda$ under the two-step projection, and 
let \decompositionScheme-decomposition scheme be constructed by \Cref{condition:csp-conditional-measure-analysis}.
Let $c^*$ be the additional variable, 
% Let $A\subseteq \*\Sigma'$ and $S = \vbl(A)$.
and let $T\ge |V|$. 
Consider the witness graph $\witnessGraph$, the bad component $\cbad = \cbad(\*r)$, and the bad tree $\tbad = \tbad(\*r) = \=T(\cbad)$, where $\*r = (r_t)_{t=-T+1}^0 \in \bigotimes_{t = -T+1}^0 \Sigma_{v_{i(t)}} \cup \set{\bot}$ is constructed as in \Cref{Alg:complex-GD-decomposed}. 

Then, for any finite $2$-tree $\+T$ in $\witnessGraph$ containing $\witnessTreeRoot$ with $|\+T| > T/(2|V|) - 1$, we have 
\begin{equation*}%\label{eq:$2$-tree-modulus-bound-subset}
\sum_{\substack{\+C: \=T(\+C) = \+T}} \sum_{\*y_{\BadTS{\+C}}} \abs{\dpsi{\sigma}\tp{\cbad = \+C \land \*o_{\BadTS{\+C}} = \*y_{\BadTS{\+C}}}}\leq \conditionalDecayConstant^{|\+T|} \conditionalTriangleConstant^{4\Delta^2 k^5 |\+T| + |V|}.%\quad \forall \+C^*\subseteq \=T^{-1}(\+T).
\end{equation*}
\end{lemma}

We use the above lemma and the next lemma to show \cref{eq:csp-convergence-limit} and prove \Cref{lemma:csp-projected-measure-marginal-convergence}.
% Recall that in our induction of showing the zero-freeness, $c^*$ is a constraint in the original CSP formula. So we have that $\Phi' = (V, \*Q, C\cup\set{c^*})$ is a $(k, \Delta)$-CSP formula. 
Recall that $\Phi' = (V, \*Q, C\cup\set{c^*})$ is still a $(k, \Delta)$-CSP formula.
We include the next lemma about the number of fixed-size $2$-trees.
\begin{lemma}[{\cite[Lemma 4.10]{liu2024phase}}]
    \label{lemma:csp-bad-tree-num-bound}
Let $\+T_j$ denote the set of $2$-trees of size $j$ in $\witnessGraph$ containing the vertex $\witnessTreeRoot$. Then, we have $|\+T_j|\le \tp{4\mathrm{e} \Delta^2 k^4}^{j-1}$.
\end{lemma}
Now, we prove \Cref{lemma:csp-projected-measure-marginal-convergence}.
\begin{proof}[Proof of \Cref{lemma:csp-projected-measure-marginal-convergence}]
    By the analysis in the beginning of this subsubsection, it suffices to show \cref{eq:csp-convergence-limit}. By \Cref{lemma:csp-large-bad-tree-fixed-size-decay}, it holds that \cref{eq:csp-convergence-limit} can be upper bounded by $\lim_{T\to \infty} \sum_{i= T/(2|V|)-2}^{\infty} \tp{4\mathrm{e} \Delta^2 k^4}^{i-1}\cdot \conditionalDecayConstant^{i} \conditionalTriangleConstant^{4\Delta^2 k^5 i + |V|}$. Combined with the \cref{item:csp-conditional-meaure-analysis-numerical} in \Cref{condition:csp-conditional-measure-analysis}, it holds that \cref{eq:csp-convergence-limit} can be further upper bounded by 
    \[\conditionalTriangleConstant^{|V|}\lim_{T\to\infty} \sum_{i= T/(2|V|)-2}^{\infty}\tp{\frac{1}{4}}^i.\] 
    So this limit is $0$.
\end{proof} 

Finally, we finish this subsubsection with a proof of \Cref{lemma:csp-large-bad-tree-fixed-size-decay}.
\begin{proof}[Proof of \Cref{lemma:csp-large-bad-tree-fixed-size-decay}]
Recall that in \Cref{Alg:complex-GD-decomposed}, for an integer $t$, we use $o_t\in \Sigma_{v_{i(t)}}$ to denote the update value at the time $t$. Let $\*o = (o_t)_{t = -T + 1}^0$ and for any component $\+C$ in the witness graph $\witnessGraph$, let $\BadTS{\+C} = \bigcup_{x\in \+C} x(0)$ containing all timestamps in the component $\+C$.
We use $\*o_{\BadTS{\+C}}$ to denote the update values for the timestamps in the component $C$.
Let $\*y_{\BadTS{\+C}} \in \bigotimes_{t\in \BadTS{\+C}} \Sigma_{v_{i(t)}}$ be a sequence such that for any timestamp $t\in \BadTS{\+C}$, $y_t\in \Sigma_{v_{i(t)}}$. Let $\+Y = \bigotimes_{t\in \BadTS{\+C}} \Sigma_{v_{i(t)}}$ be the set of all possible $\*y_{\BadTS{\+C}}$'s.

Note for large $2$-trees, i.e., $\abs{\+T} > T/(2\abs{V})-2$, the $o_t$'s within the possible bad component might depend on the initial state (\Cref{lemma:csp-conditional-measure-characterization}).
% We need further information of $o_t$'s for $t\in \{-T+1, -T+2,\dots, -T+|V|\}$.
%  to consider the enumeration of all update values for timestamps within the distance-$2$ neighbors of the bad tree $\+T$.
So we enumerate the update values for timestamps in $\{-T+1, -T+2,\dots, -T+|V|\}$ and use the triangle inequality to consider their contributions.
We claim that this part contributes at most $\conditionalTriangleConstant^{|V|}$.
To see this, recall the definition of $\conditionalTriangleConstant$ in \Cref{condition:csp-conditional-measure-analysis}, $\conditionalTriangleConstant$ maximizes all possible update values of $r$ and $o$ under arbitrary assignments of neighbors. Then this claim follows from an induction process from $-T+1$ to $-T + \abs{V}$.

% Then we use a induction method to bound this quantity, 

% Then we bound the RSH by using the induction method on timestamps in $\BadTS{\+C}$. 

Fix the $o_t$'s for $t\in \{-T+1, -T+2, \dots, -T+|V|\}$. According to the deterministic process in \Cref{definition:2-tree-construction} for the construction of the $2$-tree $\tbad = \=T(\cbad)$, all $\+C \in \=T^{-1}(\tbad)$ must only contain vertices in $\witnessGraph$ that are within distance-$1$ of $\tbad$, and to determine $\cbad$, it suffices to check for all vertices in $\witnessGraph$ that are within distance-$2$ of $\tbad$ by \Cref{definition:csp-conditional-measure-bad-things}.

For vertices in $\+T$, let $\BadTS{\+T}\defeq \bigcup_{x\in \+T}x(0)$ be the set of all timestamps in $\+T$.
To make the event $\cbad = \+C$ happen, it holds that the constraints associated in $\+T$ are not satisfied by the $\*r$. We enumerate $\*r_{\BadTS{\+T}}$.
%  By the triangle inequality, the contributions of this part is at most $\conditionalDecayConstant^{|\+T|-1}$ (\Cref{condition:csp-conditional-measure-analysis}).

Let $D(\+T)$ be the vertices in $\witnessGraph$ within the distance-$2$ neighbors of $\+T$.
Let $\BadTS{D(\+T)}\defeq \bigcup_{x\in D(\+T)} x(0)$ be the set of all timestamps in $D(\+T)$.
Then we enumerate $r_t$'s in $\BadTS{D(\+T)}\setminus \BadTS{\+T}$. Based on the $r_t$'s in $\BadTS{D(\+T)}$, we use the \Cref{definition:csp-conditional-measure-bad-things} to construct the bad component $\cbad$. If the construction fails or $\=T(\cbad)\neq \+T$, we simply ignore this case. And we only consider the case that the construction of the bad component is successful and $\=T(\cbad) = \+T$. 
Then we enumerate the $o_t$'s in $\BadTS{\cbad}$ from the smallest timestamp to the largest timestamp. Let $t$ be the current timestamp and let $v = v_{i(t)}$. We consider the following three cases:
\begin{enumerate}
    \item if $o_t$ is assigned by the enumeration for $t\in \{-T+1,-T+2,\dots, -T+|V|\}$, we consider the next $t$;
    \item if $r_t \neq \bot$, then $o_t = r_t$ by \Cref{observation:csp-decomposition-scheme-r-t-o-t};
    \item if $r_t = \bot$. We consider each constraint $c\in C$ with $v\in \vbl(c)$. 
    Note that $t> -T + |V|$.
    So if $(\ts(c, t), c)\notin \cbad$, it holds that $c$ is satisfied. We denote $X$ as the set of constraints $c$'s that $v\in c$ and $(\ts(c, t),c) \in \cbad$. Let $\BadTS{X}\defeq \bigcup_{c\in X} \ts(c, t)$, and we know that $o_{\BadTS{X}}$ is fixed by previous enumerations. So we can enumerate $o_t$ based on $o_{\BadTS{X}}$.
        % we further have the following twe cases:
    %     \begin{enumerate}
    %         \item $c$ is satisfied. And we consider the next constraint
    %         \item 
    %     \end{enumerate}
\end{enumerate}

Finally, by the triangle inequality and the above analysis, we upper bound the final contributions from $\BadTS{\+T}$ by $\conditionalDecayConstant^{|\+T|}$ (\cref{item:csp-conditional-meaure-analysis-decay} in \Cref{condition:csp-conditional-measure-analysis}).
Similarly, combined \Cref{lemma:csp-witness-graph-degree-bound}, we upper bound the final contributions from $\BadTS{D(\+T)}\setminus \BadTS{\+T}$ by $\conditionalTriangleConstant^{4\Delta^2 k^5 |\+T|}$ (\cref{item:csp-conditional-meaure-analysis-triangle} in \Cref{condition:csp-conditional-measure-analysis}).

So we have that 
\begin{align*}
\sum_{\+C: \=T(\+C)=\+T}\sum_{\*y_{\BadTS{\+C}} \in \+Y}\abs{\dpsi{\sigma}\tp{\*o_{\BadTS{\+C}} = \*y_{\BadTS{\+C}}\land \+C^{\-{bad}}= \+C}} \le \conditionalDecayConstant^{|\+T|} \conditionalTriangleConstant^{4\Delta^2 k^5 |\+T| + |V|}.
\end{align*}
\end{proof}

\subsubsection{Bounding the norm of marginal measures}
\label{subsubsection:bounding-norm}
In this subsubsection, we follow a similar strategy in the last subsubsection to prove \Cref{lemma:csp-projected-measure-norm-sum-bound}.

Let $S$ be a set of constraints and $\tau \in \Sigma_{\vbl(S)}$ be a partial assignment defined on $\vbl(S)$, and let $A$ be the event that (1) $\sbad(\sigma_0) = S$ and (2) for any $v\in \vbl(S)$, we have $\sigma_0(v) = \tau_v$.
Let $\tau'\in\supp(\psi)$.
By \Cref{lemma:csp-projected-measure-marginal-convergence} and \Cref{lemma:bounding-marginal-measure}, it suffices to upper bound
\[
\lim_{T\to \infty} \abs{\sum_{\+T: |\+T|\le T/(2|V|) - 2}\dpsi{\tau'}\tp{\sigma_0 \in A\land \tbad = \+T} }.
\]
By a similar argument in the previous subsubsection, it suffices to upper bound 
\begin{equation}
\label{eq:csp-norm-bound-limit}
\lim_{T\to\infty} \sum_{\+T: |\+T|\le T/(2|V|) - 2} \sum_{C: \=T(C) = \+T} \sum_{\bm{y}_{\BadTS{C}}} \abs{\dpsi{\tau'}\tp{\cbad = C\land \bm{o}_{\BadTS{C}} = \bm{y}_{\BadTS{C}}}}.
\end{equation}

We use the next lemma to analyze contributions from small bad trees with a fixed size by a similar argument in the proof of \Cref{lemma:csp-large-bad-tree-fixed-size-decay}.
\begin{lemma}
\label{lemma:csp-small-bad-trees-fixed-size-decay}
Suppose the \Cref{condition:csp-conditional-measure-analysis}, and 
let \decompositionScheme-decomposition scheme be constructed by \Cref{condition:csp-conditional-measure-analysis}.
Let $c^*$ be the additional variable, 
and let $T\ge |V|$. 
Consider the witness graph $\witnessGraph$, the bad component $\cbad = \cbad(\*r)$, and the bad tree $\tbad = \tbad(\*r) = \=T(\cbad)$, where $\*r = (r_t)_{t=-T+1}^0 \in \bigotimes_{t = -T+1}^0 \Sigma_{v_{i(t)}} \cup \set{\bot}$ is constructed as in \Cref{Alg:complex-GD-decomposed}. 

Then, for any finite $2$-tree $\+T$ in $\witnessGraph$ containing $\witnessTreeRoot$ with $|\+T| \le T/(2|V|) - 1$, we have 
\begin{equation*}%\label{eq:$2$-tree-modulus-bound-subset}
\sum_{\substack{\+C: \=T(\+C) = \+T}} \sum_{\*y_{\BadTS{\+C}}} \abs{\dpsi{\sigma}\tp{\cbad = \+C \land \*o_{\BadTS{\+C}} = \*y_{\BadTS{\+C}}}}\leq \conditionalDecayConstant^{|\+T|} \conditionalTriangleConstant^{4\Delta^2 k^5 |\+T|}.
\end{equation*}
\end{lemma}

We first use the above lemma and \Cref{lemma:csp-witness-graph-degree-bound} to bound the above limit and prove \Cref{lemma:csp-projected-measure-norm-sum-bound}.
% \begin{proof}[Proof of \Cref{lemma:csp-projected-measure-norm-sum-bound}]
% It suffices to bound \cref{eq:csp-norm-bound-limit}. By \Cref{lemma:csp-small-bad-trees-fixed-size-decay}, it holds that \cref{eq:csp-norm-bound-limit} can be upper bounded by $\lim_{T\to \infty} \sum_{i=1}^{T/(2|V|)-2} \tp{4\mathrm{e} \Delta^2 k^4}^{i-1}\cdot \conditionalDecayConstant^{i} \conditionalTriangleConstant^{4\Delta^2 k^5 i}$. Combining with the \cref{item:csp-conditional-meaure-analysis-numerical} of \Cref{condition:csp-conditional-measure-analysis}, it can be upper bounded by $\sum_{i=1}^{T/(2|V|)-2} \tp{\frac{1}{4}}^i < 1$. So this limit is upper bounded by $1$.
% \end{proof}
\begin{proof}[Proof of \Cref{lemma:csp-projected-measure-norm-sum-bound}]
By \Cref{lemma:csp-projected-measure-marginal-convergence}, it suffices to consider the contributions from small bad trees. By \Cref{lemma:csp-small-bad-trees-fixed-size-decay} and \Cref{lemma:2-tree-number-bound}, it holds that 
\begin{align*}
&~\sum_{\substack{\+T: |\+T| \le T/(2|V|) - 2}} \sum_{\substack{\+C: \=T(\+C) = \+T}} \sum_{\*y_{\BadTS{\+C}}} \abs{\dpsi{\sigma}\tp{\cbad = \+C \land \*o_{\BadTS{\+C}} = \*y_{\BadTS{\+C}}}} \\
\le&~ \sum_{i = 1}^{T/(2|V|)-2} \tp{4\mathrm{e}\Delta^2 k^4}^{i}\cdot \conditionalDecayConstant^{i} \conditionalTriangleConstant^{4\Delta^2 k^5 i + |V|}.
\end{align*}
% \begin{align*}
% &~\sum_{\substack{\+T: |\+T| > T/(2|V|) - 2}} \sum_{\substack{\+C: \=T(\+C) = \+T}} \sum_{\*y_{\BadTS{\+C}}} \abs{\dpsi{\sigma}\tp{\cbad = \+C \land \*o_{\BadTS{\+C}} = \*y_{\BadTS{\+C}}}}\\
% \le&~ \sum_{i = T/(2|V|)-2}^{T} \tp{4\mathrm{e}\Delta^2 k^4}^{i-1}\cdot \conditionalDecayConstant^{i - 1} \conditionalTriangleConstant^{4\Delta^2 k^5 (i-1) + 4\Delta^2 k^5 + |V|}.
% \end{align*}
Rearranging the RHS gives that $\sum_{i = 1}^{T/(2|V|)-2} \tp{4\mathrm{e}\Delta^2 k^4\cdot \conditionalDecayConstant\cdot \conditionalTriangleConstant^{4\Delta^2 k^5}}^{i-1}$. Combined with \cref{item:csp-conditional-meaure-analysis-numerical} in \Cref{condition:csp-conditional-measure-analysis}, it holds that it can be upper bounded by $\sum_0^{i = T/(2|V|)-2} \tp{\frac{1}{4}}^{i}$. Putting $T\to \infty$, this is upper bounded by $\frac{1}{2}$. So this lemma follows.
\end{proof}

Finally, we give the omitted proof of \Cref{lemma:csp-small-bad-trees-fixed-size-decay}.
\begin{proof}[Proof of \Cref{lemma:csp-small-bad-trees-fixed-size-decay}]
Recall that in \Cref{Alg:complex-GD-decomposed}, for an integer $t$, we use $o_t\in \Sigma_{v_{i(t)}}$ to denote the update value at the time $t$. Let $\*o = (o_t)_{t = -T + 1}^0$ and for any component $\+C$ in the witness graph $\witnessGraph$, let $\BadTS{\+C} = \bigcup_{x\in \+C} x(0)$ containing all timestamps involved in the component $\+C$.
We use $\*o_{\BadTS{\+C}}$ to denote the update values for the timestamps in the component $C$.
Let $\*y_{\BadTS{\+C}} \in \bigotimes_{t\in \BadTS{\+C}} \Sigma_{v_{i(t)}}$ be a sequence such that for any timestamp $t\in \BadTS{\+C}$, $y_t\in \Sigma_{v_{i(t)}}$. Let $\+Y = \bigotimes_{t\in \BadTS{\+C}} \Sigma_{v_{i(t)}}$ be the set of all possible $\*y_{\BadTS{\+C}}$'s.

Note for small $2$-trees, i.e., $\abs{\+T} \le T/(2\abs{V})-2$, the $o_t$'s within the possible bad component do not depend on the initial state (\Cref{lemma:csp-conditional-measure-characterization}).

According to the deterministic process in \Cref{definition:2-tree-construction} for the construction of the $2$-tree $\tbad = \=T(\cbad)$, all $\+C \in \=T^{-1}(\tbad)$ must only contain vertices in $\witnessGraph$ that are within distance-$1$ of $\tbad$, and to determine $\cbad$, it suffices to check for all vertices in $\witnessGraph$ that are within distance-$2$ of $\tbad$ by \Cref{definition:csp-conditional-measure-bad-things}.

For vertices in $\+T$, let $\BadTS{\+T}\defeq \bigcup_{x\in \+T}x(0)$ be the set of all timestamps in $\+T$.
To make the event $\cbad = \+C$ happen, it holds that the constraints associated in $\+T$ are not satisfied by the $\*r$. We enumerate $\*r_{\BadTS{\+T}}$.
%  By the triangle inequality, the contributions of this part is at most $\conditionalDecayConstant^{|\+T|-1}$ (\Cref{condition:csp-conditional-measure-analysis}).

Let $D(\+T)$ be the vertices in $\witnessGraph$ within the distance-$2$ neighbors of $\+T$.
Let $\BadTS{D(\+T)}\defeq \bigcup_{x\in D(\+T)} x(0)$ be the set of all timestamps in $D(\+T)$.
Then we enumerate $r_t$'s in $\BadTS{D(\+T)}\setminus \BadTS{\+T}$. Based on the $r_t$'s in $\BadTS{D(\+T)}$, we use the \Cref{definition:csp-conditional-measure-bad-things} to construct the bad component $\cbad$. If the construction fails or $\=T(\cbad)\neq \+T$, we simply ignore this case. And we only consider the case that the construction of the bad component is successful and $\=T(\cbad) = \+T$. 
Then we enumerate the $o_t$'s in $\BadTS{\cbad}$ from the smallest timestamp to the largest timestamp. Let $t$ be the current timestamp and let $v = v_{i(t)}$. We consider the following two cases:
\begin{enumerate}
    % \item if $o_t$ is assigned by the enumeration for $t\in \{-T+1,-T+2,\dots, -T+|V|\}$, we consider the next $t$;
    \item if $r_t \neq \bot$, then $o_t = r_t$ by \Cref{observation:csp-decomposition-scheme-r-t-o-t};
    \item if $r_t = \bot$. We consider each constraint $c\in C$ with $v\in \vbl(c)$. 
    Note that $t> -T + |V|$.
    So if $(\ts(c, t), c)\notin \cbad$, it holds that $c$ is satisfied. We denote $X$ as the set of constraints $c$'s that $v\in c$ and $(\ts(c, t),c) \in \cbad$. Let $\BadTS{X}\defeq \bigcup_{c\in X} \ts(c, t)$, and we know that $o_{\BadTS{X}}$ is fixed by previous enumerations. So we can enumerate $o_t$ based on $o_{\BadTS{X}}$.
        % we further have the following twe cases:
    %     \begin{enumerate}
    %         \item $c$ is satisfied. And we consider the next constraint
    %         \item 
    %     \end{enumerate}
\end{enumerate}

Finally, by the triangle inequality and the above analysis, we upper bound the final contributions from $\BadTS{\+T}$ by $\conditionalDecayConstant^{|\+T|}$ (\cref{item:csp-conditional-meaure-analysis-decay} in \Cref{condition:csp-conditional-measure-analysis}).
Similarly, combined with \Cref{lemma:csp-witness-graph-degree-bound}, we upper bound the final contributions from $\BadTS{D(\+T)}\setminus \BadTS{\+T}$ by $\conditionalTriangleConstant^{4\Delta^2 k^5 |\+T|}$ (\cref{item:csp-conditional-meaure-analysis-triangle} in \Cref{condition:csp-conditional-measure-analysis}).

So we have that 
\begin{align*}
\sum_{\+C: \=T(\+C)=\+T}\sum_{\*y_{\BadTS{\+C}} \in \+Y}\abs{\dpsi{\sigma}\tp{\*o_{\BadTS{\+C}} = \*y_{\BadTS{\+C}}\land \+C^{\-{bad}}= \+C}}\le \conditionalDecayConstant^{|\+T|} \conditionalTriangleConstant^{4\Delta^2 k^5 |\+T|}.
\end{align*}
\end{proof}

% \begin{proof}[Proof of \Cref{lemma:csp-projected-measure-norm-sum-bound}]
% By \Cref{lemma:csp-projected-measure-marginal-convergence}, it suffices to consider the contributions from small bad trees. By \Cref{lemma:csp-small-bad-trees-fixed-size-decay} and \Cref{lemma:2-tree-number-bound}, it holds that 
% \begin{align*}
% &~\sum_{\substack{\+T: |\+T| \le T/(2|V|) - 2}} \sum_{\substack{\+C: \=T(\+C) = \+T}} \sum_{\*y_{\BadTS{\+C}}} \abs{\dpsi{\sigma}\tp{\cbad = \+C \land \*o_{\BadTS{\+C}} = \*y_{\BadTS{\+C}}}} \\
% \le&~ \sum_{i = 1}^{T/(2|V|)-2} \tp{4\mathrm{e}\Delta^2 k^4}^{i}\cdot \conditionalDecayConstant^{i} \conditionalTriangleConstant^{4\Delta^2 k^5 i + |V|}.
% \end{align*}
% % \begin{align*}
% % &~\sum_{\substack{\+T: |\+T| > T/(2|V|) - 2}} \sum_{\substack{\+C: \=T(\+C) = \+T}} \sum_{\*y_{\BadTS{\+C}}} \abs{\dpsi{\sigma}\tp{\cbad = \+C \land \*o_{\BadTS{\+C}} = \*y_{\BadTS{\+C}}}}\\
% % \le&~ \sum_{i = T/(2|V|)-2}^{T} \tp{4\mathrm{e}\Delta^2 k^4}^{i-1}\cdot \conditionalDecayConstant^{i - 1} \conditionalTriangleConstant^{4\Delta^2 k^5 (i-1) + 4\Delta^2 k^5 + |V|}.
% % \end{align*}
% Rearranging the RHS gives that $\sum_{i = 1}^{T/(2|V|)-2} \tp{4\mathrm{e}\Delta^2 k^4\cdot \conditionalDecayConstant\cdot \conditionalTriangleConstant^{4\Delta^2 k^5}}^{i-1}$. By combining with \cref{item:csp-conditional-meaure-analysis-numerical} in \Cref{condition:csp-conditional-measure-analysis}, it holds that it can be upper bounded by $\sum_0^{i = T/(2|V|)-2} \tp{\frac{1}{4}}^{i}$. Putting $T\to \infty$, this is upper bounded by $\frac{1}{2}$. So this lemma follows.
% \end{proof}

\section{Applications}
\label{section:applications}
In this section, we apply the framework in \Cref{section:CSP} and consider the zero-freeness of hypergraph $q$-coloring and $(k, \Delta)$-CNF formulas, proving \Cref{theorem:coloring-one-special-zero-free-intro} and \Cref{theorem:cnf-zero-freeness-intro} respectively.

\subsection{Hypergraph $q$-coloring with one special color}
\label{section:application-coloring-one-special-color}
% Let $H = (V, \mathcal{E})$ be a $k$-uniform hypergraph with maximum degree $\Delta$. We study the problem of hypergraph $q$-colorings on $H$, where each edge must be non-monochromatic. 
In this subsection, we consider the complex extension of hypergraph $q$-coloring with a complex external field on the color $1$. We show that under certain conditions, the complex partition function is zero-free. 
We formally define the partition function as follows.
{
\color{black}
\begin{definition}[Complex extension for hypergraph $q$-coloring with one special color]
\label{definition:coloring-one-special-complex-extensions}
Let $H=(V, \+E)$ be a $k$-uniform hypergraph with maximum degree $\Delta$.  
% A coloring on the vertices of $H$ is said to be \emph{proper} if every hyperedge is not monochromatic.
%  To study Lee-Yang zeros, %in hypergraph $q$-colorings on $H$,
% we introduce a complex external field on the color $1$. 
% The choice of color $1$ is arbitrary due to symmetry.
The \emph{partition function} is given by:
\[
Z^{\-{co}}_H(\lambda)\defeq\sum_{\sigma\in[q]^V: \sigma \text{ is a proper coloring in }H} \lambda^{ \abs{ \sigma^{-1}\left(
1\right)}} , \qquad \hbox{where }\sigma^{-1}\left(
1\right)\defeq\set{v\in V\mid \sigma(v) = 1}. 
\]
When $Z^{\-{co}}_H(\lambda)\neq 0$, we naturally define a \emph{complex measure} $\mu = \mu(H, \lambda)$, i.e.,$\forall \sigma\in [q]^V, \mu(\sigma)\defeq \frac{\lambda^{|\sigma^{-1}(1)|}}{Z^{\-{co}}_H(\lambda)}$.
And we denote its \emph{support set} as $\supp(\mu) \defeq \{\sigma\in \Omega_{\Phi}\mid \mu(\sigma)\neq 0\}$.
% \end{definition}
\end{definition}
}

In order to apply techniques in \Cref{section:CSP},
we consider the $(k, \Delta)$-CSP formula $\Phi = (V, [q]^V, C)$ that encodes the hypergraph $q$-colorings on $H$. Each constraint $c \in C$ corresponds to a hyperedge and enforces that the variables in $\vbl(c)$ (i.e., the vertices in the hyperedge) are not assigned the same color.
% \begin{definition}[Complex extension of hypergraph $q$-coloring
Then we use \Cref{theorem:sufficient-condition-zero-freeness} to prove the zero-freeness, and we use \Cref{lemma:coloring-one-special-verify-sufficient-condition-zero-freeness-condition} to verify the \Cref{condition:csp-conditional-measure-analysis}.

We first establish the state-compression scheme (\Cref{definition:state-compression}) where we project $q$ colors into different ``color buckets'' and the first color bucket only contains the color $1$. And we add the complex external field only to the first color bucket. We use the superscript $\coF$ to denote the projected symbol.
% we add the complex external field to the colors in one ``color bucket''. 
%\color{red}Let $w_1$ and $w_2$ be two integers. Let $B$ be an integer which is the number of color buckets.\color{black}
\begin{definition}[State-compression scheme for hypergraph $q$-coloring with one special color]
\label{definition:coloring-one-special-state-compression}
For each $v\in V$, let $\Sigma_v = \{1^\coF, 2^\coF, \dots, B^\coF, (B+1)^\coF\}$, and let $\*f = (f_v)_{v\in V}$. 
And for each variable $v\in V$, $f_v$ is a mapping from its domain $[q]$ to a finite \emph{alphabet} $\Sigma_v$.
%  and $g_v$ is the other mapping from $\Sigma_v$ to the other \emph{alphabet} $\Sigma_v$.
Let $\*\Sigma\defeq \bigotimes_{v\in V}\Sigma_v$. And for any $\Lambda \subseteq V$, we denote $\*\Sigma_\Lambda \defeq \bigotimes_{v\in \Lambda} \Sigma_v$.
% And let $\*h = (h_v)_{v\in V}$ be the composition of $\*f$ and $\*g$ where $\forall v\in V$, $h_v = g_v \circ f_v$.

For each $v\in V$, we set $f_v(1) = 1^\coF$, and for any $i\in \set{2, 3,\dots, q}$, we set $f_v(i) = (((i-2)\mod B) + 2)^\coF$. 
% And for each $i\in [w_1]$, we set $g_v(i^\coF) = ((i\mod w_2) + 1)^\coS$.
\end{definition}

\begin{fact}[Size of color buckets]
\label{fact:coloring-one-special-size}
Let $s \defeq \lfloor (q-1)/B\rfloor$.
%  and $s_2 = \lfloor q/w_2\rfloor$. 
Based on the state-compression scheme in \Cref{definition:coloring-one-special-state-compression}, it holds that for each color bucket $j^\coF$ where $j\in \set{2, 3,\dots, B+1}$ and for each variable $v\in V$, we have $s \le |f^{-1}_v(j^\coF)| \le s + 1$.
\end{fact}

Note that the partition function induced by the above projection with the complex external field $\lambda$ on the first color bucket $1^\coF$ is equivalent to \Cref{definition:coloring-one-special-complex-extensions}.
We give the following sufficient condition for zero-freeness.
\begin{condition}
\label{condition:coloring-one-special-main-condition}
It holds that
\begin{enumerate}
    \item $s^k \ge 608\mathrm{e}q^2 \Delta^3 k^5$;\label{item:coloring-one-special-local-uniformity-constraint}
    \item $ 16\mathrm{e}^2 \Delta^2 k^4 \cdot q\tp{\frac{4s}{q}}^k \le 1 $;\label{item:coloring-one-special-LLL-constraint}
    \item let $\lambda$ be the complex parameter defined in \Cref{definition:coloring-one-special-complex-extensions} and there is a $\lambda_c\in [0,1]$ such that $|\lambda - \lambda_c|\le \gamma$ where $\gamma = \coGamma$.
    %  it satisfies $\abs{\lambda-1}\le \coGamma$. %is in the disk centered at $1$ with radius $\gamma = \coGamma$, i.e.,  $\lambda \in \+D(1, \coGamma)$.
    \label{item:coloring-one-special-lambda-constraints}
\end{enumerate}
\end{condition}

\begin{lemma}
\label{lemma:coloring-one-special-verify-sufficient-condition-zero-freeness-condition}
% Let $H = (V, \+E)$ be a $k$-uniform hypergraph with maximum degree $\Delta$. 
% Let $\Phi = (V, [q]^V, C)$ be a CSP formula that encodes the hypergraph $q$-coloring with the state-compression scheme $\*f$ (\Cref{definition:coloring-one-special-state-compression}) and the complex external field $\lambda$ on the color bucket $1^\coF$. Let $C = \set{c_1, c_2, \dots, c_m}$ and for any  $i\le m$, let $C_i = \set{c_1, c_2, \dots, c_i}$, $\Phi_i = (V, [q]^V, C_i)$.
Let $H = (V, E)$ be a $k$-uniform hypergraph with maximum degree $\Delta$.
Let $\Phi = (V, [q]^V, C)$ be the CSP formula encoding the hypergraph $q$-coloring
under the state-compression scheme $f$
(\Cref{definition:coloring-one-special-state-compression})
and the complex external field $\lambda$ on the color bucket $1^{\coF}$.
Write $C = \{c_1, c_2, \dots, c_m\}$, and for each $i \le m$, define
$C_i = \{c_1, c_2, \dots, c_i\}$ and $\Phi_i = (V, [q]^V, C_i)$.

If \Cref{condition:coloring-one-special-main-condition} holds, $\Phi_1, \Phi_2, \dots, \Phi_m$ satisfy \Cref{condition:csp-conditional-measure-analysis}.
\end{lemma}

The next theorem about zero-freeness follows directly from the above lemma.

\begin{theorem}[Zero-freeness]
\label{theorem:co-one-special-zero-freeness}
Let $H = (V, \+E)$ be a $k$-uniform hypergraph with maximum degree $\Delta$. 
Suppose that \Cref{condition:coloring-one-special-main-condition} holds. Then $Z^{\-{co}}_H(\lambda) \neq 0$.
\end{theorem}
\begin{proof}
Let $\Phi = (V, [q]^V, C)$ be a CSP formula that encodes the hypergraph $q$-coloring on $H$ with the state-compression scheme $\*f$  (\Cref{definition:coloring-one-special-state-compression}) and the complex external field $\lambda$ on the color bucket $1^\coF$.
We use \Cref{theorem:sufficient-condition-zero-freeness} to prove this theorem. We first verify \cref{eq:csp-Z_0-is-not-0}. By \Cref{condition:coloring-one-special-main-condition}-(\ref{item:coloring-one-special-lambda-constraints}), it holds that for any $v\in V$, $\sum_{i\in [q]} \`R(\=I[f_v(i) \neq 1^\coF] + \lambda \cdot \=I[f_v(i) = 1^\coF]) > 0$. So \cref{eq:csp-Z_0-is-not-0} holds.
Next, the \Cref{condition:csp-conditional-measure-analysis} holds by \Cref{lemma:coloring-one-special-verify-sufficient-condition-zero-freeness-condition}.
Then, this theorem follows from \Cref{theorem:sufficient-condition-zero-freeness}.
\end{proof}
% \color{blue}
We first use the above theorem to prove \Cref{theorem:coloring-one-special-zero-free-intro}.
% \color{red}
\begin{proof}[Proof of \Cref{theorem:coloring-one-special-zero-free-intro}]
Let $\gamma = \coGamma$.
We set $B = \lfloor q^{2/5} \rfloor$, $s = \lfloor (q-1)/B \rfloor$.
Next, given $k\ge 50$ and $\coloringCondition$, for any $\lambda_c\in [0,1]$ and $\lambda: |\lambda - \lambda_c|\le \gamma $, we verify \Cref{condition:coloring-one-special-main-condition}.

For \Cref{condition:coloring-one-special-main-condition}-(\ref{item:coloring-one-special-local-uniformity-constraint}),
we observe that $s = \lfloor \frac{q - 1}{\lfloor q^{2/5}\rfloor}\rfloor \ge \frac{q-1}{q^{2/5}} - 1 \ge q^{3/5}-2 \ge \frac{q^{3/5}}{2}$ where the last inequality is due to $q\ge 700$.
In order to satisfy $s^k \ge 608 \mathrm{e}q^2 \Delta^3 k^5$,  it suffices if $(q^{3/5}/2)^k \ge 608\mathrm{e} q^2 \Delta^3 k^5$ which is equivalent to $q\ge (608\mathrm{e}2^k k^5)^{\frac{5}{3k-10}}\Delta^{\frac{15}{3k-10}}$. Combined with $k\ge 50$, it suffices if $q\ge 10\Delta^{\frac{5}{k-4}}$.

Next, we consider \Cref{condition:coloring-one-special-main-condition}-(\ref{item:coloring-one-special-LLL-constraint}) and verify that  $16\mathrm{e}^2 \Delta^2 k^4\cdot q(4s/q)^k\le 1$. By the fact that $s = \lfloor \frac{q-1}{\lfloor q^{2/5}\rfloor}\rfloor \le 2\lfloor \frac{q-1}{q^{2/5}}\rfloor \le 2 q^{3/5}$. 
% And we have that $s + 2 \le 2 q^{3/5}\le 4q^{3/5}$.
So, it suffices if $16\mathrm{e}^2 \Delta^2 k^4 8^k \le q^{2k/5 - 1}$. After rearranging, it is equivalent to $q\ge (16\mathrm{e}^2 8^k k^4)^{\frac{5}{2k-5}} \Delta^{\frac{10}{2k-5}}$. Combined with $k\ge 50$, it suffices if $q\ge 700 \Delta^{\frac{10}{2k - 5}}$.

Combined with $\coloringCondition$, \Cref{condition:coloring-one-special-main-condition} holds.
% Then by \Cref{definition:coloring- projection-scheme}, it holds that for any $v$, $f^{-1}_v(1^{\coF}) = \lceil \frac{q}{\lfloor q^{2/5}\rfloor}\rceil$. 
By \Cref{theorem:co-one-special-zero-freeness}, \Cref{theorem:coloring-one-special-zero-free-intro} holds.
\end{proof}
The remaining subsection is devoted to the proof of \Cref{lemma:coloring-one-special-verify-sufficient-condition-zero-freeness-condition}.
Let $\Phi$ be the CSP formula that encodes the hypergraph $q$-coloring on the hypergraph $H$.
Suppose \Cref{condition:coloring-one-special-main-condition} holds.
Note that under the constraint-wise self-reduction, new CSP formulas are also $(k, \Delta)$-CSP formulas satisfying \Cref{condition:coloring-one-special-main-condition} and they also encode hypergraph $q$-colorings for different hypergraphs.
Hence, we prove a stronger result that for any CSP formula encoding a hypergraph $q$-coloring, if \Cref{condition:coloring-one-special-main-condition} holds for the formula, then \Cref{condition:csp-conditional-measure-analysis} holds.
Thus, we have that $\Phi_1, \Phi_2, \dots, \Phi_m$ satisfy \Cref{condition:csp-conditional-measure-analysis}.

We first verify \Cref{condition:csp-conditional-measure-analysis}-(\ref{item:csp-transition-well-define}).
% Let $\psi$ be the projected measure under the projection $\*f$ as defined in \Cref{definition:csp-projected-measure-under-f}.
% We first verify that under \Cref{condition:coloring-one-special-main-condition}, if $\psi$ is well-defined, then \Cref{condition:csp-conditional-measure-analysis}-(\ref{item:csp-transition-well-define}) holds.
\begin{lemma}[Well-definedness]
\label{lemma:coloring-one-special-well-definedness}
% Let $H = (V, \+E)$ be a $k$-uniform hypergraph with maximum degree $\Delta$. 
Let $\Phi = (V, [q]^V, C)$ be a $(k, \Delta)$-CSP formula that encodes a hypergraph $q$-coloring, $\*f$ be the state-compression scheme (\Cref{definition:coloring-one-special-state-compression}) and the complex external field $\lambda$ on the color bucket $1^\coF$. Let $c^*\in C$ with the largest index number and let $\Phi' = (V, [q]^V, C\setminus \set{c^*})$.

Suppose the projected measure $\psi$ induced by $\Phi'$ and $\*f$ has a nonzero partition function, i.e., $Z(\Phi', \lambda, \*f, 1^\coF)\neq 0$ and \Cref{condition:coloring-one-special-main-condition} holds, then \Cref{condition:csp-conditional-measure-analysis}-(\ref{item:csp-transition-well-define}) holds.
\end{lemma}

In order to analyze the transition measure of the projected measure $\psi$, we first introduce the next simplified CSP formula.
Given a variable $v$ and an extendable partial assignment $\tau\in\*\Sigma_{V\setminus\set{v}}$ defined on $V\setminus \set{v}$ under the projection. We use the partial assignment $\tau$ to simplify the CSP formula $\Phi'$, i.e., removing all satisfied constraints and all redundant variables that are not in any constraint.
Let $\Phi^\tau = (V^{\tau}, \*Q^{\tau}, C^{\tau})$ be the simplified CSP formula. 
For any variable $u$ in $V\setminus \set{v}$, we set $Q^{\tau}_u = f_u^{-1}(\tau_u)$.
And we set $Q^{\tau}_v = [q]$.
For any constraint $c$ with $v\in\vbl(c)$ and $\forall u\in \vbl(c)\setminus\set{v}$ such that $\tau_u = 1^\coF$, we remove the constraint $c$ and remove $f_v^{-1}(1^\coF)$ from the alphabet $Q^{\tau}_v$.

If $Q^{\tau}_v = [q]\setminus f_v^{-1}(1^\coF)$, we say $v$ is \emph{blocked}. Otherwise $Q^{\tau}_v = [q]$, we say $v$ is \emph{free}.
We have the next lemma.
\begin{lemma}
\label{lemma:coloring-simplify-CSP-by-pinning}
For any variable $v$ and any extendable partial assignment $\tau\in\*\Sigma_{V\setminus\set{v}}$ defined on $V\setminus \set{v}$ under the projection, let $\Phi^\tau=(V^\tau, \*Q^{\tau}, C^{\tau})$ be the simplified CSP formula. Then for any variable $u\in V^{\tau} \setminus \set{v}$, it holds that $|Q_u|\ge s$.
\end{lemma}
\begin{proof}
Combined with \Cref{fact:coloring-one-special-size}, it suffices to show that for any constraint $c\in C^{\tau}$, it holds that $c$ does not contain the color bucket $1^\coF$.

For constraints $c$ with $v\not\in \vbl(c)$, it holds that for any variables in this constraint, they are in the same color bucket. This is because that if it contains different color buckets, $c$ is already satisfied, so $c$ is removed.
Then, we show that $c$ does not contain the color bucket $1^\coF$. If it contains, recall that there is only one color in the color bucket $1^\coF$, so $c$ must be violated. This contradicts to the fact that $\tau$ is extendable.

For constraints $c$ with $v\in \vbl(c)$, by similar reasons, for any variables in $c$ except for the variable $v$, they are in the same color bucket.
Then recall the construction of $\Phi^{\tau}$, if a such constraint $c$ contains the color bucket $1^\coF$, $c$ is also removed.

This lemma follows from the above two cases.
\end{proof}
Next, we introduce the following lemma which shows that when considering the transition measure of the projected measure $\psi$, it suffices to consider the complex external field only on the variable being updated.
\begin{lemma}
\label{lemma:coloring-one-special-transition-measure-expression}
For any variable $v\in V$, for any extendable partial assignment $\tau\in \*\Sigma_{V\setminus \{v\}}$ defined on $V\setminus \{v\}$ under the projection, let $\Phi^\tau = (V^\tau, \*Q^{\tau}, C^{\tau})$ be the simplified CSP formula.
Let $(X_i)_{i\in [B+1]}$ be the number of satisfying assignments of $\Phi^{\tau}$ with the value of $v$ is consistent with $i^\coF$.
% For any $i\in[q]$, let $X_i$ be the number of satisfying assignments under the alphabet $[q]^V$ that are consistent with $\tau$ and the value of $v$ is consistent with $i^\coF$.

Then it holds that
\[
\forall i\in [B+1], \quad \psi^{\tau}_v(i^\coF) = \frac{\=I[i^\coF = 1^\coF]\cdot \lambda X_i + \=I[i^\coF \neq 1^\coF]\cdot X_i }{\sum_{j\in [B+1]}\tp{\=I[j^\coF = 1^\coF]\cdot \lambda X_j + \=I[j^\coF \neq 1^\coF]\cdot X_j}}.
\]
\end{lemma}
\begin{proof}
For any $i\in[B+1]$, let $(\tilde{X}_i)_{i\in [B+1]}$ be the number of satisfying assignments of the original CSP formula $\Phi'$ under the alphabet $[q]^V$ that are consistent with $\tau$ and the value of $v$ is consistent with $i^\coF$.
We claim that there exists a constant integer $\zeta\ge 1$, such that for any $i\in[B+1]$, we have $\tilde{X}_i = \zeta \cdot X_i$.

By \Cref{definition:coloring-one-special-complex-extensions}, we only add complex external fields on the color bucket $1^\coF$. And all contributions of external fields from $V\setminus \set{v}$ are fixed by the partial assignment $\tau$. So we only need to consider the external field of the variable $v$. So this lemma follows.

Finally, we prove the claim. For any $i\in [B+1]$, we show that every assignment that contributes to $X_i$ corresponds to $\zeta$ different assignments that contributes to $\tilde{X}_i$ by adding new variables. And every assignment that contributes to $\tilde{X}_i$ corresponds to one assignment that contributes to $X_i$ by removing variables.

Let $\zeta$ be the number of partial assignment on the variables in $V\setminus V^{\tau}$ which are consistent to the partial assignment $\tau$.
Note that these variables are actually independent of the CSP formula $\Phi'$. To see this, in the simplified CSP formula $\Phi^{\tau}$, they are not in any constraints.

For any assignment $\sigma\in[q]^{V^\tau}$ contributes to $X_i$, we enumerate partial assignments on the variables in $V\setminus V^{\tau}$, which are consistent to the partial assignment $\tau$. Then we append them into $\sigma$. Then these assignments are satisfying assignments with respect to $\Phi'$.

For any assignment $\sigma\in[q]^V$ contributes to $\tilde{X}_i$, we remove the values on $V\setminus V^{\tau}$. We show that $\sigma$ is a satisfying assignment with respect to $\Phi^\tau$. Note that $C^{\tau} \subseteq C$, so $\sigma$ satisfies all constraints in $C^{\tau}$. Then we consider the $Q_v$. If $v$ is blocked ($Q_v = [q]\setminus \set{f_v^{-1}(1^\coF)}$), by the construction of $\Phi^{\tau}$, for $i = 1$, it holds that $\tilde{X}_i = X_i = 0$, and for $i\neq 1$,  it holds that $\sigma_v \not\in f_v^{-1}(1^\coF)$. If $v$ is free ($Q_v = [q]$), it holds trivially that $\sigma_v\in [q]^V$.

Combined the above two cases, we have that there exists an integer $\zeta \ge 1$ such that $\forall i \in [B+1], \tilde{X}_i = \zeta\cdot X_i$.
\end{proof}

We also include the following useful lemma established in~\cite{guo2018counting,feng2022improved}.
\begin{lemma}[\text{\cite[Lemma 7]{guo2018counting} and \cite[Lemma 6]{feng2022improved}}]
\label{lemma:coloring-local-uniformity}
Let $\Phi = (V, \*Q, C)$ encodes a hypergraph $q$-coloring.
Suppose $q_0 \le |Q_v|\le q_1$ for any $v\in V$. 
Let $\mu$ be the uniform distribution over all proper hypergraph colorings.
For any $\varrho\ge k \ge 2$, if $q_0^k \ge \mathrm{e}\cdot q_1 \cdot \varrho \cdot \Delta$, then for any $v\in V$ and any color $c\in Q_v$,
\[
\frac{1}{\abs{Q_v}}\tp{1 - \frac{1}{\varrho}}\le \mu_v(c)\le \frac{1}{\abs{Q_v}}\tp{1 + \frac{4}{\varrho}}.
\]
\end{lemma}

Next, we use the above two lemmas to prove \Cref{lemma:coloring-one-special-well-definedness}.
\begin{proof}[Proof of \Cref{lemma:coloring-one-special-well-definedness}]
Under the assumption that $\psi$ is well-defined, we show that for any $v\in V$, any extendable partial assignment $\tau \in \*\Sigma_{V\setminus \{v\}}$ under the
 projection, $\psi^\tau_v$ is well-defined.

Let $\Phi^\tau = (V^\tau, \*Q^{\tau}, C^{\tau})$ be the simplified CSP formula.
Let $(X_i)_{i\in [B+1]}$ be the number of satisfying assignments of $\Phi^{\tau}$ with the value of $v$ is consistent with $i^\coF$.
% We define $X_1, X_2, \dots, X_{B+1}$, where $X_i$ denotes the number of hypergraph $q$-colorings of $\Phi$ that are consistent with $\tau$ and the color of $v$ belongs to that color bucket $i^{\coF}$.
Note that $\tau \in \Sigma_{V\setminus \set{v}}$ is an extendable partial assignment under the projection. So it holds that there exists $j\in[B+1]$, such that $X_j > 0$.

By \Cref{lemma:coloring-one-special-transition-measure-expression}, we have
\[
\forall i\in [B+1], \quad \psi^{\tau}_v(i^\coF) = \frac{\=I[i^\coF = 1^\coF]\cdot \lambda X_i + \=I[i^\coF \neq 1^\coF]\cdot X_i }{\sum_{j\in [B+1]}\tp{\=I[j^\coF = 1^\coF]\cdot \lambda X_j + \=I[j^\coF \neq 1^\coF]\cdot X_j}}.
\]
So it suffices to show that $\sum_{j\in [B+1]}\tp{\=I[j^\coF = 1^\coF]\cdot \lambda X_j + \=I[j^\coF \neq 1^\coF]\cdot X_j}  = \lambda \cdot X_1 + \sum_{j\in [B+1]\setminus \set{1}} X_j\neq 0$.
For a complex number $x$, we use $\`R(x)$ to denote its real part.
It suffices to show that $\`R(\lambda \cdot X_1 + \sum_{j\in [B+1]\setminus \set{1}} X_j) \ge -\gamma \cdot X_1 + \sum_{j\in [B+1]\setminus \set{1}} X_j > 0$.
We show this by a proof of contradiction. 
Assume that $-\gamma \cdot X_1 + \sum_{j\in [B+1]\setminus \set{1}} X_j\le 0$. 
If $X_1 = 0$, recall that there exists a $j\in [B+1]$ with $X_j > 0$, so $-\gamma \cdot X_1 + \sum_{j\in [B+1]\setminus \set{1}} X_j > 0$ which reaches a contradiction.
Hence, we assume that $X_1 \ge 1$.
By $\sum_{j\in [B+1]\setminus \set{1}} X_j\le \gamma \cdot X_1$, it holds that $\frac{X_1}{\sum_{j\in [B+1]}X_j}\ge \frac{1}{1+\gamma}$.

Next, we use \Cref{lemma:coloring-local-uniformity} to show an upper bound such that $\frac{X_1}{\sum_{j\in [B+1]}X_j} < \frac{1}{1+\gamma}$ which reaches a contradiction.
% In order to apply \Cref{lemma:coloring-local-uniformity}, we first use the partial assignment (color bucket assignment) to simplify the CSP formula, i.e., removing satisfied constraints.
% And if a variable is not in any constraint, we also remove this variable.
% Note that this simplified CSP formula $\Phi'= (V', \*Q, C')$ encodes a new hypergraph coloring problem.
% In this simplified CSP formula, we consider two cases (1) there exists a constraint $c$ with $v\in\vbl(c)$, it holds that all variables except for $v$ are in the color bucket $1^\coF$, and (2) for any constraint $c$ with $v\in\vbl(c)$, it holds that there exists at least one variable $v$ whose color bucket is not $1^\coF$.
If $v$ is blocked, we actually have $\frac{X_1}{\sum_{j\in [B+1]}X_j} = 0$.
If $v$ is free, note that in the simplified CSP formula $\Phi^\tau$, for any $u\in V^\tau$, we have $|Q^{\tau}_u| \ge s$ (\Cref{lemma:coloring-simplify-CSP-by-pinning}). 
% To see this, by \Cref{fact:coloring-one-special-size}, it suffices to show that for any variable $u \in V'$, it holds that $\tau_u \neq 1^\coF$. 
% Note that for any constraint $c$ with $v\not\in \vbl(c)$, it holds that all variables in $c$ are in the same color bucket since we already remove all satisfied constraints. Recall that the first color bucket $1^\coF$ contains only the color $1$. 
% Hence, for these constraints, they do not contain variables with the color bucket $1^\coF$. This is because, if they do, the original coloring that is consistent with the color bucket assignment is not a proper coloring.
% For constraints containing $v$, by similar reasons, variables in a such constraint except $v$ are in the same color bucket.
% Combined with the assumption that any constraints containing $v$ have no color bucket $1^\coF$, we have that these constraints do not contain the color bucket $1^\coF$.
% Thus, there is no color bucket $1^\coF$. So we can apply \Cref{lemma:coloring-local-uniformity}.
Let $\varrho = \coRHO$.
Note that $s^k \ge \mathrm{e} q \varrho \Delta$ holds by \Cref{condition:coloring-one-special-main-condition}-(\ref{item:coloring-one-special-local-uniformity-constraint}).
By \Cref{lemma:coloring-local-uniformity}, it holds that $\frac{X_1}{\sum_{j\in [B+1]}X_j} \le \frac{1}{q}\tp{1 + \frac{4}{\varrho}} < \frac{1}{1 + \gamma}$.

Therefore we have that $\psi_v^\tau$ is well-defined and \Cref{condition:csp-conditional-measure-analysis}-(\ref{item:csp-transition-well-define}) holds.
\end{proof}

Now, we construct the \decompositionScheme-decomposition scheme that satisfies \Cref{condition:csp-conditional-measure-analysis}.
And after that, we provide some intuitions of it (\Cref{remark:coloring-one-special-decomposition-scheme}).

\begin{definition}[\decompositionScheme-decomposition scheme]
\label{definition:coloring-one-special-decomposition-scheme}
Let $\Phi = (V, [q]^V, C)$ be a $(k, \Delta)$-CSP formula that encodes a hypergraph $q$-coloring, $\*f$ be the state-compression scheme (\Cref{definition:coloring-one-special-state-compression}) and the complex external field $\lambda$ on the color bucket $1^\coF$. Let $c^*\in C$ with the largest index number and let $\Phi' = (V, [q]^V, C\setminus \set{c^*})$.

Suppose the projected measure $\psi$ induced by $\Phi'$ and $\*f$ has a nonzero partition function, i.e., $Z(\Phi', \lambda, \*f, 1^\coF)\neq 0$ and \Cref{condition:coloring-one-special-main-condition} holds.
% Let $H = (V, \+E)$ be a $k$-uniform hypergraph with maximum degree $\Delta$. 
% Let $\Phi = (V, [q]^V, C)$ be a CSP formula that encodes the hypergraph $q$-coloring with the state-compression scheme $\*f$  (\Cref{definition:coloring-one-special-state-compression}) and the complex external field $\lambda$ on the color bucket $1^\coF$.
% Suppose the projected measure $\psi$ induced by $\*f$ has a nonzero partition function and \Cref{condition:coloring-one-special-main-condition} holds.
Let $\varrho =  \coRHO$.
For any variable $v \in V$, and for any $x\in \Sigma_v\cup\set{\bot}$, we set
\[
b_v(x)\defeq \begin{cases}
    0 & x = 1^\coF\\
    \frac{|f_v^{-1}(x)|}{q + (\lambda - 1)}\tp{1 - \frac{1}{\varrho}} & x\neq 1^\coF \\
    \frac{\lambda}{q-1+\lambda} + \frac{q-1}{\varrho(q-1+\lambda)} & x=\bot \\
\end{cases}.
\]

\end{definition}
We now provide some intuitions about how we set $b_v$.
\begin{remark}
\label{remark:coloring-one-special-decomposition-scheme}
Consider one transition step of the complex systematic scan Glauber dynamics.
Let $v\in V$ be the variable that we are updating, and let $\tau \in \*\Sigma_{V\setminus \{v\}}$ be the current extendable partial assignment under the projection.
Let $\Phi^\tau = (V^\tau, \*Q^{\tau}, C^{\tau})$ be the simplified CSP formula.
Let $(X_i)_{i\in [B+1]}$ be the number of satisfying assignments of $\Phi^{\tau}$ with the value of $v$ is consistent with $i^\coF$.
% We define $X_1, X_2, \dots, X_{B+1}$, where $X_i$ denotes the number of hypergraph $q$-colorings of $\Phi$ that are consistent with $\tau$ and the color of $v$ belongs to that color bucket $i^{\coF}$.
By \Cref{lemma:coloring-one-special-transition-measure-expression}, we have
\[
\forall i\in [B+1], \quad \psi^{\tau}_v(i^\coF) = \frac{\=I[i^\coF = 1^\coF]\cdot \lambda X_i + \=I[i^\coF \neq 1^\coF]\cdot X_i }{\sum_{j\in [B+1]}\tp{\=I[j^\coF = 1^\coF]\cdot \lambda X_j + \=I[j^\coF \neq 1^\coF]\cdot X_j}}.
\]
% Next, we use $\tau\in\Sigma_{V\setminus \set{v}}$ to simplify the CSP formula $\Phi$, i.e., removing all constraints that are already satisfied by $\tau$.
% And if a variable is not in any constraint, then we remove this variable.
% Let $\Phi'$ be the simplified CSP formula, let $\+C'$ be the set of constraints of $\Phi'$. 
% For any constraint $c\in \+C'$, if $v\notin \vbl(c)$, then it holds that for any variable $u\in\vbl(c)$, we have their $\tau_u$'s are the same.
% If $v\in \vbl(c)$, then it holds that for any variable $u\in\vbl(c)\setminus {v}$, we have their color buckets $\tau_u$'s are the same.
For the purpose of illustration, we assume that $v$ is free.
% We further assume that $\lambda \approx 1$.
By \Cref{lemma:coloring-local-uniformity}, we expect that $\frac{X_1}{\sum_{j\in[B+1]}X_i} = \frac{1}{q}$, and for any $i\in \set{2, 3, \dots, B+1}$, we expect $\frac{X_i}{\sum_{j\in[B+1]}X_i} = \frac{|f^{-1}_v(i^\coF)|}{q}$. Then the setting of $b_v$ follows by replacing these $X_i$' into the expression of $\psi_v^{\tau}$.
\end{remark}
Then, we show that \Cref{condition:coloring-one-special-main-condition} implies \Cref{condition:csp-conditional-measure-analysis}.
\begin{lemma}
\label{lemma:coloring-one-special-verify-decay-condition}
Let $\Phi = (V, [q]^V, C)$ be a $(k, \Delta)$-CSP formula that encodes a hypergraph $q$-coloring, $\*f$ be the state-compression scheme (\Cref{definition:coloring-one-special-state-compression}) and the complex external field $\lambda$ on the color bucket $1^\coF$. Let $c^*\in C$ with the largest index number and let $\Phi' = (V, [q]^V, C\setminus \set{c^*})$.

Suppose the projected measure $\psi$ induced by $\Phi'$ and $\*f$ has a nonzero partition function, i.e., $Z(\Phi', \lambda, \*f, 1^\coF)\neq 0$ and \Cref{condition:coloring-one-special-main-condition} holds, let \decompositionScheme be the decomposition scheme defined in \Cref{definition:coloring-one-special-decomposition-scheme}, then \Cref{condition:csp-conditional-measure-analysis}-(\ref{item:csp-decomposition-decay}) holds.
\end{lemma}
Before proving the above lemma, we first use it to prove \Cref{lemma:coloring-one-special-verify-sufficient-condition-zero-freeness-condition}.
\begin{proof}[Proof of \Cref{lemma:coloring-one-special-verify-sufficient-condition-zero-freeness-condition}]
Note that under the constraint-wise self-reduction, $\Phi_1, \Phi_2, \dots, \Phi_{m}$ are also $(k, \Delta)$-CSP formulas and encode hypergraph $q$-colorings.
Then lemma follows directly from \Cref{lemma:coloring-one-special-well-definedness} and \Cref{lemma:coloring-one-special-verify-decay-condition}.
\end{proof}

Finally, we prove \Cref{lemma:coloring-one-special-verify-decay-condition}.
\begin{proof}[Proof of \Cref{lemma:coloring-one-special-verify-decay-condition}]
Recall the definition of $\conditionalDecayConstant$ in \cref{item:csp-conditional-meaure-analysis-decay} of \Cref{condition:csp-conditional-measure-analysis}, and the definition of $\conditionalTriangleConstant$ in \cref{item:csp-conditional-meaure-analysis-triangle} of \Cref{condition:csp-conditional-measure-analysis}.
We claim that $\conditionalDecayConstant\le \mathrm{e}q\tp{\frac{4s}{q}}^k$, and $\conditionalTriangleConstant \le 1 + \frac{1}{4\Delta^2 k^5}$. Then by \Cref{condition:coloring-one-special-main-condition}-(\ref{item:coloring-one-special-LLL-constraint}), the \cref{item:csp-conditional-meaure-analysis-numerical} in \Cref{condition:csp-conditional-measure-analysis} holds.

Next, we bound $\conditionalDecayConstant$ and $\conditionalTriangleConstant$. 
We upper bound these quantities by considering two cases which we define sooner. For any extendable partial assignment $\tau\in \+F_{V\setminus \set{v}}$ under the projected alphabet, 
let $\Phi^\tau = (V^\tau, \*Q^{\tau}, C^{\tau})$ be the simplified CSP formula.
% Let $(X_i)_{i\in [B+1]}$ be the number of satisfying assignments of $\Phi^{\tau}$ with the value of $v$ is consistent with $i^\coF$.
% we use $\tau$ to simplify the CSP formula $\Phi$, i.e., removing all constraints that are already satisfied by $\tau$. 
% And for any variable that is not in any constraint, we also remove it.
% Let $\Phi'$ be the simplified CSP formula, let $\+C'$ be the set of constraints of $\Phi'$. Then, we establish few properties of $\Phi'$.
% For any constraint $c\in \+C'$, if $v\notin \vbl(c)$, then it holds that for any variable $u\in\vbl(c)$, we have their color buckets $\tau_u$'s are the same.
% If $v\in \vbl(c)$, then it holds that for any variable $u\in\vbl(c)\setminus {v}$, we have their color buckets $\tau_u$'s are the same.
% Recall that by \Cref{definition:coloring-one-special-decomposition-scheme}, there is exactly one color in the color bucket $1^\coF$. We claim that for any constraint $c\in\+C'$ with $v\notin \vbl(c)$, it holds that for any variable $u\in\vbl(c)$, we have $\tau_u \neq 1^\coF$. To see this, we know that their $\tau_u$'s are the same, and if $\tau_u = 1^{\coF}$, then this constraint is violated. So we have $\tau_u \neq 1^{\coF}$. 

We now consider the following two cases: (1) $v$ is blocked
(2) $v$ is free.
% for any $c\in\+C'$ with $v\in \vbl(c)$, it holds that there exists $u\in\vbl(c)$, $\tau_u \neq 1^{\coF}$.
% We call the first case as $v$ is \emph{blocked} and the second case as $v$ is \emph{free}.

\begin{itemize}
    \item {\textbf{$v$ is blocked:}
%\paragraph{}
We handle the first case in which the value of $v$ can not be updated as $1^\coF$. 
% We further remove all the constraints $c\in\+C'$, such that $v\in\vbl(c)$ and for any $u\in\vbl(c)\setminus \set{v}$, we have $\tau_u = 1^{\coF}$. Let the resulting CSP formula be $\Phi''$. Note that for any constraint $c$ in $\Phi''$, we have that for any $u\in\vbl(c)$, $\tau_u\neq 1^{\coF}$.

For any $i\in\set{2, 3, \dots, B+1}$, we abuse the notation of $X_i$'s, we define $X_i$ as the probability that the value of $v$ is in the color bucket $i^\coF$ under the uniform distribution over all satisfying assignment of $\Phi^{\tau}$.
Let $\*X = (X_i)_{i\in \set{2, 3, \dots, B+1}}$.
By \Cref{lemma:coloring-simplify-CSP-by-pinning} and \Cref{lemma:coloring-local-uniformity} where we set $\varrho = \coRHO$, we have that for any $i\in \set{2, 3, \dots, B+1}$,
\begin{equation}
\label{eq:coloring-one-special-local-uniformity-Xi}
\frac{|f^{-1}_v(i^\coF)|}{q-1}\tp{1-\frac{1}{\varrho}}\le X_i \le \frac{|f^{-1}_v(i^\coF)|}{q-1}\tp{1 + \frac{4}{\varrho}}.
\end{equation}
For any $v\in V$, let $J$ be defined as follows:
\[
J\defeq \max_{\*X,\lambda} \sum_{i\in [B+1]\setminus\set{1}} \abs{X_i - \frac{|f_v^{-1}(i^{\coF})|/q}{1 + (\lambda - 1)/q}\tp{1 - \frac{1}{\varrho}}}.
\]
By \Cref{definition:coloring-one-special-decomposition-scheme} and the triangle inequality, it holds that
\begin{equation}
\label{eq:coloring-one-special-blocked-J-formula}
|b_v(\bot)|\cdot \max_{\tau\in \+F_{V\setminus \{v\}}} \sum_{x\in \Sigma_v} |\psiAdaptive{v}{\tau}(x)|\le J.
\end{equation}

Now, we bound $J$, by the above inequality and the triangle inequality, we have
% \begin{equation}
% \label{eq:coloring-J1}
\begin{align*}
J\le&~ \max_{\lambda}\frac{1}{|1+(\lambda-1)/q|}\max_{\*X, \lambda} \sum_{i\in [B+1]\setminus \set{1}} \abs{X_i - \frac{|f^{-1}_v(i^\coF)|}{q} + \frac{(\lambda-1)X_i}{q} + \frac{|f^{-1}_v(i^\coF)|}{\varrho q}} \\
\le&~ \max_{\lambda}\frac{1}{|1+(\lambda-1)/q|}\tp{\max_{\*X, \lambda} \sum_{i\in [B+1]\setminus \set{1}} \abs{\frac{(q+\lambda-1)X_i}{q} - \frac{|f^{-1}_v(i^\coF)|}{q}} + \frac{1}{\varrho}},\\
\end{align*}
where the second inequality is due to $\sum_{i\in [B+1]\setminus \set{1}} \frac{|f^{-1}_v(i^\coF)|}{q} \le 1$. 
Combined with the fact that $|\lambda-\lambda_c|\le \gamma$ and $\lambda_c\in [0,1]$, we have that
\begin{equation}
\label{eq:coloring-one-special-blocked-J}
J\le~ \frac{1}{1 + (\lambda_c - \gamma - 1)/q}\tp{\max_{\*X, \lambda} \sum_{i\in [B+1]\setminus \set{1}} \abs{\frac{(q+\lambda-1)X_i}{q} - \frac{|f^{-1}_v(i^\coF)|}{q}} + \frac{1}{\varrho}}.
\end{equation}
Next, we upper bound the right-hand side. We first consider the summation 
\[\max_{\*X, \lambda} \sum_{i\in [B+1]\setminus \set{1}} \abs{\frac{(q+\lambda-1)X_i}{q} - \frac{|f^{-1}_v(i^\coF)|}{q}}.\]
By $|\lambda - \lambda_c|\le \gamma$ and \cref{eq:coloring-one-special-local-uniformity-Xi}, it holds that 
\begin{align*}
\abs{\frac{(q+\lambda-1)X_i}{q} - \frac{|f^{-1}_v(i^\coF)|}{q}} \le&~ \max\left\{\abs{\frac{q+\lambda_c + \gamma - 1}{q}\frac{|f_v^{-1}(i^\coF)|}{q - 1} \tp{1 + \frac{4}{\varrho}} - \frac{|f_v^{-1}(i^\coF)|}{q}},\right. \\
&~\quad \left.\abs{\frac{q+\lambda_c - \gamma - 1}{q}\frac{|f_v^{-1}(i^\coF)|}{q - 1} \tp{1 - \frac{1}{\varrho}} - \frac{|f_v^{-1}(i^\coF)|}{q}}\right\}
\end{align*}
We upper bound these two terms respectively. For the first term, by the triangle inequality, we have that 
\begin{align*}
\abs{\frac{q+\lambda_c + \gamma - 1}{q}\frac{|f_v^{-1}(i^\coF)|}{q-1} \tp{1 + \frac{4}{\varrho}} - \frac{|f_v^{-1}(i^\coF)|}{q}} \le \tp{\abs{ \frac{\lambda_c + \gamma}{q - 1} } + \frac{8}{\varrho} }\cdot \frac{|f_v^{-1}(i^\coF)|}{q}.
\end{align*}
We bound the second term similarly,
\begin{align*}
\abs{\frac{q+\lambda_c - \gamma - 1}{q}\frac{|f_v^{-1}(i^\coF)|}{q-1} \tp{1 - \frac{1}{\varrho}} - \frac{|f_v^{-1}(i^\coF)|}{q}}\le \tp{\abs{ \frac{\lambda_c - \gamma}{q} } + \frac{8}{\varrho} }\cdot \frac{|f_v^{-1}(i^\coF)|}{q}.
\end{align*}
Combined with \cref{eq:coloring-one-special-blocked-J-formula}, \cref{eq:coloring-one-special-blocked-J} and $\lambda_c\in [0,1]$, we have that 
\begin{equation}
% \label{eq:coloring-one-special-blocked-J-final}
\label{eq:coloring-one-special-blocked-J-final}
|b_v(\bot)|\cdot \max_{\tau\in \+F_{V\setminus \{v\}}} \sum_{x\in \Sigma_v} |\psiAdaptive{v}{\tau}(x)|\le \frac{(\lambda_c + \gamma)/q + 9/\varrho}{1 + (\lambda_c - \gamma - 1)/q} = \frac{ \lambda_c + \gamma + 9q/\varrho}{q + \lambda_c - \gamma - 1}.
\end{equation}

Next for any $i\in [B+1]$, we upper bound $|b_v(i^\coF)|$. By \Cref{definition:coloring-one-special-decomposition-scheme}, it holds that 
\begin{equation}
\label{eq:coloring-one-special-b-single}
\abs{b_v(i^\coF)}\le \abs{\frac{|f_v^{-1}(i^\coF)|}{q + \lambda - 1}\tp{1 - \frac{1}{\varrho}}} \le \frac{|f_v^{-1}(i^\coF)|}{q + \lambda_c - \gamma - 1}\tp{1 - \frac{1}{\varrho}} \le \frac{|f_v^{-1}(i^\coF)|}{q + \lambda_c - \gamma - 1}.
\end{equation}
And 
\begin{equation}
\label{eq:coloring-one-special-b-sum}
\sum_{i\in [B+1]\setminus \set{1}}\abs{b_v(i^\coF)} \le \frac{q-1}{q + \lambda_c - \gamma - 1}\tp{1 - \frac{1}{\varrho}} \le \frac{q-1}{q + \lambda_c - \gamma - 1}.    
\end{equation}

Next, we bound $\conditionalDecayConstant$.
Recall the definition of $\conditionalDecayConstant$ in \cref{item:csp-conditional-meaure-analysis-decay} of \Cref{condition:csp-conditional-measure-analysis}. 
Combined with \cref{eq:coloring-one-special-b-single}, \Cref{fact:coloring-one-special-size} and \cref{eq:coloring-one-special-blocked-J-final}, it holds that 
\begin{equation*}
\label{eq:coloring-one-special-blocked-conditional-decay}
\begin{aligned}
\conditionalDecayConstant \le&~ q\cdot \tp{ \frac{ s + 1 }{ q + \lambda_c - \gamma - 1 } + \frac{\lambda_c + \gamma + 9q/\varrho}{q + \lambda_c - \gamma - 1} }^k \\
\le&~ q\cdot \tp{ \frac{ s + 1 + \lambda_c + \gamma + 9q/\varrho }{ q + \lambda_c - \gamma - 1 } }^k \le q\cdot \tp{\frac{s + 4}{q - 2}}^k.
\end{aligned}
\end{equation*}
When $s\ge 4$ and $q\ge 2$, it holds that $\conditionalDecayConstant \le q(4s/q)^k$.

Then, we bound $\conditionalTriangleConstant$.
Recall the definition of $\conditionalTriangleConstant$ in \cref{item:csp-conditional-meaure-analysis-triangle} of \Cref{condition:csp-conditional-measure-analysis}. 
Combining \cref{eq:coloring-one-special-blocked-J-final} and \cref{eq:coloring-one-special-b-sum}, it holds that 
\begin{equation*}
\label{eq:coloring-one-special-blocked-triangle}
\begin{aligned}
\conditionalTriangleConstant \le&~ \frac{q-1}{q + \lambda_c - \gamma - 1} + \frac{ \lambda_c + \gamma + 9q/\varrho}{q + \lambda_c - \gamma - 1} \\
\le&~ \frac{q - 1 + \lambda_c + \gamma + 9q/\varrho}{q + \lambda_c - \gamma - 1} \le 1 + \frac{9q}{\varrho}.
\end{aligned}
\end{equation*}
Combined with $\varrho = \coRHO$, it holds that $\conditionalTriangleConstant \le 1 + \frac{1}{4\Delta^2 k^5}$.
}

\item{\textbf{$v$ is free:}
Next, we handle the second case.
Again, we abuse the notation of $X_i$'s, we define $X_i$ be the probability that the value of $v$ belongs to the color bucket $i^{\coF}$ under the uniform distribution over all satisfying assignments of $\Phi^{\tau}$.
Let $\*X = (X_i)_{i\in [B+1]}$.
By \Cref{fact:coloring-one-special-size} and \Cref{lemma:coloring-local-uniformity} where we set $\varrho = \coRHO$, we have that for any $i\in [B+1]$,
\begin{equation}
\label{eq:coloring-one-special-local-uniformity-Xi-free}
\frac{|f^{-1}_v(i^\coF)|}{q}\tp{1-\frac{1}{\varrho}}\le X_i \le \frac{|f^{-1}_v(i^\coF)|}{q}\tp{1 + \frac{4}{\varrho}}.
\end{equation}
For any $v\in V$, let $J$ be defined as follows:
\[
J\defeq \max_{\*X, \lambda}\tp{ \abs{\frac{\lambda X_1}{\lambda X_1 + \sum_{j\in [B+1]\setminus \set{1}} X_j }} + \sum_{i\in [B+1]\setminus\set{1}} \abs{\frac{X_i}{\lambda X_1 + \sum_{j\in [B+1]\setminus \set{1}} X_j } - \frac{|f_v^{-1}(i^{\coF})|/q}{1 + (\lambda - 1)/q}\tp{1 - \frac{1}{\varrho}}}}.
\]
By \Cref{lemma:coloring-one-special-transition-measure-expression}, \Cref{definition:decomposition-scheme} and 
\Cref{definition:coloring-one-special-decomposition-scheme}, it holds that 
\begin{equation}
\label{eq:coloring-one-special-free-J-formula}
|b_v(\bot)|\cdot \max_{\tau\in \+F_{V\setminus \{v\}}} \sum_{x\in \Sigma_v} |\psiAdaptive{v}{\tau}(x)|\le J.
\end{equation}
As we do before, we give an upper bound of $J$. We first consider the first part $\max_{\*X, \lambda} \abs{\frac{\lambda X_1}{\lambda X_1 + \sum_{j\in [B+1]\setminus \set{1}} X_j }}$. By the fact that $\sum_{i\in[B+1]}X_i = 1$. It equals to $\max_{\*X, \lambda} \abs{\frac{\lambda X_1}{1 + (\lambda - 1) X_1 }}$. 
Recall that $|\lambda - \lambda_c|\le \gamma$. Combined with \cref{eq:coloring-one-special-local-uniformity-Xi-free}, it holds that 
\begin{equation}
\label{eq:coloring-one-special-free-X1}
\max_{\*X, \lambda} \abs{\frac{\lambda X_1}{\lambda X_1 + \sum_{j\in [B+1]\setminus \set{1}} X_j }}\le \max_{\*X, \lambda} \abs{\frac{\lambda X_1}{1 + (\lambda - 1) X_1 }} \le \frac{(\lambda_c + \gamma)\cdot 1/q \cdot (1 + 4/\varrho)}{1 + (\lambda_c - \gamma  - 1) \cdot 1/q \cdot (1 + 4/\varrho) }.
\end{equation}
% \begin{equation}
% \begin{aligned}
% &~ \frac{(\lambda_c + \gamma)\cdot 1/q \cdot (1 + 4/\varrho)}{1 + (\lambda_c - \gamma  - 1) \cdot 1/q \cdot (1 + 4/\varrho) } \\
% =&~ \frac{\lambda_c + \gamma}{q + \lambda_c - \gamma - 1} + 
% \end{aligned}
% \end{equation}
Next, we bound the second term of $J$:
\begin{equation}
\label{eq:coloring-one-special-free-J-part2-step1}
\begin{aligned}
&~\max_{\*X, \lambda} \sum_{i\in [B+1]\setminus \set{1}}\abs{\frac{X_i}{1 + (\lambda-1) X_1 } - \frac{|f_v^{-1}(i^{\coF})|/q}{1 + (\lambda - 1)/q}\tp{1 - \frac{1}{\varrho}}}\\
=&~ \max_{\*X, \lambda} \sum_{i\in [B+1]\setminus \set{1}}\abs{\frac{X_i \cdot ( 1 + (\lambda - 1)/q ) - ( 1 + (\lambda - 1)X_1 ) \cdot |f_v^{-1}(i^\coF)|/q \cdot (1 - 1/\varrho) }{ ( 1 + (\lambda - 1) X_1 ) ( 1 + (\lambda - 1)/q ) }} \\
\le&~ \frac{ \max_{\*X, \lambda} \sum_{i\in [B+1]\setminus \set{1}} \abs{ X_i \cdot ( 1 + (\lambda - 1)/q ) - ( 1 + (\lambda - 1)X_1 ) \cdot |f_v^{-1}(i^\coF)|/q \cdot (1 - 1/\varrho) } }{ \min_{\*X, \lambda} |(1 + (\lambda - 1) X_1)(1+(\lambda - 1)/q)|},
\end{aligned}
\end{equation}
where the last inequality is due to $|\lambda - \lambda_c| \le \gamma$, $\lambda_c \in [0, 1]$ and \cref{eq:coloring-one-special-local-uniformity-Xi-free}.
Then, we consider the denominator of \cref{eq:coloring-one-special-free-J-part2-step1}.
Recall that $\gamma= \coGamma$ (\Cref{condition:coloring-one-special-main-condition}-(\ref{item:coloring-one-special-lambda-constraints})) and $\varrho = \coRHO$. We now give an upper bound of the denominator:
\begin{equation}
\label{eq:coloring-one-special-free-J-part2-denominator}
\begin{aligned}
\frac{ 1 }{ \min_{\*X, \lambda} |(1 + (\lambda - 1) X_1)(1+(\lambda - 1)/q)|} \le \frac{ 1 }{ ( 1 - 4/q ) ( 1 - 2/q ) } \le 4,
\end{aligned}
\end{equation}
where the last inequality holds when $q \ge 8$.
Next, we consider the numerator of \cref{eq:coloring-one-special-free-J-part2-step1}. For each element of the summation, by the triangle's inequality, it holds that 
\begin{equation}
\label{eq:coloring-one-special-free-J-part2-numerator}
\begin{aligned}
&~\abs{ X_i \cdot \tp{ 1 + \frac{\lambda - 1}{q} } - ( 1 + (\lambda - 1)X_1 ) \cdot \frac{|f_v^{-1}(i^\coF)|}{q} \cdot \tp{ 1 - \frac{1}{\varrho} } } \\
\le&~ \abs{ \tp{ X_i - \frac{|f^{-1}_v(i^\coF)|}{q}\cdot \tp{1 - \frac{1}{\varrho}} } } + \frac{ \abs{\lambda - 1} }{ q } \abs{X_i - X_1 \cdot |f_v^{-1}(i^\coF)|\cdot \tp{1-\frac{1}{\varrho}} }.
\end{aligned}
\end{equation}
Combined with \cref{eq:coloring-one-special-local-uniformity-Xi-free}, we bound the first part as $|f_v^{-1}(i^\coF)|/q\cdot 5/\varrho$.
Combined with $|\lambda - \lambda_c|\le \gamma$, we bound the second part as follows:
\begin{equation}
\label{eq:coloring-one-special-free-J-part2-numerator-part2}
\frac{ \abs{\lambda - 1} }{ q } \abs{X_i - X_1 \cdot |f_v^{-1}(i^\coF)|\cdot \tp{1-\frac{1}{\varrho}} }\le \frac{ 1  + \gamma }{ q } \cdot \frac{|f_v^{-1}(i^\coF)|}{q} \cdot \frac{6}{\varrho}.
\end{equation}
Combining \cref{eq:coloring-one-special-free-J-part2-numerator} and \cref{eq:coloring-one-special-free-J-part2-numerator-part2}, we have that 
\begin{equation}
\label{eq:coloring-one-special-free-J-numerator-final}
\begin{aligned}
&~ \max_{\*X, \lambda} \sum_{i\in [B+1]\setminus \set{1}} \abs{ X_i \cdot \tp{ 1 + \frac{\lambda - 1}{q} } - ( 1 + (\lambda - 1)X_1 ) \cdot \frac{|f_v^{-1}(i^\coF)|}{q} \cdot \tp{ 1 - \frac{1}{\varrho} } } \\
\le&~ \sum_{i\in [B+1]\setminus \set{1}} \frac{|f_v^{-1}(i^\coF)|}{q\cdot \varrho} \cdot \tp{ 5 + \frac{6(1 + \gamma)}{q} } \le \frac{17}{\varrho}.
\end{aligned}
\end{equation}
Combining \cref{eq:coloring-one-special-free-J-formula}, \cref{eq:coloring-one-special-free-X1}, \cref{eq:coloring-one-special-free-J-part2-numerator} and \cref{eq:coloring-one-special-free-J-numerator-final}, we have that 
\begin{equation}
\label{eq:coloring-one-special-free-bot}
|b_v(\bot)|\cdot \max_{\tau\in \+F_{V\setminus \{v\}}} \sum_{x\in \Sigma_v} |\psiAdaptive{v}{\tau}(x)| \le \frac{(\lambda_c + \gamma) \cdot (1 + 4/\varrho)}{q + (\lambda_c - \gamma  - 1) \cdot (1 + 4/\varrho) } + \frac{68}{\varrho}.
\end{equation}

Next for any $i\in [B+1]$, recall the upper bound of $|b_v(i^\coF)|$ (\cref{eq:coloring-one-special-b-single}) and $\sum_{i\in [B+1]}|b_v(i^\coF)|$ (\cref{eq:coloring-one-special-b-sum}).
Then, we upper bound $\conditionalDecayConstant$ (recall its definition in \cref{item:csp-conditional-meaure-analysis-decay} of \Cref{condition:csp-conditional-measure-analysis}). 
Combining \cref{eq:coloring-one-special-free-bot}, \cref{eq:coloring-one-special-b-single}, \Cref{fact:coloring-one-special-size} and $\lambda_c \in [0, 1]$, it holds that 
\begin{equation}
\label{eq:coloring-one-special-free-conditional-decay}
\begin{aligned}
\conditionalDecayConstant \le&~ q\cdot \tp{\frac{s + 1}{q + \lambda_c - \gamma - 1}\tp{1 - \frac{1}{\varrho}} + \frac{(\lambda_c + \gamma) \cdot (1 + 4/\varrho)}{q + (\lambda_c - \gamma  - 1) \cdot (1 + 4/\varrho) } + \frac{68}{\varrho}}^k \\
\le&~ q\cdot \tp{ \frac{ s + 1 + (\lambda_c + \gamma) \cdot ( 1 + 4/\varrho ) }{ q + ( \lambda_c - \gamma  - 1 ) \cdot ( 1 + 4/\varrho ) }
 + \frac{68}{\varrho} }^k \\
\le&~ q\cdot \tp{ \frac{ s + 1 + (1 + \gamma) \cdot ( 1 + 4/\varrho ) }{ q - (  \gamma  + 1 ) \cdot ( 1 + 4/\varrho ) } 
 + \frac{68}{\varrho} }^k \le q\cdot \tp{\frac{s + 4}{q - 2} + \frac{68}{\varrho}}^k.
\end{aligned}
\end{equation}
When $s\ge 4$ and $q \ge 4$, it holds that $\conditionalDecayConstant \le q\cdot \tp{\frac{4s}{q} + \frac{68}{\varrho}}^k$. Combined with $\varrho = \coRHO$, we have that $\conditionalDecayConstant \le \mathrm{e}q (4s/q)^k$.

Then, we upper bound $\conditionalTriangleConstant$ (recall its definition in \cref{item:csp-conditional-meaure-analysis-triangle} of \Cref{condition:csp-conditional-measure-analysis}).
Combined with \cref{eq:coloring-one-special-free-bot}, \cref{eq:coloring-one-special-b-sum}, $\lambda_c \in [0,1]$, $\gamma = \coGamma$ and $\varrho = \coRHO$, it holds that 
\begin{equation}
\begin{aligned}
\conditionalTriangleConstant \le&~ \frac{q-1}{q + \lambda_c - \gamma - 1}\tp{1 - \frac{1}{\varrho}} + \frac{(\lambda_c + \gamma) \cdot (1 + 4/\varrho)}{q + (\lambda_c - \gamma  - 1) \cdot (1 + 4/\varrho) } + \frac{68}{\varrho} \\
\le&~ \frac{ ( q - 1 ) \cdot ( 1 - 1/\varrho ) + (\lambda_c + \gamma) \cdot (1 + 4/\varrho) }{ q + ( \lambda_c - \gamma - 1 ) \cdot ( 1 + 4/\varrho )} + \frac{68}{\varrho} \\
\le&~ \frac{ q - 1  + (\lambda_c + \gamma) \cdot (1 + 4/\varrho) }{ q + ( \lambda_c - \gamma - 1 ) \cdot ( 1 + 4/\varrho )} + \frac{68}{\varrho} \\
=&~ 1 + \frac{ 2\gamma (1 + 4/\varrho) + 4/\varrho }{q + ( \lambda_c - \gamma - 1 ) \cdot ( 1 + 4/\varrho ) } + \frac{68}{\varrho} \le 1 + 2 \gamma + \frac{76}{\varrho} \le 1 + \frac{1}{4\Delta^2 k^5}.
\end{aligned}
\end{equation}
}
\end{itemize}

In both cases, we have $\conditionalDecayConstant\le \mathrm{e}q\tp{\frac{4s}{q}}^k$ and $\conditionalTriangleConstant \le 1 + \frac{1}{4\Delta^2 k^5}$. Combined with \Cref{condition:coloring-one-special-main-condition}-(\ref{item:coloring-one-special-LLL-constraint}), the \cref{item:csp-conditional-meaure-analysis-numerical} in \Cref{condition:csp-conditional-measure-analysis} holds.
\end{proof}

\subsection{\texorpdfstring{$(k, \Delta)$}{(k, d)}-CNF formulas}
In this subsection, we consider the zero-freeness of $(k, \Delta)$-CNF formulas. 
A $(k, \Delta)$-CNF is a formula $\Phi = (V, \set{0,1}^V, C)$ where each clause has exactly $k$ variables and each variable belongs to at most $\Delta$ clauses.
We use $0$ to denote ``False'' and $1$ to denote ``True''.
Let $\Omega_{\Phi} \subseteq \{0, 1\}^V$ be the set of solutions to the formula $\Phi$.
For a set of variables $\+M\subseteq V$, we introduce the complex external field on these variables $\+M$. 
We define the partition function formally.

\begin{definition}[Complex extension for CNF formulas]
\label{definition:k-cnf-complex-extensions}
Given a CNF formula $\Phi = (V, \set{0,1}^V, C)$ and a subset of variables $\+M\subseteq V$.
We denote the \emph{partition function} 
\[
Z^{\-{CNF}}_{\Phi, \+M}(\lambda) \defeq \sum_{\sigma \in \Omega_{\Phi}} \lambda^{\abs{\+M \cap \sigma^{-1}(1)}}, \qquad \hbox{where }\sigma^{-1}(1)\defeq\set{v\in V\mid \sigma(v)=1}. 
\]
% $Z^{\-{CNF}} = Z^{\-{CNF}}(\Phi, \lambda,  \*f, \kcnfPFI)$ as the sum of all assignments' weights, i.e.  $Z^{\-{CNF}}\defeq \sum_{\sigma\in \Omega_{\Phi}} w(\sigma)$.
When $Z^{\-{CNF}}_{\Phi, \+M}(\lambda)\neq 0$, we naturally define a \emph{complex measure} $\mu = \mu(\Phi, \lambda, \+M)$, i.e.,$\forall \sigma\in\Omega_{\Phi}, \mu(\sigma)\defeq \frac{\lambda^{|\+M \cap \sigma^{-1}(1)|}}{Z^{\-{CNF}}_{\Phi, \+M}(\lambda)}$.
And we denote its \emph{support set} as $\supp(\mu) \defeq \{\sigma\in \Omega_{\Phi}\mid \mu(\sigma)\neq 0\}$.
% \end{definition}
\end{definition}

When the CNF formula $\Phi$, the subset $\+M$, and the complex external field $\lambda$ are clear from the context, we may omit the subscript $\Phi$, $\+M$, and $\lambda$.
Let $\log$ denote $\log_2$. In this subsection, we prove the following theorem by verifying \Cref{condition:csp-conditional-measure-analysis} and applying \Cref{theorem:sufficient-condition-zero-freeness}.
\begin{theorem}[Zero-freeness for $(k, \Delta)$-CNF]
\label{theorem:cnf-zero-freeness-intro}
Let $k$ and $\Delta \ge 2$ be two integers. Let $\gamma = \cnfGamma$.
Suppose 
\begin{align}
    \kcnfCondition.
    \label{eq:CNF-cond}
\end{align}

% suppose $2^{k/(36\ln(2))}\ge 400^3\mathrm{e}^8 \Delta k^2/(1-1/\mathrm{e})$.
For any $(k, \Delta)$-CNF formula $\Phi$,
there exists a set of variables $\+M\subseteq V$ with $|V|/6\le |\+M|\le 4|V|/5$ such that $Z^{\-{CNF}}_{\Phi, \+M}(\lambda) \neq 0$ for any  $|\lambda - 1|\le \gamma$.

Furthermore, under the same condition on $k$ and $\Delta$ as above, let $\lambda_c \ge 0$ satisfy
$\left( \frac{\max\{1,\lambda_c+\gamma\}}{1+\lambda_c-\gamma} \right)^{0.17k}\le \frac{1}{16\mathrm{e}^3 \Delta^2 k^4}$.
Then for any $\lambda$ with $|\lambda-\lambda_c| \le \gamma$,
it holds that $Z^{\mathrm{CNF}}_{\Phi,\mathcal{M}}(\lambda) \neq 0$.

\end{theorem}
% \yxtodo{Check this.}
We remark that unlike the zero-free region of the hypergraph $q$ coloring, the zero-free regime of the CNF formula does not contain $0$.
In the remaining part of this subsection, we prove \Cref{theorem:cnf-zero-freeness-intro}. 
We use \Cref{lemma:k-cnf-verify-sufficient-condition-for-general-zero-freeness-thm} to verify the \Cref{condition:csp-conditional-measure-analysis} in \Cref{theorem:sufficient-condition-zero-freeness}.
We first introduce some useful notations. 
% Let $\Phi = (V, \*Q, C)$ be a CNF formula, where $V$ is the set of Boolean variables and $C$ is the set of clauses where $\*Q = \{0, 1\}^V$, and we use $0$ to denote ``False'' and $1$ to denote ``True''.

% We say $(V, E)$ is a $(k, \Delta)$-CNF with maximum degree $\Delta$ if each clause has exactly the size $k$ and each variable belongs to at most $\Delta$ clauses.
The subset $\+M$ is constructed by the \emph{state-compression scheme} which is also known as the ``mark/unmark'' method for CNF formulas.
It plays a very important role in \emph{sampling/counting Lov\'asz local lemma}~\cite{moitra2019approximate, guo2018counting,  feng2021fast, feng2021sampling, jain2021sampling, Vishesh21towards, he2021perfect, galanis2022fast, he2022sampling, he2023deterministic, he2023improved,chen2023algorithms}. 
% Next, we define the \emph{complex extension of $\Phi$}.
We divide the variables into two parts: marked variables $\+M$ and unmarked variables $V\setminus \+M$, and we add the complex external field $\lambda$ on marked variables $\+M$.

Later, in order to show the zero-freeness, we consider a projected measure on only marked variables $\+M$, and a complex systematic scan Glauber dynamics on this projected measure. Then we use \Cref{theorem:sufficient-condition-zero-freeness} to show the zero-freeness. We verify the conditions in \Cref{condition:csp-conditional-measure-analysis} in \Cref{lemma:k-cnf-verify-sufficient-condition-for-general-zero-freeness-thm}.

To formalize the state-compression scheme, we project the values on unmarked variables $V\setminus \+M$ into a new symbol $\kcnfPF$. So that the value on unmarked variables $V\setminus \+M$ under the projection is always $\kcnfPF$ and the Glauber dynamics only changes values on marked variables $\+M$.

\begin{definition}[State-compression scheme for CNF formulas]
\label{definition:k-cnf-state-compression}
Let $\+M\subseteq V$ be the set of \emph{marked variables}. 
Let $\*f = (f_v)_{v\in V}$ be a \emph{projection}. 
For each variable $v\in V$, $f_v$ is a mapping from the domain $\set{0,1}$ to a finite \emph{alphabet} $\Sigma_v$.
Let $\*\Sigma\defeq \bigotimes_{v\in V}\Sigma_v$. And for any $\Lambda \subseteq V$, we denote $\*\Sigma_\Lambda \defeq \bigotimes_{v\in\Lambda} \Sigma_v$.
% And let $\*h = (h_v)_{v\in V}$ be the composition of $\*f$ and $\*g$ where $\forall v\in V$, $h_v = g_v \circ f_v$.

For each variable $v$ in $\+M$, let $\Sigma_v = \set{\kcnfPFO, \kcnfPFI}$, and let $f_v(0) = \kcnfPFO$ and $f_v(1) = \kcnfPFI$.

For each variable $v$ in $V\setminus \+M$, let $\Sigma_v = \set{\kcnfPF}$, and let $f_v(0) = f_v(1)= \kcnfPF$.

\end{definition}

% \todo{cleanup notations}
Note that the partition function defined by the above projection with the complex external field $\lambda$ on the symbol $1^\coF$ is equivalent to \Cref{definition:k-cnf-complex-extensions}.
Let $\markedK \ge 1, \unmarkedK\ge 1$ be two integers satisfying that $\markedK + \unmarkedK \le k$. Define $0\le \alpha, \beta \le 1$ as $\alpha\defeq \frac{\markedK}{k}$ and $\beta\defeq \frac{\unmarkedK}{k}$. 
% Let $k_{\-{mark}}: 1 \le k_{\-{mark}} \le k$ be a number. 
% Let $k_r$, $k_g$ and $k_b$ be three positive numbers and $k_r + k_g + k_b = k$.
We give the following sufficient condition for zero-freeness.
% \color{red}
\begin{condition}
\label{condition:k-cnf-main-condition}
Let $\Phi = (V, \set{0, 1}^V, C)$ be a $(k, \Delta)$-CNF formula.
Let $\+M \subseteq V$ be the set of marked variables, and let $\lambda$ be the complex parameter defined in \Cref{definition:k-cnf-complex-extensions}.
Let $\lambda_c \ge 0$ be a real number.
It holds that

\begin{enumerate}
    \item for any clause, there are at least $\markedK$ marked variables and at least $\unmarkedK$ unmarked variables; Furthermore, there are at least $\alpha\cdot |V|$ marked variables and at least $\beta\cdot |V|$ unmarked variables in total; \label{item:k-cnf-marked-unmarked}
    \item  \[2^{k}\ge \tp{4\mathrm{e}\Delta k}^{\frac{6\ln(2)\cdot (1 + \alpha - \beta)}{(1-\alpha - \beta)^2}},\quad \tp{\frac{\max\set{1, \lambda_c + \gamma}}{1 + \lambda_c - \gamma}}^{\markedK} \le \frac{1}{16\mathrm{e}^3 \Delta^2 k^4},\quad 2^{\unmarkedK} \ge 4000\mathrm{e}\Delta^3 k^5; \]
%     \item $k_r\ge 36\ln(2)$ and 
%     \[
% \forall j\in \{r, g, b\},\quad 2^{k_j/(12\ln 2)}\ge \frac{400^3\mathrm{e}^8 \Delta k^2}{1 - 1/\mathrm{e}}.
% \]
\label{item:k-cnf-LLL-constraint}
    \item $\lambda$ is in the disk centered at $\lambda_c$ with radius $\gamma = \cnfGamma$, i.e.,  $|\lambda - \lambda_c|\le \gamma$.\label{item:k-cnf-lambda-constraints}
\end{enumerate}
\end{condition}
% \color{black}

\begin{lemma}
\label{lemma:k-cnf-verify-sufficient-condition-for-general-zero-freeness-thm}
Let $\Phi = (V, \set{0,1}^V, C)$ be a $(k, \Delta)$-CNF formula with the state-compression scheme $\*f$ (\Cref{definition:k-cnf-state-compression}), and the complex external field $\lambda$ on the projected symbol $\kcnfPFI$.
Let $\+M\subseteq V$.
Write $C = \set{c_1, c_2, \dots, c_m}$. For any $i\le m$, define $C_i = \set{c_1, c_2, \dots, c_i}$ and $\Phi_i = (V, \set{0,1}^V, C_i)$.

Suppose that $(\Phi, \Lambda, \lambda)$ satisfies \Cref{condition:k-cnf-main-condition}, then $\Phi_1, \Phi_2, \dots, \Phi_m$ satisfy \Cref{condition:csp-conditional-measure-analysis}.
% then $Z^{\-{CNF}}\neq 0$.
\end{lemma}

By \Cref{theorem:sufficient-condition-zero-freeness} and \Cref{lemma:k-cnf-verify-sufficient-condition-for-general-zero-freeness-thm}, we have the following zero-freeness result for $(k, \Delta)$-CNF formulas.

\begin{theorem}[Zero-freeness]
\label{theorem:k-cnf-zero-freeness}
Let $\Phi = (V, \set{0,1}^V, C)$ be a $(k, \Delta)$-CNF formula, $\+M\subseteq V$ and $\lambda$ be the the complex external field .
Suppose $(\Phi, \Lambda, \lambda)$ satisfies \Cref{condition:k-cnf-main-condition},
then $Z^{\-{CNF}}_{\Phi, \+M}(\lambda)\neq 0$.
\end{theorem}
\begin{proof}
Let $\*f$ be the state-compression scheme (\Cref{definition:k-cnf-state-compression}).
%  and the complex external field $\lambda$ on the projected symbol $\kcnfPFI$.
% Suppose \Cref{condition:k-cnf-main-condition} holds,
% then $Z^{\-{CNF}}_{\Phi, \+M}(\lambda)\neq 0$.
Hence we have $Z^{\-{CNF}}_{\Phi, \+M}(\lambda) = Z(\Phi, \*f, \lambda, 1^\coF)$.
It suffices to verify the conditions in \Cref{theorem:sufficient-condition-zero-freeness}. For a complex number $x$, we use $\`R(x)$ to denote its real part. First about \cref{eq:csp-Z_0-is-not-0}.
By \Cref{condition:k-cnf-main-condition}-(\ref{item:k-cnf-lambda-constraints}), we have that 
\[
\forall v\in V,\quad \sum_{\sigma\in \set{0, 1}} \`R(\=I[f_v(\sigma) \neq \kcnfPFI] + \lambda \cdot \=I[f_v(\sigma) = \kcnfPFI]) > 0.
\]
So we have that \cref{eq:csp-Z_0-is-not-0} in \Cref{theorem:sufficient-condition-zero-freeness} holds.

Next, \Cref{condition:csp-conditional-measure-analysis} holds by \Cref{lemma:k-cnf-verify-sufficient-condition-for-general-zero-freeness-thm}.
Combined with \Cref{theorem:sufficient-condition-zero-freeness}, we have $Z^{\-{CNF}}_{\Phi, \+M}(\lambda)\neq 0$.
\end{proof}

Then, we use \Cref{theorem:k-cnf-zero-freeness} to prove \Cref{theorem:cnf-zero-freeness-intro}.
We first include the next useful lemma to show that \Cref{condition:k-cnf-main-condition}-(\ref{item:k-cnf-LLL-constraint}) implies \Cref{condition:k-cnf-main-condition}-(\ref{item:k-cnf-marked-unmarked})

\begin{lemma}\label{lemma:k-cnf-moser-tardos}
    Let $\Phi = (V, \set{0,1}^V, C)$ be a $(k, \Delta)$-CNF formula.
    Suppose \Cref{condition:k-cnf-main-condition}-(\ref{item:k-cnf-LLL-constraint}) holds. 
    There is an algorithm such that for any $\delta> 0$, with probability at least $1 - \delta$, it returns the set $\+M\subseteq V$ satisfying 
    \Cref{condition:k-cnf-main-condition}-(\ref{item:k-cnf-marked-unmarked})
    with time complexity $O\tp{|V|^2 \cdot |C|\cdot \log \frac{1}{\delta}}$.
\end{lemma}
\begin{proof}
We use the Moser-Tardos algorithm (\Cref{theorem:moser-tardos}) to find $\+M$.
We first check the conditions of the Moser-Tardos algorithm (\Cref{theorem:moser-tardos}).
Let $C$ be the set of clauses, and let $\+P[\cdot]$ denote the product distribution that every variable is put in $\+M$ with probability $\frac{1+\alpha - \beta}{2}$ independently.
Let $n = |V|$ and $m=|C|$.

For each clause $c\in C$, let $A_c$ be the event that for the clause $c$ there are less than $\markedK$ variables in $\+M$ or less than $\unmarkedK$ variables not in $\+M$.
Let $B$ be the event that there are less than $\alpha\cdot n$ marked variables or less than $\beta\cdot n$ unmarked variables in total.
By the Chernoff bound, we have that
\[
\forall c\in C,\quad \+P\inbr{A_c} \le 2\exp\tp{-\frac{(1-\alpha-\beta)^2}{6(1+\alpha-\beta)}\cdot k} = 2\tp{\frac{1}{2}}^{\frac{(1-\alpha-\beta)^2}{6\ln(2)\cdot (1+\alpha-\beta)}\cdot k},
\]
and similarly, we have 
\[\+P\inbr{B} \le 2\exp\tp{-\frac{(1-\alpha-\beta)^2}{6(1+\alpha-\beta)}\cdot n} = 2\tp{\frac{1}{2}}^{\frac{(1-\alpha-\beta)^2}{6\ln(2)\cdot (1+\alpha-\beta)}\cdot n}.\]
Next, we define the function $x$. For any clause $c$, we set $x(A_c)\defeq \frac{1}{2\Delta k}$, and we set $x(B) = \exp\tp{-\frac{m}{\Delta}}$.
For any event $A_c$ we use $\Gamma(A_c)$ to denote the set of events that are not independent of $A_c$ except for $A_c$ itself.
We first verify that for any clause $c$, we have $\+P\inbr{A_c}\le x(A_c)(1-x(B))\prod_{A_{c'}\in \Gamma(A_c)} (1-x(A_{c'}))$. 
By the fact that $m\ge \Delta$, we have $1-x(B)\ge 1-\exp(-1)\ge 1/2$.
For any clause $c$, we have that 
\[
x(A_c) (1-x(B)) \prod_{A_{c'}\in\Gamma(A_c)}\tp{1 - x(A_{c'})} \ge \frac{1}{2\Delta k}\cdot \tp{1 - \frac{1}{2\Delta k}}^{\Delta k} \ge \frac{1}{2\mathrm{e}\Delta k}.
\]
By the assumption, we have $\+P\inbr{A_c} \le \frac{1}{2\mathrm{e}\Delta k}\le x(A_c)(1-x(B))\prod_{A_{c'}\in\Gamma(A_c)} (1-x(A_{c'}))$.
Next we verify that $\+P\inbr{B} \le x(B)\prod_{A_c} \tp{1 - x(A_c)}$. By the fact that $\Delta k \ge 1$, we have 
\begin{align*}
x(B)\prod_{A_c} \tp{1 - x(A_c)} = x(B)\cdot \tp{1 - \frac{1}{\Delta k}}^{m}
\ge x(B)\cdot \exp\tp{-\frac{m}{2\Delta k}} \ge \exp\tp{-\frac{m}{\Delta} - \frac{m}{2\Delta k}}.
\end{align*}
By the assumption and $n\ge k$, $k\ge 3$, we have $\+P\inbr{B}\le 2\cdot \tp{4\mathrm{e}\Delta k}^{-n/k} \le \exp\tp{-\frac{2n}{k}}$. By the fact that $n\Delta \ge km$. We have that $\+P\inbr{B}\le \exp\tp{-\frac{2n}{k}} \le \exp\tp{-\frac{2m}{\Delta}}\le x(B)\prod_{A_c}\tp{1-x(A_c)}$.

Since the total number of clauses is at most $\Delta n$, the expected number of resampling steps is at most 
\[
\sum_{c\in C}\frac{x(A_c)}{1 - x(A_c)} + \frac{x(B)}{1 - x(B)} \le 3n.
\]
We run the Moser-Tardos algorithm (\Cref{theorem:moser-tardos}) for $6n$ resampling steps, by Markov's inequality, the algorithm returns the sets $\+M$ satisfying \Cref{condition:k-cnf-main-condition}-(\ref{item:k-cnf-marked-unmarked}) with probability at least $\frac{1}{2}$. So by running $\lceil \log\frac{1}{\delta}\rceil$ Moser-Tardos algorithms independently, then with probability at least $1 - \delta$, we find the set $\+M$ within totally $6n\lceil \log\frac{1}{\delta}\rceil$ resampling steps.
% Finally, we set $\+M_b = V\backslash(\+M_r\cup\+M_g)$ to satisfy the first constraint in \Cref{condition:k-cnf-main-condition}-(\ref{item:k-cnf-M-contraints}) .

Note that in each resampling step, we need to resample at most $n$ variables and check whether $m+1$ bad event occurs. Hence, the total time complexity is $O\tp{n^2 m \log\frac{1}{\delta}}$.

\end{proof}

Now, we prove \Cref{theorem:cnf-zero-freeness-intro}.
% , then we give some useful notations and lemmas to prove \Cref{theorem:k-cnf-zero-freeness}.
\begin{proof}[Proof of \Cref{theorem:cnf-zero-freeness-intro}]
We use \Cref{lemma:k-cnf-verify-sufficient-condition-for-general-zero-freeness-thm} to verify \Cref{condition:csp-conditional-measure-analysis}. Then we use \Cref{theorem:sufficient-condition-zero-freeness} to prove this lemma.
By \Cref{lemma:k-cnf-verify-sufficient-condition-for-general-zero-freeness-thm}, it suffices to verify \Cref{condition:k-cnf-main-condition}. 

We first consider the case that $\lambda_c = 1$.
\Cref{condition:k-cnf-main-condition}-(\ref{item:k-cnf-lambda-constraints}) holds by the assumption. 
By \Cref{lemma:k-cnf-moser-tardos}, it suffices to show \Cref{condition:k-cnf-main-condition}-(\ref{item:k-cnf-LLL-constraint}).

Let $\alpha = 0.171562$, $\beta = 0.257342$. 
In order to have $2^k\ge (4\mathrm{e}\Delta k)^{\frac{6\ln(2)(1+\alpha-\beta)}{(1-\alpha-\beta)^2}}$, it suffices to require that 
\[
k\ge 12\log(\Delta) + 12\log(k) + 41.
\]
In order to have $\tp{\frac{1 + \gamma}{2 - \gamma}}^{\markedK} \le \frac{1}{16\mathrm{e}^3 \Delta^2 k^4}$, recall that $\gamma = \cnfGamma$ and $\markedK \ge \alpha k$, it suffices to require that 
\[
k\ge 12\log(\Delta) + 24\log(k) + 57.
\]
In order to have $2^{\unmarkedK} \ge 4000\mathrm{e}\Delta^3 k^5$, recall that $\unmarkedK \ge \beta k$, it suffices to require that
\[
k \ge 12\log(\Delta) + 20\log(k) + 53.
\]
Combined with \cref{eq:CNF-cond}, \Cref{condition:k-cnf-main-condition} holds.

For general $\lambda_c \ge 0$, assume additionally that $\tp{\frac{\max\set{1, \lambda_c + \gamma}}{1 + \lambda_c - \gamma}}^{\markedK} \le \frac{1}{16\mathrm{e}^3 \Delta^2 k^4}$.
Then the verification of \Cref{condition:k-cnf-main-condition} is identical.

Finally, this theorem follows from \Cref{theorem:k-cnf-zero-freeness}.
\end{proof}
Next, the remaining subsection is mainly devoted to the proof of \Cref{lemma:k-cnf-verify-sufficient-condition-for-general-zero-freeness-thm}. 
Recall that we write the constraint set as $C = \set{c_1, c_2, \dots, c_m}$.
Also recall that for any $i\le m$, $C_i = \set{c_1, c_2, \dots, c_i}$ and $\Phi_i = (V, \set{0, 1}^V, C_i)$.
Suppose that $(\Phi, \Lambda, \lambda)$ satisfies \Cref{condition:k-cnf-main-condition}, it can be verified that for any $i\le m$, $(\Phi_i, \Lambda, \lambda)$ satisfies \Cref{condition:k-cnf-main-condition}.
Thus, we show a stronger result that if any formula satisfies \Cref{condition:k-cnf-main-condition}, then it satisfies \Cref{condition:csp-conditional-measure-analysis}.
% We prove \Cref{lemma:k-cnf-verify-sufficient-condition-for-general-zero-freeness-thm} by verifying \Cref{condition:csp-conditional-measure-analysis}.

We start with \Cref{condition:csp-conditional-measure-analysis}-(\ref{item:csp-transition-well-define}).
\begin{lemma}[Well-definedness]
\label{lemma:k-cnf-well-definedness}
Let $\Phi = (V, \set{0,1}^V, C)$ be a $(k, \Delta)$-CNF formula with the state-compression scheme $\*f$ (\Cref{definition:k-cnf-state-compression}), and the complex external field $\lambda$ on the projected symbol $\kcnfPFI$. Let $c^* \in C$ with the largest index number, and let $\Phi' = (V, \set{0, 1}^V, C\setminus \set{c^*})$

Suppose the projected measure $\psi$ induced by $\Phi'$ and $\*f$ has a nonzero partition function, i.e., $Z(\Phi', \lambda, \*f, 1^\coF)\neq 0$ and $(\Phi, \Lambda, \lambda)$ satisfies \Cref{condition:k-cnf-main-condition}, then \Cref{condition:csp-conditional-measure-analysis}-(\ref{item:csp-transition-well-define}) holds.
\end{lemma}
In order to prove \Cref{lemma:k-cnf-well-definedness}, we introduce the following lemma which shows that when considering the transition measure of the projected measure $\psi$, it suffices to consider the complex external field only on the variable being updated.

\begin{lemma}
\label{lemma:k-cnf-transition-measure-expression}
For any variable $v\in \+M$, for any extendable partial assignment $\tau\in \*\Sigma_{V\setminus \{v\}}$ defined on $V\setminus \{v\}$ under the projection, let $(X_i)_{i\in \set{0, 1}}$ be the number of satisfying assignments under the alphabet $\set{0,1}^V$ that are consistent with $\tau$ and the value of $v$ is consistent with $i^\kcnfPF$. 
% Let $X_1$ be the number of satisfying assignment under the alphabet $\*Q$ that are consistent with $\tau$ and the value of $v$ is $1$.

Then it holds that
\[
\psi^{\tau}_v(\kcnfPFO) = \frac{X_0}{X_0 + \lambda X_1}, \quad \psi^{\tau}_v(\kcnfPFI) = \frac{\lambda X_1}{X_0 + \lambda X_1}.
\]
\end{lemma}
\begin{proof}
By the definition of the complex extensions of CNF formulas (\Cref{definition:k-cnf-complex-extensions}), we only add complex external fields on the projected symbol $\kcnfPFI$. So for any satisfying assignment $\sigma\in \set{0,1}^V$ that is consistent with $\tau \in \*\Sigma_{V\setminus \set{v}}$, the contributions of the complex external fields of $\lambda$ to the weight $\lambda^{|\+M \cap \sigma^{-1}(1)|}$ are always the same except for that of the variable $v$. So we only consider the $\lambda$ at the variable $v$ when we consider the transition measure $\psi^{\tau}_v$. So this lemma follows.
%  It holds that 
% \[
% \psi^{\tau}_v(\kcnfPFO) = \frac{X_0}{X_0 + \lambda X_1}, \quad \psi^{\tau}_v(\kcnfPFI) = \frac{\lambda X_1}{X_0 + \lambda X_1}.
% \]
\end{proof}

We also include the next useful lemma from~\cite{feng2021fast} which establishes the \emph{local uniformity} for marginal probability.
% We also include the following useful lemma established in \cite
\begin{lemma}[\text{\cite[Corollary 2.2]{feng2021fast}}]
\label{lemma:k-cnf-local-uniformity}
Let $\Phi = (V, \set{0,1}^V, C)$ be a CNF formula. Assume that each clause contains at least $k_1$ variables and at most $k_2$ variables, and each variable belongs to at most $\Delta$ clauses.
Let $\`X$ be the uniform distribution over all satisfying assignments of $\Phi$.
For any $s\ge k_2$, if $2^{k_1}\ge 2\mathrm{e}\Delta s$, then there exists a satisfying assignment for $\Phi$ and for any $v\in V$,
\[\max\set{\Pr[\sigma\sim \`X]{\sigma_v=0},\Pr[\sigma\sim \`X]{\sigma_v=1} }\le \frac{1}{2}\exp\tp{\frac{1}{s}}.\]
% where $\mu$ is the uniform distribution of all satisfying assignments for $\Phi$.
\end{lemma}

Now we prove \Cref{lemma:k-cnf-well-definedness}.
\begin{proof}[Proof of \Cref{lemma:k-cnf-well-definedness}]
Under the assumption that $\psi$ is well-defined, by our state-compression scheme (\Cref{definition:k-cnf-state-compression}), it suffices to check for any $v\in \+M$, for any extendable partial assignment $\tau\in \*\Sigma_{V\setminus \{v\}}$ defined on $V\setminus \{v\}$ under the projection, $\psi^{\tau}_v$ is well-defined.

Let $X_0$ be the number of satisfying assignments of $\Phi'$ that are consistent with $\tau$, and $v$ has the value $0$. Let $X_1$ be the number of satisfying assignments of $\Phi'$ that are consistent with $\tau$, and $v$ has the value $1$.
By \Cref{lemma:k-cnf-transition-measure-expression}, it holds that
\[
\psi^{\tau}_v(\kcnfPFO) = \frac{X_0}{X_0 + \lambda X_1}, \quad \psi^{\tau}_v(\kcnfPFI) = \frac{\lambda X_1}{X_0 + \lambda X_1}.
\]
Next, we show they are well-defined ($\lambda X_1 + X_0 \neq 0$).
For a complex number $x$, we use $\`R(x)$ to denote its real part.
It suffices to show that $\`R(\lambda X_1 + X_0) > 0$.
Recall that there exists a real number $\lambda_c \ge 0$ and $|\lambda - \lambda_c| \le \gamma$. Hence $\`R(\lambda) \ge -\gamma$ and $\`R(\lambda X_1 + X_0) = \`R(\lambda) X_1 + X_0 \ge X_0 - X_1\cdot \gamma$. It suffices to show that $X_0 > X_1 \cdot \gamma$. We show this by a proof of contradiction. Assume that $X_0 \le X_1 \cdot \gamma$. 
If $X_1 = 0$, recall that $\tau$ is an extendable partial assignment, so $X_0 \ge 1$. So we have $X_0 > X_1\cdot \gamma$ which reaches a contradiction. Therefore we further assume that $X_1 \ge 1$.
By $X_0 \le X_1 \cdot \gamma$, we have that 
$\frac{X_1}{X_0 + X_1} \ge \frac{1}{1+\gamma}$.

However, we can use \Cref{lemma:k-cnf-local-uniformity} to establish an upper bound $\frac{X_1}{X_0 + X_1} < \frac{1}{1+\gamma}$. Hence, we reach a contradiction.
Note that $(\Phi, \Lambda, \lambda)$ satisfies \Cref{condition:k-cnf-main-condition}, then $(\Phi', \Lambda, \lambda)$ also satisfies \Cref{condition:k-cnf-main-condition}.
In order to apply \Cref{lemma:k-cnf-local-uniformity}, we use the partial assignment $\tau\in\Sigma_{V\setminus \set{v}}$ to simplify the CNF formula $\Phi'$, i.e., removing satisfied clauses and remove marked variables in $V\setminus \set{v}$ as they are fixed. 
Note that for any unmarked variable $u$, it holds that projected symbol $\tau_u = \coF$ which does not provide any information to the original value of $u$. So unmarked variables are not fixed.
By \Cref{condition:k-cnf-main-condition}-(\ref{item:k-cnf-marked-unmarked}), in the simplified CNF formula $\Phi''$, for any clause $c$, there are at least $\unmarkedK$ variables and at most $k$ variables.
Let $s = \cnfS$, note $s\ge k$ and $2^{\unmarkedK} \ge 2\mathrm{e}\Delta s$.
Combined with \Cref{condition:k-cnf-main-condition}-(\ref{item:k-cnf-LLL-constraint}), it holds that $\frac{X_1}{X_0 + x_1} \le \frac{1}{2}\exp\tp{\frac{1}{s}} < \frac{1}{1+\gamma}$.

Hence, $\psi_v^\tau$ is well-defined and \Cref{condition:csp-conditional-measure-analysis}-(\ref{item:csp-transition-well-define}) holds.
\end{proof}

Next, in order to verify \Cref{condition:csp-conditional-measure-analysis}-(\ref{item:csp-decomposition-decay}), we construct the \decompositionScheme-decomposition.
And after that, we provide some intuitions of it (\Cref{remark:k-cnf-decomposition-scheme}).
\begin{definition}[\decompositionScheme-decomposition scheme for $(k, \Delta)$-CNF formulas]
\label{definition:k-CNF-decomposition-scheme}
Let $\Phi = (V, \set{0,1}^V, C)$ be a $(k, \Delta)$-CNF formula with the state-compression scheme $\*f$ (\Cref{definition:k-cnf-state-compression}), and the complex external field $\lambda$ on the projected symbol $\kcnfPFI$. Let $\+M\subseteq V$ be the set of marked variables. Let $c^* \in C$ with the largest index number, and let $\Phi' = (V, \set{0, 1}^V, C\setminus \set{c^*})$.

Suppose the projected measure $\psi$ induced by $\Phi'$ and $\*f$ has a nonzero partition function, i.e., $Z(\Phi', \lambda, \*f, 1^\coF)\neq 0$ and $(\Phi, \Lambda, \lambda)$ satisfies \Cref{condition:k-cnf-main-condition}.
Let $s = \cnfS$.
For any unmarked variable $v\in V\setminus \+M$, we set $b_v(\kcnfPF) \defeq 1, b_v(\bot) \defeq 0$.
For any marked variable $v\in \+M$, we set
\[
b_v(x) \defeq \begin{cases}
\frac{1-1/2\exp(1/s)}{1+ 1/2\cdot\exp(1/s)\cdot(\lambda-1)}& x = \kcnfPFO\\[1.2ex]
\frac{\lambda \cdot (1-1/2\cdot\exp(1/s))}{1 + 1/2\cdot\exp(1/s)\cdot (\lambda-1)}&x=\kcnfPFI\\[1.2ex]
\frac{\lambda (\exp(1/s)-1)}{1 + 1/2\cdot\exp(1/s)\cdot (\lambda-1)}& x = \bot
\end{cases}.
\]
\end{definition}
We now provide some intuitions about how we set $b_v$ for marked variables $v\in \+M$.
% We now provide some intuitions about how we set $b_v$ for $v\in \+M_r\setminus \Lambda$. These intuitions extend to other cases naturally.
\begin{remark}
    \label{remark:k-cnf-decomposition-scheme}
Consider one transition step of the complex systematic scan Glauber dynamics, 
let $v\in \+M$ be the variable that we are updating, and let $\tau$ be the current assignment on $V \setminus \{v\}$ which is an extendable partial assignment under the projection.

For any $j\in \{0,1\}$, let $\psi^{\tau}_v(j^\kcnfPF)$ be the complex measure that the value of $v$ is updated to $j^\kcnfPF$ in this transition step, and 
let 
% $X_j \defeq \lambda^{\abs{v\in \+N(v)\mid \tau(v)=1}}$.
$X_j$ be the number of satisfying assignments on the whole variables set $V$ that are consistent with $\tau$ for the variables in $V\setminus \set{v}$ and the assignment of $v$ is $j$.
By \Cref{lemma:k-cnf-transition-measure-expression}, it holds that
\[
\psi^{\tau}_v(\kcnfPFO) = \frac{X_0}{X_0 + \lambda X_1}, \quad \psi^{\tau}_v(\kcnfPFI) = \frac{\lambda X_1}{X_0 + \lambda X_1}.
\]
By the local uniformity (\Cref{lemma:k-cnf-local-uniformity}), we expect $X_0 \approx X_1 \approx 1 - \frac{1}{2}\exp\tp{\frac{1}{s}}$,
which leads to the setting of $b_v$.
\end{remark}

Next, we verify the \Cref{condition:csp-conditional-measure-analysis}-(\ref{item:csp-decomposition-decay}).
\begin{lemma}
\label{lemma:k-cnf-verify-decay-condition}
Let $\Phi = (V, \set{0,1}^V, C)$ be a $(k, \Delta)$-CNF formula with the state-compression scheme $\*f$ (\Cref{definition:k-cnf-state-compression}), and the complex external field $\lambda$ on the projected symbol $\kcnfPFI$. Let $\+M\subseteq V$ be the set of marked variables. Let $c^* \in C$ with the largest index number, and let $\Phi' = (V, \set{0, 1}^V, C\setminus \set{c^*})$.

Suppose the projected measure $\psi$ induced by $\Phi'$ and $\*f$ has a nonzero partition function, i.e., $Z(\Phi', \lambda, \*f, 1^\coF)\neq 0$ and $(\Phi, \Lambda, \lambda)$ satisfies \Cref{condition:k-cnf-main-condition}. Let \decompositionScheme be the decomposition scheme defined in \Cref{definition:k-CNF-decomposition-scheme}, then \Cref{condition:csp-conditional-measure-analysis}-(\ref{item:csp-decomposition-decay}) holds.

% Let $\Phi = (V, \set{0,1}^V, C)$ be a $(k, \Delta)$-CNF formula with the state-compression scheme $\*f$ (\Cref{definition:k-cnf-state-compression}), and the complex external field $\lambda$ on the projected symbol $\kcnfPFI$.

% Suppose the projected measure $\psi$ induced by $\*f$ has a nonzero partition function, i.e., $Z^{\-{CNF}}_{\Phi, \+M}(\lambda)\neq 0$ and \Cref{condition:k-cnf-main-condition} hold, let \decompositionScheme be the decomposition scheme defined in \Cref{definition:k-CNF-decomposition-scheme}, then \Cref{condition:csp-conditional-measure-analysis}-(\ref{item:csp-decomposition-decay}) holds.
\end{lemma}
Before we prove the above lemma, we first use it to prove \Cref{lemma:k-cnf-verify-sufficient-condition-for-general-zero-freeness-thm}.
\begin{proof}[Proof of \Cref{lemma:k-cnf-verify-sufficient-condition-for-general-zero-freeness-thm}]
Note that $(\Phi_1, \Lambda, \lambda), (\Phi_2, \Lambda, \lambda), \dots, (\Phi_m, \Lambda, \lambda)$ satisfy \Cref{condition:k-cnf-main-condition}.
This lemma follows directly from \Cref{lemma:k-cnf-well-definedness} and \Cref{lemma:k-cnf-verify-decay-condition}.
\end{proof}

Finally, we prove \Cref{lemma:k-cnf-verify-decay-condition}.
\begin{proof}[Proof of \Cref{lemma:k-cnf-verify-decay-condition}]
Let $\gamma = \cnfGamma$.
Recall that there exists a real number $\lambda_c \ge 0$, such that $\abs{\lambda - \lambda_c} \le \gamma$ or equivalently, $\lambda\in \+D(\lambda_c, \gamma)$.
Recall the definition of $\conditionalDecayConstant$ in \cref{item:csp-conditional-meaure-analysis-decay} of \Cref{condition:csp-conditional-measure-analysis} and the definition of $\conditionalTriangleConstant$ in \cref{item:csp-conditional-meaure-analysis-triangle} of \Cref{condition:csp-conditional-measure-analysis}.
We claim that $\conditionalDecayConstant\le \mathrm{e} \tp{\frac{\max\set{1, \lambda_c + \gamma}}{1 + \lambda_c - \gamma}}^{\markedK}$ and $\conditionalTriangleConstant \le 1 + \frac{1}{4\Delta^2 k^5}$. 
Then this lemma follows from \Cref{condition:k-cnf-main-condition}-(\ref{item:k-cnf-LLL-constraint}):
\[
4\mathrm{e} \Delta^2 k^4\cdot \conditionalDecayConstant\cdot \conditionalTriangleConstant^{4\Delta^2 k^5} \le 4\mathrm{e}^3 \Delta^2 k^4 \cdot \tp{\frac{\max\set{1, \lambda_c + \gamma}}{1 + \lambda_c - \gamma}}^{\markedK} \le \frac{1}{4}.
\]
Now we bound $\conditionalDecayConstant$ and $\conditionalTriangleConstant$ and prove this claim.

For each variable $v$, let $\+F_{V\setminus \{v\}} \subseteq \*\Sigma_{V\setminus \{v\}} $ be the set of extendable partial assignments under the projection, we first bound the following quantity,
\begin{equation}
\label{eq:k-cnf-resample-part}
|b_v(\bot)|\cdot \max_{\sigma\in \+F_{V\setminus \{v\}}} (|\psiAdaptive{v}{\sigma}(\kcnfPFO)| + |\psiAdaptive{v}{\sigma}(\kcnfPFI)|).
\end{equation}
And after this, we consider the upper bound of $\max_{v\in \+M} |b_v(0^\coF)|$ and $\max_{v\in \+M} |b_v(1^\coF)|$. 
Then we upper bound $\conditionalDecayConstant$ and $\conditionalTriangleConstant$ by these quantities.

We first consider \cref{eq:k-cnf-resample-part}.
Recall that if $v$ is an unmarked variable, then $|b_v(\bot)|=0$ and the above quantity is $0$. So we assume that $v$ is a marked variable.

By \Cref{lemma:k-cnf-local-uniformity}, let $\`X$ be the uniform distribution on all satisfying assignments in $\Omega_\Phi$ that are consistent with $\sigma \in \+F_{V\setminus \{v\}}$. 
We use the partial assignment $\sigma$ to simplify the CNF formula $\Phi'$, i.e., removing satisfied clauses and remove marked variables in $V\setminus \set{v}$ as they are fixed.
Note that for any unmarked variable $u$, it holds that projected symbol $\tau_u = \coF$ which does not provide any information to the original value of $u$. So unmarked variables are not fixed. 
By \Cref{condition:k-cnf-main-condition}-(\ref{item:k-cnf-marked-unmarked}), in the simplified CNF formula $\Phi''$, for any clause $c$, there are at least $\unmarkedK$ variables and at most $k$ variables.

Let $s = \cnfS$, note $s\ge k$ and $2^{\unmarkedK} \ge 2\mathrm{e}\Delta s$ (\Cref{condition:k-cnf-main-condition}-(\ref{item:k-cnf-LLL-constraint})), it holds that 
\[\max\set{\Pr[\tau\sim \`X]{\tau_v=0},\Pr[\tau\sim \`X]{\tau_v=1} }\le \frac{1}{2}\exp\tp{\frac{1}{s}}.\]
Let $x = \Pr[\tau\sim \`X]{\tau_v=1}$. It holds that 
\begin{align*}
&~|b_v(\bot)|\cdot \max_{\sigma\in \+F_{V\setminus \{v\}}} (|\psiAdaptive{v}{\sigma}(\kcnfPFO)| + |\psiAdaptive{v}{\sigma}(\kcnfPFI)|)\\
\le&~ \max_{\substack{1 - 1/2\exp(1/s)\le x\le 1/2\exp(1/s)\\ \lambda \in \+D(\lambda_c, \cnfGamma)}} \tp{\abs{\frac{1 - x}{1 + (\lambda - 1)x} - b_v(0)} + \abs{\frac{\lambda x}{1 + (\lambda - 1)x} - b_v(1)}}.
\end{align*}
We define $J_0$ and $J_1$ as follows.
\begin{align*}
  J_0 \defeq&~ \max_{\substack{1 - 1/2\exp(1/s)\le x\le 1/2\exp(1/s)\\ \lambda \in \+D(\lambda_c, \cnfGamma)}}\abs{\frac{1 - x}{1 + (\lambda - 1)x} - \frac{1-1/2\exp(1/s)}{1+ 1/2\cdot\exp(1/s)\cdot(\lambda-1)}},\\
  J_1 \defeq&~ \max_{\substack{1 - 1/2\exp(1/s)\le x\le 1/2\exp(1/s)\\ \lambda \in \+D(\lambda_c, \cnfGamma)}} \abs{\frac{\lambda x}{1 + (\lambda - 1)x} - \frac{\lambda \cdot (1-1/2\cdot\exp(1/s))}{1 + 1/2\cdot\exp(1/s)\cdot (\lambda-1)}}.  
% \end{aligned}
\end{align*}

We first consider $J_0$, for simplicity we omit the constraint $1 - 1/2\exp(1/s)\le x\le 1/2\exp(1/s)$ and $\lambda \in \+D(\lambda_c, \cnfGamma)$; it holds that
\[
J_0 = \max_{x, \lambda} \abs{\frac{1/2\cdot\lambda (\exp(1/s)-2x)}{(1+(\lambda-1)x)(1 + (\lambda-1)/2\cdot\exp(1/s))}}.
\]
Recall that $\gamma = \cnfGamma$. We upper bound $J_0$ by considering the following two cases, (1) $|\lambda| \le 2\gamma$, and (2) $|\lambda| > 2\gamma$.

For the first case that $|\lambda| \le 2 \gamma$, it holds that 
\[
J_0 \le  \frac{2 \gamma}{(1 - (1+2\gamma) x) (1 - (1+2\gamma)/2\cdot \exp(1/s))} \le 32\gamma,
\]
where the second inequality is due to the fact that with $s = \cnfS$ and $|\lambda| \le 2\cdot \cnfGamma$, we have $(1+2\gamma)x \le \frac{3}{4}$ and $(1+2\gamma)/2\cdot \exp(1/s) \le \frac{1}{4}$.

For the second case $|\lambda| > 2\gamma$, we have $\`R(\lambda) > 0$. Hence, we have 
\begin{align*}
J_0 =&~ \max_{x, \lambda} \abs{\frac{1/2\cdot\lambda (\exp(1/s)-2x)}{(1+(\lambda-1)x)(1 + (\lambda-1)/2\cdot\exp(1/s))}}\\
=&~ \max_{x, \lambda} \abs{\frac{1/2 \cdot (\exp(1/s) - 2x)}{((1-x)/\lambda + x) ( 1 + (\lambda - 1)/2\cdot \exp(1/s) )}}\\
\le&~ \max_{x} \frac{1/2 \cdot |\exp(1/s) - 2x|}{ x\cdot (1 - 1/2\exp(1/s))},
\end{align*}
the last inequality is due to the fact that $|(1 - x)/\lambda + x|\ge x$ and $|1 + (\lambda - 1)/2 \exp(1/s)| \ge 1 - 1/2\exp(1/s)$.
With the fact that $s = \cnfS$, we have that in this case $J_0 \le 8|\exp(1/s) - 2x|$.
Combined, we have that 
\begin{equation}
\label{eq:k-cnf-bound-J0-final}
J_0 \le 32 \gamma + 8 \max_{x}|\exp(1/s) - 2x|.
\end{equation}
We next upper bound $J_1$ similarly by considering the following two cases, (1) $|\lambda| \le 2\gamma$, and (2) $|\lambda| > 2\gamma$.
\[J_1 =  \max_{x, \lambda} \abs{\frac{\lambda ( (\exp(1/s) - 2)(1/2 - x) + \lambda x (\exp(1/s) - 1) )}{(1+(\lambda-1)x)(1 + (\lambda-1)/2\cdot\exp(1/s))}}.\]
For $|\lambda| \le 2\gamma$, it holds that $(1+(\lambda-1)x)(1 + (\lambda-1)/2\cdot\exp(1/s)) \ge (1 - (1 + 2\gamma)x)(1 - (1 + 2\gamma)/2 \cdot \exp(1/s))$. And we have $|\lambda ( (\exp(1/s) - 2)(1/2 - x) + \lambda x (\exp(1/s) - 1) )|\le 2\gamma( 2|1/2 - x| + 2\gamma\cdot |1 - \exp(1/s)| )$. Recall that $\gamma = \cnfGamma$ and $s = \cnfS$, so we have $(1+2\gamma)x\le \frac{3}{4}$, $(1+2\gamma)/2\cdot \exp(1/s)\le \frac{3}{4}$, $|1/2-x|\le \frac{1}{2}$ and $2\gamma |1-\exp(1/s)|\le 1$.
Combined, we have 
\begin{align*}
    J_1\le&~ 64 \gamma.
\end{align*}

For $|\lambda| > 2\gamma$, note that $\`R(\lambda) > 0$, we have that  
\begin{equation}
\label{eq:k-cnf-bound-J1-decomposition}
J_1 \le \max_{x, \lambda} \abs{\frac{\lambda(\exp(1/s)-2)(1/2-x)}{(1+(\lambda-1)x)(1 + (\lambda-1)/2\exp(1/s))}} + \max_{x, \lambda} \abs{\frac{\lambda^2 x (\exp(1/s)-1)}{(1+(\lambda-1)x)(1 + (\lambda-1)/2\exp(1/s))}}.
\end{equation}
For the first part of \cref{eq:k-cnf-bound-J1-decomposition}, we divide the fraction by $\lambda$ for both the numerator and the denominator. For the numerator, it can be upper bounded by $2 |1/2 - x|$. For the denominator, it holds that $\abs{\tp{\frac{1 - x}{\lambda} + x}\tp{1 - \frac{1}{2}\exp\tp{\frac{1}{s}} + \frac{\lambda}{2}\exp\tp{\frac{1}{s}}}} \ge |x\cdot (1 - 1/2\exp(1/s))|$. Combined with $s = \cnfS$ and $1 - 1/2\exp(1/s)\le x\le 1/2\exp(1/s)$, the denominator can be lower bounded by $\frac{1}{16}$. Hence, we can upper bound this part by $32\max_{x}|1/2 - x|$.

For the second part of \cref{eq:k-cnf-bound-J1-decomposition}, we divide the fraction by $\lambda^2$ for both the numerator and the denominator. For the numerator, it can be upper bounded by $|1 - \exp(1/s)|$. For the denominator, it holds that $\abs{ \tp{ \frac{1-x}{\lambda} + x } \tp{ \frac{1 - 1/2\exp(1/s)}{\lambda} + 1/2 \exp(1/s)} } \ge 1/2\exp(1/s) x$. Combined with $s = \cnfS$ and $1 - 1/2\exp(1/s)\le x \le 1/2\exp(1/s)$, the denominator can be lower bounded by $\frac{1}{16}$. Hence, this part can be upper bounded by $16|1-\exp(1/s)|$.

Combined, we upper bound $J_1$ as follows:
\begin{equation}
\label{eq:k-cnf-bound-J1-final}
J_1 \le 64\gamma + 16 |1-\exp(1/s)| + 32\max_{x} |1/2-x|.
\end{equation}
Combining \cref{eq:k-cnf-bound-J0-final}, \cref{eq:k-cnf-bound-J1-final} and the triangle inequality, we have that 
\begin{equation}
\label{eq:k-cnf-resample-bound}
\begin{aligned}
|b_v(\bot)|\cdot \max_{\sigma\in \+F_{V\setminus \{v\}}} (|\psiAdaptive{v}{\sigma}(\kcnfPFO)| + |\psiAdaptive{v}{\sigma}(\kcnfPFI)|) \le&~ 96\gamma + 24|1 - \exp(1/s)|+48\max_{x}|1/2-x| \\
\le&~ 96\gamma + \frac{192}{s},
\end{aligned}
\end{equation}
where the second inequality is due to $s = \cnfS$ and $1-1/2\exp(1/s)\le x \le 1/2\exp(1/s)$.
Next, for any $j\in\set{0 ,1}$, we upper bound $\max_{\substack{v\in\+M}} |b_v(j^\coF)|$. Recall the decomposition scheme defined in \Cref{definition:k-CNF-decomposition-scheme}.
Also recall that $\gamma = \cnfGamma$ and $|\lambda - \lambda_c| \le \gamma$.
W.l.o.g.\@, we consider $j = 0^\coF$, the other case can be handled similarly.
It suffices to upper bound 
\begin{align*}
\max_{\lambda} \frac{ \tp{1 - 1/2\exp(1/s)} }{\abs{1 - 1/2\exp(1/s) + \lambda/2 \exp(1/s)} } 
\le&~ \frac{1 - 1/2\exp(1/s)}{ 1 - 1/2\exp(1/s) + (\lambda_c - \gamma)/2 \exp(1/s) } \\
=&~ \frac{1}{1 + \lambda_c - \gamma}\cdot \frac{ (1 + \lambda_c - \gamma)(1 - 1/2\exp(1/s)) }{ 1 - 1/2 \exp(1/s) + (\lambda_c - \gamma)/2 \exp(1/s) } \\
=&~ \frac{1}{1 + \lambda_c - \gamma}\cdot\tp{1 + \frac{ (\lambda_c - \gamma)(1 - \exp(1/s)) }{ 1 - 1/2 \exp(1/s) + (\lambda_c - \gamma)/2 \exp(1/s) }}. \\
\end{align*}
Next, we upper bound the second term of the second part.
We consider two cases: (1) $|\lambda_c - \gamma| \le 2\gamma$ and (2) $\lambda_c - \gamma > 2\gamma$.

For case (1), it holds that 
\[
\frac{ (\lambda_c - \gamma)(1 - \exp(1/s)) }{ 1 - 1/2 \exp(1/s) + (\lambda_c - \gamma)/2 \exp(1/s) } \le \frac{ 2\gamma(\exp(1/s)- 1) }{ 1 - 1/2 \exp(1/s) - \gamma \exp(1/s) } \le 8\gamma,
\] 
where the last inequality is due to the fact that $s =\cnfS$ and $\gamma = \cnfGamma$ such that $\exp(1/s)-1\le 1$ and $1 - 1/2\exp(1/s) - \gamma \exp(1/s) \ge \frac{1}{4}$.

For case (2), note that $\lambda_c - \gamma > 2\gamma > 0$, $1 - 1/2\exp(1/s) + (\lambda_c - \gamma)/2\exp(1/s) > 0$ and $1 - \exp(1/s) < 0$. Hence, we can upper bound this term by $0$.

Combined, we have that 
\begin{equation}
\label{eq:k-cnf-b-v-0-bound}
\max_{v\in V}|b_{v}(0^\coF)| \le \frac{1}{1 + \lambda_c - \gamma} \cdot \tp{1 + 8\gamma}\le \frac{1}{1 + \lambda_c - \gamma} + 16\gamma.
\end{equation}
And similarly, we have 
\begin{equation}
\label{eq:k-cnf-b-v-1-bound}
\max_{v\in V}|b_{v}(1^\coF)| \le \frac{\lambda_c + \gamma}{1 + \lambda_c - \gamma} \cdot \tp{1 + 8\gamma} \le \frac{\lambda_c + \gamma}{1 + \lambda_c - \gamma} + 16\gamma.
\end{equation}
Now, we are ready to upper bound $\conditionalDecayConstant$ and $\conditionalTriangleConstant$ (defined in \cref{item:csp-conditional-meaure-analysis-decay} and \cref{item:csp-conditional-meaure-analysis-triangle}).
We first bound $\conditionalDecayConstant$.
For any constraint $c\in C$, it holds that there is only one partial assignment on $\vbl(c)$ violating the constraint $c$, i.e., $|c^{-1}(\false)|=1$.
Also recall that $\Phi'$ is a $(k, \Delta)$-CNF formula. Hence every constraint depends on $k$ variables. And by \Cref{condition:k-cnf-main-condition}-(\Cref{item:k-cnf-marked-unmarked}), each constraint contains at least $\markedK$ marked variables.
Combined with \cref{eq:k-cnf-resample-bound}, \cref{eq:k-cnf-b-v-0-bound} and \cref{eq:k-cnf-b-v-1-bound}, we have 
\begin{equation}
\label{eq:k-cnf-decay-upper-bound}
\conditionalDecayConstant \le \tp{\frac{\max\set{1, \lambda_c + \gamma}}{1 + \lambda_c - \gamma}+ 112\gamma + \frac{192}{s}}^{\markedK} \le  \mathrm{e} \tp{\frac{\max\set{1, \lambda_c + \gamma}}{1 + \lambda_c - \gamma}}^{\markedK},
\end{equation}
where the inequality holds because $\frac{\max\set{1, \lambda_c + \gamma}}{1 + \lambda_c - \gamma} \ge \frac{1}{2}$ and $(1 + 224\gamma + \frac{284}{s})^{\markedK} \le \mathrm{e}$ (with $\gamma = \cnfGamma$ and $s = \cnfS$).

Next, we bound $\conditionalTriangleConstant$. For any unmarked variables $v$, they contribute $1$ as $v$'s projected alphabet containing $\coF$ with $b_v(\coF) = 1$ and $b_v(\bot) = 0$.
Hence, we only consider marked variables. For any unmarked variable $v$, it holds that the projected alphabet $\Sigma_v = \set{0^\coF, 1^\coF}$. And the decomposition scheme is defined in \Cref{definition:k-CNF-decomposition-scheme}. Combined with \cref{eq:k-cnf-resample-bound}, \cref{eq:k-cnf-b-v-0-bound} and \cref{eq:k-cnf-b-v-1-bound}, we have 
\begin{equation}
\label{eq:k-cnf-triangle-upper-bound}
\conditionalTriangleConstant \le \frac{1 + \lambda_c + \gamma}{1 + \lambda_c - \gamma} + 128\gamma + \frac{192}{s} \le 1 + 132\gamma + \frac{192}{s} \le 1 + \frac{1}{4\Delta^2 k^5},
\end{equation}
where the second inequality holds because $\frac{2\gamma}{1 + \lambda_c - \gamma} \le 4\gamma$ and the last inequality holds because $\gamma = \cnfGamma$ and $s = \cnfS$.
\end{proof}

\section*{Acknowledgements}
We thank Chunyang Wang and Yitong Yin for their insightful discussions and helpful comments.
\bibliographystyle{alpha}
\bibliography{main}

\newpage
\appendix
% \section{Going beyond local uniformity: Fisher zeros for the Ising model}
% \label{section:function-decomposition}
% \input{function_decomposition.tex}

\section{Missing proofs for Analytic percolation}
\label{section:missing-analytic-percolation}
In this section, we provide the proofs for \Cref{lemma:convergence,lemma:bounding-marginal-measure}.
% Now we prove the sufficient condition for the convergence of a complex systematic scan Glauber dynamics~\Cref{lemma:convergence}.
\begin{proof}[Proof of~\Cref{lemma:convergence}]
    % Recall that applying the $\bm{b}$-decomposition scheme does not affect the induced measure on $\sigma_0$. 
    We claim that for any two initial configurations $\sigma,\sigma'\in \supp(\mu)$, the following always holds:
    \begin{equation}\label{eq:convergence-equivalent-condition}
      \lim\limits_{T\to \infty} \abs{\dmu{\sigma}(\sigma_0 \in A)-\dmu{\sigma'}(\sigma_0 \in A)}=0.
    \end{equation}
     Assuming \cref{eq:convergence-equivalent-condition}, we compare the chain $\dmu{\sigma}$ starting from an arbitrary initial configuration $\sigma\in \-{supp}(\mu)$ with the stationary chain (\Cref{definition:stationary-start}). By the triangle inequality, 
    \begin{align*}
    \abs{\dmu{\sigma}(\sigma_0 \in A) - \mu(A)}&=
        \abs{\dmu{\sigma}(\sigma_0 \in A) - \sum\limits_{\sigma' \in \-{supp}(\mu)}\mu(\sigma')\dmu{\sigma'}(\sigma_0 \in A)  } \\
        &\le \sum\limits_{\sigma' \in \-{supp}(\mu)}\abs{\mu(\sigma')}\abs{\dmu{\sigma}(\sigma_0 \in A) - \dmu{\sigma'}(\sigma_0 \in A)}.
    \end{align*}
    According to \cref{eq:convergence-equivalent-condition}, as $T\to \infty$, the right-hand side diminishes to $0$. Thus, the complex systematic scan Glauber dynamics converges to $\mu$.
    
    Now we complete the proof by establishing \cref{eq:convergence-equivalent-condition}. For any $T\ge1$, let $B(T)$ be the set of non-witness sequences satisfying \Cref{condition:convergence}. We have: 
% \begin{align*}
%     &\lim\limits_{T\to \infty}\abs{\dmu{\sigma}(\sigma_0=\tau)-\dmu{\sigma'}(\sigma_0=\tau)}\\
%   (\star)  \quad= &\lim\limits_{T\to \infty}\abs{\sum\limits_{\bm{\rho}}\dmu{\sigma}(\bm{r}=\bm{\rho})\cdot \left(\dmu{\sigma}(\sigma_0=\tau\mid \bm{r}=\bm{\rho})-\dmu{\sigma'}(\sigma_0=\tau\mid \bm{r}=\bm{\rho})\right)}\\
%      (\blacktriangle) \quad=&\lim\limits_{T\to \infty}\abs{\sum\limits_{\*\rho\in B(T)}\dmu{\sigma}(\bm{r}=\bm{\rho})\cdot \left(\dmu{\sigma}(\sigma_0=\tau\mid \bm{r}=\bm{\rho})-\dmu{\sigma'}(\sigma_0=\tau\mid \bm{r}=\bm{\rho})\right)}\\
%    (\blacksquare) \quad\leq&0,
%     \end{align*}
\begin{align*}
  &\lim\limits_{T\to \infty}\abs{\dmu{\sigma}(\sigma_0\in A)-\dmu{\sigma'}(\sigma_0\in A)}\\
(\star)  \quad= &\lim\limits_{T\to \infty}\abs{\sum\limits_{\bm{\rho}}\left(\dmu{\sigma}(\sigma_0\in A\land \bm{r}=\bm{\rho})-\dmu{\sigma'}(\sigma_0\in A\land \bm{r}=\bm{\rho})\right)}\\
(\blacktriangle)  \quad= &\lim\limits_{T\to \infty}\abs{\sum\limits_{\bm{\rho}\in B(T)}\left(\dmu{\sigma}(\sigma_0\in A\land \bm{r}=\bm{\rho})-\dmu{\sigma'}(\sigma_0\in A\land \bm{r}=\bm{\rho})\right)}\\
  (\blacksquare) \quad\leq&0,
  \end{align*}
    which implies \cref{eq:convergence-equivalent-condition}. Here, the $(\star)$ inequality follows from the law of total measure, along with the observation that $\mu^{\-{GD}}_{T,\sigma,\bm{b}}(\bm{r}=\bm{\rho})$ does not depend on $\sigma$. 
    The $(\blacktriangle)$ equality follows from \Cref{definition:witness-sequence} and that all $\*\rho\notin B(T)$ satisfy $\rho\Rightarrow A$. 
    The $(\blacksquare)$ inequality follows from the triangle inequality and \Cref{condition:convergence}.
    This completes the proof.
\end{proof}

Now we prove \Cref{lemma:bounding-marginal-measure}.
\begin{proof}[Proof of \Cref{lemma:bounding-marginal-measure}]
We have that
\begin{align*}
\abs{\mu(A)} = &\abs{\sum\limits_{\sigma \in \-{supp}(\mu)} \mu(\sigma) \left( \sum\limits_{\rho \not\in B(T)} \dmu{\sigma}(\sigma_0\in A \land \*r = \*\rho) + \sum\limits_{\rho \in B(T)}\dmu{\sigma}(\sigma_0 \in A \land \*r = \*\rho) \right) } \\
\le &\abs{\sum\limits_{\sigma \in \-{supp}(\mu)} \mu(\sigma)\sum\limits_{\rho \not\in B(T)} \dmu{\sigma}(\sigma_0\in A \land \*r = \*\rho)} \\
&+ \abs{\sum\limits_{\sigma \in \-{supp}(\mu)} \mu(\sigma)\sum\limits_{\rho \in B(T)}\dmu{\sigma}(\sigma_0 \in A \land \*r = \*\rho)},
\end{align*}
where the second inequality is due to the triangle inequality.
For any $\*\rho\not\in B(T)$, since $\*\rho$ is a witness sequence for $A$,  we have that for any  $\sigma,\tau  \in \-{supp}(\mu)$,
\[
  \sum\limits_{\rho \not\in B(T)} \dmu{\tau}(\sigma_0\in A \land \*r = \*\rho) = \sum\limits_{\rho \not\in B(T)} \dmu{\sigma}(\sigma_0\in A \land \*r = \*\rho).  
\]
By this equation and since  $\mu$ is a complex normalized measure, the previous bound for $\abs{\mu(A)}$ can be expressed as follows, after fixing an arbitrary $\tau\in  \-{supp}(\mu)$:
\begin{align*}
\abs{\mu(A)} &\le \abs{\sum\limits_{\rho \not\in B(T)} \dmu{\tau}(\sigma_0\in A \land \*r = \*\rho)} + \abs{\sum\limits_{\sigma \in \-{supp}(\mu)} \mu(\sigma)\sum\limits_{\rho \in B(T)}\dmu{\sigma}(\sigma_0 \in A \land \*r = \*\rho)}\\
&\le \abs{\sum\limits_{\rho \not\in B(T)} \dmu{\tau}(\sigma_0\in A \land \*r = \*\rho)} + \sum\limits_{\sigma \in \-{supp}(\mu)} \abs{\mu(\sigma)}\abs{\sum\limits_{\rho \in B(T)}\dmu{\sigma}(\sigma_0 \in A \land \*r = \*\rho)}.
\end{align*}
As $T \to \infty$, according to \Cref{condition:convergence}, we know that for any $\sigma \in \-{supp}(\mu)$, 
\[\lim_{T \to \infty}\abs{\sum\limits_{\*\rho \in B(T)}\dmu{\sigma}(\sigma_0 \in A \land \*r = \*\rho)} = 0.\] 
Therefore, as $T\to\infty$, we have
\begin{equation*}
\abs{\mu(A)} \le \lim_{T\to\infty} \abs{\sum\limits_{\rho \not\in B(T)} \mu^{\-{GD}}_{T, \tau, \*b}(\sigma_0 \in A \land \*r = \*\rho)}.
%%= \abs{\sum\limits_{\rho \not\in B(T)} \mu^{\-{GD}}_{T,\tau,\bm{b}}(\bm{r}=\bm{\rho}) \cdot \mu^{\-{GD}}_{T, \tau, \*b}(\sigma_0 \in A \mid  \*r = \*\rho)}.
\end{equation*}
\end{proof}

\section{Concentration inequality for CSP formulas}
\label{section:concentration_inequality}
In this section, we derive a concentration inequality on the number of variables taking the value $1$ for atomic CSP formulas through Chebyshev's inequality and prove \Cref{theorem:concentration-atomic-csp}.
% Let $\Phi = (V, [q]^V, C)$ be an atomic CSP formula. 

To use Chebyshev's inequality, it suffices to derive an lower bound of the expectation and an upper bound for the variance.
For the lower bound of the expectation, we utilize the \emph{local uniformity}.
For the variance, by the definition of variance, we consider the per-variable variances and covariances (see \cref{eq:concentration-variance-decomposition}).
For the per-variable variances, we again utilize the \emph{local uniformity}.
 For the covariances, we employ the recursive coupling from \cite{wang2024sampling} to upper bound the \emph{total influence}.

% For simplicity, we assume $\Phi$ is an atomic $(k, \Delta)$-CSP formula and for any $v\in V$, $Q_v = [q]$. It can be verified that the following condition implies the above condition. We use the following stronger condition to derive the concentration inequality.

% In this section, we give a stronger version of the concentration. 
% Recall that state-compression scheme defined in \Cref{definition:state-compression} with the projection $\*f = (f_v)_{v\in V}$ and the projected alphabet $\*\Sigma = (\Sigma_v)_{v\in V}$. 
% We consider the number of variables whose projected symbol is $1^\coF$.
% When $f_v^{-1}(1^\coF) = \set{1}$, the number of variables whose projected symbol is $1^\coF$ is equivalent to the number of variables whose value is $1$.
% We further assume that for any $v\in V$, if $1^coF\in \Sigma_v$, then it holds that $$

% Let $\mu$ be the distribution over the uniform distribution on $\Omega_{\Phi}$ and let $\sigma \sim \mu$.

% Instead of establishing a concentration inequality on the number of variables taking var
Let $\Phi = (V, [q]^V, C)$ be an $(k, \Delta)$-CSP formula.
Let $\lambda  > 0$ be the external field on the value $1$, and let the partition function be $Z = Z(\Phi, \lambda) = \sum_{\sigma\in \Omega_{\Phi}}\lambda^{|\sigma^{-1}(1)|}$.
Let $\mu$ be the corresponding Gibbs distribution, and let $\sigma \sim \mu$.
For any $v\in V$, let $\+X_v = \=I[\sigma_v = 1]$ be a random variable.
Let $\+X = \sum_{v\in V} \+X_v$.
We show that under the next condition, we can derive a concentration inequality.
% \todo{modify $\zeta$.}
\begin{condition}
\label{condition:general-csp-chebyshev-condition}
$\Phi = (V, [q]^V, C)$ is an atomic $(k, \Delta)$-CSP formula.
Let $r_{\max}\defeq \frac{\max\set{1, \lambda}}{q - 1 + \lambda}$. It holds that
\[
(8\mathrm{e})^3\cdot \max\set{1, 1/\lambda}\cdot  r_{\max}^{k-1}\cdot (\Delta k+ 1)^{2 + \zeta} \le 1, \hbox{ where } \zeta = \frac{2\ln{(2 - r_{\max})}}{\ln{(1/r_{\max})} - \ln{(2 - r_{\max})}}.
\]
% where 
% \[
% \zeta = \frac{2\ln{(2 - 1/q)}}{\ln{(q)} - \ln{(2 - 1/q)}}.
% \]
\end{condition}
We remark that for constant $\lambda$ and sufficiently large $q$, the above condition is asymptotically implied by
\[q\gtrsim \max\set{\tp{\frac{1}{\lambda}}^{\frac{1}{k-1}}, \lambda}\cdot  \Delta^{\frac{2 + o_{q}(1)}{k-1}}.\]
% Let $\mu$ be the uniform distribution on $\Omega_\Phi$ and note $\+P[\cdot\mid C] = \mu(\cdot)$.
% For each variable $v\in V$, we use $\sigma_v$ to denote its value and we define the random variable $X_v \defeq \=I[\sigma_v = 1]$, $\+X \defeq \sum_{v\in V}X_v$.
% Next, we give our concentration inequality for $\+X$.
\begin{theorem}
\label{theorem:concentration-atomic-csp-lambda}
Let $\Phi=(V,[q]^V, C)$ be an atomic $(k, \Delta)$-CSP formula, and let $\mu$ be the Gibbs distribution over $\Omega_\Phi$ with the external field $\lambda$. Suppose that \Cref{condition:general-csp-chebyshev-condition} holds, then we have
    \[
    \forall \delta > 0,\quad \Pr[\mu]{|\+X - \=E_{\mu}[\+X]|\ge \delta\cdot \=E_{\mu}[\+X]}\le \frac{4\Delta k(\Delta k + 1)}{\delta^2 \cdot \tp{\frac{\lambda}{q-1+\lambda}\cdot (1-1/(\Delta k))}^2 \cdot |V|}.
    \]
\end{theorem}
Before proving this theorem, we first show that this theorem implies \Cref{theorem:concentration-atomic-csp}.
\begin{proof}[Proof of \Cref{theorem:concentration-atomic-csp}]
Note that the uniform distribution over $\Omega_\Phi$ is the Gibbs distribution with the external field $\lambda = 1$.
And in this case, it holds that $r_{\max} = 1/q$. 
\Cref{theorem:concentration-atomic-csp} follows.
\end{proof}

Next, we use Chebyshev's inequality to prove the above theorem. It suffices to bound the expectation $\=E_{\mu}[\+X]$ and the variance $\Var[\mu]{\+X}$. 
We first bound the expectation, recall the Lov\'asz local lemma (\Cref{theorem:LLL}), its condition \cref{eq:LLL} and \Cref{theorem:HSS}.

\begin{lemma}
\label{lemma:general-csp-expectation-bound}
    Suppose \Cref{condition:general-csp-chebyshev-condition}, let $\mu$ be the Gibbs distribution over $\Omega_\Phi$ with the external field $\lambda$. For any variable $v\in V$, we have that 
    \[
    \frac{\lambda}{q - 1 + \lambda} - \frac{\lambda}{(q-1+\lambda)\Delta k}\le \Pr[\mu]{X_v = 1}\le \frac{\lambda}{q - 1 + \lambda} + \frac{\lambda}{(q-1+\lambda)\Delta k}.
    \]
\end{lemma}
\begin{proof}
Recall that $\sigma \sim \mu$.
% Let $\+A$ be the event that $\sigma_v \neq 1$. 
% Let $\+P_1$ be the distribution on $[q]$
Let $\+P$ be the product distribution over $[q]^V$.
We first verify the condition \cref{eq:LLL}. For each constraint $c\in C$, let $B_c$ be the event that $c$ is not satisfied, i.e., $B_c = \lnot c$ and we set $x(B_c) \defeq \frac{\min\set{1, \lambda}}{2(\Delta k)^2}\cdot \frac{\max\set{1, \lambda}}{q - 1 + \lambda} = \frac{\lambda}{2(q-1+\lambda) (\Delta k)^2}$. Next, by \Cref{condition:general-csp-chebyshev-condition}, it can be verified that 
\[
\+P[B_c] \le \tp{\frac{\lambda}{q - 1 + \lambda}}^k \le \frac{\min\set{1, \lambda}\cdot \max\set{1, \lambda}}{2 (q - 1 + \lambda) (\Delta k)^2} \tp{1 - \frac{\min\set{1, \lambda}\cdot \max\set{1, \lambda}}{2 (q - 1 + \lambda) (\Delta k)^2}}^{\Delta k}.
\]
With the fact that $|\Gamma(B_c)|\le k\Delta$, it holds that $\+P[B_c] \le x(B_c)\cdot \prod_{B'\in \Gamma(B_c)}(1 - x(B'))$.
% \todo{Check this}
Finally, by \Cref{theorem:HSS}, for any variable $v$, we have
\[
\Pr[\mu]{X_v\neq 1} \le \tp{1 - \frac{\lambda}{q - 1 + \lambda}} \tp{1 - \frac{\lambda}{2 (q - 1 + \lambda) (\Delta k)^2}}^{-\Delta k}\le 1 - \frac{\lambda}{q - 1 + \lambda} + \frac{\lambda}{(q-1+\lambda)\Delta k},
\]
and
\[
\Pr[\mu]{X_v =  1} \le \frac{\lambda}{q - 1 + \lambda}\tp{1 - \frac{\lambda}{2 (q - 1 + \lambda) (\Delta k)^2}}^{-\Delta k} \le \frac{\lambda}{q - 1 + \lambda} + \frac{\lambda}{(q-1+\lambda) \Delta k}.
\]
% So we have $\+P[\sigma_v = o \mid C] \ge \frac{1}{|Q_v|} - \frac{1}{|Q_v|D}$.
\end{proof}
As a corollary, we have that 
\begin{corollary}
\label{corollary:general-csp-expectation-bound}
Suppose \Cref{condition:general-csp-chebyshev-condition} and let $\mu$ be the Gibbs distribution over $\Omega_\Phi$ with the external field $\lambda$. We have that $\E[\mu]{\+X}\ge \tp{\frac{\lambda}{q - 1 + \lambda} - \frac{\lambda}{(q-1+\lambda)\Delta k}}\cdot |V| = \frac{\lambda}{q - 1 + \lambda}\cdot (1-1/(\Delta k))\cdot |V|$.
\end{corollary}
Then, we bound the variance. By definition of variance, it holds that 
\begin{equation}
\label{eq:concentration-variance-decomposition}
\Var[\mu]{\+X} = \sum_{v\in V}\Var[\mu]{X_v} + \sum_{u\neq v} \Cov[\mu]{X_u, X_v}.
\end{equation}
Next we deal with the per-variable variances.
\begin{lemma}
\label{lemma:general-csp-local-variance}
Suppose \Cref{condition:general-csp-chebyshev-condition} and let $\mu$ be the Gibbs distribution over $\Omega_\Phi$ with the external field $\lambda$. For any variable $v$, we have that $\Var[\mu]{X_v}\le \frac{\lambda}{q - 1 + \lambda} + \frac{\lambda}{(q-1+\lambda)\Delta k}$.
\end{lemma}
\begin{proof}
This lemma follows from the definition of $\Var[\mu]{X_v} = \Pr[\mu]{X_v = 1} - \Pr[\mu]{X_v = 1}^2$ and \Cref{lemma:general-csp-expectation-bound}.
\end{proof}
Then we bound $\sum_{u\neq v} \Cov[\mu]{X_u, X_v}$.
Note that for $u, v\in V$,  we have 
\[\Cov[\mu]{X_u, X_v} = \Pr[\mu]{X_u = X_v = 1} - \Pr[\mu]{X_u = 1}\Pr[\mu]{X_v = 1} \le |\Pr[\mu]{X_u = 1\mid X_v = 1} - \Pr[\mu]{X_u = 1}|.\]
Next, for any $v\in V$, we bound $\sum_{v\neq u}|\Pr[\mu]{X_u = 1\mid X_v = 1} - \Pr[\mu]{X_u = 1}|$ by the recursive coupling between $\mu(\cdot \mid \sigma_v = 1)$ and $\mu(\cdot)$ in~\cite{wang2024sampling}. We first sketch the intuition. Let $(\+A, \+B)$ be the realization of this coupling.
We use $d_{\-{Ham}}(\+A, \+B)$ to denote the Hamming distance between $\+A$ and $\+B$. We have 
\begin{align*}
\sum_{v\neq u}|\Pr[\mu]{X_u = 1\mid X_v = 1} - \Pr[\mu]{X_u = 1}|&\le \sum_{v\neq u} \=E_{(\+A, \+B)}|\=I[\+A_v = 1] - \=I[\+B_v = 1]|\\
&\le \=E_{(\+A, \+B)}[d_{\-{Ham}}(\+A, \+B)],
\end{align*}
where the first inequality is due to the coupling lemma and the last inequality is due to the linearity of expectation. Next, we include the following lemma to bound $\=E_{(\+A, \+B)}[d_{\-{Ham}}(\+A, \+B)]$. For two CSP formulas $\Phi = (V, \*Q, C)$ and $\Phi' = (V, \*Q, C\setminus \set{c})$ with $c\in C$, we abuse the notation $\mu$ a little, denoting $\mu_{C}$ as the Gibbs distribution over all satisfying assignments of $\Phi$ with the external field $\lambda$, and $\mu_{C\setminus \set{c}}$ as the Gibbs distribution over all satisfying assignments of $\Phi'$ with the external field $\lambda$.

% Let $\+P$ be the uniform (product) distribution over all possible assignments in $[q]^V$. 
% We include the following condition from~\cite[Condition 1]{wang2024sampling}.
%  We need the following parameters for $\Phi$:
% \begin{itemize}
%     % \item minimum domain size $q_{\min}\defeq \min_{v\in V}|Q_v|$;
%     % \item maximum domain size $q_{\max}\defeq \max_{v\in V}|Q_v|$;
%     \item dependency degree $D\defeq \max_{c\in C}|\{c'\in C\setminus \{c\}\mid \vbl(c)\cap\vbl(c')\neq \emptyset\}|$;
%     \item maximum violation probability $p\defeq \max_{c\in C}\+P[\lnot c]$.
% \end{itemize}

We include the following condition, which is a specification of~\cite[Condition 1]{wang2024sampling}. 
We remark that the statement is different from \cite[Condition 1]{wang2024sampling}.
In~\cite{wang2024sampling}, they consider the uniform distribution over all satisfying assignment which is equivalent to the case that $\lambda = 1$.
Adding the external field $\lambda$ change the violating probability implicitly and the definition of $\zeta$ explicitly.
By running through the same proof, we derive the following condition.
We also remark that when $\lambda = 1$, $r_{\max}$ defined in the next condition is equivalent to $1/q$. Hence we can derive the statement in \cite[Condition 1]{wang2024sampling}.
\begin{condition}[{\cite[Condition 1]{wang2024sampling}}]
\label{condition:general-csp-sampling-condition}
$\Phi = (V, [q]^V, C)$ is a CSP formula.
\begin{itemize}
\item Let $k_{\min}$ be the minimum size of constraints, i.e., $k_{\min}\defeq \min_{c\in C}|\vbl(c)|$;
\item let $D$ be the dependency degree, i.e., $D\defeq \max_{c\in C} |\set{c'\in C\setminus \set{c} \mid \vbl(c) \cap \vbl(c') \neq \emptyset}|$;
\item let $r_{\max}\defeq \frac{\max\set{1, \lambda}}{q - 1 + \lambda}$, and let $p \defeq r_{\max}^{k_{\min}}$. Note that $p$ is an upper bound of the violating probability of any constraints in $C$ under the product distribution.
\end{itemize}
It holds that
\[
(8\mathrm{e})^3\cdot p\cdot (D + 1)^{2 + \zeta} \le 1, \hbox{ where } \zeta = \frac{2\ln{(2 - r_{\max})}}{\ln{(1/r_{\max})} - \ln{(2 - r_{\max})}}.
\]
% where 
% \[
% \zeta = \frac{2\ln{(2 - 1/q_{\min})}}{\ln{(q_{\min})} - \ln{(2 - 1/q_{\min})}}.
% \]
\end{condition}
We remark that our \Cref{condition:general-csp-chebyshev-condition} is stronger than \Cref{condition:general-csp-sampling-condition}.
By running the same proof of \cite[Theorem 3.1]{wang2024sampling}, we have the following theorem.
\begin{theorem}[{\cite[Theorem 3.1]{wang2024sampling}}]
\label{theorem:general-csp-hamming-distance-bound}
Suppose the CSP formula $\Phi=(V,[q]^V, C)$ satisfies \Cref{condition:general-csp-sampling-condition}.
Let $k$ be the maximum size of constrains, i.e., $k\defeq \max_{c\in C}|\vbl(c)|$.
% , and let $\Delta$ be the maximum number of constraints which a variable belongs to, i.e., $\Delta\defeq \max_{v\in V}|\set{c\in C\mid v\in \vbl(c)}|$.
Let $c_0\in C$ be an arbitrary constraint. There exists a coupling $(\+A, \+B)$ of $\mu_C$ and $\mu_{C\setminus \{c_0\}}$, such that for any integer $K\ge 1$,
\[
\Pr{d_{\-{Ham}}(\+A, \+B)\ge k\cdot (D + 1)\cdot K}\le 2^{-K},
\]
where $d_{\-{Ham}}(\+A, \+B)$ denotes the Hamming distance between $\+A$ and $\+B$.
\end{theorem}

As a corollary, we have that
\begin{lemma}
\label{lemma:general-csp-covariance}
Let the CSP formula $\Phi=(V,[q]^V, C)$ satisfy \Cref{condition:general-csp-chebyshev-condition}. 
Let $v\in V$ be an arbitrary variable. There exists a coupling $(\+A, \+B)$ of $\mu(\cdot)$ and $\mu(\cdot\mid \sigma_v = 1)$, such that 
\[
\=E[d_{\-{Ham}}(\+A, \+B)]\le 2\Delta k \cdot (\Delta k + 1).
\]
\end{lemma}
\begin{proof}
Let $C'$ be the new constraint set that is simplified by setting $\sigma_v = 1$, i.e., we remove constraints that are already satisfied by $\sigma_v = 1$ and for other constraints we remove the variable $v$. Note each constraint in $C'$ contains at least $k-1$ variables. Now we use the triangle inequality to construct a coupling between $\mu_C$ and $\mu_{C'}$. Let $S = C\oplus C'$ be the symmetric difference between $C$ and $C'$. And note $|S|\le 2\Delta$. We first remove constraints in $C\setminus C'$ one by one. Then we add constraints in $C'\setminus C$ one by one.  For step $i$, let $\+A_i, \+B_i$ be the coupling in \Cref{theorem:general-csp-hamming-distance-bound}. And by \Cref{condition:general-csp-chebyshev-condition}, it can be verified that \Cref{condition:general-csp-sampling-condition} holds. Next, by \Cref{theorem:general-csp-hamming-distance-bound}, we have that 
\[
\=E_{(\+A_i, \+B_i)}[d_{\-{Ham}}(\+A_i, \+B_i)] \le \sum_{K = 1}^{\infty} \frac{k\cdot (\Delta k + 1)}{2^{K}}\le k\cdot(\Delta k + 1).
\]
Then this lemma follows from the triangle inequality.
\end{proof}
Finally, we prove \Cref{theorem:concentration-atomic-csp-lambda}.
\begin{proof}[Proof of \Cref{theorem:concentration-atomic-csp-lambda}]
    By \Cref{corollary:general-csp-expectation-bound}, it holds that $\=E_{\mu}[\+X] \ge \frac{\lambda}{q - 1 + \lambda}\cdot (1-1/(\Delta k))\cdot |V|$.
    By \Cref{lemma:general-csp-local-variance} and \Cref{lemma:general-csp-covariance}, it holds that $\Var[\mu]{\+X} \le 4\Delta k(\Delta k + 1) |V|$. By Chebyshev's inequality, it holds that
    \[
    \forall \delta > 0,\quad \Pr[\mu]{|\+X - \=E_{\mu}[\+X]|\ge \delta\cdot \=E_{\mu}[\+X]}\le \frac{4\Delta k(\Delta k + 1)}{\delta^2 \cdot \tp{\frac{\lambda}{q-1+\lambda}\cdot (1-1/(\Delta k))}^2 \cdot |V|}.
    \]
\end{proof}

\section{Central limit theorem}
\label{section:clt}
In this section, we establish a framework for the central limit theorem of CSP formulas. As an application, we derive the central limit theorem for hypergraph $q$-coloring and prove \Cref{theorem:coloring-one-special-clt-intro}.
% In this section, we derive the central limit theorem for hypergraph $q$-coloring and prove \Cref{theorem:coloring-one-special-clt-intro}.
% In this section, we derive the central limit theorem for hypergraph $q$-coloring and the $(k,\Delta)$-CNF formula, and prove \Cref{theorem:coloring-one-special-clt-intro}, \Cref{theorem:coloring-clt-intro} and \Cref{theorem:k-cnf-clt}.
% We focus on the number of variable that have a special projected symbol.
% We first establish a central limit theorem for the projected CSP formula, then we lift it back to the original CSP formula.
To establish the central limit theorem, we make use of the following lemma from \cite{michelen2024central}.
\begin{lemma}[{\cite[Theorem 1.2]{michelen2024central}}]\label{lemma:zero-free-to-clt-co}
    Let $X \in \{0,\ldots, n\}$ be a random variable with mean $\bar{\mu}$, standard deviation $\bar{\sigma}$ and probability generating function $g$ and set $X^*=(X - \bar{\mu})\bar{\sigma}^{-1}$. For $\delta \in (0, 1)$ such that $\abs{1 - \zeta} \ge \delta$ for all roots $\zeta$ of $g$,
    \[
    \sup_{t\in \=R} \abs{\Pr{X^* \le t} - \Pr{\+Z \le t}} \le O\left( \frac{\log n}{\delta \bar{\sigma}} \right),
    \]
    where $\+Z \sim N(0, 1)$ is a standard Gaussian random variable.
\end{lemma}
It suffices to prove the zero-freeness and the lower bound of variance. For the zero-freeness, we established a framework in \Cref{section:CSP}. 
And in this section, we first establish a condition to provide the lower bound of the variance (\Cref{lemma:clt-variance-bound}). We mainly use the \emph{local uniformity} (\Cref{lemma:clt-marginal-bound}) and the bound of \emph{total influence} via the \emph{recursive coupling} (\Cref{lemma:clt-bounded-total-influence}).
% After that, we focus on the hypergraph $q$-coloring, and show its concentration inequality.

% Given a $k$-uniform hypergraph with maximum degree $\Delta$, we consider the hypergraph $q$-coloring on this hypergraph. Let $\Phi = (V, \*Q, C)$ be the corresponding CSP formula where for any variable $v\in V$, $Q_v = [q]$. We consider the uniform distribution on $\Omega_\Phi$ and we denote it as $\mu$.
% Recall the two-step projection scheme with $(\*f, \*g)$ in \Cref{definition:coloring-projection-scheme} and their composition $\*h$. 

% \subsection{Central limit theorem}
Let $\Phi=(V,[q]^V, C)$ be an atomic $(k, \Delta)$-CSP formula. 
% For simplicity, we assume that for any variable $v\in V$, $Q_v = [q]$. 
Again, we consider the state-compression scheme (\Cref{definition:state-compression}) with the projection $\*f = (f_v)_{v\in V}$ where $f_v: [q]\to \Sigma_v$. 
% We introduce an external field $\lambda \in (0, 1]$ on the projected symbol $1^\coF$.
% Recall the weight function and the partition function defined in \Cref{definition:complex-extensions-csp}.
% Let $\mu$ be the Gibbs distribution over $\Omega_\Phi$.
We consider the Gibbs distribution on $\Omega_\Phi$ with the external field $\lambda$ on the projected symbol $1^\coF$, and denote the Gibbs distribution as $\mu$ (\Cref{definition:complex-extensions-csp} with $\lambda$ is a real number).

% Recall the state-compression scheme (\Cref{definition:state-compression}), let $\*f$ be projections. And recall that $(\Sigma_v)_{v\in V}$ are the alphabets under the projection.
% Let $1^{\coF} \in \bigcup_{v\in V}\Sigma_v$ and for any variable $v\in V$, let $\sigma_v \in Q_v$ be the assignment of $v$. 

We consider the central limit theorem on the number of variables that have the projected symbol $1^\coF$.
Let $\sigma \sim \mu$ following $\mu$.
Let $Y_v$ be the random variable that indicates whether the projected symbol of $v$ is $1^\coF$; i.e., if $\sigma_v \in f^{-1}_v(1^\coF)$, then $Y_v = 1$; otherwise, $Y_v = 0$. Let $\+Y = \sum_{v\in V} Y_v$. 
% And without further specifications, we consider the uniform distribution over all satisfying assignments of $\Phi$.

For simplicity, we assume that for any $v\in V$, if $1^\coF \in \Sigma_v$, then $|f^{-1}_v(1^\coF)| = q^*$ and $0<q^*<q$. We first provide the bound of variance.
{
% \color{red}
    % Assume $p$ is the upper bound of the violating probability.
\begin{condition}
\label{condition:clt-condition}
$\Phi = (V, [q]^V, C)$ is an atomic $(k, \Delta)$-CSP formula. 
Let $r_{\max}\defeq \frac{\max\set{1, \lambda}}{q - q^* + q^* \lambda}$.
It holds that
% \[
% 2 \mathrm{e}^2\cdot \max\set{1, 1/\lambda}\cdot r_{\max}^{k-1}\cdot \tp{\Delta k}^2 \le 1.
% \]
% \[
% 16\mathrm{e}^2 \cdot (q - q^* + q^* \lambda)^2/\lambda \cdot r_{\max}^{\frac{2(k-1)}{2 + \zeta}} \cdot (\Delta k + 1)^4 \le 1.
% \]
\[
16\mathrm{e}^2 \cdot \frac{(q - q^* + q^* \lambda)^2}{\lambda\cdot(q-q^*)} \cdot r_{\max}^{\frac{2(k-1)}{2 + \zeta}} \cdot (\Delta k + 1)^4 \le 1 , \quad \hbox{where } \zeta = \frac{2\ln{(2 - r_{\max})}}{\ln{(1/r_{\max})} - \ln{(2 - r_{\max})}}.
\]
We remark that for constant $\lambda, q^*$ and sufficiently large $q$, the above condition is asymptotically implied by 
\[
q\gtrsim \max\set{\tp{\frac{1}{\lambda}}^{\frac{1+o_q(1)}{k-2}}, \lambda} \Delta^{\frac{4 + o_{q}(1)}{k-2}}.
\]
% \[
% 16\mathrm{e}^2 \cdot (q - q^* + q^* \lambda)^2/\min\set{1, \lambda} \cdot r_{\max}^{k-1} \cdot (\Delta k + 1)^{2(2+\zeta)} \le 1.
% \]
% \[
% 32\mathrm{e}^2 q^3 \Delta^4 k^4 \tp{\frac{1}{q}}^{\frac{2(k-1)}{2 + \zeta}} \le 1, \qquad \hbox{where } \zeta = \frac{2\ln{(2 - r_{\max})}}{\ln{(1/r_{\max})} - \ln{(2 - r_{\max})}}.
% \]
\end{condition}
\begin{lemma}
\label{lemma:clt-variance-bound}
Suppose \Cref{condition:clt-condition}, let $\mu$ be the Gibbs distribution over $\Omega_{\Phi}$ with the external field $\lambda$. Let $N^*$ be the number of variables whose projected alphabet contains $1^\coF$, i.e., $N^* = |\{v\in V\mid 1^\coF \in \Sigma_v \}|$. It holds that 
\[
\Var[\mu]{\+Y} \ge \frac{1}{2} \cdot \frac{ (q^* \lambda)\cdot (q-q^*) }{ \tp{q-q^*+q^*\lambda}^2 } \cdot \tp{\frac{\Delta k - 1}{\Delta k}}^2 \cdot N^*
\]
\end{lemma}

Before proving this lemma, we establish the central limit theorem for hypergraph coloring as in \Cref{theorem:coloring-one-special-clt-intro}, which we restate below.
Given a $k$-uniform hypergraph $H = (V, \+E)$ with the maximum degree $\Delta$. We consider the hypergraph $q$-colorings on $H$. Let $\coloringOneRandomVariable$ be the random variable that equals the number of vertices whose color is $1$ following the uniform distribution over possible colorings.

\CLTColorOneIntro*
\begin{proof}[Proof of \Cref{theorem:coloring-one-special-clt-intro}]
We combine \Cref{lemma:zero-free-to-clt-co} and \Cref{lemma:clt-variance-bound}.
Formally,  the probability generating function $g$ in \Cref{lemma:zero-free-to-clt-co} is chosen as follows.
Let $\Phi = (V, [q]^V, C)$ be the corresponding CSP formula of the hypergraph $q$-coloring on $H$. 
By \Cref{theorem:coloring-one-special-zero-free-intro}, let $\lambda \in \+D(1, \coGamma)$. We know that the partition function $Z_H^{\-{co}}(\lambda)\neq 0$. 
Let $g(\lambda) = \frac{Z^{\text{co}}_H(\lambda)}{Z^{\text{co}}_H(1)}$ be a probability generating function.
It can be verified that the probability distribution with respect to $g$ is exactly the law of $\coloringOneRandomVariable$. 
The zero-freeness is implied by \Cref{theorem:coloring-one-special-zero-free-intro}.

Next, we verify the lower bound of $\coloringOneRandomVariable$'s variance using \Cref{lemma:clt-variance-bound}. Note that $\Var{\coloringOneRandomVariable} = \Var[\mu]{\+Y}$. In order to use \Cref{lemma:clt-variance-bound}, we first construct a new CSP $\Phi' = (V, [q]^V, C')$ that is atomic.
Specifically, for each constraint $c\in C$ that restricts one hyperedge from being monochromatic, we construct $q$ constraints in $C'$, each of which restricts one monochromatic color assignment. 
Note that  each constraint in $C'$ has exactly $k$ variables, and each variable belongs to at most $q\Delta$ constraints.
It can be verified that $\Phi$ and $\Phi'$ equivalently encode the hypergraph $q$-coloring on the hypergraph $H$.

% Finally, this lemma follows from the assumption $32\mathrm{e}^2 q^7 \Delta^4 k^4 \tp{\frac{1}{q}}^{\frac{2(k-1)}{2 + \zeta}} \le 1$, \Cref{lemma:clt-variance-bound}, and \Cref{lemma:zero-free-to-clt-co}.
% We use \Cref{lemma:clt-coloring-one-special} to prove this theorem.
In order to use \Cref{lemma:clt-variance-bound}, it suffices to verify that $k\ge 50$, $\coloringCondition$ imply \Cref{condition:clt-condition}.
Note that in the case that $q^* = 1$ and $\lambda = 1$, it suffices to show that $16\mathrm{e}^2 q^2 \tp{\frac{1}{q}}^{\frac{2(k-1)}{2+\zeta}}\cdot (q\Delta k + 1)^4 \le 1$ where $\zeta = \frac{2\ln(2 - 1/q)}{\ln(q) - \ln(2 - 1/q)}$.
Rearranging $16\mathrm{e}^2 q^2 \tp{\frac{1}{q}}^{\frac{2(k-1)}{2+\zeta}}\cdot (q\Delta k + 1)^4 \le 1$ gives that $q \ge \tp{16\mathrm{e}^2}^{\frac{2+\zeta}{2k - 14 - 6\zeta}} \tp{\Delta k + \frac{1}{q}}^{\frac{8 + 4\zeta}{2k - 14 - 6 \zeta}}$. It suffices to show that $q \ge \tp{256\mathrm{e}^2}^{\frac{2+\zeta}{2k - 14 - 6\zeta}} \tp{\Delta k}^{\frac{8 + 4\zeta}{2k - 14 - 6 \zeta}}$. By the fact that $\Delta \ge 1$ and $k \ge 50$, we have that $q \ge 700$, hence we have that $\zeta \le 0.23638$. Combined with $k\ge 50$, it suffices to show that $q \ge 2 \Delta^{\frac{8 + 4\zeta}{2k - 14 - 6 \zeta}}$. By the fact that $\zeta \le 0.23638$, it suffices to show that $q \ge \Delta^{\frac{4.5}{k - 8}}$. Hence \Cref{condition:clt-condition} holds.
By \Cref{lemma:clt-variance-bound}, it holds that $\Var[\mu]{\coloringOneRandomVariable} = \Theta_{q, \Delta, k}(|V|)$. Hence, this theorem follows from \Cref{lemma:zero-free-to-clt-co}.

%  Rearranging $32\mathrm{e}^2 q^7 \Delta^4 k^4 \tp{\frac{1}{q}}^{\frac{2(k-1)}{2 + \zeta}} \le 1$ gives that $q\ge (32\mathrm{e}^2 k^4)^{\frac{2+\zeta}{2k-16-7\zeta}} \Delta^{\frac{8+4\zeta}{2k-16-7\zeta}}$. By the fact that $\Delta \ge 1$, we have that $q\ge 120$ and $\zeta \approx 0.336205$. So it suffices to show that $q\ge 2\Delta^{\frac{4.7}{k-9.2}}$. By $\coloringCondition$, this holds.
\end{proof}

}
\subsection{Bounding the variance}
Now, we consider how to bound the variance and prove \Cref{lemma:clt-variance-bound}.
Recall that $\mu$ is the Gibbs distribution over $\Omega_{\Phi}$ with the external field $\lambda$.
Note that 
\begin{equation}
\label{eq:clt-variance-decomposition}
\Var[\mu]{\+Y} = \sum_{v\in V} \Var[\mu]{Y_v} + \sum_{\substack{u, v\in V \\ u\neq v}} \Cov[\mu]{Y_u, Y_v}.
\end{equation}
We establish a lower bound on the variance by providing a lower bound for the sum of per-variable variances $\sum_{v\in V} \Var[\mu]{Y_v}$ through the \emph{local uniformity} and an upper bound on the absolute value of the covariances through the \emph{total influence}.

The next lemma establish \emph{local uniformity}.

% Then, we use the \emph{local uniformity} to bound $\var_{\mu}[Y_v]$ (\Cref{lemma:clt-marginal-bound}), and we use the recursive coupling in~\cite{wang2024sampling} to bound the \emph{total influence}, so that we can upper bound the covariance (\cref{eq:clt-covariace-lower-bound}, \Cref{lemma:clt-bounded-total-influence}). 
\begin{lemma}
 \label{lemma:clt-marginal-bound}
    Suppose \Cref{condition:clt-condition}, then for any variable $v\in V$ with $1^\coF \in \Sigma_v$, we have that 
    \[
    \frac{q^* \lambda}{q - q^* + q^* \lambda} - \frac{q^* \lambda}{ (q - q^* + q^* \lambda) \Delta k} \le \Pr[\mu]{Y_v = 1}\le \frac{q^* \lambda}{q - q^* + q^* \lambda} + \frac{q-q^*}{ (q - q^* + q^* \lambda) \Delta k}.
    \]
\end{lemma}
\begin{proof}
We use \Cref{theorem:HSS} to prove this lemma.
Let $\sigma \sim \mu$ follow the Gibbs distribution $\mu$ with the external field $\lambda$.
% Let $\+A$ be the event that the projected symbol of $\sigma_v$ is not $1^\coF$, i.e. , $\sigma_v \in f^{-1}_v(1^\coF)$. 
Let $\+P$ be the product distribution over $[q]^V$.
We first verify the condition \cref{eq:LLL}. 
Recall that $r_{\max}\defeq \frac{\max\set{1, \lambda}}{q - q^* + q^* \lambda}$.
For each constraint $c\in C$, let $B_c$ be the event that the constraint $c$ is violated, i.e., $B_c = \lnot c$ 
and we set $x(B_c) \defeq \mathrm{e}\cdot r_{\max}^k$. 
% and we set $x(B_c) \defeq \frac{\min\set{1, \lambda}}{2(\Delta k)^2}\cdot \frac{\max\set{1, \lambda}}{q - q^* + \lambda q^*} = \frac{\lambda}{2(q - q^* + q^* \lambda) (\Delta k)^2}$. 
% Recall that $p$ is the upper bound of the violating probability.
By \Cref{condition:clt-condition}, it can be verified that 
\[
\+P[B_c] \le r_{\max}^k \le \mathrm{e}\cdot r_{\max}^k \cdot \tp{1 - \mathrm{e}\cdot r_{\max}^k}^{\Delta k}.
\]
With the fact that $|\Gamma(B_c)|\le k\Delta$, it holds that $\+P[B_c] \le x(B_c)\cdot \prod_{B'\in \Gamma(B_c)}(1 - x(B'))$.
Finally, by \Cref{theorem:HSS}, for any variable $v$, we have
\[
\Pr[\mu]{Y_v = 0} \le \frac{q - q^*}{q - q^* + q^* \lambda} \tp{1 - \mathrm{e}r_{\max}^k}^{-\Delta k}\le \frac{q - q^*}{q - q^* + q^* \lambda}  + \frac{q^* \lambda}{(q - q^* + q^* \lambda)\Delta k},
\]
and
\[
\Pr[\mu]{Y_v =  1} \le \frac{q^* \lambda}{q - q^* + q^* \lambda}\tp{1 - \mathrm{e}r_{\max}^k}^{-\Delta k} \le \frac{q^* \lambda}{q - q^* + q^* \lambda} + \frac{ q-q^* }{ (q - q^* + q^* \lambda) \Delta k}.
\]
% So we have $\+P[\sigma_v = o \mid C] \ge \frac{1}{|Q_v|} - \frac{1}{|Q_v|D}$.
\end{proof}

As a corollary, we have the following lower bound for $\Var[\mu]{Y_v}$.
\begin{corollary}
\label{lemma:clt-local-variance}
Suppose \Cref{condition:clt-condition}, then for any variable $v$ with $1^\coF \in \Sigma_v$, we have that 
\[ \Var[\mu]{Y_v} \ge \frac{ (q^* \lambda)\cdot (q-q^*) }{ \tp{q-q^*+q^*\lambda}^2 } \cdot \tp{\frac{\Delta k - 1}{\Delta k}}^2.\]
\end{corollary}
\begin{proof}
By the definition of $\Var[\mu]{Y_v}$, it holds that $\Var[\mu]{Y_v} = \Pr[\mu]{Y_v = 1} - \Pr[\mu]{Y_v=1}^2$.
Then, the lower bound follows from plugging in the bounds from \Cref{lemma:clt-marginal-bound}.
\end{proof}

% Then we bound the covariance. 
% {\color{red}
To deal with the covariances in~\cref{eq:clt-variance-decomposition}, we again leverage the recursive coupling in~\cite{wang2024sampling} to bound the covariance.
However, we need a much stronger upper bound of the \emph{total influence} in order to derive the lower bound of the variance, different from the previous section, we need to modify the conditions in~\cite{wang2024sampling}.
We first outline the basic idea of how the total influence provides a bound for the covariances. We remark that this part is almost the same as what we did in the previous section.
After that, we provide a high-level idea of how we modify the proofs in~\cite{wang2024sampling}.

% We first outline the basic idea. This part is almost the same as we do in the previous section. The difference here is that we need the lower bound of the variance so that we need to modify the conditions in~\cite{wang2024sampling}.
% Recall that $\mu$ is the uniform distribution over all satisfying assignments of $\Phi$.
Let $u,v \in V$ and $u\neq v$. 
\begin{align*}
\Cov[\mu]{Y_u, Y_v} =&~ \Pr[\mu]{\sigma_u \in f^{-1}_{u}(1^\coF) \land \sigma_v \in \complexColorSet{v}} - \Pr[\mu]{\sigma_u \in \complexColorSet{u}}\cdot \Pr[\mu]{\sigma_v \in \complexColorSet{v}}\\
=&~ \sum_{x\in \complexColorSet{u}} \Pr[\mu]{\sigma_u = x} \tp{\Pr[\mu]{\sigma_v \in \complexColorSet{v} \mid \sigma_u = x} - \Pr[\mu]{\sigma_v \in \complexColorSet{v}}} \\
\ge&~ -\sum_{x\in \complexColorSet{u}}\abs{\Pr[\mu]{\sigma_v \in \complexColorSet{v} \mid \sigma_u = x} - \Pr[\mu]{\sigma_v \in \complexColorSet{v}}}.
\end{align*}
% By the triangle inequality, it holds that 
% \[
% \abs{\mu(\sigma_v \in \complexColorSet{v} \mid \sigma_u = x) - \mu(\sigma_v \in \complexColorSet{v})} \le \sum_{y\in \complexColorSet{v}}\abs{\mu(\sigma_v =y \mid \sigma_u = x) - \mu(\sigma_v = y)}.
% \]
% So we have that 
% \[
% \-{Cov}(Y_u, Y_v) \ge - \sum_{x\in \complexColorSet{u}}\sum_{y\in \complexColorSet{v}}\abs{\mu(\sigma_v =y \mid \sigma_u = x) - \mu(\sigma_v = y)}.
% \]
% The total influence is defined as 
% \[
% \+I = \max_{u\in V} \sum_{v\in V: u\neq v} |\mu(\sigma_u = 1) - |
% \]
Then it holds that
\begin{equation}\label{eq:clt-covariace-lower-bound}
\sum_{\substack{u, v\in V \\ u\neq v}} \Cov[\mu]{Y_u, Y_v} \ge - \sum_{u\in V}\sum_{x\in \complexColorSet{u}}\sum_{\substack{v\in V: u\neq v}}\abs{\Pr[\mu]{\sigma_v \in \complexColorSet{v} \mid \sigma_u = x} - \Pr[\mu]{\sigma_v \in \complexColorSet{v}}}.
\end{equation}
% Similarly, we have
% \begin{equation}\label{eq:clt-covariace-upper-bound}
% \sum_{\substack{u, v\in V \\ u\neq v}} \Cov[\mu]{Y_u, Y_v} \le \sum_{u\in V}\sum_{x\in \complexColorSet{u}}\sum_{\substack{v\in V: u\neq v}}\abs{\Pr[\mu]{\sigma_v \in \complexColorSet{v} \mid \sigma_u = x} - \Pr[\mu]{\sigma_v \in \complexColorSet{v}}}.
% \end{equation}
Later in \Cref{lemma:clt-bounded-total-influence}, for all $u\in V$ and $x\in \complexColorSet{u}$, we provide a uniform upper bound on the total influence:
\begin{equation}\label{eq:clt-total-influence}
\max_{u\in V} \sum_{\substack{v\in V: u\neq v}}\abs{\Pr[\mu]{\sigma_v \in \complexColorSet{v} \mid \sigma_u = x} - \Pr[\mu]{\sigma_v \in \complexColorSet{v}}}.
\end{equation}
We remark that this is a variant of the standard total influence.

\begin{lemma}
\label{lemma:clt-bounded-total-influence}
Let $\Phi = (V, [q]^V, C)$ be a $(k, \Delta)$-CSP formula satisfying \Cref{condition:clt-condition}, and let $\mu$ be the Gibbs distribution over $\Omega_{\Phi}$ with the external field $\lambda$. For any $x\in f^{-1}_v(1^\coF)$, it holds that the total influence \cref{eq:clt-total-influence} is upper bounded by $\frac{1}{2 q^*} \cdot \frac{ (q^* \lambda)\cdot (q-q^*) }{ \tp{q-q^*+q^*\lambda}^2 } \cdot \tp{\frac{\Delta k - 1}{\Delta k}}^2$.
% \[
% \sum_{\substack{v\in V: u\neq v}}\abs{\Pr[\mu]{\sigma_v \in \complexColorSet{v} \mid \sigma_u = x} - \Pr[\mu]{\sigma_v \in \complexColorSet{v}}} \le \frac{1}{8 q^3}.
% \]
% \begin{align*}
%     &~\sum_{\substack{v\in V: u\neq v}}\abs{\Pr[\mu]{\sigma_v \in \complexColorSet{v} \mid \sigma_u = x} - \Pr[\mu]{\sigma_v \in \complexColorSet{v}}}\\
%     \le&~
%     \frac{1}{2} \cdot \frac{q^* \lambda}{q - q^* + q^* \lambda}\cdot \frac{\Delta k - 1}{\Delta k}\cdot \tp{ 1- \frac{q^* \lambda}{q - q^* + q^* \lambda} \cdot \frac{\Delta k + 1}{\Delta k}}.
% \end{align*}
\end{lemma}
We prove this lemma in the next subsection.
We first use it to prove \Cref{lemma:clt-variance-bound}.
\begin{proof}[Proof of \Cref{lemma:clt-variance-bound}]
    Recall \cref{eq:clt-variance-decomposition}. Combined with \Cref{lemma:clt-local-variance} we have that
    \[ \sum_{v\in V}\Var[\mu]{Y_v} \ge \frac{ (q^* \lambda)\cdot (q-q^*) }{ \tp{q-q^*+q^*\lambda}^2 } \cdot \tp{\frac{\Delta k - 1}{\Delta k}}^2 \cdot N^*.\]
    Combined with \Cref{lemma:clt-bounded-total-influence}, \cref{eq:clt-covariace-lower-bound}, we have 
    \[
    \Var[\mu]{\+Y} \ge \frac{1}{2} \cdot \frac{ (q^* \lambda)\cdot (q-q^*) }{ \tp{q-q^*+q^*\lambda}^2 } \cdot \tp{\frac{\Delta k - 1}{\Delta k}}^2 \cdot N^*.
    \]
    % By \Cref{lemma:clt-bounded-total-influence}, we have 
    % \[
    % \sum_{u\in V}\sum_{x\in \complexColorSet{u}}\sum_{\substack{v\in V: u\neq v}}\abs{\mu(\sigma_v \in \complexColorSet{v} \mid \sigma_u = x) - \mu(\sigma_v \in \complexColorSet{v})} \le N^*\cdot q^*\cdot \frac{1}{8q^3}\ge -\frac{N^*}{8q^2}.
    % \]
    % Combined with \cref{eq:clt-covariace-lower-bound} and \cref{eq:clt-covariace-upper-bound}, we have $ \frac{N^*}{8q^2}\le  \Var[\mu]{\+Y} \le 5N^*$.
\end{proof}

\subsection{Bounding the total influence}
\label{subsection:clt-bound-total-influence}
In this subsection, we prove \Cref{lemma:clt-bounded-total-influence}. We remark that this proof is similar to that of \Cref{lemma:general-csp-covariance}.

The key is to leverage the recursive coupling in~\cite{wang2024sampling}. For two CSP formulas $\Phi = (V, [q]^V, C)$ and $\Phi' = (V, [q]^V, C\setminus \set{c})$ with $c\in C$, we again slightly abuse the notation $\mu$, denoting $\mu_{C}$ as the Gibbs distribution over all satisfying assignments of $\Phi$, and $\mu_{C\setminus \set{c}}$ as the Gibbs distribution over all satisfying assignments of $\Phi'$.
% We first include the following useful lemma in~\cite{wang2024sampling}.
% We sketch its proof in \Cref{section:missing-k-cnf}. 
Similar to the proof of \Cref{lemma:general-csp-covariance}, we construct a coupling between the original CSP formula and a new CSP formula conditioned on a partial assignment on one variable.
Formally, given $v\in V$, $x\in [q]$, and a CSP formula $\Phi = (V, [q]^V, C)$ , we consider a new CSP formula $\Phi' = (V\setminus \set{v}, [q]^{V\setminus \set{v}}, C')$ conditioned on $\sigma_v = x$.
It is straightforward to see that if the original $\Phi$ satisfies \Cref{condition:clt-condition}, then the new $\Phi'$ satisfies the following condition:
\begin{condition}
\label{condition:clt-condition-under-one-pinning}
$\Phi = (V, [q]^V, C)$ is a CSP formula that each constraint contains either $k-1$ or $k$ variables and each variable belong to at most $\Delta$ constraints. 
Let $r_{\max}\defeq \frac{\max\set{1, \lambda}}{q - q^* + q^* \lambda}$.
It holds that
\[
16\mathrm{e}^2 \cdot \frac{(q - q^* + q^* \lambda)^2}{\lambda\cdot(q-q^*)} \cdot r_{\max}^{\frac{2(k-1)}{2 + \zeta}} \cdot (\Delta k + 1)^4 \le 1 , \quad \hbox{where } \zeta = \frac{2\ln{(2 - r_{\max})}}{\ln{(1/r_{\max})} - \ln{(2 - r_{\max})}}.
\]
\end{condition}
We also note that the condition is preserved under simplifying the formula itself.
\begin{observation}
\label{observation:clt-condition-to-condition-under-pinning}
 If $\Phi'$ satisfies \Cref{condition:clt-condition-under-one-pinning}, the formula $\Phi''$, obtained by removing all satisfied constraints and the pinned variable $v$ from other constraints, also satisfies \Cref{condition:clt-condition-under-one-pinning}.
\end{observation}
Under \Cref{condition:clt-condition-under-one-pinning}, there is a coupling with bounded Hamming distance.
\begin{lemma}
\label{lemma:clt-bounded-hamming-distance}
% Suppose that 
% \[
% 32\mathrm{e}^2 q^3 \Delta^4 k^4 \tp{\frac{1}{q}}^{\frac{2(k-1)}{2 + \zeta}} \le 1,
% \]
% where $\zeta = \frac{2\ln(2 - 1/q)}{\ln(q) - \ln(2 - 1/q)}$. 
% Suppose \Cref{condition:clt-condition}.
Let $\Phi=(V, [q]^V, C)$ be a CSP formula satisfying \Cref{condition:clt-condition-under-one-pinning}. Then, for any constraint $c\in C$, there exists a coupling $(\+A, \+B)$ of $\mu_{C}$ and $\mu_{C\backslash\{c\}}$, such that 
\[\=E_{(\+A, \+B)}[d_{\-{Ham}}(\+A, \+B)] \le \frac{1}{4 q^* \Delta } \cdot \frac{ (q^* \lambda)\cdot (q-q^*) }{ \tp{q-q^*+q^*\lambda}^2 } \cdot \tp{\frac{\Delta k - 1}{\Delta k}}^2,\]
% $\=E_{(\+A, \+B)}[d_{\-{Ham}}(\+A, \+B)] \le \frac{1}{8 q^3 \Delta}$,
% for any integer $K\ge 1$,
% \[
% \-{Pr}[d_{\-{Ham}}(\+A, \+B)\ge k\cdot (\Delta k + 1)\cdot K] \le 2^{-K},
% \]
where $d_{\-{Ham}}(\+A, \+B)\defeq \sum_{v\in V} \=I[\+A_v\neq \+B_v]$ denotes the Hamming distance between $\+A$ and $\+B$.
\end{lemma}
We first use this lemma to prove \Cref{lemma:clt-bounded-total-influence}.
\begin{proof}[Proof of \Cref{lemma:clt-bounded-total-influence}]
% Let $C$ be the constraint set. Recall that each constraint in $C$ has exactly $k$ variables and each variable belongs to at most $\Delta$ constraints.
We simplify $C$ using the partial assignment $\sigma_u = x$ and denote the new constraint set as $C'$, i.e., we remove all constraints that are already satisfied by $\sigma_u = x$ and eliminate the variable $u$ from the other constraints.
Note that in this new CSP formula, each constraint in $C'$ has at least $k-1$ variables and at most $k$ variables, and each variable belongs to at most $\Delta$ constraints.
By \Cref{observation:clt-condition-to-condition-under-pinning}, this new CSP formula satisfies \Cref{condition:clt-condition-under-one-pinning}.

% We use $\mu_{C'}$ to denote the uniform distribution on all satisfying assignments of the constraint set $C'$.
% For a constraint $c$ in $C'$, let $\mu_{C'\backslash\{c\}}$ be the uniform distribution on all satisfying assignments of the constraint set $C'\backslash\{c\}$.
% Let $C$ be the constraint set, and let $C'$ be the new constraint set that is simplified by setting the $\sigma_u = x$.
% Recall each constraint in $C$ has exactly $k$ variables and each variable belongs to at most $\Delta$ constraints.
% And each constraint in $C'$ has at least $k-1$ variables and at most $k$ variables, and each variable belongs to at most $\Delta$ constraints.
% Let $\mu$ be the uniform distribution on all satisfying assignments of $C$, and 
Let $\mu'$ be the Gibbs distribution on all satisfying assignments of $C'$ with the external field $\lambda$. 
Note that $\mu'(\cdot) = \mu(\cdot \mid \sigma_u = x)$.
We claim that there is a coupling $(\+A, \+B)$ of $\mu$ and $\mu'$, such that 
\[
\=E_{(\+A, \+B)} [d_{\-{Ham}}(\+A, \+B)] \le \frac{1}{2q^*} \cdot \frac{ (q^* \lambda)\cdot (q-q^*) }{ \tp{q-q^*+q^*\lambda}^2 } \cdot \tp{\frac{\Delta k - 1}{\Delta k}}^2.
\]
% $\color{red}\=E_{(\+A, \+B)} [d_{\-{Ham}}(\+A, \+B)] \le \frac{1}{8 q^3}$.
\Cref{lemma:clt-bounded-total-influence} follows from this coupling, because the total influence can be bounded by hamming distance:
\begin{align*}
&~\sum_{\substack{v\in V: u\neq v}}\abs{\mu(\sigma_v \in \complexColorSet{v} \mid \sigma_u = x) - \mu(\sigma_v \in \complexColorSet{v})}\\
\le&~ \sum_{\substack{v\in V: u\neq v}} \=E_{(\+A, \+B)}\inbr{\abs{\=I[\+A_v \in \complexColorSet{v}] - \=I[\+B_v \in \complexColorSet{v}]}}\\
=&~ \=E_{(\+A, \+B)}\inbr{\sum_{\substack{v\in V: u\neq v}}\abs{\=I[\+A_v\in \complexColorSet{v}] - \=I[\+B_v \in \complexColorSet{v}]}}\\
\le&~ \=E_{(\+A, \+B)}[d_{\-{Ham}}(\+A, \+B)],
\end{align*}
where the first inequality is due to the standard coupling lemma and the equality is due to the linearity of expectation.

Finally, we show the claimed coupling.
Recall that each variable in $C$ or $C'$ belongs to at most $\Delta$ constraints. Let $C\oplus C'$ be the symmetric difference between $C$ and $C'$. We have that $|C\oplus C'| \le 2\Delta$. 
With this fact and \Cref{lemma:clt-bounded-hamming-distance}, we bound $\=E_{(\+A, \+B)} [d_{\-{Ham}}(\+A, \+B)]$ by a triangle inequality.
We first remove all constraints in $C\backslash C'$ one by one and then add constraints in $C'\backslash C$ one by one. By the triangle inequality, it holds that
\[
\=E_{(\+A, \+B)}[d_{\-{Ham}}(\+A, \+B)] \le \frac{1}{2 q^*} \cdot \frac{ (q^* \lambda)\cdot (q-q^*) }{ \tp{q-q^*+q^*\lambda}^2 } \cdot \tp{\frac{\Delta k - 1}{\Delta k}}^2.
\]
\end{proof}

We conclude this subsection by sketching the proof of \Cref{lemma:clt-bounded-hamming-distance}, which is almost identical to the proof of~\cite[Lemma 3.10]{wang2024sampling}.
% For a $(k, q\Delta)$-CSP $\Phi = (V,\*Q, C)$.
% Let $c\in C$ be a constraint. 
% And we use $\mu_C$ and $\mu_{C\setminus \{c\}}$ to denote the uniform distribution on $\Omega_C$ and $\Omega_{C\setminus\{c\}}$ respectively.
We use the coupling $(\+A, \+B)$ between $\mu_{C}$ and $\mu_{C\setminus \{c\}}$ in~\cite[Algorithm 1]{wang2024sampling}.
They defined a set $B$ as a witness for the discrepancy of the coupling procedure in~\cite[Definition 3.4]{wang2024sampling}.
By~\cite[Lemma 3.5]{wang2024sampling}, they show that for any realization $(\+A, \+B)$ of the coupling with $d_{\-{Ham}}(\+A, \+B) = K$, it holds that $|B|\ge \lfloor K/(k\cdot (\Delta k + 1))\rfloor$. And $|B| = 0$ if and only if $d_{\-{Ham}}(X, Y) = 0$. 

Hence, if $|B|= K$, then we have $d_{\-{Ham}}(\+A, \+B)\le k\cdot(\Delta k + 1)\cdot K$. So it holds that 
\begin{align*}
\=E_{(\+A, \+B)}[d_{\-{Ham}}(\+A, \+B)] \le k\cdot (\Delta k + 1)\cdot \sum_{i = 1}^{|V|} \Pr{|B|\ge i}.
\end{align*}

Recall that $r_{\max} = \frac{\max\set{1, \lambda}}{q - q^* + q^* \lambda}$ and $\zeta = \frac{2\ln{(2 - r_{\max})}}{\ln{(1/r_{\max})} - \ln{(2 - r_{\max})}}$.
By running through the proof of~\cite[Lemma 3.10]{wang2024sampling}, we have that 
% \begin{align*}
% \Pr{|B|\ge i} \le&~ \tp{\mathrm{e}(\Delta k)^2\cdot r_{\max}^{\frac{2(k-1)}{2 + \zeta}}\cdot \tp{1 - \mathrm{e}\cdot r_{\max}^{k-1}}^{-2(\Delta k + 1)}}^i\\
% \le&~  \tp{(\mathrm{e}\Delta k)^2\cdot r_{\max}^{\frac{2(k-1)}{2+\zeta}}}^i\\
% \le&~ \tp{\frac{1}{16 q^3 \Delta k (\Delta k + 1)}}^i,
% \end{align*}
\begin{align*}
\Pr{|B|\ge i} \le&~ \tp{\mathrm{e}(\Delta k)^2\cdot r_{\max}^{\frac{2(k-1)}{2 + \zeta}}\cdot \tp{1 - \mathrm{e}\cdot r_{\max}^{k-1}}^{-2(\Delta k + 1)}}^i \le  \tp{(\mathrm{e}\Delta k)^2\cdot r_{\max}^{\frac{2(k-1)}{2+\zeta}}}^i,
\end{align*}
where the last inequality is due to \Cref{condition:clt-condition}. Hence, combined with \Cref{condition:clt-condition}, we have that 
\[
\sum_{i = 1}^{|C|} \Pr{|B|\ge i} \le \frac{(\mathrm{e} \Delta k)^2\cdot r_{\max}^{\frac{2(k - 1)}{2 + \zeta}}}{1 - (\mathrm{e} \Delta k)^2\cdot r_{\max}^{\frac{2(k - 1)}{2 + \zeta}}} \le 2 (\mathrm{e} \Delta k)^2\cdot r_{\max}^{\frac{2(k - 1)}{2 + \zeta}}.
\]
Hence, we have
\[\sum_{i = 1}^{|C|} \Pr{|B|\ge i}  \le \frac{1}{4q^* \Delta \cdot (k(k\Delta + 1))} \cdot \frac{ (q^* \lambda)\cdot (q-q^*) }{ \tp{q-q^*+q^*\lambda}^2 } \cdot \tp{\frac{\Delta k - 1}{\Delta k}}^2.\]
% \[
% \frac{(\mathrm{e}\Delta k)^2}{(\Delta k+1)^{2x/(4.83)}\cdot (8^5 \mathrm{e}^5)^{2/(2+\zeta)}}
% \]
% \[2x/4.83 = 5,\quad  x = 12.075.\]

Combining inequalities mentioned above, we have that 
\[
\=E_{(\+A, \+B)}[d_{\-{Ham}}(\+A, \+B)] \le \frac{1}{4q^* \Delta } \cdot \frac{ (q^* \lambda)\cdot (q-q^*) }{ \tp{q-q^*+q^*\lambda}^2 } \cdot \tp{\frac{\Delta k - 1}{\Delta k}}^2.
\]
% \begin{align*}
% \=E_{(\+A, \+B)}[d_{\-{Ham}}(\+A, \+B)] &\le~ k\cdot (\Delta k + 1)\cdot \sum_{i = 1}^{|C|} \Pr{|B|\ge i}\\
% &\le~ k\cdot (\Delta k+1)\cdot \sum_{i=1}^{\infty} \tp{\frac{1}{16 q^3 \Delta k (\Delta k + 1)}}^i\\
% &\le~ \frac{1}{8 q^3 \Delta}.
% \end{align*}

% Now, we are ready to prove \Cref{lemma:clt-variance-bound}.

\section{Local central limit theorem}
\label{section:lclt}

In this section, we derive a local central limit theorem for CSP formulas from the central limit theorem.
As an application, we prove a local CLT for hypergraph coloring with one special color in \Cref{subsection:lclt-application}.
%  \todo{yyx: fill.}
We follow the same high-level ideas in~\cite{jain2022approximate} (See also~\cite{liu2024phase}).  
We remark that in~\cite{jain2022approximate,liu2024phase}, the external field $\lambda$ need not be a constant and can depend on $|V|$. This allows them to derive an algorithm for counting fix-sized independent sets.
For general CSP formulas such as hypergraph colorings, handling $\lambda$ that are arbitrarily small seems to require new ideas, and we focus on constant external field $\lambda$ in this work.

Let $\Phi = (V, [q]^V, C)$ be an atomic $(k, \Delta)$-CSP formula. Our local CLT concerns $\Phi$ under projection.
% Note that $\Phi$ may not be atomic.\todo{check this.}
Again, we consider the state-compression scheme (\Cref{definition:state-compression}) with the projection $\*f = (f_v)_{v\in V}$ where $f_v: [q]\to \Sigma_v$. 
Let $\mu$ be the Gibbs distribution on $\Omega_\Phi$ with the external field $\lambda$ on the projected symbol $1^\coF$ (\Cref{definition:complex-extensions-csp} with $\lambda$ is a real number).
We consider the local CLT on the number of variables assigned the projected symbol $1^\coF$.
Let $\sigma \sim \mu$ following $\mu$, and let $Y_v$ be the random variable that indicates whether the projected symbol of $v$ is $1^\coF$; i.e., if $\sigma_v \in f^{-1}_v(1^\coF)$, then $Y_v = 1$; otherwise, $Y_v = 0$. Let $\+Y = \sum_{v\in V} Y_v$ be the summation. 
We assume that for any variable $v\in V$ if $1^\coF \in \Sigma_v$, it holds that $|f^{-1}_v(1^\coF)| = q^*$ where $0 < q^* < q$.
Let $N^*$ be the number of variables satisfying $1^\coF \in \Sigma_v$, i.e., $N^* = \abs{\set{v\in V \mid 1^\coF \in \Sigma_v}}$.

In this section, we establish a local CLT for $\+Y$. 
Our first ingredient is an abstract form of central limit theorem, which we established in \Cref{section:clt}.
\begin{condition}[Central limit theorem]
\label{condition:lclt-clt-condition}
$\Phi = (V, [q]^V, C)$ is a $(k, \Delta)$-CSP formula with the Gibbs distribution $\mu$ and the external field $\lambda$ on the projected symbol $1^\coF$.
Let $\bar{\mu}$ be the expectation of $\+Y$, and $\bar{\sigma}$ be the standard deviation of $\+Y$. We define $\+Y^*\defeq (\+Y - \bar{\mu})\bar{\sigma}^{-1}$. It holds that 
\[
\sup_{t\in \=R}|\Pr{\+Y^* \le t} - \Pr{\+Z \le t}| \le \lcltCltTail,
\]
where $\+Z\sim N(0, 1)$ is a standard Gaussian random variable and $\xi$ is a function.
\end{condition}
Our second ingredient comes from LLL-type conditions. Before that, we introduce the \emph{two-step projection scheme}.
\begin{definition}[Two-step projection scheme]
\label{definition:lclt-two-step-projection-scheme}
Let $\*h = (h_v)_{v\in V}$ be a new projection where $h_v: \Sigma_v \to \Sigma'_v$. We call $\*f$ is the \emph{first projection} and $\*h$ is the \emph{second projection}.
For any subset $S\subseteq V$, $\*\Sigma'_S=\bigotimes_{v\in S} \Sigma'_{v}$. We say $\*\Sigma'_S$ is the \emph{domain/alphabet under the two-step projection}.
For any constraint $c\in C$ and any assignment $\tau \in \*\Sigma'_V$ under the two-step projection, we say $c$ is \emph{satisfied} by $\tau$ if and only if for any assignment $\tau' \in [q]^V$ with $\forall v\in V$, $h(\tau'(v)) = \tau(v)$, we have $c(\tau') = \true$.
\end{definition}

\begin{condition}[LLL-type condition]
\label{condition:lclt-LLL-condition}
Given an atomic $(k, \Delta)$-CSP formula $\Phi = (V, [q]^V, C)$  with a two-step projection $\*f, \*h$, let $\+P$ be the product distribution on $V$ with the external field $\lambda$ on the projected symbol $1^\coF$.
\begin{enumerate}
\item 
Let $\lcltProjectedVioProb$ denote the maximum probability that a constraint is not satisfied by an assignment under the two-step projection following the product distribution, i.e., 
\[\lcltProjectedVioProb \defeq  \max_{c\in C} \Pr[\tau \sim \+P]{c \text{ is not satisfied by } (\*h \circ \*f)(\tau)}.\]
It holds that $2\mathrm{e}^2\cdot \lcltProjectedVioProb\cdot (\Delta k)^2 \le 1$.\label{item:lclt-LLL-condition-Proj}
\item 
Let $\lcltConditionVioProb$ denote the maximum probability that a constraint is violated conditioned on an assignment under the two-step projection following the product distribution, i.e., 
\[\lcltConditionVioProb \defeq \max_{c\in C} \max_{J \in \*\Sigma'} \Pr[\tau \sim \+P]{c \text{ is violated by } \tau \mid \*h(\tau) = J}.\]
It holds that $2\mathrm{e}\cdot \lcltConditionVioProb\cdot q\Delta k \le 1$. \label{item:lclt-LLL-condition-Cond}
\item
Let $r_{\max}\defeq \frac{\max\set{1, \lambda}}{q - q^* + q^* \lambda}$. It holds that
\[
(8\mathrm{e})^3\cdot \frac{q-q^*+q^*\lambda}{\min\set{q - q^*, q^* \lambda}}\cdot  r_{\max}^{k-1}\cdot (\Delta k+ 1)^{2 + \zeta} \le 1, \hbox{ where } \zeta = \frac{2\ln{(2 - r_{\max})}}{\ln{(1/r_{\max})} - \ln{(2 - r_{\max})}}.\]\label{item:lclt-LLL-condition-concentration}
\end{enumerate}

% \[p = \max_{c\in C} \Pr[\tau \sim \+P]{c \text{ is not satisfied by } (\*h \circ \*f)(\tau)}.\]
% It holds that 
% \[ 2\mathrm{e}^2\cdot p\cdot (\Delta k)^2 \le 1. \]
\end{condition}
As an example, the above condition is
asymptotically implied by $q\gtrsim \max\set{\tp{\frac{1}{\lambda}}^{\frac{1}{k-7/4}}, \lambda} \Delta^{\frac{4+o_q(1)}{k-7/4}}$ for
constant $\lambda$ and sufficiently large $q$, by choosing a uniform two-step projections in which the first
projection maps $q$ symbols to $\Theta(q^{3/4})$ symbols and the second
projection further maps them to $\Theta(q^{1/2})$ symbols. 

Now we are ready to state our local CLT result.
% \todo{maybe add a theorem here.}
\begin{theorem}
\label{theorem:lclt}
Suppose \Cref{condition:lclt-clt-condition} and \Cref{condition:lclt-LLL-condition} holds. Let $\+N(x) = \mathrm{e}^{-x^2/2}/\sqrt{2\pi}$ denote the density of the standard normal distribution. If $N^* = \Theta(V)$, we have 
% \begin{align*}
% &~\sup_{t\in \=Z} \abs{\Pr{\+Y = t} - \bar{\sigma}^{-1} N((t-\bar{\mu})/\bar{\sigma})} \\
% \le&~ \frac{1}{\bar{\sigma}} \int_{-\pi \bar{\sigma}}^{\pi \bar{\sigma}} \min\set{(2|t|\eta + 1) \lcltCltTail  + 4\mathrm{e}^{-\eta^2/2}, \frac{4\lcltConcentrationProb}{N^*}+ \frac{1}{|V|} + \exp\tp{- \frac{\lcltHighPhaseConstant N^*}{\log(|V|)^3} \tp{\frac{t}{\bar{\sigma}}}^2}} \d t + \mathrm{e}^{-\frac{\pi^2 \bar{\sigma}^2}{2}}.
% \end{align*}
\begin{align*}
\sup_{t\in \=Z} \abs{\Pr{\+Y = t} - \bar{\sigma}^{-1} N((t-\bar{\mu})/\bar{\sigma})} = O\tp{\zeta \cdot \frac{\log^{5/2} |V|}{\sqrt{|V|}}} + \mathrm{e}^{-\pi^2 \bar{\sigma}^2 /2 }.
\end{align*}
\end{theorem}
The key to lifting a CLT to a local CLT is the following inverse Fourier transform bound.
\begin{lemma}
    \label{lemma:establish-lclt-by-characteristic-function}
    Let $X$ be a random variable supported on the lattice $\+L = \alpha + \beta \=Z$ and let $\+N(x) = \mathrm{e}^{-x^2/2}/\sqrt{2\pi}$ denote the density of the standard normal distribution. Then 
    \[
    \sup_{x\in\+L} \abs{\beta\+N(x) - \Pr{X = x}} \le \beta \int_{-\pi/\beta}^{\pi/\beta} \abs{\=E\inbr{\mathrm{e}^{itX}} - \=E\inbr{\mathrm{e}^{it\+Z}}} \d t + \mathrm{e}^{-\pi^2/(2\beta^2)},
    \]
    where $\+Z\sim N(0, 1)$ is a standard normal Gaussian random variable.
\end{lemma}
At a high level, we closely follow the arguments in~\cite{jain2022approximate, liu2024phase}. 
The main idea is that the high Fourier phases, $\abs{\E{\mathrm{e}^{it \+Y^*}}}$ with large $t$'s, are negligible via a combinatorial argument. For low Fourier phases, we can use the central limit theorem to show that they are negligible.  
% \todo{yyx: state that $\lambda$ is a constant. And breaking this barrier needs new ideas.}
% We remark that to obtain a local central limit theorem for $\lambda \in (0,1]$, one needs a central limit theorem for $\lambda \in (0, 1]$. 
% However, the same proof in \Cref{section:clt} breaks down when $\lambda$ is small, for example, when $\lambda = \frac{\log n}{n}$, the lopsided Lov\'asz local lemma fails to provide a non-trivial marginal lower bound. Instead, we utilize the nature of hypergraph coloring to establish a new lower bound of variance for $\lambda \in (0, 1]$.
\begin{lemma}
\label{lemma:lclt-low-fourier-phases-bound}
% Fix $k\ge 2$, $\Delta \ge 3$ and $q$. Let $H = (V, \+E)$ be a $k$-uniform hypergraph with maximum degree $\Delta$.
% Let $n = |V|$, $\bar{\mu} = \=E[\+Y]$ and $\sigma^2 = \Var{\+Y}$. Let $\bar{\+Y} = (\+Y - \bar{\mu})/\sigma$, $\+Z\sim N(0, 1)$ be a standard Gaussian random variable.
Suppose \Cref{condition:lclt-clt-condition} holds.
For any $t\in \=R$ and any $\eta > 0$, we have that
\[
\abs{ \E{\mathrm{e}^{it\+Y^*}} - \E{\mathrm{e}^{it \+Z}} } \le (2|t|\eta + 1)\cdot \lcltCltTail  + 4\mathrm{e}^{-\eta^2/2}.
\]
% \[
% \abs{}
% \]
% Let $$
\end{lemma}
\begin{proof}
Let $\+Y'$ be $\+Y^*$ convolved with a centered Gaussian of infinitesimally small variance so that $\+Y'$ has a density function with respect to the Lebesgue measure on $\=R$; it suffices to consider $\+Y'$ and then pass to the limit.
Next, we express $\E{\mathrm{e}^{it \+Y'}}$ into an integration. Let $p_{\+Y'}$ be the law of $\+Y'$.
\begin{align*}
\E{\mathrm{e}^{it\+Y'}} =&~ \int_{\infty}^{\infty} \mathrm{e}^{itz} p_{\+Y'}(z) \d z \\
    =&~ \int_{|z|\le \eta} \mathrm{e}^{itz} p_{\+Y'}(z)\d z \pm \mathrm{e}^{i\theta'} \Pr{|\+Y'| \ge \eta} \\
    =&~ \inbr{\mathrm{e}^{itz} \tp{\int_{-\eta}^z p_{\+Y'}(z')\d z'}}_{z = -\eta}^{z = \eta} - \int_{-\eta}^{\eta} it\mathrm{e}^{itz} \tp{\int_{-\eta}^z p_{\+Y'}(z') \d z'} \d z \pm \mathrm{e}^{i\theta'} \Pr{|\+Y'| \ge \eta}.
\end{align*}
Hence
\begin{align*}
    % \E{\mathrm{e}^{it\+Y'}} =&~ \int_{\infty}^{\infty} \mathrm{e}^{itz} p_{\+Y'}(z) \d z \\
    % =&~ \int_{|z|\le \eta} \mathrm{e}^{itz} p_{\+Y'}(z)\d z \pm \mathrm{e}^{i\theta'} \Pr{|\+Y'| \ge \eta}\\
    % =&~ \inbr{\mathrm{e}^{itz} \tp{\int_{-\eta}^z p_{\+Y'}(z')\d z'}}_{z = -\eta}^{z = \eta} - \int_{-\eta}^{\eta} it\mathrm{e}^{itz} \tp{\int_{-\eta}^z p_{\+Y'}(z') \d z'} \d z \pm \mathrm{e}^{i\theta'} \Pr{|\+Y'| \ge \eta}\\
    \E{\mathrm{e}^{it\+Y'}} =&~ \mathrm{e}^{it\eta} - \int_{-\eta}^{\eta} it \mathrm{e}^{itz} \Pr{\+Y' \in [-\eta, z]} \d z \pm \mathrm{e}^{i \theta'} \Pr{|\+Y'| \ge \eta} - \mathrm{e}^{it \eta} \Pr{|\+Y'| \ge \eta} \\
    =&~ \mathrm{e}^{it\eta} - \int_{-\eta}^{\eta} it \mathrm{e}^{itz} \Pr{\+Y' \in [-\eta, z]} \d z + \mathrm{e}^{i \theta} \Pr{|\+Y'| \ge \eta},
\end{align*}
for some $\theta', \theta \in [0, 2\pi)$. 
We also apply the same calculation to $\+Z$ instead of $\+Y'$ and taking the difference, by the triangle inequality, we have that 
\begin{align*}
&~ \abs{ \E{\mathrm{e}^{it\+Y'}} - \E{\mathrm{e}^{it \+Z}} } \\
\le&~ |t|\int_{-\eta}^{\eta} \abs{\Pr{\+Y' \in [-\eta, z]} - \Pr{\+Z \in [-\eta, z]}} \d z + \Pr{|\+Y'| \ge \eta} + \Pr{|\+Z| \ge \eta}. 
\end{align*}
Combined with the central limit theorem of $\+Y^*$ (\Cref{condition:lclt-clt-condition}), it holds that 
\begin{align*}
\abs{ \E{\mathrm{e}^{it\+Y^*}} - \E{\mathrm{e}^{it \+Z}} } \le (2|t|\eta + 1)\cdot \lcltCltTail  + 4\mathrm{e}^{-\eta^2/2}.
\end{align*}
\end{proof}
Next, we control the high Fourier phases, $\abs{\E{\mathrm{e}^{it\+Y^*}}}$ with large $t$'s.
Note that $\+Y^* = (\+Y - \bar{\mu})/\bar{\sigma}$. Hence $\abs{\E{\mathrm{e}^{it\+Y^*}}} = \abs{\E{\mathrm{e}^{it \+Y / \bar{\sigma}}}}$.
The basic idea is that we construct an event $J$ such that conditioned $J$, $\+Y$ can be factorized into a product of $\ell$ independent random variables $X_1, X_2, \dots, X_\ell$, i.e., $\E{\mathrm{e}^{it \+Y/\bar{\sigma}} \mid J} = \prod_{j = 1}^\ell \E{\mathrm{e}^{i t X_j}\mid J}$.
We show that when $J$ is ``good'', we can control the contribution of $X_1, X_2, \dots, X_\ell$, and with high probability, $J$ is ``good''.
Hence, we have 
\begin{equation}
\label{eq:lclt-high-fourier-phases-condition-on-J}
\abs{\E{\mathrm{e}^{it \+Y/\bar{\sigma} }}} = \abs{\E{ \prod_{j = 1}^{\ell} \E{ \mathrm{e}^{it X_j/\bar{\sigma}} \mid J}}} \le \Pr{J \text{ is not ``good''}} + \max_{J \text{ is ``good''}} \prod_{j = 1}^\ell \abs{\E{\mathrm{e}^{it X_j / \bar{\sigma}} \mid J}}.
\end{equation}

Next, for any $J$ and any $j\in [\ell]$, let $X_j'$ be an independent copy of $X_j$. It holds that 
\begin{equation*}
% \label{eq:lclt-high-fourier-phases-each-X_j}
\begin{aligned}
\abs{\E{\mathrm{e}^{it X_j/\bar{\sigma}}\mid J}}^2 \le&~ \E{\mathrm{e}^{it(X_j - X_j')/\bar{\sigma}}\mid J}\\
=&~ \Pr{X_j = X_j' \mid J} + \sum_{k = 1} \Pr{\abs{X_j - X_j'} = k \mid J}\cos(k\cdot t/\bar{\sigma}) \\
\le&~ \Pr{X_j = X_j' \mid J} + \sum_{k = 2} \Pr{\abs{X_j - X_j'} = k \mid J} + 2\Pr{X_j - X_j' = 1 \mid J}\cos(t/\bar{\sigma}). \\
\end{aligned}
\end{equation*}
Hence, it holds that 
\begin{equation}
\label{eq:lclt-high-fourier-phases-each-X_j}
\begin{aligned}
\abs{\E{\mathrm{e}^{it X_j/\bar{\sigma}} \mid J}}^2 \le 1 - 2 \Pr{X_j - X_j' = 1 \mid J}(1 - \cos(t/\bar{\sigma})) \le 1 - \frac{1}{4} \Pr{X_j - X_j' = 1 \mid J} (t/\bar{\sigma})^2.
\end{aligned}
\end{equation}
Next, we construct the event $J$ by the two-step projection scheme (\Cref{definition:lclt-two-step-projection-scheme}). 
$J$ is actually an assignment of all variables in $V$ under the two-step projection, i.e., $J \in \*\Sigma'_V$.
% Recall that $\*f = (f_v)_{v\in V}$ where $f_v: [q]\to \Sigma_v$ is the projection. We introduce the other projection $\*h = (h_v)_{v\in V}$ where $h_v: \Sigma_v \to \Sigma'_v$. We call $\*f$ is the \emph{first projection} and $\*h$ is the \emph{second projection}.
% For any subset $S\subseteq V$, $\*\Sigma'_S=\bigotimes_{v\in S} \Sigma'_{v}$. We say $\*\Sigma'_S$ is the \emph{domain under the two-step projection}.
% For any constraint $c\in C$ and any assignment $\tau \in \*\Sigma'_V$ under the two-step projection, we say $c$ is \emph{satisfied} by $\tau$ if and only if for any assignment $\tau' \in [q]^V$ with $\forall v\in V$, $h(\tau'(v)) = \tau(v)$, we have $c(\tau') = \true$.
% Then $J \in \*\Sigma'_V$ is defined as an assignment of all variables in $V$ under the two-step projection. 
Next, we use $J$ to simplify the CSP formula, i.e., we remove all satisfied constraints in $C$. Let the simplified CSP formula be $\Phi' = (V, [q]^V, C')$.
$\Phi'$ shatters into $\ell$ disjoint sub-CSP formulas $\Phi_1 = (V_1, [q]^{V_1}, C_1), \Phi_2 = (V_2, [q]^{V_2}, C_2), \dots, \Phi_\ell = (V_{\ell}, [q]^{V_{\ell}}, C_{\ell})$. We further assume that the dependency graph for any sub-CSP formula is connected.
% and let $G_{\Phi'}$ be its corresponding dependency graph.
% $G_{\Phi'}$ may shatters into $\ell$ maximal connected components $G_1 = (V_1, E_1), G_2 = (V_2, E_2), \dots, G_\ell = (V_\ell, E_{\ell})$.

We say $J$ is \emph{good} if and only if (1) for any $j\in [\ell]$, it holds that $|C_j|\le \connectedComponentSize$, and (2) there are least $\ell^* \defeq \left \lceil \frac{\lcltExpectation N^*}{2 \connectedComponentVarSize}\right \rceil$ sub-CSP formulas containing a variable $v$ with $1^\coF \in h^{-1}_v(J_v)$ where $\lcltExpectation$ is a constant with respect to $q, q^*, \lambda, \Delta$ and $k$ defined in \Cref{lemma:lclt-concentration-expectation-lower-bound}. 
% \todo{check this.}
 
% We show that under the next condition, we can show the diminishing of the high Fourier phases.
% \begin{condition}
% \label{condition:lclt-good-condition}
% $\Phi = (V, [q]^V, C)$ is an atomic $(k, \Delta)$-CSP formula with the two-step projection $\*f, \*h$.
% Let $p$ denote the maximum probability that a constraint is not satisfied by a partial assignment under the two-step projection following the product distribution, i.e., $p = \max_{c\in C} \Pr[\tau \sim \+P]{c \text{ is not satisfied by } (\*h \circ \*f)(\tau)}$.
% % \[p = \max_{c\in C} \Pr[\tau \sim \+P]{c \text{ is not satisfied by } (\*h \circ \*f)(\tau)}.\]
% It holds that 
% \[ 2\mathrm{e}^2\cdot p\cdot (\Delta k)^2 \le 1. \]
% \end{condition}
We show that with high probability, $J$ is good. And we defer its proof to \Cref{subsection:lclt-J-is-good}.
\begin{lemma}
\label{lemma:lclt-J-is-good}
Suppose \Cref{condition:lclt-LLL-condition}. It holds that 
\[\Pr{J \text{ is good}} \ge 1 - \frac{4\lcltConcentrationProb}{N^*} - \frac{1}{|V|}, \hbox{ where  } \lcltConcentrationProb = \frac{4q^* \Delta k(\Delta k + 1)}{\tp{\frac{q^*\lambda}{q-q^*+q^* \lambda}\cdot (1-1/(\Delta k))}^2},\]
\end{lemma}
Next, for those $\Phi_j$'s with a variable $v$ satisfying $1^\coF \in h^{-1}_v(J_v)$, we upper bound $\abs{\E{\mathrm{e}^{it X_j} \mid J}}^2$ by giving a lower bound of $\Pr{X_j - X_j' = 1 \mid J}$.
% Next, we show the upper bound of last part of \cref{eq:lclt-high-fourier-phases-condition-on-J}. Instead of bounding each $\abs{\E{\mathrm{e}^{it X_j} \mid J}}^2$ in \cref{eq:lclt-high-fourier-phases-each-X_j} directly by bounding $\Pr{X_j - X_j' = 1 \mid J}$. We actually, show that there is a large portion of $j$'s so that we can upper bound their $\abs{\E{\mathrm{e}^{it X_j} \mid J}}^2$'s. 
We defer its proof to \Cref{subsection:lclt-X_j-bound}.
\begin{lemma}
\label{lemma:lclt-X_j-bound}
Suppose \Cref{condition:lclt-LLL-condition}. For any good $J$, and any $\Phi_j$ with a variable $v$ satisfying $1^\coF \in h^{-1}_v(J_v)$, it holds that $\Pr{X_j - X_j' = 1 \mid J} \ge \frac{\min\set{\lambda, 1/\lambda}\cdot \mathrm{e}^{-1}}{\tp{2\connectedComponentVarSize}^2}$.
\end{lemma}

Combining the two lemmas above, we give an upper bound of the high Fourier phases.
\begin{lemma}
\label{lemma:lclt-high-fourier-phases}
Suppose \Cref{condition:lclt-clt-condition} and \Cref{condition:lclt-LLL-condition}. For any $t\in [-\pi \bar{\sigma}, \pi \bar{\sigma}]$, there exists a constant $\lcltHighPhaseConstant$, such that $\abs{\E{\mathrm{e}^{it \+Y^*}}} \le \frac{4\lcltConcentrationProb}{N^*}+ \frac{1}{|V|} + \exp\tp{- \frac{\lcltHighPhaseConstant\cdot N^*}{\log(|V|)^3} \cdot (t/\bar{\sigma})^2}$.
\end{lemma}
\begin{proof}
Recall that $\ell^* = \left \lceil \frac{\lcltExpectation N^*}{2 \connectedComponentVarSize}\right \rceil$.
Combined with \cref{eq:lclt-high-fourier-phases-condition-on-J} and \cref{eq:lclt-high-fourier-phases-each-X_j}, it holds that 
\begin{align*}
\abs{\E{\mathrm{e}^{it \+Y^*}}} \le&~ \frac{4\lcltConcentrationProb}{N^*}+ \frac{1}{|V|}+ \tp{1 - \frac{1}{4}\cdot \frac{\min\set{\lambda, 1/\lambda}\cdot \mathrm{e}^{-1}}{\tp{2\connectedComponentVarSize}^2}\cdot (t/\bar{\sigma})^2}^{\ell^*} \\
\le&~ \frac{4\lcltConcentrationProb}{N^*}+ \frac{1}{|V|} + \exp\tp{- \frac{\ell^*}{4}\cdot \frac{\min\set{\lambda, 1/\lambda}\cdot \mathrm{e}^{-1}}{\tp{2\connectedComponentVarSize}^2}\cdot (t/\bar{\sigma})^2}.
\end{align*}
\end{proof}
% \begin{definition}[Two-step projection]
% \label{definition:lclt-two-step-projection}

% \end{definition}

% \[
% \abs{\E{}}
% \]
% Let $\+Y_1^*$ and $\+Y_2^*$ be two i.i.d.\@ random variables following the law of $\+Y^*$.
Now, we are ready to derive the local central limit theorem (\Cref{theorem:lclt}).
\begin{proof}[Proof of \Cref{theorem:lclt}]
Applying \Cref{lemma:establish-lclt-by-characteristic-function} to $\+Y^* = (\+Y - \bar{\mu})/\bar{\sigma} \in \alpha + \beta \=Z$, where $\alpha = -\bar{\mu}$ and $\beta = 1/\bar{\sigma}$, we have
\[
\sup_{t\in \+L} \abs{\beta N(t) - \Pr{\+Y^* = t}} \le \frac{1}{\bar{\sigma}}\int_{-\pi \bar{\sigma}}^{\pi \bar{\sigma}} \abs{\E{\mathrm{e}^{it \+Y^*}} - \E{\mathrm{e}^{it \+Z}}}\d t + \mathrm{e}^{-\pi^2 \bar{\sigma}^2/2}.
\] 
Combined with \Cref{lemma:lclt-low-fourier-phases-bound} and \Cref{lemma:lclt-high-fourier-phases}, for any $\eta > 0$, we have that 
\begin{align*}
&~\frac{1}{\bar{\sigma}} \int_{-\pi \bar{\sigma}}^{\pi \bar{\sigma}} \abs{\E{\mathrm{e}^{it \+Y^*}} - \E{\mathrm{e}^{it \+Z}}}\d t \\
\le&~ \frac{1}{\bar{\sigma}} \int_{-\pi \bar{\sigma}}^{\pi \bar{\sigma}} \min\set{(2|t|\eta + 1)\cdot \lcltCltTail  + 4\mathrm{e}^{-\eta^2/2}, \frac{4\lcltConcentrationProb}{N^*}+ \frac{1}{|V|} + \exp\tp{- \frac{\lcltHighPhaseConstant\cdot N^*}{\log(|V|)^3} \cdot (t/\bar{\sigma})^2}} \d t. \\
\end{align*}
Next, we set $\eta = \sqrt{2 \log |V|}$ and truncate the integral at $t^* = \Theta\tp{\bar{\sigma}\cdot \frac{ \log^2 |V| }{\sqrt{|V|}}}$. 
Hence, we have that 
\begin{align*}
&~\frac{1}{\bar{\sigma}} \int_{-\pi \bar{\sigma}}^{\pi \bar{\sigma}} \abs{\E{\mathrm{e}^{it \+Y^*}} - \E{\mathrm{e}^{it \+Z}}}\d t \\
\le&~ \frac{1}{\bar{\sigma}} \int_{-t^*}^{t^*} 4|t|\cdot \sqrt{2 \log |V|}\cdot \xi \d t + \frac{1}{\bar{\sigma}} \int_{t^* \le |t| \le \pi \bar{\sigma}} \exp\tp{- \frac{\lcltHighPhaseConstant\cdot N^*}{\log(|V|)^3} \cdot (t/\bar{\sigma})^2} \d t + \Theta\tp{\frac{1}{|V|}}\\
\le&~ O\tp{\xi \cdot \frac{\log^{5/2} |V|}{\sqrt{|V|}}}. \qedhere
\end{align*}
\end{proof}
\subsection{Local central limit theorem for hypergraph coloring}
\label{subsection:lclt-application}
In this subsection, we prove the local central limit theorem for hypergraph coloring with one special color.
% We first prove \Cref{theorem:coloring-one-special-clt-intro}.
Given a $k$-uniform hypergraph $H = (V, \+E)$ with the maximum degree $\Delta$. We consider hypergraph $q$-colorings on $H$. Let $\coloringOneRandomVariable$ be the random variable that counts the number of vertices whose color is $1$ in a uniformly random proper coloring.

\LCLTColorOneIntro*
% \begin{theorem}
% \label{theorem:lclt-coloring-one-special-}
% Fix any integers $k\ge 50$ and $\coloringCondition$. Let $\bar{\mu}$ be the expectation of $\coloringOneRandomVariable$, and $\bar{\sigma}$ be the standard deviation of $\coloringOneRandomVariable$. It holds that 
% \[
% \sup_{t\in \=Z} \abs{\Pr{\coloringOneRandomVariable = t} - \bar{\sigma}^{-1}\+N((t - \bar{\mu})/\bar{\sigma}) \Pr{\+Z \le t}} \le \Theta_{q, \Delta, k}\tp{\frac{\ln^{7/2}|V|}{|V|}},
% \]
% where $\+N(x) = \mathrm{e}^{-x^2/2}/\sqrt{2\pi}$ denote the density of the standard normal distribution.
% \end{theorem}
\begin{proof}[Proof of \Cref{theorem:coloring-one-special-lclt-intro}]
We use \Cref{theorem:lclt} to prove this theorem.
Note that by running through the same proof in \Cref{section:clt}, it holds that $\bar{\sigma} = \Theta\tp{\sqrt{|V|}}$ and $N^* = |V|$.

% We use \Cref{lemma:zero-free-to-clt-co} to prove this lemma.
% We first construct the probability generating function $g$.
Let $\Phi = (V, [q]^V, C)$ be the corresponding atomic $(k, q\Delta)$-CSP formula of the hypergraph $q$-coloring on $H$. 
It suffices to verify \Cref{condition:lclt-clt-condition} and \Cref{condition:lclt-LLL-condition}. Note that by \Cref{theorem:coloring-one-special-clt-intro}, \Cref{condition:lclt-clt-condition} holds that $\xi = O\tp{\frac{\ln |V|}{\sqrt{|V|}}}$.
% \Cref{condition:lclt-clt-condition} is satisfied by \Cref{theorem:coloring-one-special-clt-intro} with $\xi = \Theta\tp{\frac{\ln |V|}{\sqrt{|V|}}}$.

In order to verify \Cref{condition:lclt-LLL-condition}. We first define the projection $\*f$ and $\*h$.
Let $B_1, B_2$ be two integers.
For any variable $v\in V$, we set $f_v(1) \defeq 1^\coF$ and $f_v(j) \defeq ((j-2) \mod B_1 + 2)^\coF$ for any $j \neq 1$.
We use the superscript $\coF$ to denote the projected symbol under $\*f$.
Note that under $\*f$, $q$ colors are distributed as $B_1 + 1$ projected symbols. 
And except for the $1^\coF$, each projected symbol corresponds to at least $\lfloor \frac{q-1}{B_1} \rfloor$ and at most $\lceil \frac{q-1}{B_1} \rceil$ original colors.

For any variable $v \in V$, we set $h_v(1^\coF) = 1^\coS$ and $h_v(j^\coF) \defeq ((j - 2) \mod B_2 + 1)^\coS$ for any $j\in \set{2, 3, \dots, B_1 + 1}$. We use the superscript $\coS$ to denote the two-step projected symbol.
Note that each two-step projected symbol corresponds to at least $\lfloor \frac{B_1}{B_2}\rfloor$ and at most $\lceil \frac{B_1}{B_2} \rceil + 1$ projected symbols, or at least $\tp{\lfloor \frac{B_1}{B_2} \rfloor - 1}\cdot \lfloor \frac{q-1}{B_1} \rfloor$ and at most $\tp{\lceil \frac{B_1}{B_2} \rceil + 1}\cdot \lceil \frac{q-1}{B_1} \rceil$ original colors.

Because $\mu$ is the uniform distribution over possible colorings, we have $\lambda = 1$.
And under the setting of $\*f$, we have $q^* = 1$, $N^* = |V|$.
Combined, we have that 
\[
\lcltProjectedVioProb \le \tp{\frac{\tp{\lceil \frac{B_1}{B_2} \rceil + 1}\cdot \lceil \frac{q-1}{B_1} \rceil}{q}}^k, \quad \lcltConditionVioProb \le \tp{\frac{1}{ \tp{\lfloor \frac{B_1}{B_2} \rfloor - 1}\cdot \lfloor \frac{q-1}{B_1} \rfloor }}^k, \quad \text{and }r_{\max} = \frac{1}{q}.
\]
Set $B_1  = \lfloor q^{3/4} \rfloor$ and $B_2 = \lfloor q^{1/2} \rfloor$. We now verify \Cref{condition:lclt-LLL-condition} respectively.

For \Cref{condition:lclt-LLL-condition}-(\ref{item:lclt-LLL-condition-Proj}), combined with the fact that $q\ge 700$, it holds that 
\[
\lcltProjectedVioProb \le \tp{\frac{(\lceil 4 q^{1/4} \rceil + 1)\cdot \lceil \frac{q - 1}{\lfloor q^{3/4} \rfloor}\rceil } {q}}^k \le \tp{\frac{20 q^{1/2}}{q}}^k.
\]
It suffices if $2\mathrm{e}^2 \cdot (20/q^{1/2})^k\cdot (\Delta k)^2 \le 1$. Rearranging this inequality gives that $q \ge 400\tp{2\mathrm{e}^2 k^2}^{2/k}\cdot \Delta^{4/k}$.
By the fact that $k\ge 50$, it suffices to verify that $q\ge 640 \Delta^{4/k}$ which is satisfied by the assumption.

For \Cref{condition:lclt-LLL-condition}-(\ref{item:lclt-LLL-condition-Cond}), combined with the fact that $q \ge 700$, it holds that 
\[
\lcltConditionVioProb \le \tp{\frac{1}{(\frac{1}{2}q^{1/4}-1)\cdot \frac{1}{4}q^{1/4}}}^k \le \tp{\frac{16}{q^{1/2}}}^k.
\]
It suffices to satisfy that $2\mathrm{e}\cdot (16/q^{1/2})^k \cdot q\Delta k \le 1$. Rearranging this inequality gives that $q \ge \tp{2\mathrm{e}\cdot 16^k \cdot k}^{\frac{2}{k-2}}\cdot \Delta^{\frac{2}{k-2}}$. Combined with the fact that $k\ge 50$, it suffices to verify that $q \ge 410\Delta^{\frac{2}{k-2}}$ which is satisfied by the assumption.

For \Cref{condition:lclt-LLL-condition}-(\ref{item:lclt-LLL-condition-concentration}), it suffices to verify that $(8\mathrm{e})^3 q^{-k+2} (\Delta k + 1)^{2 + \zeta} \le 1$ where $\zeta = \frac{2\ln(2-1/q)}{\ln(q) - \ln(2 - 1/q)}$. By the fac that $q \ge 700$, it holds that $\zeta \le 0.23628$. Hence, it suffices to verify that $(8\mathrm{e})^3 q^{-k + 2} (\Delta k + 1)^{2.24} \le 1$. So it suffices to verify that $q \ge \tp{(8\mathrm{e})^3 \cdot 2^{2.24}\cdot k^{2.24}}^{\frac{1}{k-2}}\Delta^{\frac{2.24}{\Delta - 2}}$ which is satisfied by the assumption.

% By \Cref{theorem:lclt}, it holds that 
% \begin{align*}
% &~\sup_{t\in \=Z} \abs{\Pr{\+Y = t} - \bar{\sigma}^{-1} N((t-\bar{\mu})/\bar{\sigma})} \\
% \le&~ \frac{1}{\bar{\sigma}} \int_{-\pi \bar{\sigma}}^{\pi \bar{\sigma}} \min\set{(2|t|\eta + 1) \Theta\tp{\frac{\ln |V|}{|V|}}  + 4\mathrm{e}^{-\eta^2/2}, \frac{4\lcltConcentrationProb + 1}{|V|} + \exp\tp{- \frac{\lcltHighPhaseConstant |V|}{\log(|V|)^3} \tp{\frac{t}{\bar{\sigma}}}^2}} \d t + \mathrm{e}^{-\frac{\pi^2 \bar{\sigma}^2}{2}},
% \end{align*}
% where $\bar{\sigma} = \Theta\tp{\sqrt{|V|}}$.
% We set $\eta = \sqrt{2\ln |V|}$. We have that 
% \begin{align*}
% &~\frac{1}{\bar{\sigma}} \int_{-\pi \bar{\sigma}}^{\pi \bar{\sigma}} \min\set{(2|t|\eta + 1) \Theta\tp{\frac{\ln |V|}{|V|}}  + 4\mathrm{e}^{-\eta^2/2}, \frac{4\lcltConcentrationProb + 1}{|V|} + \exp\tp{- \frac{\lcltHighPhaseConstant |V|}{\log(|V|)^3} \tp{\frac{t}{\bar{\sigma}}}^2}} \d t\\
% =&~ \frac{1}{\bar{\sigma}} \int_{-\pi \bar{\sigma}}^{\pi \bar{\sigma}} \min\set{ |t|\cdot \Theta\tp{\frac{\ln^{3/2} |V|}{|V|}}, \exp\tp{- \frac{\lcltHighPhaseConstant |V|}{\log(|V|)^3} \tp{\frac{t}{\bar{\sigma}}}^2}} \d t + \Theta\tp{\frac{\ln |V|}{|V|}}. \\
% =&~ \frac{1}{\bar{\sigma}} \int_{-C\cdot \log^2 |V|}^{C\cdot \log^2 |V|} |t|\cdot \Theta\tp{\frac{\ln^{3/2} |V|}{|V|}} \d t + \Theta\tp{\frac{\ln |V|}{|V|}} \\
% =&~ \Theta\tp{\frac{\ln^{7/2}|V|}{|V|}}.\qedhere
% \end{align*}
\end{proof}

\subsection{Proof of \Cref{lemma:lclt-J-is-good}}
\label{subsection:lclt-J-is-good}
In this subsection, we prove \Cref{lemma:lclt-J-is-good}.
Recall that after simplifying the CSP formula $\Phi = (V, [q]^V, C)$ using $J$, it shatters into $\ell$ sub-CSP formulas. 
% Recall that the dependency graph $G_{\Phi'}$ of the simplified CSP formula $\Phi'$ shatters into $\ell$ maximal connected components.
For any $v\in V$, let $\Phi_v = (V_v, [q]^{V_v}, C_v)$ denote the sub-CSP formula containing the variable $v$.
We first give a tail bound for $\Pr{|C_v|\ge i}$.
% With the above local uniformity, we show that with high probability it holds that $|C_v| = O(\log |C|)$.
\begin{lemma}
\label{lemma:lclt-bounded-connected-component-size}
Suppose \Cref{condition:lclt-LLL-condition}-(\ref{item:lclt-LLL-condition-Proj}) and \Cref{condition:lclt-LLL-condition}-(\ref{item:lclt-LLL-condition-concentration}).
For any variable $v\in V$ and any $i\ge 1$, it holds that 
\[\Pr{|C_v|\ge i} \le 2\Delta \tp{\frac{1}{2}}^{\lfloor \frac{i}{\Delta k + 1} \rfloor}.\]
\end{lemma}
\begin{proof}
Given $\tau \sim \mu$ and $J$ constructed by the two-step projection $J = (\*h \circ \*f)(\tau)$. We can construct $\Phi_v$ by the following BFS procedure:
\begin{enumerate}
    \item let $V_v = \set{v}$ and $C_v = \emptyset$;
    \item let $c \in C\setminus C_v$ be a constraint satisfying that (1) $c$ is not satisfied by $J$, (2) $c$ crosses $V_v$, i.e., $\vbl(c)\cap V_v \neq \emptyset$ and $\vbl(c)\cap V\setminus V_v \neq \emptyset$. If such constraint $c$ does not exist, terminate this procedure;
    \item Set $V_v \gets V_v \cup \vbl(c)$ and $C_v \gets C_v \cup \set{c}$. Go to step (2).
\end{enumerate}
Next, for any $i\ge 1$, we consider the upper bound of $\Pr{|C_v|\ge i}$ using the $2$-tree argument.
Consider the dependency graph $G = (V, E)$ of $\Phi_v$.
Recall that the vertex set $V$ of $G$ is the set of constraints, and we say two constraints $c_1, c_2$ are adjacent if and only if $\vbl(c_1) \cap \vbl(c_2) \neq \emptyset$ and $c_1 \neq c_2$. 
Let $r$ be an arbitrary constraint in $C_v$ containing the variable $v$. Note that if $|C|\neq 0$, this $r$ exists.
Let $\+T \subseteq C_v$ be a $2$-tree on the dependency graph $G$ constructed by $C_v$ containing the $r$ using \Cref{definition:2-tree}.
It holds $|\+T| \ge \frac{|C_v|}{\Delta k + 1}$.
By the BFS procedure, it holds that for any constraint $c$ in the $2$-tree $\+T$, $c$ is not satisfied by $J$. 
Recall \Cref{condition:lclt-LLL-condition}-(\ref{item:lclt-LLL-condition-Proj}) and \Cref{condition:lclt-LLL-condition}-(\ref{item:lclt-LLL-condition-concentration}).
Combined with the fact that variables in $\+T$ are disjoint and \Cref{theorem:HSS}, letting $x(c) = \mathrm{e} r_{\max}^k$ for any $c\in C$, these constraints together contribute at most $\tp{\lcltProjectedVioProb\cdot \tp{1 - x(c)}^{-\Delta k}}^{|\+T|} \le \tp{\mathrm{e} \lcltProjectedVioProb}^{|\+T|}$.
Hence, we have $\Pr{|C_v| \ge i} \le \Delta \cdot \sum_{\+T: |\+T|\ge \lfloor\frac{i}{\Delta k + 1}\rfloor} \tp{\mathrm{e} \lcltProjectedVioProb}^{|\+T|}$ where the first factor $\Delta$ enumerate possible $r$.
% \[
% \Pr{|C_v| \ge i} \le \Delta \cdot \sum_{\+T: |\+T|\ge \frac{i}{\Delta k + 1}} ?,
% \]
% where the first factor $\Delta$ enumerate possible $r$.
Combined with \Cref{condition:lclt-LLL-condition}-(\ref{item:lclt-LLL-condition-Proj}) and \Cref{lemma:2-tree-number-bound}, we have that 
\[
\Pr{|C_v| \ge i} \le \Delta \cdot \sum_{j = \lfloor\frac{i}{\Delta k + 1}\rfloor} \tp{\mathrm{e} (\Delta k)^2 \cdot \mathrm{e}\lcltProjectedVioProb}^j \le 2\Delta \tp{\frac{1}{2}}^{\lfloor \frac{i}{\Delta k + 1} \rfloor}.
\]
\end{proof}

Next, we establish a concentration inequality on the number of variables satisfying that $1^\coF \in h_v^{-1}(J_v)$.
Let $\tau \sim \mu$ follow the Gibbs distribution $\mu$ with the external field $\lambda$ on the projected symbol $1^\coF$.
Recall that for any variable $v\in V$, $Y_v$ is the random variable indicating $\tau_v \in f_v^{-1}(1^\coF)$ or not and $\+Y = \sum_{v\in V} Y_v$.
We define new random variables with respect to the two-step projection.
Let $W_v$ be the random variable indicating $1^\coF \in h_v^{-1}(J_v)$ or not. Let $\+W = \sum_{v\in V} W_v$.

We establish a concentration inequality by the Chebyshev's inequality.
For the lower bound of expectation $\E{\+W}$, we obtain the next lemma by running through the same proof of \Cref{lemma:clt-marginal-bound} with different conditions.
\begin{lemma}
\label{lemma:lclt-concentration-expectation-lower-bound}
Suppose \Cref{condition:lclt-LLL-condition}-(\ref{item:lclt-LLL-condition-concentration}), it holds that
\[ \E{\+W} \ge \lcltExpectation N^*, \hbox{where } \lcltExpectation = \frac{q^* \lambda}{q - q^* + q^* \lambda} - \frac{q^* \lambda}{(q - q^* + q^* \lambda)\Delta k}.\]
% $$, where $$.
\end{lemma}
% \yxtodo{Check condition.}
For the upper bound of the variance $\Var{\+W}$, we obtain the following lemma by running through a similar proof in \Cref{section:clt} with different conditions.
\begin{lemma}
\label{lemma:lclt-concentration-variance-upper-bound}
Suppose \Cref{condition:lclt-LLL-condition}-(\ref{item:lclt-LLL-condition-concentration}), it holds that 
\[
\Var{\+W} \le 4 q^* \Delta k (\Delta k + 1) N^*.
\]
\end{lemma}
% \yxtodo{Check condition.}
% By running an almost same proof in \Cref{section:concentration_inequality}, we have the following concentration inequality.
Combined, we have that
\begin{lemma}
\label{lemma:lclt-concentration-two-step-projection-W}
Suppose \Cref{condition:lclt-LLL-condition}-(\ref{item:lclt-LLL-condition-concentration}), we have
\[
\forall \delta > 0,\quad \Pr{|\+W - \E{\+W}|\ge \delta\cdot \E{\+W}}\le \frac{\lcltConcentrationProb}{\delta^2 \cdot N^*},\hbox{ where } \lcltConcentrationProb = \frac{4q^* \Delta k(\Delta k + 1)}{\tp{\frac{q^*\lambda}{q-q^*+q^* \lambda}\cdot (1-1/(\Delta k))}^2}.
\]
% where $\lcltConcentrationProb = \frac{4q^* \Delta k(\Delta k + 1)}{\tp{\frac{q^*\lambda}{q-q^*+q^* \lambda}\cdot (1-1/(q^* \Delta k))}^2}$.
\end{lemma}

Finally, we prove \Cref{lemma:lclt-J-is-good}.
\begin{proof}[Proof of \Cref{lemma:lclt-J-is-good}]

In \Cref{lemma:lclt-bounded-connected-component-size}, let $i = \connectedComponentSize$. It holds that for any $v\in V$, we have $\Pr{|C_v|\ge i} \le \frac{1}{|V|^2}$. By a union bound, it holds that with probability at least $1 - \frac{1}{|V|}$, $\forall v\in V, |C_v| \le \connectedComponentSize$.

In \Cref{lemma:lclt-concentration-two-step-projection-W}, set $\delta = \frac{1}{2}$, it holds that with probability at least $1 - \frac{4 \lcltConcentrationProb}{N^*}$, we have $\+W \ge \frac{\lcltExpectation\cdot N^*}{2}$. 
Combined with the fact that $\forall v\in V, |C_v| \le \connectedComponentSize$,
there are least $\ell^* = \left \lceil \frac{\lcltExpectation N^*}{2 \connectedComponentVarSize}\right \rceil$ sub-CSP formulas containing a variable $v$ with $1^\coF \in h^{-1}_v(J_v)$.

By a union bound of these two events, with probability at least $1 - \frac{4\lcltConcentrationProb}{N^*} + \frac{1}{|V|}$, we have $J$ is good.
\end{proof}
\subsection{Proof of \Cref{lemma:lclt-X_j-bound}}
\label{subsection:lclt-X_j-bound}
In this subsection, we prove \Cref{lemma:lclt-X_j-bound}.
% The idea is that we first show that with high probability, there is a large portion of CSP formulas $\Phi_j$'s satisfying that there exists at least one variable $v$ in $\Phi_j$ such that $1^\coF \in h_v^{-1}(J_v)$ which means $X_j$ is not a constant $0$.
% We show this property by establishing a concentration inequality on the number of variables satisfying that $1^\coF \in h_v^{-1}(J_v)$. This part follows the similar proof in \Cref{section:concentration_inequality}.
% Next, for these CSP formulas, let $v$ be a variable satisfying  $1^\coF \in h_v^{-1}(J_v)$.
At a high level, the idea is to introduce new constraints so that $v$ becomes ``independent'' in the CSP formula.
In particular, for any satisfying assignment of this new CSP formula, we can arbitrarily change $v$'s value.
Let $\mu'$ be the new Gibbs distribution of this new CSP formula.
Combined with the fact that $J$ is good, then there exists a number $x$ such that $\Pr[\mu']{X_j = x\mid J} \ge \frac{1}{\connectedComponentVarSize}$.
Then we show that $\Pr[\mu]{X_j = x \mid J}$ is near $\Pr[\mu']{X_j = x \mid J}$.
Note that we can change $v$'s value, we get a lower bound of $\Pr{X_j - X_j' = 1 \mid J}$.
\begin{proof}[Proof of \Cref{lemma:lclt-X_j-bound}]
Let $v$ be the variable in $V_j$ satisfying that $1^\coF \in h_v^{-1}(J_v)$.
For any constraint $c\in C_j$ with $v\in \vbl(c)$, we add $q$ copies of $c$ into $C'$, each forbids a possible value of $v$.
Let $\Phi' = (V', [q]^{V'}, C')$ be the new CSP formula with $V' = V_j$. 
Note that $\abs{C' \setminus C} \le q\Delta$ and for any constraint $c\in C'$, there are at most $q\Delta$ constraints that intersect $c$, i.e., $\abs{\set{c'\in C'\setminus \set{c} \mid \vbl(c')\cap \vbl(c) \neq \emptyset}} \le q\Delta$.

Let $\mu$ be the Gibbs distribution of $\Phi$ and $\mu'$ be the Gibbs distribution of $\Phi'$. Next, we upper bound $\Pr[\mu]{\Phi' \text{ is satisfied}\mid J}$ by the telescoping method. Let $C'\setminus C = \set{c_1, c_2, \dots, c_m}$. For any $0\le i \le m$, let $\Phi'_i = (V_j, [q]^{V_j}, C_i)$ where $C_i = C \cup \set{c_1, c_2, \dots, c_i}$. Let $\mu_i$ be the Gibbs distribution of $\Phi'_i$. Hence, we have $\Pr[\mu]{\Phi' \text{ is satisfied}\mid J} = \prod_{i = 1}^{m} \Pr[\mu_{i-1}]{\Phi'_{i} \text{ is satisfied} \mid J} = \prod_{i = 1}^{m} \tp{1 - \Pr[\mu_{i-1}]{\Phi'_{i} \text{ is not satisfied}\mid J}}$. Next, we upper bound the probability that $\Phi'_{i} \text{ is not satisfied}$ by \Cref{theorem:HSS}. 
Recall \Cref{condition:lclt-LLL-condition}-(\ref{item:lclt-LLL-condition-Cond}).
For any $c\in C'$, set $x(c) = \mathrm{e} \lcltConditionVioProb$. 
% \yxtodo{check this.} 
Hence, we have 
\[\Pr[\mu_{i-1}]{\Phi'_{i} \text{ is not satisfied}\mid J} \le \lcltConditionVioProb \tp{1 - \mathrm{e} \lcltConditionVioProb}^{-q\Delta k} \le \mathrm{e} \lcltConditionVioProb.\]
So, we have $\Pr[\mu]{\Phi' \text{ is satisfied}\mid J} \ge \tp{1 - \mathrm{e} \lcltConditionVioProb}^{q\Delta k} \ge \mathrm{e}^{-1}$.
% \yxtodo{modify the intuition in the beginning of this section}.
Recall that $X_j$ is the number of variables in $\Phi_j$ whose projected symbol under $\*f$ is $1^\coF$. Combined with the fact that $J$ is good, we have that $0 \le X_j \le \connectedComponentVarSize$.
Hence, we have that there exists $x \ge 0$, such that $\Pr[\mu']{X_j = x\mid J} \ge \frac{1}{\connectedComponentVarSize}$.
We also have that $\max\set{\Pr[\mu']{X_j = x \land Y_v = 1 \mid J}, \Pr[\mu']{X_j = x \land Y_v = 0 \mid J}} \ge \frac{1}{2\connectedComponentVarSize}$.
W.l.o.g.\@, we assume that $\Pr[\mu']{X_j = x \land Y_v = 1 \mid J} \ge \frac{1}{2\connectedComponentVarSize}$. Note that $v$ is an independent variable, by changing $v$'s value, we have $\Pr[\mu']{X_j = x - 1 \land Y_v = 0 \mid J} \ge \frac{\min\set{\lambda, 1/\lambda}}{2\connectedComponentVarSize}$. Hence, we have that $\Pr[\mu']{X_j - X_j'\mid J} \ge \frac{\min\set{\lambda, 1/\lambda}}{\tp{2\connectedComponentVarSize}^2}$. Combined with $\Pr[\mu]{\Phi' \text{ is satisfied}\mid J} \ge \mathrm{e}^{-1}$, we have that $\Pr[\mu]{X_j - X_j'\mid J} \ge \frac{\min\set{\lambda, 1/\lambda}\cdot \mathrm{e}^{-1}}{\tp{2\connectedComponentVarSize}^2}$.
\end{proof}
% \begin{lemma}
% \label{lemma:lclt-high-fourier-phases-bound}
% Suppose \Cref{condition:lclt-clt-condition}. For any $t\in [-\pi \sigma, \pi \sigma]$, we have that 
% \[
% \abs{\E{\mathrm{e}^{it \+Y^*}}} \le ?.
% \]
% \end{lemma}
% \begin{proof}

% \end{proof}

% \input{algorithm_new.tex}
% \section{Max CSP}
% \input{appendix_max_csp.tex}
% 

\section{Fisher zeros}
\label{section:fisher-zeros}
In this section, we derive the Fisher zeros for CSP formulas.
%  by providing a sufficient condition. 
Let $\Phi = (V, \*Q, C)$ be a $(k, \Delta)$-CSP formula and let $\beta \in \=C$. 
We define the partition function with respect to $\beta$ as 
\begin{equation}
\label{eq:fisher-partition-function}
Z^{\-{fs}} = Z^{\-{fs}}(\Phi, \beta) \defeq \sum_{\sigma\in \*Q} \beta^{\abs{ \set{c\in C\mid c \text{ is violated by } \sigma }}}.
\end{equation}
% \begin{theorem}
% \end{theorem}
We remark that $\beta$ can be seen as a ``penalty'' when a constraint is violated.
Two interesting points are $\beta = 0$, which refers to the uniform distribution over all satisfying assignment, and $\beta = 1$ which refers to the simple product distribution.

To obtain zero-freeness in $\beta$, we follow a similar reduction in \cite[Section 6]{liu2024phase} that reduces the Fisher zeros into Lee-Yang zeros (of another CSP) by introducing new variables. Then we apply the framework in \Cref{section:CSP}, specifically, we verify \Cref{condition:csp-conditional-measure-analysis} and apply \Cref{theorem:sufficient-condition-zero-freeness}.

We first provide the reduction.
We remark that when $\beta = 1$, the partition function $Z^{\-{fs}}(\Phi, 1) = |\*Q|$ which is trivially not zero. Hence, we focus on the case $\beta \neq 1$.
\begin{definition}[Reduction from the Fisher zeros to Lee-Yang zeros]
\label{definition:fisher-reduction}
Given a CSP formula $\Phi = (V, \*Q, C)$ with a complex number $\beta\in \=C\setminus\set{1}$, we construct a new CSP formula $\Phi' = (V', \*Q', C')$ and a new complex number $\lambda = \frac{\beta}{1 - \beta}$, denoted by $(\Phi', \lambda) = \-{red}(\Phi, \beta)$ as follows:
\begin{itemize}
% \item $V' = V \cup \set{v_c \mid c\in C}$;
\item for any constraint $c\in C$, we construct a new constraint $c'$ and add a new variable $v_{c'}$ whose domain is $\set{0, 1}$. 
Let $\tau = c^{-1}(\text{False})$ be the violating partial assignment of the constraint $c$. Then we set $\tau_{v_{c'}} = 0$ and $c{'^{-1}}(\text{False}) = \tau$. Let $C' = \set{c' \mid c \in C}$.
\item Let $V' = V\cup \set{v_{c'} \mid c'\in C'}$.
\end{itemize}
\end{definition}
We call variables in $V$ as \emph{original variables} and variables in $V'\setminus V$ as \emph{new variables}.
Recall that $\Omega_{\Phi'}$ is the set of all satisfying assignments of the CSP formula $\Phi'$.
Let $Z^{\-{ly}} = Z^{\-{ly}}(\Phi', \lambda)$ be the partition function with respect to Lee-Yang zeros defined as 
\begin{equation}
\label{eq:lee-yang-partition-function}
Z^{\-{ly}}(\Phi', \lambda)\defeq \sum_{\sigma \in \Omega_{\Phi'}} \lambda^{\abs{\set{ v_{c'} = 1 \mid c'\in C' }}}.
\end{equation}
\begin{lemma}[Properties of the reduction]
\label{lemma:fisher-reduction-properties}
Let $\Phi=(V, \*Q, C)$ be a $(k, \Delta)$-CSP formula and $\beta \in \=C\setminus \set{1}$ be a complex number. Let $(\Phi', \lambda) = \-{red}(\Phi, \beta)$. Then,
\begin{enumerate}
\item $\Phi'$ is a $(k+1, \Delta)$-CSP formula. \label{item:reduction-property-1}
\item $Z^{\-{fs}}(\Phi, \beta) = (1 - \beta)^{|C|}\cdot Z^{\-{ly}}(\Phi', \lambda)$.\label{item:reduction-property-2}
\end{enumerate}
\end{lemma}
\begin{proof}
For \cref{item:reduction-property-1}, note that in \Cref{definition:fisher-reduction}, we introduce a new variable for each constraint and this new variable only belongs to this constraint. Hence, \cref{item:reduction-property-1} follows.

We then prove \cref{item:reduction-property-2}. For any logical expression $P$, we define the Iverson bracket $[P] = 1$ if $P$ is true, otherwise $[P] = 0$. 

Recall \cref{eq:lee-yang-partition-function}.
Also recall that $V$ is the original variable set of $\Phi$.
For the reduced CSP formula $\Phi'$. we enumerate partial assignments on the original variables $V$. If a constraint $c'$ is already satisfied by the partial assignment, then the new variable $v_{c'}$ can choose $0$ or $1$, otherwise we have $v_{c'}=1$.
Recall $\*Q$ is the original alphabet of $\Phi$, $|C|=|C'|$.
By this reason, we reformulate $Z^{\-{ly}}$ as
\begin{equation*}
\begin{aligned}
Z^{\-{ly}} &= \sum_{\sigma\in \*Q} \prod_{c'\in C'} \tp{ (1+\lambda)[c' \text{ is satisfied by } \sigma] + \lambda [c' \text{ is not satisfied by } \sigma] }. \\
\end{aligned}
\end{equation*}
Recall that $\lambda = \frac{\beta}{1 - \beta}$. Thus, we have
\begin{equation*}
\begin{aligned}
Z^{\-{ly}}&= \tp{\frac{1}{1-\beta}}^{|C|}\sum_{\sigma\in \*Q}\prod_{c'\in C'} \tp{[c' \text{ is satisfied by } \sigma] + \beta\cdot [c' \text{ is not satisfied by } \sigma] }.
\end{aligned}
\end{equation*}
Combined with \cref{eq:fisher-partition-function}, \cref{item:reduction-property-2} holds and this lemma follows.
\end{proof}
By the above lemma, to obtain the zero-freeness, it suffices to show that $Z^{\-{ly}}(\Phi', \lambda)\neq 0$. We provide the next theorem by leveraging the framework in \Cref{section:CSP}.
\begin{theorem}[Zero-freeness for Fisher zeros]
\label{theorem:fisher-zero-freeness}
Let $\Phi = (V, \*Q, C)$ be a $(k, \Delta)$-CSP formula and $\beta \in \=C\setminus \set{1}$ be a complex number. Let $(\Phi', \lambda) = \-{red}(\Phi, \beta)$ be constructed by \Cref{definition:fisher-reduction} where $\Phi' = (V', \*Q', C')$. Write $C' = \set{c'_1, c'_2, \dots, c'_m}$. For any $i\le m$, we define $C'_i = \set{c'_1, c'_2, \dots, c'_i}$ and $\Phi'_i = (V', \*Q', C'_i)$.
Let $\*f$ be a state-compression scheme (\Cref{definition:state-compression}) and let $\Sigma_{V'}$ be the projected domain.

Suppose that the following holds: 
\begin{enumerate}
\item for any constraint $c'\in C'$, $f_{v_{c'}}(0) = 0^\coF$ and  $f_{v_{c'}}(1) = 1^\coF$;\label{item:fisher-zero-freeness-property-1}
\item for any original variable $v\in V$, the projected domain does not contain $1^\coF$, i.e., $\forall v\in V,  1^\coF \not\in \Sigma_{v}$; \label{item:fisher-zero-freeness-property-2}
\item $\Phi'$ satisfies \cref{eq:csp-Z_0-is-not-0} and $\Phi_1', \Phi_2', \dots, \Phi'_m$ satisfy \Cref{condition:csp-conditional-measure-analysis}, \label{item:fisher-zero-freeness-property-3}
\end{enumerate}
then we have $Z^{\-fs}(\Phi, \beta) \neq 0$.
\end{theorem}
\begin{proof}
By \Cref{lemma:fisher-reduction-properties}, it holds that $Z^{\-{fs}}(\Phi, \beta) = (1-\beta)^{|C|}\cdot Z^{\-{ly}}(\Phi', \lambda)$. It suffices to show $Z^{\-{ly}}(\Phi', \lambda)\neq 0$. Recall the complex extension of CSP formulas \Cref{definition:complex-extensions-csp}. Combined with \cref{item:fisher-zero-freeness-property-1} and \cref{item:fisher-zero-freeness-property-2}, we have $Z^{\-{ly}}(\Phi', \lambda) = Z(\Phi', \lambda, \*f, 1^\coF)$.

Finally, by \Cref{theorem:sufficient-condition-zero-freeness} and \cref{item:fisher-zero-freeness-property-3}, we have $Z(\Phi', \lambda, \*f, 1^\coF)\neq 0$. Hence, this theorem follows.
\end{proof}

By running similar calculations for $Z^{\-{fs}}(\Phi, \beta)$ as in \Cref{section:application-coloring-one-special-color}, one can show that under the condition $k\ge 50$ and $q\ge 20\Delta^{\frac{5}{k-2.5}}$, there is a zero-free strip around $[0,1]$ for the hypergraph $q$-coloring. 
As a corollary, there is an FPTAS for approximating the number of hypergraph $q$-colorings through Barvinok's interpolation method~\cite{barvinok2016combinatorics,patel2017deterministic,Liu2017TheIP}.
We remark that currently the best FPTAS~\cite{wang2024sampling}, which is based on the recursive coupling and Moitra's LP approach, works for $k\ge 8$ and $q\ge 70\Delta^{\frac{2 + \zeta}{k - 2 - \zeta}}$ with $q\ge \exp(O(1/\zeta))$.
We leave it as an open problem to close this gap using zero-freeness.

\end{document}